\documentclass[dsc,numbers]{coppe}
\usepackage[utf8]{inputenc}
\usepackage{amsmath,amssymb}
\usepackage{hyperref}
\usepackage{macros_markus}
\usepackage{graphicx}                                   
\usepackage{url}    
\usepackage[tight,footnotesize]{subfigure}  
\usepackage{amsthm}
\usepackage{multirow}
\usepackage{color, colortbl}  
\definecolor{LightCyan}{rgb}{0.88,1,1}
\definecolor{Gray}{rgb}{0.82,0.82,0.82}

\usepackage{geometry}
\theoremstyle{plain}
\newtheorem{thm}{Theorem}
\newtheorem{cor}{Corollary}

\makeindex

\makelosymbols
\makeloabbreviations

\begin{document}
  \title{On Data-Selective Learning}
  \foreigntitle{Aprendizado sob Seleção de dados}
  \author{Hamed}{Yazdanpanah}
  \advisor{Prof.}{Paulo Sergio}{Ramirez Diniz}{Ph.D.}
  \advisor{Prof.}{Markus Vinicius}{Santos Lima}{D.Sc.}

  \examiner{Prof.}{Paulo Sergio Ramirez Diniz}{Ph.D.}
  \examiner{Prof.}{Markus Vinicius Santos Lima}{D.Sc.}
  \examiner{Prof.}{Marcello Luiz Rodrigues de Campos}{Ph.D.}
  \examiner{Prof.}{José Antonio Apolinário Jr.}{D.Sc.}
  \examiner{Prof.}{Mário Sarcinelli Filho}{D.Sc.}
  \examiner{Prof.}{Cássio Guimarães Lopes}{Ph.D.}
  \department{PEE}
  \date{03}{2018}


  \maketitle

  \frontmatter
  \dedication{ }

\begin{flushright}
To my parents, Mohammad and Mina, and Ana Clara\\for their love, attention, and support.
\end{flushright}
  \chapter*{Acknowledgments}

I would like to express my sincere gratitude to my advisor, Professor Paulo S. R. Diniz, for the continuous support, guidance, patience, motivation, and immense knowledge. Specially, I would like to thank him for his generous support and patience during my illness that lasted for about one year. Also, his extreme competence and friendly comprehension inspire me to be a better professional and friend. In fact, he is a remarkable example of a Brazilian. I could not have imagined having a better advisor for my Ph.D. study.

Also, I would like to thank Professor Markus V. S. Lima, my other advisor. He helped me for all details of my thesis. In fact, I am grateful for having his guidance during my study. He was always keen to know what I was doing and how I was proceeding. He always inspired me to be a serious and diligent researcher. I thank him for being not only my advisor, but also a friend.

Beside my advisors, I would like to thank my thesis committee: Prof. Marcello L. R. de Campos, Prof. José A. Apolinário Jr., Prof. Mário S. Filho, and Prof. Cássio G. Lopes for their encouragement, insightful comments and suggestions. My thesis benefited from their valuable comments. Moreover, I would like to express my sincere gratitude to Prof. José A. Apolinário Jr. for his invaluable comments on Chapter 7 of the text.

My sincere thanks also goes to Prof. Sergio L. Netto and Prof. Eduardo A. B. da Silva for offering me a research project in their group. I have
learned a lot from them during the project.

I would like to thank the professors of the Programa de Engenharia Elétrica (PEE) who have contributed to my education. In particular, I am grateful to Prof. Wallace A. Martins for the courses he taught.

Also, I would like to thank the staff of the SMT Lab. I am particularly grateful to Michelle Nogueira for her support and assistance during my Ph.D. study. Moreover, I thank the university staff, in particular, Daniele C. O. da Silva and Mauricio de Carvalho Machado for their help.

My sincere thanks also goes to Camila Gussen and all friends of the SMT Lab. They make the SMT Lab a pleasant and collaborative workplace. Also, I would like to thank Prof. Tadeu Ferreira for his special attention and help.

A very special gratitude goes out to Coordenação de Aperfeiçoamento de Pessoal de Nível Superior (CAPES), Conselho Nacional de Desenvolvimento Científico e Tecnológico (CNPq), and Fundação de Amparo à Pesquisa do Estado do Rio de Janeiro (FAPERJ) for the financial support.

I am really grateful to my lovely girlfriend, Ana Clara, and her family for all their love, patience, and help. Her love motivates me to continue my studies in Brazil and to choose this beautiful country as my home. Her continuous encouragement, unfailing emotional support, and permanent attention played fundamental roles throughout my years of study.

I am deeply grateful to my parents for giving birth to me at the first place and supporting me spiritually throughout my life. I can never pay them back the sacrifice they made for me. My father, Mohammad Yazdanpanah, and my mother, Mina Alizadeh, have provided me through moral and emotional support during my education. Finally, I must express my very profound gratitude to my brother and my sister for providing me with support and continuous encouragement through the process of researching and writing this thesis. This accomplishment would not have been possible without my family. Thank you.
  \begin{abstract}

Filtros adaptativos são aplicados em diversos aparelhos eletrônicos e de comunicação, como {\it smartphones}, fone de ouvido avançados, DSP chips, antenas inteligentes e sistemas de teleconferência. Eles também têm aplicação em várias áreas como identificação de sistemas, equalização de canal, cancelamento de eco, cancelamento de interferência, previsão de sinal e mercado de ações. Desse modo, reduzir o consumo de energia de algoritmos adaptativos tem importância significativa, especialmente em tecnologias verdes  e aparelhos que usam bateria.

Nesta tese, filtros adaptativos com seleção de dados, em particular filtros adaptativos da família {\it set-membership} (SM), são apresentados para cumprir essa missão. No presente trabalho objetivamos apresentar novos algoritmos, baseados nos clássicos, a fim de aperfeiçoar seus desempenhos e, ao mesmo tempo, reduzir o número de operações aritméticas exigidas. Dessa forma, primeiro analisamos a robustez dos filtros adaptativos SM clássicos. Segundo, estendemos o SM aos números trinions e quaternions. Terceiro, foram utilizadas também duas famílias de algoritmos, SM filtering e {\it partial-updating}, de uma maneira elegante, visando reduzir energia ao máximo possível e obter um desempenho competitivo em termos de estabilidade. Quarto, a tese propõe novos filtros adaptativos baseado em algoritmos {\it least-mean-square} (LMS) e mínimos quadrados recursivos com complexidade computacional baixa para espaços esparsos. Finalmente, derivamos alguns algoritmos {\it feature} LMS para explorar a esparsidade escondida nos parâmetros.

\end{abstract}

  \begin{foreignabstract}

Adaptive filters are applied in several electronic and communication devices like smartphones, advanced headphones, DSP chips, smart antenna, and teleconference systems. Also, they have application in many areas such as system identification, channel equalization, noise reduction, echo cancellation, interference cancellation, signal prediction, and stock market. Therefore, reducing the energy consumption of the adaptive filtering algorithms has great importance, particularly in green technologies and in devices using battery.

In this thesis, data-selective adaptive filters, in particular the set-membership (SM) adaptive filters, are the tools to reach the goal. There are well known SM adaptive filters in literature. This work introduces new algorithms based on the classical ones in order to improve their performances and reduce the number of required arithmetic operations at the same time. Therefore, firstly, we analyze the robustness of the classical SM adaptive filtering algorithms. Secondly, we extend the SM technique to trinion and quaternion systems. Thirdly, by combining SM filtering and partial-updating, we introduce a new improved set-membership affine projection algorithm with constrained step size to improve its stability behavior. Fourthly, we propose some new least-mean-square (LMS) based and recursive least-squares based adaptive filtering algorithms with low computational complexity for sparse systems. Finally, we derive some feature LMS algorithms to exploit the hidden sparsity in the parameters.

\end{foreignabstract}

  \tableofcontents
  \listoffigures
  \listoftables
  \printlosymbols
  \printloabbreviations
  
  \mainmatter
  \chapter{Introduction}



In the last decades, the volume of data to be processed and kept for storage has been proliferated, mainly due to the increased availability of low-cost sensors and storage devices. As examples, we can mention the usage of multiple antennas in multiple-input and multiple-output wireless communication systems, the application of multiple audio devices in speech enhancement and audio signal processing, and the employment of echo cancellers in small or handheld communication devices. Moreover, these technological features are continuously spreading.

Our world is overwhelmed by data and to benefit from them in our daily life, we need to process the data correctly. A significant amount of data, however, brings about no new information in order that only part of it is particularly useful~\cite{Berberidis_censor_data_tsp2016,Wang_Big_data_GlobalSIP2014}. Therefore, we are compelled to improve our ability to evaluate the importance of the received data. This capability is called {\it data selection}. It enables the derivation of {\it data-selective adaptive filters}, which can neglect undesired data in a smart way. These filters are designed to reject the redundant data and perform their modeling tasks utilizing a small fraction of the available data.

Data-selective adaptive filters evaluate, select, and process data at each iteration of their learning process. These filters assess the data and choose only the ones bringing about some innovation. This property of the data-selective adaptive filters distinguishes them from the family of classical adaptive filters, which consider all data. In particular, these data-selective adaptive filters improve the accuracy of the estimator and decrease the computational complexity at the same time~\cite{Hamed_robustnessSMNLMS_sam2016,Hamed_robustnessSM_EURASIP2017,Markus_sparseSMAP_tsp2014}.

In this thesis, to apply the data selection, we employ the {\it set-membership filtering} (SMF)\abbrev{SMF}{Set-Membership Filtering} approach~\cite{Gollamudi_smf_letter1998,Diniz_adaptiveFiltering_book2013}. The set-membership (SM) adaptive filtering algorithm aims at estimating
the system such that the magnitude of the estimation output error is upper bounded by a predetermined positive constant called the threshold. The threshold is usually chosen based on {\it a priori} information about the sources of uncertainty. A comparison between traditional and SM adaptive filters was performed in~\cite{Diniz_adaptiveFiltering_book2013,Markus-phdthesis}, where the results had shown that the algorithms employing the SMF\abbrev{SMF}{Set-Membership Filtering} strategy require lower computational resources as compared to the conventional adaptive filters. The SMF\abbrev{SMF}{Set-Membership Filtering} algorithms, however, are not so widely used since there is some lack of analysis tools, and there is a limited number of set-membership adaptive filtering algorithms available. This thesis introduces new algorithms employing the SMF\abbrev{SMF}{Set-Membership Filtering} approach and provides some analysis tools.

This chapter is organized as follows. Section~\ref{sec:motivation-chap1} contains the main motivations. The targets of this thesis are given in Section~\ref{sec:target-chap1}. Section~\ref{sec:profile-chap1} describes the contributions of this thesis. Finally, the notation is explained in Section~\ref{sec:notation-chap1}.

\section{Motivations} \label{sec:motivation-chap1}

The area of {\it Digital Signal Processing} takes part in our daily lives for decades now, since it is at the core of virtually all electronic gadget we have been utilizing, ranging from medical equipment to mobile phones. If we have full information about the signals,
we can apply the most suitable algorithm (a digital filter for instance) to process the signals. However, if we do not know the statistical properties of the signals, a possible solution is to utilize an adaptive filter that automatically modifies its characteristics to match the behavior of the observed data. 

Adaptive filters~\cite{Diniz_adaptiveFiltering_book2013,Sayed_adaptiveFilters_book2008,Haykin_adaptiveFiltering_book2002} are utilized in several electronic and communication devices, such as smartphones, advanced headphones, DSP chips, smart antennas, and microphone arrays for teleconference systems. Also, they have application in many areas such as system identification~\cite{Raffaello_rls_dcd_eusipco2016}, channel equalization~\cite{Diniz_semiblind_ds_iscas2008}, noise reduction~\cite{Andersen_atf_taslp2016}, echo cancellation~\cite{Ruiz_acoustic_ec_its2014}, interference cancellation~\cite{Rodrigo_multi-antenna_twc2013}, signal prediction~\cite{Hamed_smtrinion-tcssII2016}, acoustic images~\cite{Ehrenfried_damas_aiaa2007}, stock market~\cite{Zheng_stock_market_icca2010}, etc. Due to the diversity of applications of adaptive signal processing, traditional adaptive filters cannot meet the needs of every application. An ideal adaptive filter would have low processing time, high accuracy in the learning process, low energy consumption, low memory usage, etc. These properties, however, conflict with each other.

An adaptive filter uses an algorithm to adjust its coefficients. An algorithm is a procedure to modify the coefficients in order to minimize a prescribed criterion. The algorithm is characterized by defining the search method, the objective function, and the error signal nature. The traditional algorithms in adaptive filtering implement coefficient updates at each iteration. However, when the adaptive filter learns from the observed data and reaches its steady state, it is desirable that the adaptive filter has the ability to reduce its energy consumption since there is less information to be learned. Here appears the importance of data-selective adaptive filters since they assess the input data, then according to the innovation they decide to perform an update or not.

After defining the set-membership adaptive filtering algorithms as a subset of the data-selective adaptive filters, many works have shown how effective these algorithms are in reducing the energy consumption. In some environments they can decrease the number of updates by 80$\%$~\cite{Diniz_adaptiveFiltering_book2013,Markus-phdthesis}. This thesis, however, shows that there is room for improvements regarding the reduction in the number of arithmetic operations and energy consumption, as discussed in Chapters 5 and 6.


\section{Targets} \label{sec:target-chap1}

The targets of this thesis are:
\begin{itemize}
\item To analyze the performance of some existing set-membership adaptive filtering algorithms to confirm their competitive performance as compared to the classical adaptive filtering approaches;
\item To develop data-selective adaptive filtering algorithms beyond the real and complex numbers, and examine the advantage of the set-membership technique in different mathematical number systems;
\item To improve some existing set-membership adaptive filtering algorithms to bring about improvements in performance and  computational complexity;
\item To introduce some new sparsity-aware set-membership adaptive filtering algorithms with low computational burden;
\item To exploit the hidden sparsity in the linear combination of parameters of adaptive filters.
\end{itemize}
In a nutshell, in this thesis, we improve and analyze data-selective adaptive filtering algorithms.


\section{Thesis Contributions} \label{sec:profile-chap1}

In this thesis, we analyze the robustness of classical set-membership adaptive filtering algorithms and extend these conventional algorithms for the trinion and the quaternion systems. In addition, we introduce an improved version of a set-membership adaptive filtering algorithm along with the partial updating strategy. Moreover, we develop some algorithms for sparse systems utilizing the SMF\abbrev{SMF}{Set-Membership Filtering} technique. Finally, we try to exploit the hidden sparsity in systems with lowpass and highpass frequencies. To address such topics, the text is hereinafter organized as follows.

Chapter 2 introduces some conventional adaptive filtering algorithms, such as the least-mean-square (LMS), the normalized LMS (NLMS), the affine projection (AP), and the recursive least-squares (RLS) ones. Then, we review the set estimation theory in adaptive signal processing and presents the set-membership filtering (SMF)\abbrev{SMF}{Set-Membership Filtering} strategy. Also, we describe a short review of the set-membership normalized least-mean-square (SM-NLMS) \abbrev{SM-NLMS}{Set-Membership Normalized LMS}and the set-membership affine projection (SM-AP) algorithms. 

In Chapter 3, we address the robustness, in the sense of $l_2$-stability, of the SM-NLMS \abbrev{SM-NLMS}{Set-Membership Normalized LMS}and the SM-AP\abbrev{SM-AP}{Set-Membership Affine Projection} algorithms. For the SM-NLMS \abbrev{SM-NLMS}{Set-Membership Normalized LMS}algorithm, we demonstrate that it is robust regardless the choice of its parameters and that the SM-NLMS \abbrev{SM-NLMS}{Set-Membership Normalized LMS}enhances the parameter estimation in most of the iterations in which an update occurs, two advantages over the classical NLMS algorithm. Moreover, we also prove that if the noise bound is known, then we can set the SM-NLMS \abbrev{SM-NLMS}{Set-Membership Normalized LMS}so that it never degrades the estimate. As for the SM-AP\abbrev{SM-AP}{Set-Membership Affine Projection} algorithm, we demonstrate that its robustness depends on a judicious choice of one of its parameters: the constraint vector (CV). We prove the existence of CVs satisfying the robustness condition, but practical choices remain unknown. We also demonstrate that both the SM-AP\abbrev{SM-AP}{Set-Membership Affine Projection} and the SM-NLMS \abbrev{SM-NLMS}{Set-Membership Normalized LMS}algorithms do not diverge, even when their parameters are selected naively, provided the additional noise is bounded. Furthermore, numerical results that corroborate our analyses are presented.

In Chapter 4, we introduce new data-selective adaptive filtering algorithms for trinion and quaternion systems $\mathbb{T}$ and $\mathbb{H}$. The work advances the set-membership trinion- and quaternion-valued normalized least-mean-square (SMTNLMS\abbrev{SMTNLMS}{Set-Membership Trinion-Valued NLMS} and SMQNLMS)\abbrev{SMQNLMS}{Set-Membership Quaternion-Valued NLMS} and the set-membership trinion- and quaternion-valued affine projection (SMTAP\abbrev{SMTAP}{Set-Membership Trinion-Valued AP} and SMQAP)\abbrev{SMQAP}{Set-Membership Quaternion-Valued AP} algorithms. Also, as special cases, we obtain trinion- and quaternion-valued algorithms not employing the set-membership strategy. Prediction simulations based on recorded wind data are provided, showing the improved performance of the proposed algorithms regarding reduced computational load. Moreover, we study the application of quaternion-valued adaptive filtering algorithms to adaptive beamforming.

Usually, set-membership algorithms implement updates more regularly during the early iterations in stationary environments. Therefore, if these updates exhibit high computational complexity, an alternative solution is needed. A possible approach to partly control the computational complexity is to apply partial update technique, where only a subset of the adaptive filter coefficients is updated at each iteration. In Chapter 5, we present an improved set-membership partial-update affine projection (I-SM-PUAP)\abbrev{I-SM-PUAP}{Improved SM-PUAP} algorithm, aiming at accelerating the convergence rate, and decreasing the update rate of the set-membership partial-update affine projection (SM-PUAP)\abbrev{SM-PUAP}{Set-Membership Partial-Update AP} algorithm. To meet these targets, we constrain the weight vector perturbation to be bounded by a hypersphere instead of the threshold hyperplanes as in the standard algorithm. We use the distance between the present weight vector and the expected update in the standard SM-AP\abbrev{SM-AP}{Set-Membership Affine Projection} algorithm to construct the hypersphere. Through this strategy, the new algorithm shows better behavior in the early iterations. Simulation results verify the excellent performance of the proposed algorithm related to the convergence rate and the required number of updates.

In Chapter 6, we derive two LMS-based\abbrev{LMS}{Least-Mean-Square} algorithms, namely the simple set-membership affine projection (S-SM-AP)\abbrev{S-SM-AP}{Simple SM-AP} and the improved S-SM-AP\abbrev{IS-SM-AP}{Improved S-SM-AP} (IS-SM-AP), in order to exploit the sparsity of an unknown system while focusing on having low computational cost. To achieve this goal, the proposed algorithms apply a discard function on the weight vector to disregard the coefficients close to zero during the update process. In addition, the IS-SM-AP\abbrev{IS-SM-AP}{Improved S-SM-AP} algorithm reduces the overall number of computations required by the adaptive filter even further by replacing small coefficients with zero. Moreover, we introduce the $l_0$ norm RLS ($l_0$-RLS)\abbrev{$l_0$-RLS}{$l_0$ Norm RLS} and the RLS\abbrev{RLS}{Recursive Least-Squares} algorithm for sparse models (S-RLS)\abbrev{S-RLS}{RLS Algorithm for Sparse System}. Also, we derive the data-selective version of these RLS-based\abbrev{RLS}{Recursive Least-Squares} algorithms. Simulation results show similar performance when comparing the proposed algorithms with some existing state-of-the-art sparsity-aware algorithms while the proposed algorithms require lower computational complexity.   

When our target is to detect and exploit sparsity in the model parameters, in many situations, the sparsity is hidden in the relations among these coefficients so that some suitable tools are required to reveal the potential sparsity. Chapter 7 proposes a set of least-mean-square (LMS)\abbrev{LMS}{Least-Mean-Square} type algorithms, collectively called feature LMS (F-LMS)\abbrev{F-LMS}{Feature LMS} algorithms, setting forth a hidden feature of the unknown parameters, which ultimately would improve convergence speed and steady-state mean-squared error. The fundamental idea is to apply linear transformations, by means of the so-called feature matrices, to reveal the sparsity hidden in the coefficient vector, followed by a sparsity-promoting penalty function to exploit such sparsity. Some F-LMS\abbrev{F-LMS}{Feature LMS} algorithms for lowpass and highpass systems are also introduced by using simple feature matrices that require only trivial operations. Simulation results demonstrate that the proposed F-LMS\abbrev{F-LMS}{Feature LMS} algorithms bring about several performance improvements whenever the hidden sparsity of the parameters is exposed.

Finally, chapter 8 highlights the conclusions of the work, and gives some clues for future works regarding the topics addressed in the thesis.


\section{Notation} \label{sec:notation-chap1}

In this section, we introduce most of the usual notation utilized in this thesis. However, in order to avoid confusing the reader, we evade presenting here the definition of the rare notation in this text, and we introduce them only at the vital moments.

Equalities are shown by $=$, and when they refer to a definition, we use $\triangleq$.\symbl{$\triangleq$}{Definition} The real, nonnegative real, nature, integer, complex, trinion, and quaternion numbers are denoted by $\mathbb{R}$, $\mathbb{R}_+$, $\mathbb{N}$, $\mathbb{Z}$, $\mathbb{C}$, $\mathbb{T}$, and $\mathbb{H}$, respectively. \symbl{$\mathbb{R}$}{Set of real numbers} \symbl{$\mathbb{R}_+$}{Set of nonnegative real numbers} \symbl{$\mathbb{N}$}{Set of natural numbers} \symbl{$\mathbb{Z}$}{Set of integer numbers} \symbl{$\mathbb{C}$}{Set of complex numbers} \symbl{$\mathbb{T}$}{Set of trinion numbers} \symbl{$\mathbb{H}$}{Set of quaternion numbers}

Moreover, scalars are represented by lowercase letters (e.g., $x$), vectors by lowercase boldface letters (e.g., $\xbf$), and matrices by uppercase boldface letters (e.g., $\Xbf$). The
symbols $(\cdot)^T$ and $(\cdot)^H$ stand for the transposition\symbl{$(\cdot)^T$}{Transposition of $(\cdot)$} and Hermitian operators,\symbl{$(\cdot)^H$}{Hermitian transposition of $(\cdot)$} respectively. Also, all vectors are column vectors in order that the inner product between two vectors $\xbf$ and $\ybf$ is defined as $\xbf^T\ybf$ or $\xbf^H\ybf$.

We represent the trace operator by ${\rm tr}(\cdot)$.\symbl{${\rm tr}(\cdot)$}{Trace of matrix} The identity matrix and zero vector (matrix) are denoted by $\Ibf$ \symbl{$\Ibf$}{Identity matrix} and ${\bf 0}$,\symbl{${\bf 0}$}{Zero vector or zero matrix} respectively. Also, ${\rm diag}(\xbf)$ stands for a diagonal matrix with vector $\xbf$ on its diagonal and zero outside it.\symbl{${\rm diag}(\xbf)$}{Diagonal matrix with $\xbf$ on its diagonal} Furthermore, $\mathbb{P}[\cdot]$ and $\mathbb{E}[\cdot]$ denote the probability\symbl{$\mathbb{P}$}{Probability operator} and the expected value operators,\symbl{$\mathbb{E}$}{Expected value operator} respectively. Also, $\|\cdot\|$ denotes the $l_2$ norm (when the norm is not defined explicitly, we are referring to the $l_2$ norm).
  \chapter{Conventional and Set-Membership Adaptive Filtering Algorithms}



The {\it point estimation theory}~\cite{Lehmann_pointEstimation_book2003} utilizes a sample data for computing a single solution as the best estimate of an unknown parameter. For decades, machine learning and adaptive filtering have been grounded in the point estimation theory~\cite{Diniz_adaptiveFiltering_book2013,Sayed_adaptiveFilters_book2008,Haykin_adaptiveFiltering_book2002,Theodoridis_Pattern_Recognition_book2008,Bishop_Pattern_Recognition_book2011}. Nowadays, the benefit of the set estimation approach, however, is becoming clearer by disclosing its advantages~\cite{Combettes_foundationSetTheoreticEstimation_procIEEE1993,Markus_edcv_eusipco2013,Combettes_noise_SetTheoretic_tsp1991}. 

In contrast with the world of theoretical models, in the real-world we live with uncertainties originate from measurement noise, quantization, interference, modeling errors, etc. Therefore, searching the solution utilizing point estimation theory sometimes results in a waste of energy and time. An alternative is to address the problem from the {\it set estimation theory}~\cite{Combettes_foundationSetTheoreticEstimation_procIEEE1993} point of view. In fact, in this approach, we search for a set of acceptable solutions instead of a unique point as a solution. 

The adaptive filtering algorithms presented in \cite{Haykin_adaptiveFiltering_book2002,Sayed_adaptiveFilters_book2008} exhibit a trade-off between convergence rate and misadjustment after transient, particularly in stationary environments. In general, fast converging algorithms lead to high variance estimators after convergence. To tackle this problem, we can apply set-membership filtering (SMF)~\abbrev{SMF}{Set-Membership Filtering}
\cite{Diniz_adaptiveFiltering_book2013,Markus-phdthesis} which is a representative of the set estimation theory. The SMF\abbrev{SMF}{Set-Membership Filtering} technique prevents unnecessary updates and reduces the computational complexity by updating the filter coefficients only when the estimation error is greater than a predetermined upper bound~\cite{Fogel_valueOfInformation_automatica1982,Deller_smi_asspmag1989,Gollamudi_smf_letter1998}.

In set-membership adaptive filters, we try to find a {\it feasibility set} such that any member in this set has the output estimation error limited by a predetermined upper bound. For this purpose, the objective function of the algorithm is related to a bounded error constraint on the filter output, such that the updates are contained in a set of acceptable solutions. The inclusion of {\it a priori} information, such as the noise bound, into the objective function leads to some noticeable advantages. As compared with the normalized least-mean-square (NLMS)\abbrev{NLMS}{Normalized LMS} and the affine projection (AP)\abbrev{AP}{Affine Projection} algorithms, their set-membership counterparts have lower computational cost, better accuracy, data selection, and robustness against noise~\cite{Gollamudi_smf_letter1998,Gollamudi_smUpdatorShared_tsp1998,Nagaraj_beacon_tsp1999,Diniz_sm_bnlms_tsp2003,Werner_sm_ap_letter2001,Hamed_robustnessSMNLMS_sam2016,Hamed_robustnessSM_EURASIP2017}.

This chapter presents a brief review of some adaptive filtering algorithms. An interested reader should refer to~\cite{Diniz_adaptiveFiltering_book2013} for more details. Section~\ref{sec:conventional_algorithms} describes the point estimation adaptive filtering algorithms.  Section~\ref{sec:SM-AF-chap2} reviews the SMF\abbrev{SMF}{Set-Membership Filtering} approach and the main set-membership algorithms. The estimation of the threshold parameter for big data applications is discussed in Section~\ref{sec:estimate_gamma_chap2}. Finally, Section~\ref{sec:conclusion-chap2} contains the conclusions.


\section{Point Estimation Adaptive Filtering \\Algorithms} \label{sec:conventional_algorithms}

In this section, we introduce some LMS-based adaptive filtering algorithms and the recursive least-squares (RLS)\abbrev{RLS}{Recursive Least-Squares} algorithm.

\subsection{Least-mean-square algorithm} \label{sub:lms-chap2}

The update equation of the least-mean-square (LMS)\abbrev{LMS}{Least-Mean-Square} algorithm is given by~\cite{Diniz_adaptiveFiltering_book2013}
\begin{align}
\wbf(k+1)=\wbf(k)+2\mu e(k)\xbf(k),
\end{align}
where $\xbf(k)=[x_0(k)~x_1(k)~\cdots~x_N(k)]^T$ and $\wbf(k)=[w_0(k)~w_1(k)~\cdots~w_N(k)]^T$ are the input signal vector and the the weight vector, respectively.\symbl{$\xbf(k)$}{Input signal vector} \symbl{$\wbf(k)$}{Coefficient vector} \symbl{$k$}{Iteration counter} The output signal is defined by $y(k)\triangleq\wbf^T(k)\xbf(k)=\xbf^T(k)\wbf(k)$,\symbl{$y(k)$}{Output signal} and $e(k)\triangleq d(k)-y(k)$ denotes the error signal,\symbl{$e(k)$}{Error signal} where $d(k)$ is the desired signal.\symbl{$d(k)$}{Desired signal} The convergence factor $\mu$ \symbl{$\mu$}{Convergence factor} should be chosen in the range $0<\mu<\frac{1}{{\rm tr}[\Rbf]}$ to guarantee the convergence, where $\Rbf\triangleq\mathbb{E}[\xbf(k)\xbf^T(k)]$ is the correlation matrix. \symbl{$\Rbf$}{Correlation matrix}


\subsection{Normalized LMS algorithm} \label{sub:nlms-chap2}

To increase the convergence rate of the LMS\abbrev{LMS}{Least-Mean-Square} algorithm without using matrix $\Rbf$, we can utilize the NLMS\abbrev{NLMS}{Normalized LMS} algorithm. The recursion rule of the NLMS\abbrev{NLMS}{Normalized LMS} algorithm is described by~\cite{Diniz_adaptiveFiltering_book2013}
\begin{align}
\wbf(k+1)=\wbf(k)+\frac{\mu_n}{\xbf^T(k)\xbf(k)+\delta}e(k)\xbf(k),
\end{align}
where $\delta$ is a small regularization factor,\symbl{$\delta$}{Regularization factor} and the step size $\mu_n$ should be selected in the range $0<\mu_n<2$.


\subsection{Affine projection algorithm} \label{sub:ap-chap2}

When the input signal is correlated, it is possible to use old data signal to improve the convergence speed of the algorithm. For this purpose, let us utilize the last $L+1$ input signal vector and form matrix $\Xbf(k)$ as
\begin{align}
\Xbf(k)=[\xbf(k)~\xbf(k-1)~\cdots~\xbf(k-L)]\in\mathbb{R}^{(N+1)\times(L+1)}.
\end{align}
Also, let us define the desired signal vector $\dbf(k)$, the output signal vector $\ybf(k)$, and the error signal vector $\ebf(k)$ as follows
\begin{align}
\dbf(k)&=[d(k)~d(k-1)~\cdots~d(k-L)]^T,\nonumber\\
\ybf(k)&\triangleq\wbf^T(k)\Xbf(k)=\Xbf^T(k)\wbf(k),\nonumber\\
\ebf(k)&\triangleq\dbf(k)-\ybf(k).
\end{align}

Then, the update rule of the affine projection (AP)\abbrev{AP}{Affine Projection} algorithm is described by~\cite{Diniz_adaptiveFiltering_book2013}
\begin{align}
\wbf(k+1)=\wbf(k)+\mu\Xbf(k)[\Xbf^T(k)\Xbf(k)]^{-1}\ebf(k),
\end{align}
where $\mu$ is the convergence factor.


\subsection{Recursive least-squares algorithm} \label{sub:rls-chap2}

Here, we review the RLS\abbrev{RLS}{Recursive Least-Squares} algorithm. The goal of this algorithm is to match the output signal to the desired signal as much as possible.

The objective function of the RLS\abbrev{RLS}{Recursive Least-Squares} algorithm is given by
\begin{align}
\zeta(k)=\sum_{i=0}^k\lambda^{k-i}\varepsilon^2(i)=\sum_{i=0}^k\lambda^{k-i}[d(i)-\xbf^T(i)\wbf(k)]^2,
\end{align}
where $\lambda$ is a forgetting factor which should be adopted in the range $0\ll\lambda\leq1$, and $\varepsilon(i)$ is called the {\it a posteriori} error.\symbl{$\varepsilon(k)$}{{\it A posteriori} error signal} Note that in the elaboration of the LMS-based\abbrev{LMS}{Least-Mean-Square} algorithms we use the {\it a priori} error, whereas for the RLS\abbrev{RLS}{Recursive Least-Squares} algorithm we utilize the {\it a posteriori} error.

If we differentiate $\zeta(k)$ with respect to $\wbf(k)$ and equate the result to zero, we get the optimal coefficient vector $\wbf(k)$~\cite{Diniz_adaptiveFiltering_book2013}
\begin{align}
\wbf(k)=\Big[\sum_{i=0}^k\lambda^{k-i}\xbf(i)\xbf^T(i)\Big]^{-1}\sum_{i=0}^k\lambda^{k-i}\xbf(i)d(i)=\Rbf_D^{-1}(k)\pbf_D(k),
\end{align}
where $\Rbf_D(k)$ and $\pbf_D(k)$ are named the deterministic correlation matrix of the input signal and the deterministic cross-correlation vector between the input and the desired signals, respectively.\symbl{$\Rbf_D(k)$}{Deterministic correlation matrix of the input signal} \symbl{$\pbf_D(k)$}{Deterministic cross-correlation vector between the input and the desired signals} By using the matrix inversion lemma~\cite{Goodwin_Dynamic_system_id_book1977}, the inverse of $\Rbf_D(k)$ can be given by\symbl{$\Sbf_D(k)$}{The inverse of $\Rbf_D(k)$}
\begin{align}
\Sbf_D(k)=\Rbf_D^{-1}(k)=\frac{1}{\lambda}\Big[\Sbf_D(k-1)-\frac{\Sbf_D(k-1)\xbf(k)\xbf^T(k)\Sbf_D(k-1)}{\lambda+\xbf^T(k)\Sbf_D(k-1)\xbf(k)}\Big].
\end{align}


\section{Set-Membership Adaptive Filtering \\Algorithms} \label{sec:SM-AF-chap2}

In this section, we firstly introduce the set-membership filtering (SMF)\abbrev{SMF}{Set-Membership Filtering} approach. Secondly, we present the SM-NLMS\abbrev{SM-NLMS}{Set-Membership Normalized LMS} algorithm. Finally, we review the SM-AP\abbrev{SM-AP}{Set-Membership Affine Projection} algorithm.


\subsection{Set-membership filtering} \label{sub:SMF-chap2}

The SMF\abbrev{SMF}{Set-Membership Filtering} approach proposed in~\cite{Gollamudi_smf_letter1998} is suitable for adaptive filtering problems that are linear in parameters. Thus, for a given input signal vector $\xbf(k)\in\mathbb{R}^{N+1}$ at iteration $k$ and the filter coefficients $\wbf\in\mathbb{R}^{N+1}$, the output signal of the filter is obtained by 
\begin{align}
y(k)=\wbf^T\xbf(k),
\end{align}
where $\xbf(k)=[x_0(k)~x_1(k)~\cdots~x_N(k)]^T$ and $\wbf=[w_0~w_1~\cdots~w_N]^T$.
For a desired signal sequence $d(k)$, the estimation error sequence $e(k)$ is computed as
\begin{align}
e(k)=d(k)-y(k).
\end{align}
The SMF\abbrev{SMF}{Set-Membership Filtering} criterion aims at estimating the parameter $\wbf$ such that the magnitude of the estimation output error is upper bounded by a constant $\gammabar\in\mathbb{R}_+$, for all possible pairs $(\xbf,d)$.\symbl{$\gammabar$}{Upper bound for the magnitude of the error signal} If the value of $\gammabar$ is suitably selected, there are various valid estimates for $\wbf$. The threshold is usually chosen based on {\it a priori} information about the sources of uncertainty. Note that any $\wbf$ leading to an output estimation error with magnitude smaller than $\gammabar$ is an acceptable solution. Hence, we obtain a set of filters rather than a single estimate.

Let us denote by ${\cal S}$ the set comprised of all possible pairs $(\xbf,d)$.\symbl{${\cal S}$}{Set comprised of all possible pairs $(\xbf,d)$} We want to find $\wbf$ such that $|e|=|d-\wbf^T\xbf|\leq\gammabar$ for all $(\xbf,d)\in{\cal S}$. Therefore, the {\it feasibility set} $\Theta$ will be defined as\symbl{$\Theta$}{Feasibility set}
\begin{align}
\Theta\triangleq\bigcap_{(\xbf,d)\in{\cal S}}\{\wbf\in\mathbb{R}^{N+1}:|d-\wbf^T\xbf|\leq\gammabar\},
\end{align}
so that the SMF\abbrev{SMF}{Set-Membership Filtering} criterion can be stated as finding $\wbf\in\Theta$.

In the case of online applications, we do not have access to all members of ${\cal S}$. Thus, we consider the practical case in which only measured data are available and develop iterative techniques. Suppose that a set of data pairs $\{(\xbf(0),d(0)),(\xbf(1),d(1)),\cdots,(\xbf(k),d(k))\}$ is available, and define the {\it constraint set} ${\cal H}(k)$ at time instant $k$ as\symbl{${\cal H}(k)$}{Constraint set at iteration $k$}
\begin{align}
{\cal H}(k)\triangleq \{\wbf\in\mathbb{R}^{N+1}:|d(k)-\wbf^T\xbf(k)|\leq\gammabar\}.
\end{align}
Also, define the {\it exact membership set} $\psi(k)$ as the intersection of the constraint sets from the beginning, i.e. the first iteration, to iteration $k$,\symbl{$\psi(k)$}{Exact membership set} or
\begin{align}
\psi(k)\triangleq \bigcap_{i=0}^k{\cal H}(i).
\end{align}
Then, $\Theta$ can be iteratively estimated via the exact membership set since $\lim_{k\rightarrow\infty}\psi(k)=\Theta$.

Figure \ref{fig:smf-chap2} shows the geometrical interpretation of the SMF\abbrev{SMF}{Set-Membership Filtering} principle. The boundaries of the constraint sets are hyperplanes, and ${\cal H}(k)$ corresponds to region between the parallel hyperplanes in the parameter space. The exact membership set represents a polytope in the parameter space. The volume of $\psi(k)$ decreases for each $k$ in which the pairs $(\xbf(k),d(k))$ bring about some innovation. Note that $\Theta\subset\psi(k)$ for all $k$, since $\Theta$ is the intersection of all possible constraint sets.

\begin{figure}[t!]
\begin{center}
\includegraphics[width=.85\linewidth]{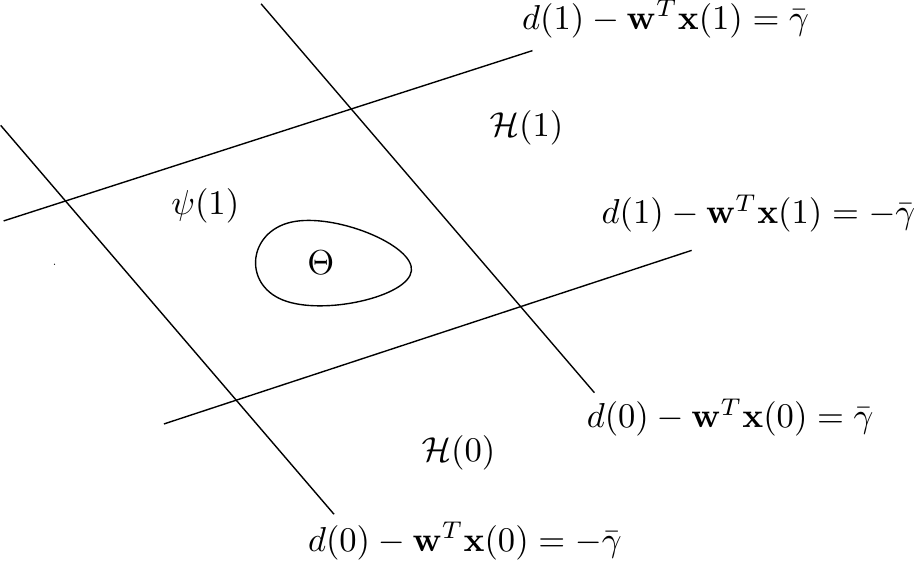}
\caption{SMF geometrical interpretation in the parameter space  $\psi(1)$ (redrawn from \cite{Markus-phdthesis}).}
\label{fig:smf-chap2}
\end{center}
\end{figure} 

The target of set-membership adaptive filtering is to obtain adaptively an estimate that belongs to the feasibility set. The simplest method is to calculate a point estimate using, for example, the information
provided by ${\cal H}(k)$ similar to the set-membership NLMS\abbrev{NLMS}{Normalized LMS} algorithm described in the following subsection, or several previous ${\cal H}(k)$ like in the SM-AP\abbrev{SM-AP}{Set-Membership Affine Projection} algorithm discussed in Subsection~\ref{sub:sm-ap-chap2}.


\subsection{Set-membership normalized LMS algorithm} \label{sub:sm-nlms-chap2}

The set-membership NLMS\abbrev{NLMS}{Normalized LMS} algorithm, first proposed in \cite{Gollamudi_smf_letter1998}, implements a test to check if the previous estimate $\wbf(k)$ lies outside the constraint set ${\cal H}(k)$. If $|d(k)-\wbf^T(k)\xbf(k)|>\gammabar$, then $\wbf(k+1)$ will be updated to the closest boundary of ${\cal H}(k)$ at a minimum distance. Figure~\ref{fig:sm-nlms-chap2} depicts the updating procedure of the SM-NLMS \abbrev{SM-NLMS}{Set-Membership Normalized LMS}algorithm.

\begin{figure}[t!]
\begin{center}
\includegraphics[width=.75\linewidth]{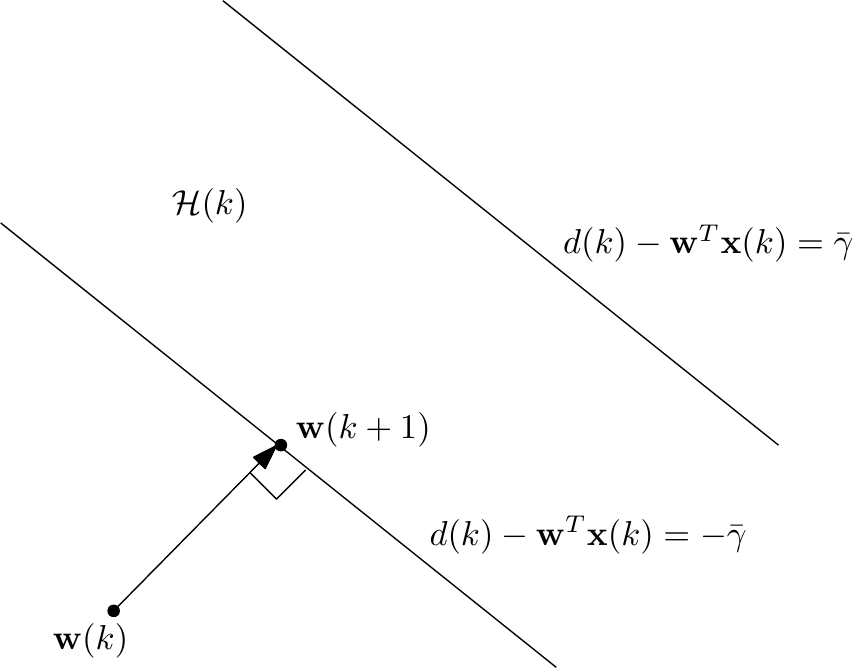}
\caption{Coefficient vector updating for the SM-NLMS algorithm (redrawn from \cite{Diniz_adaptiveFiltering_book2013}).}
\label{fig:sm-nlms-chap2}
\end{center}
\end{figure}

The SM-NLMS \abbrev{SM-NLMS}{Set-Membership Normalized LMS}algorithm has the updating rule
\begin{align}
\wbf(k+1)= \wbf(k)+\frac{\mu(k)}{\xbf^T(k)\xbf(k)+\delta}e(k)\xbf(k), \label{eq:sm-nlms-update-chap2}
\end{align}
where the variable step size $\mu(k)$ is given by 
\begin{align}
\mu(k)=\left\{\begin{array}{ll}1-\frac{\gammabar}{|e(k)|}&\text{if }|e(k)|>\gammabar,\\0&\text{otherwise},\end{array}\right. 
\end{align}
and $\delta$ is a small regularization factor. As
a rule of thumb, the value of $\gammabar$ is selected about $\sqrt{\tau\sigma_n^2}$, where $\sigma_n^2$ is the variance of the additional noise~\cite{Gollamudi_smf_letter1998,Galdino_SMNLMS_gammabar_ISCAS2006}, and $1\leq\tau\leq5$.

Note that we can introduce the NLMS algorithm through the SM-NLMS algorithm. Indeed, the NLMS\abbrev{NLMS}{Normalized LMS} algorithm with unit step size is a particular case of the SM-NLMS \abbrev{SM-NLMS}{Set-Membership Normalized LMS}algorithm by adopting $\gammabar=0$.


\subsection{Set-membership affine projection algorithm} \label{sub:sm-ap-chap2}

The exact membership set $\psi(k)$ suggests the use of more constraint sets in the update~\cite{Werner_sm_ap_letter2001}. Moreover, it is widely known that data-reusing algorithms can increase convergence speed significantly for correlated-input 
signals~\cite{Diniz_adaptiveFiltering_book2013,Haykin_adaptiveFiltering_book2002,Ozeki_ap_japan1984}. This section introduces the SM-AP\abbrev{SM-AP}{Set-Membership Affine Projection} algorithm whose updates belong to the last $L+1$ constraint sets. For this purpose, let us define the input signal matrix $\Xbf(k)$, the output signal vector $\ybf(k)$, the error signal vector $\ebf(k)$, the desired signal vector $\dbf(k)$, 
the additive noise signal vector $\nbf(k)$, and the constraint vector (CV)\abbrev{CV}{Constraint Vector} $\gammabf(k)$ as \symbl{$\Xbf(k)$}{Input signal matrix} \symbl{$\ybf(k)$}{Output signal vector} \symbl{$\ebf(k)$}{Error signal vector} \symbl{$\dbf(k)$}{Desired signal vector} \symbl{$\nbf(k)$}{Additive noise signal vector} \symbl{$\gammabf(k)$}{Constraint vector}
\begin{equation}
\begin{aligned}
\Xbf(k)&=[\xbf(k)~\xbf(k-1)~\cdots~\xbf(k-L)]                 \in\mathbb{R}^{(N+1)\times (L+1)},\\
\xbf(k)&=[x(k)~x(k-1)~\cdots~x(k-N)]^T\in\mathbb{R}^{N+1},\\
\ybf(k)&=[y(k)~y(k-1)~\cdots~y(k-L)]^T\in\mathbb{R}^{L+1},\\
\ebf(k)&=[e(k)~\epsilon(k-1)~\cdots~\epsilon(k-L)]^T             \in\mathbb{R}^{L+1},\\
\dbf(k)&=[d(k)~d(k-1)~\cdots~d(k-L)]^T                        \in\mathbb{R}^{L+1},\\
\nbf(k)&=[n(k)~n(k-1)~\cdots~n(k-L)]^T                           \in\mathbb{R}^{L+1},\\
\gammabf(k)&=[\gamma_0(k)~\gamma_1(k)~\cdots~\gamma_L(k)]^T     \in\mathbb{R}^{L+1},
 \label{eq:pack-chap2}
\end{aligned}
\end{equation}
where $N$ is the order of the adaptive filter\symbl{$N$}{Order of the FIR adaptive filter}, and $L$ is the data-reusing factor\symbl{$L$}{Data reuse factor}, i.e., 
$L$ previous data are used together with the data from the current iteration $k$. 
The output signal vector is defined as $\ybf(k)\triangleq\wbf^T(k)\Xbf(k)=\Xbf^T(k)\wbf(k)$, the desired signal vector is given by $\dbf(k)\triangleq\wbf_o^T\Xbf(k)+\nbf(k)$, where $\wbf_o$ is the optimal solution (unknown system),\symbl{$\wbf_o$}{Impulse response of the unknown system} and the error signal vector is given by $\ebf(k) \triangleq \dbf(k)-\ybf(k)$. The entries of the constraint vector 
should satisfy $| \gamma_i(k) |\leq \gammabar$, for $i=0,\ldots,L$, where $\gammabar \in \mathbb{R}_+$ is the upper bound for the 
magnitude of the error signal $e(k)$.

The objective function to be minimized in the SM-AP\abbrev{SM-AP}{Set-Membership Affine Projection} algorithm can be stated as follows: a coefficient update is implemented whenever $\wbf(k)\not\in\psi^{L+1}(k)$ in such a way that
\begin{align}
&\min\frac{1}{2}\|\wbf(k+1)-\wbf(k)\|^2\nonumber\\
&\text{subject to:}\nonumber\\
&\dbf(k)-\Xbf^T(k)\wbf(k+1)=\gammabf(k),
\end{align} 
where $\psi^{L+1}(k)$ is the intersection of the $L+1$ last constraint sets.

Figure~\ref{fig:sm-ap-chap2} shows a usual coefficient update related to the SM-AP\abbrev{SM-AP}{Set-Membership Affine Projection} algorithm in $\mathbb{R}^2$, $L=1$ and $|\gamma_i(k)|\leq\gammabar$ such that $\wbf(k+1)$ is not placed at the border of ${\cal H}(k)$.

\begin{figure}[t!]
\begin{center}
\includegraphics[width=.85\linewidth]{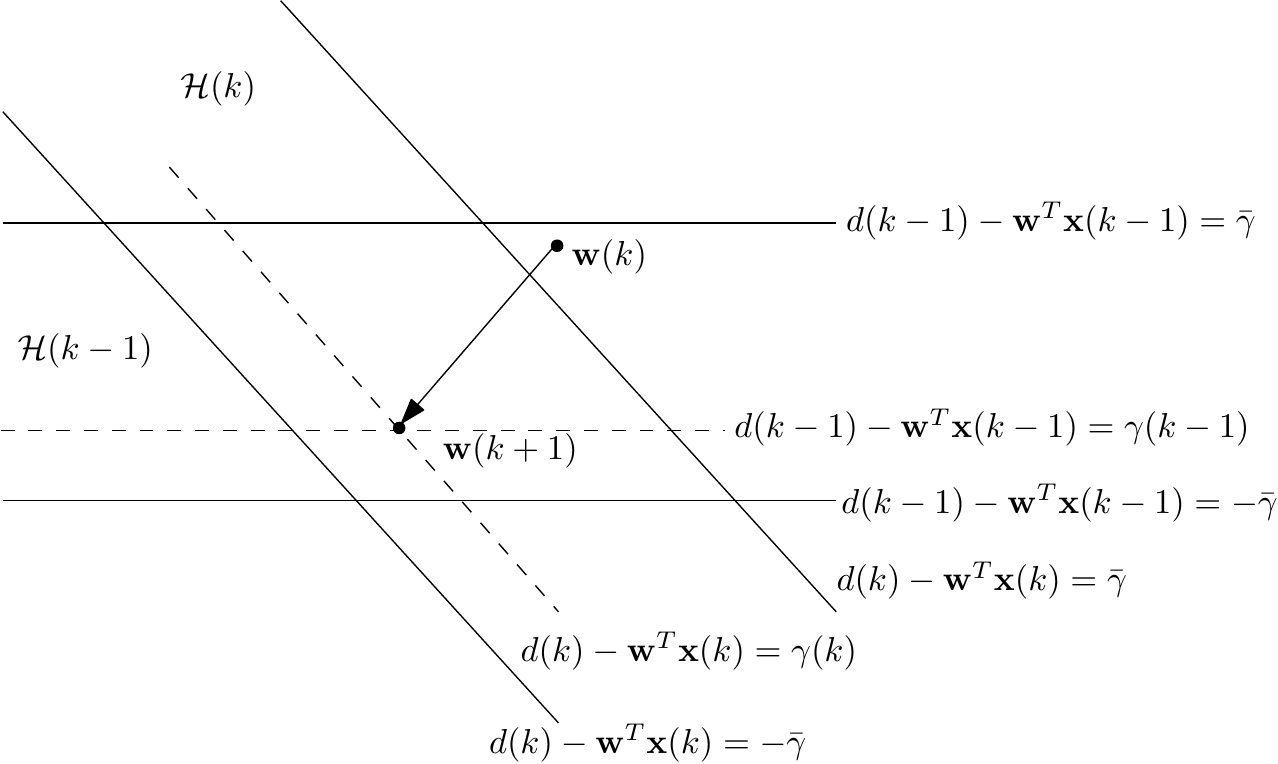}
\caption{Coefficient vector updating for the SM-AP algorithm (redrawn from \cite{Diniz_adaptiveFiltering_book2013}).}
\label{fig:sm-ap-chap2}
\end{center}
\end{figure}

By using the method of Lagrange multipliers, after some manipulations, the recursion rule of the SM-AP\abbrev{SM-AP}{Set-Membership Affine Projection} algorithm will be described as
\begin{align}
\wbf(k+1) =\left\{\begin{array}{ll}\wbf(k)+\Xbf(k)\Abf(k)(\ebf(k)-\gammabf(k))&\text{if}~|e(k)|>\gammabar ,
\\ \wbf(k)&\text{ otherwise,}\end{array}\right.   \   \label{eq:sm-ap}
\end{align}
where we assume that $\Abf(k)\triangleq(\Xbf^T(k)\Xbf(k))^{-1} \in\mathbb{R}^{L+1\times L+1}$ exists, i.e., $\Xbf^T(k)\Xbf(k)$ is a full-rank matrix. \symbl{$\Abf(k)$}{Auxiliary matrix $\Abf(k)\triangleq(\Xbf^T(k)\Xbf(k))^{-1}$}
Otherwise, we could add a regularization parameter as explained in~\cite{Diniz_adaptiveFiltering_book2013}.

Note that we can propose the AP\abbrev{AP}{Affine Projection} algorithm through the SM-AP\abbrev{SM-AP}{Set-Membership Affine Projection} algorithm. In other words, the AP\abbrev{AP}{Affine Projection} algorithm with unity step-size, aiming at improving the convergence speed of stochastic gradient algorithms, is a particular case of the SM-AP\abbrev{SM-AP}{Set-Membership Affine Projection} algorithm by selecting $\gammabar=0$.

It is worthwhile to mention that when $L=0$ and $\gamma_0(k)=\frac{\gammabar e(k)}{|e(k)|}$, the SM-AP algorithm has the SM-NLMS algorithm as special case.


\section{Estimating $\gammabar$ in the Set-Membership \\Algorithm for Big Data Application} \label{sec:estimate_gamma_chap2}

In big data applications, initially, it could be practical to prescribe a percentage of the amount of data we intend to utilize to achieve the desired performance. This percentage will be defined in accordance with our ability to analyze the data, taking into consideration the constraints on energy, computational time, and memory restrictions. After adopting a percentage of the update, our goal is to select the most informative data to be part of the corresponding selected percentage. Here, by taking the probability of updating into consideration, we will estimate the threshold in the SM-NLMS \abbrev{SM-NLMS}{Set-Membership Normalized LMS} and the SM-AP\abbrev{SM-AP}{Set-Membership Affine Projection} algorithms, which is responsible for censoring the data in accordance with the adopted percentage of the update. The content of this section is published in~\cite{Hamed_gamma_estimate_GlobalSIP2017}.

We want to obtain $\gammabar$ such that the algorithm considers the desired percentage of data to update its recursion rule. In fact, if the magnitude of the output estimation error is greater than $\gammabar$, the set-membership (SM)\abbrev{SM}{Set-Membership} algorithm will update since the current input and the desired signals carry enough innovation.

In general, for the desired update rate, $p$, we require computing
$\gammabar$ such that
\begin{align}
\mathbb{P}[|e(k)|>\gammabar]=p, \label{eq:single_gamma-chap2}
\end{align}
where $\mathbb{P}[\cdot]$ denotes the probability operator. Note that $p$ represents the update rate of the algorithm, i.e., the percentage of the data which we consider most informative data.

Given the probability density function of the error signal, then it is possible to compute $\gammabar$. Note that the error signal is the difference between the desired and the output signals, i.e.,
\begin{align}
e(k)&\triangleq d(k)-y(k)\triangleq \wbf_o^T\xbf(k)+n(k)-\wbf^T(k)\xbf(k)\nonumber\\
&=[\wbf_o-\wbf(k)]^T\xbf(k)+n(k)=\etilde(k)+n(k), \label{eq:error_signal-chap2}
\end{align}
where $\etilde(k)$ is the noiseless error signal, and $n(k)$ is the noise signal. \symbl{$\etilde(k)$}{Noiseless error signal} \symbl{$n(k)$}{Noise signal}
In the steady-state environment $\|\mathbb{E}[\wbf_o-\wbf(k)]\|_2^2<\infty$~\cite{Hamed_robustnessSM_EURASIP2017}, where $\mathbb{E}[\cdot]$ is the expected value operator and, in general, $\mathbb{E}[\wbf_o-\wbf(k)]\approx{\bf 0}$. Therefore, if you have sufficient order for the adaptive system, then in the steady-state environment the distribution of the error signal and the additive noise signal are the same. Thus, we can use the distribution of the additive noise signal in Equation~(\ref{eq:single_gamma-chap2}) to calculate the desired value of $\gammabar$.

Assuming the distribution of the noise signal is Gaussian with zero mean and variance $\sigma_n^2$,\symbl{$\sigma_n^2$}{Variance of the noise signal} an important case, we can provide a solution for the threshold for this special case. If the noiseless error signal is uncorrelated with the additional noise signal, by Equation~(\ref{eq:error_signal-chap2}), we have $\mathbb{E}[e(k)]=\mathbb{E}[\etilde(k)]+\mathbb{E}[n(k)]=0$
and ${\rm Var}[e(k)]=\mathbb{E}[\etilde^2(k)]+\sigma_n^2$, where ${\rm Var}[\cdot]$ is the variance operator.\symbl{${\rm Var}$}{Variance operator} $\mathbb{E}[\etilde^2(k)]$ is the excess of the steady-state mean-square error (EMSE)\abbrev{EMSE}{Excess of the Steady-State Mean-Square Error} that in the steady-state environment is given by~\cite{Markus_mseSMAP_icassp2010,Markus_mseSMAP_cssp2013}
\begin{align}
\mathbb{E}[\etilde^2(k)]=\frac{(L+1)[\sigma_n^2+\gammabar^2-2\gammabar\sigma_n^2\rho_0(k)]p}{[(2-p)-2(1-p)\gammabar\rho_o(k)]}\Big(\frac{1-a}{1-a^{L+1}}\Big), \label{eq:E(etilde)-chap2}
\end{align}
where 
\begin{align}
\rho_0(k)=&\sqrt{\frac{2}{\pi(2\sigma_n^2+\frac{1}{L+1}\gammabar^2)}}, \label{eq:rho_0-chap2}\\
a=&[1-p+2p\gammabar\rho_0(k)](1-p). \label{eq:a-chap2}
\end{align}

To calculate $\mathbb{E}[\etilde^2(k)]$ in Equation~(\ref{eq:E(etilde)-chap2}), we require the value of $\gammabar$, while estimating $\gammabar$ is our purpose. To address this
problem, the natural approach is estimate it using numerical
integration or Monte-Carlo methods. However, aiming at gaining some
insight, at the first moment we can assume that in the steady-state
environment $\mathbb{E}[\etilde^2(k)]=0$, and the distribution of $e(k)$ is the same as $n(k)$, in order to calculate the estimation of $\gammabar$ using Equation~(\ref{eq:single_gamma-chap2}). Then, we substitute the obtained value of $\gammabar$ in Equation~(\ref{eq:E(etilde)-chap2}) to compute $\mathbb{E}[\etilde^2(k)]$. Finally, by obtaining $\mathbb{E}[\etilde^2(k)]$, we can have a better estimation for the distribution of $e(k)$. 

Therefore, since the distribution of $e(k)$ is the same as the distribution of  $n(k)$, for the first estimation of $\gammabar$ we have
\begin{align}
\mathbb{P}[|e(k)|>\gammabar]=\mathbb{P}[|n(k)|>\gammabar]
=\mathbb{P}[n(k)<-\gammabar]+\mathbb{P}[n(k)>\gammabar]=p.
\end{align}
Then because of the symmetry in Gaussian distribution we have $\mathbb{P}[n(k)>\gammabar]=\frac{p}{2}$. Since $n(k)$ has Gaussian distribution, we need to obtain $\gammabar$ from
\begin{align}
\int_{\gammabar}^{\infty}\frac{1}{\sqrt{2\pi\sigma_n^2}}\exp(-\frac{r^2}{2\sigma_n^2})dr=\frac{p}{2}.
\end{align}
Hence, given an update rate $0\leq p\leq 1$, we may use the standard normal distribution table and find the desired $\gammabar$. As the second step, for getting a better estimation of $\gammabar$, we substitute $\gammabar$ in Equations~(\ref{eq:E(etilde)-chap2})-(\ref{eq:a-chap2}) to obtain $\mathbb{E}[\etilde^2(k)]$. We can now use the zero mean Gaussian distribution with variance $\sigma_e^2=\mathbb{E}[\etilde^2(k)]+\sigma_n^2$ as the distribution of the error signal.\symbl{$\sigma_e^2$}{Variance of the error signal} Applying this distribution to Equation~(\ref{eq:single_gamma-chap2}), we can obtain a better estimation for $\gammabar$ through the equation
\begin{align}
\int_{\gammabar}^{\infty}\frac{1}{\sqrt{2\pi\sigma_e^2}}\exp(-\frac{r^2}{2\sigma_e^2})dr=\frac{p}{2}. \label{eq:second_step}
\end{align}
By using the standard normal distribution table, from where we can find the new estimation of $\gammabar$. It is worth mentioning
that the chosen desired update rate determines a loose relative
importance of the innovation brought about by the new
incoming data set.


\section{Conclusions} \label{sec:conclusion-chap2}

In this chapter, we have reviewed some adaptive filtering algorithms which play an essential role in the following chapters. First, we have introduced the LMS\abbrev{LMS}{Least-Mean-Square}, the NLMS\abbrev{NLMS}{Normalized LMS}, the AP\abbrev{AP}{Affine Projection}, and the RLS\abbrev{RLS}{Recursive Least-Squares} algorithms. Then, we have described the SMF\abbrev{SMF}{Set-Membership Filtering} approach. By incorporating this strategy into the conventional algorithms, we implement an update when the magnitude of the output estimation error is greater than the predetermined positive constant. For this purpose, we have defined some of the involved sets such as the feasibility set, the constraint set, and the exact membership set. Then, we have described the SM-NLMS\abbrev{SM-NLMS}{Set-Membership Normalized LMS} and the SM-AP\abbrev{SM-AP}{Set-Membership Affine Projection} algorithms. Finally, for the SM-NLMS\abbrev{SM-NLMS}{Set-Membership Normalized LMS} and the SM-AP\abbrev{SM-AP}{Set-Membership Affine Projection} algorithms, we have discussed how to estimate the threshold parameter in big data applications to obtain the desired update rate.

  \chapter{On the Robustness of the Set-Membership Algorithms}

Online learning algorithms are a substantial part of Adaptive Signal Processing, thus the efficiency of the algorithms has to be assessed. The classical adaptive filtering algorithms are iterative estimation methods based on the {\it point estimation theory}~\cite{Lehmann_pointEstimation_book2003}. 
This theory focuses on searching for a unique solution that minimizes (or maximizes) some objective function. 
Two widely used classical algorithms are the normalized least-mean-square (NLMS)\abbrev{NLMS}{Normalized LMS} and the affine projection (AP)\abbrev{AP}{Affine Projection} algorithms. 
These algorithms present a trade-off between convergence rate and steady-state misadjustment, and their properties 
have been extensively studied~\cite{Diniz_adaptiveFiltering_book2013,Sayed_adaptiveFilters_book2008}. 

Two important set-membership (SM)\abbrev{SM}{Set-Membership} algorithms are the set-membership NLMS (SM-NLMS) and the set-membership AP (SM-AP)\abbrev{SM-AP}{Set-Membership Affine Projection} algorithms, proposed 
in~\cite{Gollamudi_smf_letter1998,Werner_sm_ap_letter2001}, respectively. 
These algorithms keep the advantages of their classical counterparts, but they are more accurate, more robust against noise, and 
also reduce the computational complexities due to the data selection strategy previously 
explained~\cite{Markus_mseSMAP_cssp2013,Diniz_adaptiveFiltering_book2013,Arablouei_tracking_performance_SMNLMS_APSIPA2012,Carini_Filtered_x_SMAP_icassp2006}. 
Various applications of SM\abbrev{SM}{Set-Membership} algorithms and their advantages over the classical algorithms have been discussed in the 
literature~\cite{Gollamudi_smUpdatorShared_tsp1998,Nagaraj_beacon_tsp1999,Guo_fsmf_tsp2007,Diniz_sm_pap_jasmp2007,Markus_semiblindQAM_spawc2008,Bhotto_2012_TSP,Zhang_robustSMnlms_tcas2014,Mao_smfGPS_sensors2017}.

Despite the recognized advantages of the SM\abbrev{SM}{Set-Membership} algorithms, they are not broadly used, probably 
due to the limited analysis of the properties of these algorithms. 
The steady-state mean-squared error (MSE)\abbrev{MSE}{Mean-Squared Error} analysis of the SM-NLMS\abbrev{SM-NLMS}{Set-Membership Normalized LMS} algorithm has been discussed 
in~\cite{Markus_mseSMNLMS_iswcs2010,Yamada_sm-nlmsAnalysis_tsp2009}. 
Also, the steady-state MSE\abbrev{MSE}{Mean-Squared Error} performance of the SM-AP\abbrev{SM-AP}{Set-Membership Affine Projection} algorithm has been analyzed 
in~\cite{Diniz_CSSP_2011,Markus_mseSMAP_cssp2013,Markus_mseSMAP_icassp2010}. 

The content of this chapter was published in~\cite{Hamed_robustnessSMNLMS_sam2016,Hamed_robustnessSM_EURASIP2017}. In this chapter, the robustness of the SM-NLMS \abbrev{SM-NLMS}{Set-Membership Normalized LMS}and the SM-AP\abbrev{SM-AP}{Set-Membership Affine Projection} algorithms are discussed in the sense of 
$l_2$ stability~\cite{Sayed_adaptiveFilters_book2008,Rupp_PAProbustness_tsp2011}. For the SM-NLMS\abbrev{SM-NLMS}{Set-Membership Normalized LMS} algorithm, we demonstrate that it is robust regardless the choice of its parameters and that 
the SM-NLMS\abbrev{SM-NLMS}{Set-Membership Normalized LMS} enhances the parameter estimation in most of the iterations in which an update occurs, two advantages 
over the classical NLMS \abbrev{NLMS}{Normalized LMS}algorithm.
Moreover, we also prove that if the noise bound is known, then we can set the SM-NLMS\abbrev{SM-NLMS}{Set-Membership Normalized LMS} so that it  
never degrades the estimate.
As for the SM-AP\abbrev{SM-AP}{Set-Membership Affine Projection} algorithm, we demonstrate that its robustness depends on a judicious choice of one of its parameters: 
the constraint vector (CV)\abbrev{CV}{Constraint Vector}. 
We prove the existence of CVs\abbrev{CV}{Constraint Vector} satisfying the robustness condition, but practical choices remain unknown. 
We also demonstrate that both the SM-AP\abbrev{SM-AP}{Set-Membership Affine Projection} and the SM-NLMS\abbrev{SM-NLMS}{Set-Membership Normalized LMS} algorithms do not diverge, even when their parameters are selected naively, provided that the additional noise is bounded. Section~\ref{sec:robustness-criterion} describes the robustness criterion. Section \ref{sec:discussed-algrithms-robustness} presents the algorithms discussed in this chapter. The robustness of the SM-NLMS\abbrev{SM-NLMS}{Set-Membership Normalized LMS} algorithm is studied in Section~\ref{sec:robustness-sm-nlms}, where we also discuss the cases in which the 
noise bound is assumed known and unknown. Section~\ref{sec:robustness-sm-ap} presents the local and the global robustness properties of the SM-AP\abbrev{SM-AP}{Set-Membership Affine Projection} algorithm. Section~\ref{sec:simulation-robustness} contains the simulations and numerical results. Finally, concluding remarks are drawn in Section~\ref{sec:conclusion-robustness}.


\section{Robustness Criterion}\label{sec:robustness-criterion}

At every iteration $k$, assume that the desired signal $d(k)$ is related to the unknown system $\wbf_o$ by
\begin{align}
d(k) \triangleq \underbrace{\wbf_o^T \xbf(k)}_{\triangleq y_o(k)} + n(k),
\end{align}
where $n(k)$ denotes the unknown noise and accounts for both measurement noise and modeling uncertainties or errors. 
Also, we assume that the unknown noise sequence $\{ n(k) \}$ has finite energy~\cite{Sayed_adaptiveFilters_book2008}, i.e.,
\begin{align}
\sum_{k=0}^j |n(k)|^2<\infty,\qquad {\rm for~all~} j. \label{eq:noise-condtion}
\end{align}
Suppose that we have a sequence of desired signals $\{d(k)\}$ and we intend to estimate $y_o(k)=\wbf_o^T\xbf(k)$. 
For this purpose, assume that $\hat{y}_{k|k}$ is an estimate of $y_o(k)$ and it is only dependent on $d(j)$ for $j=0,\cdots,k$. 
For a given positive number $\eta$, we aim at calculating the following estimates
$\hat{y}_{k|k} \in \{\hat{y}_{0|0},\hat{y}_{1|1},\cdots,\hat{y}_{M|M}\}$, 
such that for any $n(k)$ satisfying~\eqref{eq:noise-condtion} and any $\wbf_o$, the following criterion is satisfied:
\begin{align}
\frac{\sum\limits_{k=0}^j \|\hat{y}_{k|k}-y_o(k)\|^2}{\wbftilde^T(0)\wbftilde(0)+\sum_{k=0}^j|n(k)|^2}<\eta^2, \qquad {\rm for~all~} j=0,\cdots,M \label{eq:criterion}
\end{align}
where $\wbftilde(0) \triangleq \wbf_o-\wbf(0)$ and $\wbf(0)$ is our initial guess about $\wbf_o$. 
Note that the numerator is a measure of estimation-error energy up to iteration $j$ and the denominator includes 
the energy of disturbance up to iteration $j$ and the energy of the error $\wbftilde(0)$ that is due to the initial guess.

So, the criterion given in~\eqref{eq:criterion} requires that we adjust estimates $\{\hat{y}_{k|k}\}$ such that
the ratio of the estimation-error energy (numerator) to the energy of the uncertainties (denominator) does not exceed $\eta^2$.
When this criterion is satisfied, we say that bounded disturbance energies induce bounded estimation-error energies and, 
therefore, the obtained estimates are robust. The smaller value of $\eta$ results in the more robust solution, but the value of $\eta$ cannot be decreased freely.
The interested reader can refer to~\cite{Sayed_adaptiveFilters_book2008}, pages 719 and 720, for more details about this robustness criterion.


\section{The Set-Membership Algorithms} \label{sec:discussed-algrithms-robustness}

In this section, we remind the SM-NLMS\abbrev{SM-NLMS}{Set-Membership Normalized LMS} and the SM-AP\abbrev{SM-AP}{Set-Membership Affine Projection} algorithms, and in the following sections we deal with their robustness.


\subsection{The SM-NLMS Algorithm} \label{subsec:sm_nlms}

The SM-NLMS\abbrev{SM-NLMS}{Set-Membership Normalized LMS} algorithm is characterized by the updating rule~\cite{Diniz_adaptiveFiltering_book2013}
\begin{align}
 \wbf(k+1) = \wbf(k) + \frac{\mu(k)}{ \| \xbf(k) \|^2 + \delta } e(k) \xbf(k) ,  \label{eq:sm-nlms-robustness}
\end{align}
where 
\begin{align}  \label{eq:def_mu}
 \mu(k) \triangleq \left\{ \begin{matrix}  1 - \frac{\gammabar}{|e(k)|}  &  \text{if } |e(k)| > \gammabar   ,  \\ 
                                  0                             &  \text{otherwise}  , \end{matrix}  \right.  
\end{align}
and $\gammabar \in \mathbb{R}_+$ is the upper bound for the magnitude of the error signal 
that is acceptable and it is usually chosen as a multiple of the noise standard deviation $\sigma_n$~\cite{Markus_mseSMAP_cssp2013,Diniz_adaptiveFiltering_book2013}. The parameter $\delta \in \mathbb{R}_+$ is a regularization factor, usually chosen as 
a small constant, used to avoid singularity (divisions by $0$).


\subsection{The SM-AP Algorithm} \label{subsec:sm_ap}

The SM-AP\abbrev{SM-AP}{Set-Membership Affine Projection} algorithm is described by the recursion~\cite{Werner_sm_ap_letter2001}
\begin{align}
\wbf(k+1)=\left\{\begin{array}{ll}\wbf(k)+\Xbf(k)\Abf(k)(\ebf(k)-\gammabf(k))&\text{if}~|e(k)|>\gammabar ,
\\ \wbf(k)&\text{ otherwise,}\end{array}\right.   \   \label{eq:sm-ap}
\end{align}
where we assume that $\Abf(k)\triangleq(\Xbf^T(k)\Xbf(k))^{-1} \in\mathbb{R}^{(L+1)\times (L+1)}$ exists, i.e., $\Xbf^T(k)\Xbf(k)$ is a full-rank matrix.
Otherwise, we could add a regularization parameter as explained in~\cite{Diniz_adaptiveFiltering_book2013}.


\section{Robustness of the SM-NLMS Algorithm}\label{sec:robustness-sm-nlms}

In this section, we discuss the robustness of the set-membership NLMS (SM-NLMS) algorithm.
In Subsection~\ref{sub:robustness-sm-nlms}, we present some 
robustness properties.
We address the robustness of the SM-NLMS\abbrev{SM-NLMS}{Set-Membership Normalized LMS} algorithm for the cases of unknown noise bound and known noise bound in 
Subsections~\ref{sub:sm-nlms-unbounded-noise} and~\ref{sub:sm-nlms-bounded-noise}, respectively.
Then, in Subsection~\ref{sub:sm-nlms-time-varying-gammabar}, we introduce a time-varying error bound {aiming at achieving 
simultaneously fast convergence, low computational burden, and efficient use of the input data}.

\subsection{Robustness of {the} SM-NLMS algorithm}\label{sub:robustness-sm-nlms}

Let us consider a system identification scenario in which the unknown system is denoted by  
$\wbf_o \in \mathbb{R}^{N+1}$ and the desired (reference) signal $d(k)$ is defined as 
\begin{align}
 d(k) \triangleq \wbf_o^T \xbf(k) + n(k) ,         \label{eq:desiredSignalModel} 
\end{align}
where $n(k) \in \mathbb{R}$ represents the additive measurement noise. 

One of the main difficulties of analyzing the SM-NLMS\abbrev{SM-NLMS}{Set-Membership Normalized LMS} algorithm is its conditional statement in~\eqref{eq:def_mu}. 
We can overcome such difficulty by defining \symbl{$\mubar(k)$}{Auxiliary step size $\mubar(k)\triangleq1-\frac{\gammabar}{|e(k)|}$} \symbl{$f(e(k),\gammabar)$}{The indicator function: returns 1 if $|e(k)|>\gammabar$, otherwise returns 0} \symbl{$\alpha(k)$}{Auxiliary value $\alpha(k)\triangleq\|\xbf(k)\|^2+\delta$}
\begin{align}
 \mubar(k) \triangleq 1 - \frac{\gammabar}{|e(k)|},   \label{eq:def_mubar}
\end{align}
and the indicator function $f:\mathbb{R}\times\mathbb{R}_+ \rightarrow \{ 0,1 \}$ as 
\begin{align}
 f(e(k),\gammabar) \triangleq \left\{ \begin{matrix}
                               1   &   \text{if } |e(k)| > \gammabar  ,  \\ 
                               0   &   \text{otherwise}  .     
                               \end{matrix}  \right.   \label{eq:def_indicatorFunc}
\end{align}
In this way, the SM-NLMS\abbrev{SM-NLMS}{Set-Membership Normalized LMS} updating rule can be rewritten as 
\begin{align}
 \wbf(k+1) = \wbf(k) + \frac{\mubar(k)}{\alpha(k)} e(k) \xbf(k) f(e(k),\gammabar) , \label{eq:sm-nlms_indicator}
\end{align}
where 
\begin{align}
 \alpha(k) \triangleq  \| \xbf(k) \|^2 + \delta .  \label{eq:def_alpha}
\end{align}

Since we are interested in robustness properties, it is useful to define $\wbftilde(k) \in \mathbb{R}^{N+1}$ as \symbl{$\wbftilde(k)$}{Auxiliary vector $\wbftilde(k)\triangleq\wbf_o-\wbf(k)$}
\begin{align}
 \wbftilde(k) \triangleq \wbf_o - \wbf(k)  ,    \label{eq:def_wbftilde}
\end{align}
i.e., $\wbftilde(k)$ is a vector representing the discrepancy between the quantity we aim to estimate $\wbf_o$ and 
our current estimate $\wbf(k)$. Thus, the error signal can be rewritten as 
\begin{align}
 e(k) = d(k) - \wbf^T(k) \xbf(k) &= \wbf_o^T \xbf(k) + n(k) - \wbf^T(k) \xbf(k) \nonumber\\
  &= \underbrace{\wbftilde^T(k) \xbf(k)}_{\triangleq \etilde(k)} + n(k) ,       \label{eq:def_noiselessError-sm-nlms}
\end{align}
where $\etilde(k)$ denotes the noiseless error, i.e., the error due to a mismatch between $\wbf(k)$ and $\wbf_o$. 

By using~\eqref{eq:def_wbftilde} in~\eqref{eq:sm-nlms_indicator} we obtain
\begin{align}
 \wbftilde(k+1) = \wbftilde(k) - \frac{\mubar(k)}{\alpha(k)} e(k) \xbf(k) f(e(k),\gammabar) , 
\end{align}
which can be further expanded by decomposing $e(k)$ as in Equation~\eqref{eq:def_noiselessError-sm-nlms} yielding 
\begin{align}
 \wbftilde(k+1) = \wbftilde(k) - \frac{\mubar(k)}{\alpha(k)} \etilde(k) \xbf(k) f(e(k),\gammabar)  
                               - \frac{\mubar(k)}{\alpha(k)} n(k) \xbf(k) f(e(k),\gammabar) . \label{eq:robust_aux01}
\end{align}

By computing the energy of~\eqref{eq:robust_aux01}, the robustness property given in Theorem~\ref{thm:local_robustness-sm-nlms} can be derived after some mathematical manipulations.

\begin{thm}[Local Robustness of SM-NLMS]\label{thm:local_robustness-sm-nlms}
 For the SM-NLMS\abbrev{SM-NLMS}{Set-Membership Normalized LMS} algorithm, it always holds that
 \begin{align}
  \| \wbftilde (k+1) \|^2 = \| \wbftilde (k) \|^2 , \text{ if } f(e(k),\gammabar) = 0     \label{eq:local_robustness_f0}
 \end{align}
 or
 \begin{align}
\| \wbftilde(k+1) \|^2 + \frac{\mubar(k)}{\alpha(k)} \etilde^2(k) 
 < \| \wbftilde(k) \|^2 + \frac{\mubar(k)}{\alpha(k)} n^2(k)     \ ,                      \label{eq:local_robustness_f1}
\end{align}
if $f(e(k),\gammabar) = 1$.
\end{thm}
\begin{proof}
We start by repeating Equation~\eqref{eq:robust_aux01}, but to simplify the notation we will omit both the index $k$ and the arguments of function $f$ that appear on the right-hand side of that equation to obtain
\begin{align}
\wbftilde(k+1) = \wbftilde  - \frac{\mubar}{\alpha} \etilde \xbf f 
                              - \frac{\mubar}{\alpha} n \xbf f   . 
\end{align}
By computing the Euclidean norm of the above equation we get
\begin{align}
 \| \wbftilde(k+1) \|^2 
 =& \wbftilde^T \wbftilde - \frac{\mubar}{\alpha} \etilde \wbftilde^T \xbf f - \frac{\mubar}{\alpha} n \wbftilde^T \xbf f -\frac{\mubar}{\alpha} \etilde \xbf^T \wbftilde f + \frac{\mubar^2}{\alpha^2} \etilde^2 \xbf^T \xbf f^2 \nonumber \\
 &+ \frac{\mubar^2}{\alpha^2} \etilde n \xbf^T \xbf f^2 -\frac{\mubar}{\alpha} n \xbf^T \wbftilde f + \frac{\mubar^2}{\alpha^2} n \etilde \xbf^T \xbf f^2 + \frac{\mubar^2}{\alpha^2} n^2 \xbf^T \xbf f^2  \nonumber \\
 =& \| \wbftilde \|^2 - \frac{\mubar}{\alpha} \etilde^2 f - \frac{\mubar}{\alpha} n \etilde f -\frac{\mubar}{\alpha} \etilde^2 f + \frac{\mubar^2}{\alpha^2} \etilde^2 \| \xbf \|^2 f^2 + \frac{\mubar^2}{\alpha^2} \etilde n \| \xbf \|^2 f^2 \nonumber \\
  &-\frac{\mubar}{\alpha} n \etilde f + \frac{\mubar^2}{\alpha^2} n \etilde \| \xbf \|^2 f^2 + \frac{\mubar^2}{\alpha^2} n^2 \| \xbf \|^2 f^2  \nonumber \\
 =& \| \wbftilde \|^2 - 2 \frac{\mubar}{\alpha} \etilde^2 f - 2 \frac{\mubar}{\alpha} n \etilde f +(\etilde + n)^2 \frac{\mubar^2}{\alpha^2}  \| \xbf \|^2 f^2                                     \nonumber \\
 =& \| \wbftilde \|^2  +  (\etilde + n)^2 \frac{\mubar^2}{\alpha^2}  \| \xbf \|^2 f^2  - 2 \frac{\mubar}{\alpha} \etilde^2 f - 2 \frac{\mubar}{\alpha} n \etilde f     - \frac{\mubar}{\alpha} n^2 f + \frac{\mubar}{\alpha} n^2 f      \nonumber \\
 =& \| \wbftilde \|^2  +  (\etilde + n)^2 \frac{\mubar^2}{\alpha^2}  \| \xbf \|^2 f^2  + \frac{\mubar}{\alpha} n^2 f - (\etilde + n)^2 \frac{\mubar}{\alpha} f  -  \frac{\mubar}{\alpha} \etilde^2 f   ,   \label{eq:robustness_derivation_1}
\end{align}
where the second equality is due to the relation $\etilde = \wbftilde^T \xbf = \xbf^T \wbftilde$, as given  
in Equation~\eqref{eq:def_noiselessError-sm-nlms}. Rearranging the terms in~\eqref{eq:robustness_derivation_1} yields
\begin{align}
 \| \wbftilde(k+1) \|^2 + \frac{\mubar f}{\alpha} \etilde^2 
 = \| \wbftilde \|^2 + \frac{\mubar f}{\alpha} n^2 + c_1 c_2  ,   \label{eq:energy_relation}
\end{align}
where 
\begin{align}
 c_1 \triangleq \frac{\mubar f}{\alpha} (\etilde + n)^2     ,\qquad
 c_2 \triangleq \frac{\mubar f}{\alpha} \| \xbf \|^2 - 1    .
\end{align} 

From~\eqref{eq:energy_relation}, we observe that when $f = 0$ we have
\begin{align}
 \| \wbftilde(k+1) \|^2 = \| \wbftilde(k) \|^2 
\end{align}
as expected, since $f=0$ means that no update was performed. 
However, when $f=1$ we have $0 < \mubar < 1$ and $(\etilde + n)^2 = e^2 > \gammabar^2 > 0$. 
In addition, observe that $0 \leq \| \xbf \|^2/\alpha < 1$ due to Equation~\eqref{eq:def_alpha} and the fact that $\delta>0$. 
Combining these inequalities leads to $c_2 < 0$ and $c_1 > 0$.
Thus, when $f=1$ the product $c_1 c_2 < 0$, which leads to the inequality 
\begin{align}
 \| \wbftilde(k+1) \|^2 + \frac{\mubar}{\alpha} \etilde^2 
 < \| \wbftilde \|^2 + \frac{\mubar}{\alpha} n^2      .  
\end{align}
Returning with the omitted index $k$, for $f(e(k),\gammabar)=1$ we have
\begin{align}
 \| \wbftilde(k+1) \|^2 + \frac{\mubar(k)}{\alpha(k)} \etilde^2(k) 
 < \| \wbftilde(k) \|^2 + \frac{\mubar(k)}{\alpha(k)} n^2(k)     .   
\end{align}
\end{proof}

Theorem~\ref{thm:local_robustness-sm-nlms} presents local bounds for the energy of the coefficient deviation when running from 
an iteration to the next one. 
Indeed, \eqref{eq:local_robustness_f0} states that the coefficient deviation does not change when no coefficient update is actually implemented, whereas~\eqref{eq:local_robustness_f1} {provides} a bound for $\| \wbftilde(k+1) \|^2$ 
based on $\| \wbftilde(k) \|^2$, $\etilde^2(k)$, and $n^2(k)$, when an update occurs.
In addition, the global robustness result in Corollary~\ref{thm:global_robustness-sm-nlms} can readily be derived 
from Theorem~\ref{thm:local_robustness-sm-nlms}.\symbl{${\cal K}_{\rm up}$}{Set containing the iteration indexes in which $\wbf(k)$ is updated}

\begin{cor}[Global Robustness of SM-NLMS]\label{thm:global_robustness-sm-nlms}
Consider the SM-NLMS\abbrev{SM-NLMS}{Set-Membership Normalized LMS} algorithm running from iteration $0$ (initialization) to a given iteration $K$.
The relation 
\begin{align}
\dfrac{\| \wbftilde(K) \|^2 + \sum\limits_{k \in {\cal K}_{\rm up}}\frac{\mubar(k)}{\alpha(k)}\etilde^2(k)}{\| \wbftilde(0) \|^2 + 
\sum\limits_{k \in {\cal K}_{\rm up}}\frac{\mubar(k)}{\alpha(k)}n^2(k)}  < 1   \label{eq:global_robustness-sm-nlms}
\end{align}
holds, where ${\cal K}_{\rm up} \neq \emptyset$ is the set containing the iteration indexes $k$ in which $\wbf(k)$ is indeed updated. 
If ${\cal K}_{\rm up} = \emptyset$,\symbl{$\emptyset$}{Empty set} then $\| \wbftilde(K) \|^2 = \| \wbftilde(0) \|^2$ due to~\eqref{eq:local_robustness_f0}, 
but this case is not of practical interest since ${\cal K}_{\rm up} = \emptyset$ means that no update is performed at all. 
\end{cor}
\begin{proof}
Define the set of all iterations under analysis ${\cal K} \triangleq \{ 0, 1, 2, \ldots,$ $K-1 \}$. 
Denote as ${\cal K}_{\rm up}$ the subset of ${\cal K}$ comprised only of the iterations in which an update 
occurs, whereas ${\cal K}_{\rm up}^c \triangleq {\cal K} \setminus {\cal K}_{\rm up}$ contains the 
iteration indexes in which the filter coefficients are not updated. 

From Theorem~\ref{thm:local_robustness-sm-nlms}, \eqref{eq:local_robustness_f1} holds when $\wbf(k)$ is updated. 
By summing such inequality for all $k \in {\cal K}_{\rm up}$ we obtain 
\begin{align}
\sum_{k \in {\cal K}_{\rm up}} \Big(\| \wbftilde(k+1) \|^2 + \frac{\mubar(k)}{\alpha(k)} \etilde^2(k) \Big)  
< \sum_{k \in {\cal K}_{\rm up}} \Big(\| \wbftilde(k) \|^2 + \frac{\mubar(k)}{\alpha(k)} n^2(k)\Big). \label{eq:robustness_accumulation_f1}
\end{align}
Similarly, we can use~\eqref{eq:local_robustness_f0} to write, for all $k \in {\cal K}_{\rm up}^c$, the equality
\begin{align}
 \sum_{k \in {\cal K}_{\rm up}^c} \| \wbf(k+1) \|^2 = \sum_{k \in {\cal K}_{\rm up}^c} \| \wbf(k) \|^2. \label{eq:robustness_accumulation_f0}
\end{align}
Combining~\eqref{eq:robustness_accumulation_f1} and~\eqref{eq:robustness_accumulation_f0} leads to 
\begin{align}
\sum_{k \in {\cal K}} \| \wbftilde(k+1) \|^2 
+ \sum_{k \in {\cal K}_{\rm up}} \frac{\mubar(k)}{\alpha(k)} \etilde^2(k)  
< \sum_{k \in {\cal K}} \| \wbftilde(k) \|^2 
+ \sum_{k \in {\cal K}_{\rm up}} \frac{\mubar(k)}{\alpha(k)} n^2(k). \label{eq:robustness_accumulation-sm-nlms} 
\end{align}
But since several of the terms $\| \wbftilde(k) \|^2$ get canceled from both sides of the inequality \eqref{eq:robustness_accumulation-sm-nlms}, we find that it simplifies to
\begin{align}
\| \wbftilde(K) \|^2 + \sum_{k \in {\cal K}_{\rm up}}\frac{\mubar(k)}{\alpha(k)}\etilde^2(k)  
< \| \wbftilde(0) \|^2 + \sum_{k \in {\cal K}_{\rm up}}\frac{\mubar(k)}{\alpha(k)}n^2(k) 
\end{align}
or, assuming a nonzero denominator,
\begin{align}
 \dfrac{\| \wbftilde(K) \|^2 + \sum\limits_{k \in {\cal K}_{\rm up}}\frac{\mubar(k)}{\alpha(k)}\etilde^2(k)}{\| \wbftilde(0) \|^2 + 
 \sum\limits_{k \in {\cal K}_{\rm up}}\frac{\mubar(k)}{\alpha(k)}n^2(k)}  < 1      .
\end{align}
This relation holds for all $K$. The only assumption used in the derivation is that ${\cal K}_{\rm up}$ is a nonempty set.  Otherwise, we would have $\| \wbftilde(K) \|^2 = \| \wbftilde(0) \|^2$, which would happen only if $\wbf(k)$ is never updated, which has no practical interest. 
\end{proof}

Corollary~\ref{thm:global_robustness-sm-nlms} shows that, for the SM-NLMS\abbrev{SM-NLMS}{Set-Membership Normalized LMS} algorithm, $l_2$-stability from its uncertainties 
$\{ \wbftilde(0), \{ n(k) \}_{0\leq k\leq K} \}$ to its errors $\{ \wbftilde(K), \{ \etilde(k) \}_{0\leq k\leq K} \}$ 
is always guaranteed. 
Unlike the NLMS\abbrev{NLMS}{Normalized LMS} algorithm, in which the step-size parameter must be chosen appropriately to guarantee such $l_2$-stability, 
for the SM-NLMS\abbrev{SM-NLMS}{Set-Membership Normalized LMS} algorithm it is taken for granted (i.e., no restriction is imposed on $\gammabar$). 


\subsection{Convergence of $\{\|\wbftilde(k)\|^2\}$ with unknown  noise bound}\label{sub:sm-nlms-unbounded-noise}

The robustness results mentioned in Subsection~\ref{sub:robustness-sm-nlms} provide bounds for the evolution of  
$\{\|\wbftilde(k)\|^2\}$ in terms of other variables.
However, we have experimentally observed that the SM-NLMS\abbrev{SM-NLMS}{Set-Membership Normalized LMS} algorithm presents a well-behaved convergence of the 
sequence $\{\|\wbftilde(k)\|^2\}$, i.e.,
for most iterations we have $\|\wbftilde(k+1)\|^2 \leq \|\wbftilde(k)\|^2$.
Therefore, in this subsection, we investigate under which conditions the sequence $\{\|\wbftilde(k)\|^2\}$ 
is (and is not) decreasing.

\begin{cor}\label{cor:sm_nlms_decreasing}
 When an update occurs $($i.e., $f(e(k),\gammabar) = 1$ $)$, $\etilde^2(k) \geq n^2(k)$ implies $\| \wbftilde(k+1) \|^2  < \| \wbftilde(k) \|^2$.
\end{cor}
\begin{proof}
By rearranging the terms in~\eqref{eq:local_robustness_f1} we obtain
 \begin{align}
\| \wbftilde(k+1) \|^2 + \frac{\mubar(k)}{\alpha(k)} \left( \etilde^2(k) - n^2(k) \right) 
 < \| \wbftilde(k) \|^2,     
\end{align}
which is valid for $f(e(k),\gammabar) = 1$.
Observe that $\frac{\mubar(k)}{\alpha(k)} > 0$ since $\alpha(k) \in \mathbb{R}_+$ and $\mubar(k) \in (0,1)$ when $f(e(k),\gammabar) = 1$. 
Thus $\frac{\mubar(k)}{\alpha(k)} \left( \etilde^2(k) - n^2(k) \right)\geq0$ when $f(e(k),\gammabar) = 1$ and $\etilde^2(k) \geq n^2(k)$. 
Therefore, when an update occurs, $\etilde^2(k) \geq n^2(k)  \Rightarrow  \| \wbftilde(k+1) \|^2  < \| \wbftilde(k) \|^2$. 
\end{proof}

In words, Corollary~\ref{cor:sm_nlms_decreasing} states that the SM-NLMS\abbrev{SM-NLMS}{Set-Membership Normalized LMS} algorithm improves its estimate $\wbf(k+1)$ 
every time an update is required and the energy of the error signal $e^2(k)$ is dominated by $\etilde^2(k)$, the component of the error 
which is due to the mismatch between $\wbf(k)$ and $\wbf_o$.

Corollary~\ref{cor:sm_nlms_decreasing} also explains why the SM-NLMS\abbrev{SM-NLMS}{Set-Membership Normalized LMS} algorithm usually presents a {\it monotonic decreasing sequence} 
$\{\|\wbftilde(k)\|^2\}$ during its transient period. 
Indeed, in the early iterations, the absolute value of the error is generally large, thus $|e(k)|>\gammabar$ and $\etilde^2(k)>n^2(k)$, 
implying that $\| \wbftilde(k+1) \|^2  < \| \wbftilde(k) \|^2$.
In addition, there are a few iterations during the transient period in which the input data do not bring enough innovation so that 
no update is performed, which means that $\| \wbftilde(k+1) \|^2  = \| \wbftilde(k) \|^2$ for these few iterations. 
As a conclusion, it is very likely to have $\| \wbftilde(k+1) \|^2  \leq \| \wbftilde(k) \|^2$ for all iterations $k$ belonging 
to the transient period.

After the transient period, however, the SM-NLMS\abbrev{SM-NLMS}{Set-Membership Normalized LMS} algorithm may yield $\| \wbftilde(k+1) \|^2  > \| \wbftilde(k) \|^2$ in a few iterations. 
Although it is hard to compute how often such an event occurs, we can provide an upper bound for {the probability of this event} as follows: 
\begin{align}
\mathbb{P}[\|\wbftilde(k+1)\|^2 > \|\wbftilde(k)\|^2] &\leq \mathbb{P}[\{|e(k)|>\gammabar\}\cap\{\etilde^2(k)<n^2(k)\}]      \nonumber\\
&<\mathbb{P}[|e(k)|>\gammabar]={\rm erfc}\left(\sqrt{\frac{\tau}{2}}\right) ,  \label{eq:sm_nlms_probability}
\end{align}
where $\mathbb{P}[\cdot]$ and ${\rm erfc}(\cdot)$ are the probability operator and the complementary error 
function~\cite{Proakis_DigitalCommunications_book1995}, respectively. \symbl{${\rm erfc}(\cdot)$}{The complementary error function} 
The first inequality follows from the fact that we do not know exactly what will happen with $\| \wbftilde(k+1) \|^2$ when an update 
occurs and $\etilde^2(k)<n^2(k)$ at the same time\footnote{This is because Corollary~\ref{cor:sm_nlms_decreasing} provides a 
sufficient, but not necessary, condition for $\|\wbftilde(k+1)\|^2 < \|\wbftilde(k)\|^2$.} 
and, therefore, it corresponds to a {\it pessimistic bound}. 
The second inequality is trivial and the subsequent equality follows from~\cite{Galdino_SMNLMS_gammabar_ISCAS2006} by parameterizing $\gammabar$
as $\gammabar=\sqrt{\tau\sigma_n^2}$, where $\tau \in \mathbb{R}_+$ (typically $\tau = 5$) and by modeling the error $e(k)$ 
as a zero-mean Gaussian random variable with variance $\sigma_n^2$.

From~\eqref{eq:sm_nlms_probability}, one can observe that the probability of obtaining 
$\|\wbftilde(k+1)\|^2 > \|\wbftilde(k)\|^2$ is small. 
For instance, for $2\leq\tau\leq9$ we have $0.0027\leq{\rm erfc}\Big(\sqrt{\frac{\tau}{2}}\Big)\leq0.1579$, and 
for the usual choice $\tau=5$ we have ${\rm erfc}\Big(\sqrt{\frac{\tau}{2}}\Big)=0.0253$.

The results in this subsection show that $\| \wbftilde(k+1) \|^2  \leq \| \wbftilde(k) \|^2$ for most iterations of  
the SM-NLMS\abbrev{SM-NLMS}{Set-Membership Normalized LMS} algorithm, meaning that the SM-NLMS\abbrev{SM-NLMS}{Set-Membership Normalized LMS} algorithm uses the input data efficiently.
Indeed, having $\| \wbftilde(k+1) \|^2  > \| \wbftilde(k) \|^2$ means that the input data was used to obtain an 
estimate $\wbf(k+1)$ which is further away from the quantity we aim to estimate $\wbf_o$, which is a waste of computational resources 
(it would be better not to update at all). 
Here, we showed that this rarely happens for the SM-NLMS\abbrev{SM-NLMS}{Set-Membership Normalized LMS} algorithm, {a property not shared by} the classical algorithms, 
as it will be verified experimentally in Section~\ref{sec:simulation-robustness}.


\subsection{Convergence of $\{\|\wbftilde(k)\|^2\}$ with known noise bound}\label{sub:sm-nlms-bounded-noise}

In this subsection, we demonstrate that if the noise bound is known, then it is possible to set the threshold parameter $\gammabar$ 
of the SM-NLMS\abbrev{SM-NLMS}{Set-Membership Normalized LMS} algorithm so that $\{\|\wbftilde(k)\|^2\}$ is a monotonic decreasing sequence. 
Theorem~\ref{thm:strong-local-robustness-sm-nlms} and Corollary~\ref{cor:strong-global-robustness-sm-nlms} address 
this issue.

\begin{thm}[Strong Local Robustness of SM-NLMS]\label{thm:strong-local-robustness-sm-nlms}
 Assume the noise is bounded by a known constant $B \in \mathbb{R}_+$, i.e., $|n(k)|\leq B, \forall k$. 
 If one chooses $\gammabar \geq 2B$, then $\{\|\wbftilde(k)\|^2\}$ is a monotonic decreasing sequence, i.e.,   
 $\|\wbftilde(k+1)\|^2\leq\|\wbftilde(k)\|^2,\forall k$. 
\end{thm}
\begin{proof}
 If $f(e(k),\gammabar)=1$, then $|e(k)| = |\etilde(k) + n(k)|>\gammabar$, which means that: 
 (i)~$\etilde(k) > \gammabar - n(k)$ for the positive values of $\etilde(k)$ or 
 (ii)~$\etilde(k) < -\gammabar - n(k)$ for the negative values of $\etilde(k)$. 
 Recalling that $n(k) \in [-B,B]$ and $\gammabar \in [2B,\infty)$, now we can find the bound for $\etilde(k)$ by finding the 
 minimum of (i) and the maximum of (ii) as follows: \\
  (i) $\etilde(k) > \gammabar - n(k) \Rightarrow  \etilde_{\rm min} > \gammabar - B \geq B$; \\
 (ii) $\etilde(k) <-\gammabar - n(k) \Rightarrow  \etilde_{\rm max} <-\gammabar + B \leq -B$. \\
 Results (i) and (ii) above state that if $\gammabar \geq 2B$, then $| \etilde(k) | > B$, 
 which means that $| \etilde(k) | > | n(k) |, \forall k$.
 Consequently, by using Corollary~\ref{cor:sm_nlms_decreasing} it follows that $\|\wbftilde(k+1)\|^2 < \|\wbftilde(k)\|^2,\forall k$ in 
 which $f(e(k),\gammabar)=1$. 
 In addition, if $f(e(k),\gammabar)=0$ we have $\|\wbftilde(k+1)\|^2 = \|\wbftilde(k)\|^2$. 
 Therefore, we can conclude that $\gammabar \geq 2B \Rightarrow \|\wbftilde(k+1)\|^2\leq\|\wbftilde(k)\|^2,\forall k$.  
 \end{proof}
 
 \begin{cor}[Strong Global Robustness of SM-NLMS]\label{cor:strong-global-robustness-sm-nlms}
 Consider the SM-NLMS\abbrev{SM-NLMS}{Set-Membership Normalized LMS} algorithm running from iteration $0$ (initialization) to a given iteration $K$. 
 If $\gammabar \geq 2B$, then $\|\wbftilde(K)\|^2 \leq \|\wbftilde(0)\|^2$, in which the equality 
 holds only when no update is performed along all the iterations. 
 \end{cor}

 The proof of Corollary~\ref{cor:strong-global-robustness-sm-nlms} is omitted because it is a straightforward consequence 
 of Theorem~\ref{thm:strong-local-robustness-sm-nlms}.
 
 
 \subsection{Time-varying $\gammabar(k)$} \label{sub:sm-nlms-time-varying-gammabar}
 
 After reading Subsections~\ref{sub:sm-nlms-unbounded-noise} and~\ref{sub:sm-nlms-bounded-noise}, one might be 
 tempted to set $\gammabar$ as a high value since it reduces the number of updates, thus saving computational resources 
 {and also leading} to a well-behaved sequence $\{ \|\wbftilde(k)\|^2 \}$ that has high probability of being monotonously decreasing.
 However, a high value of $\gammabar$ leads to slow convergence, because the updates during the learning stage (transient period) are 
 less frequent and the step-size $\mu(k)$ is reduced as well. 
 Hence, $\gammabar$ represents a compromise between convergence speed and efficiency and, therefore, 
 should be chosen carefully according to the specific characteristics of the application.
 
 An alternative approach is to allow a time-varying error bound $\gammabar(k)$ generally defined as 
 $\gammabar(k) \triangleq \sqrt{\tau(k) \sigma_n^2}$, where \symbl{$\gammabar(k)$}{Time-varying error bound}
 \begin{align}\label{eq:gammabar-timevar}
  \tau(k) \triangleq\begin{cases}
               \text{Low value (e.g., $\tau(k) \in [1,5]$)}, \qquad\text{if $k \in$ transient period, }   \\
               \text{High value (e.g., $\tau(k) \in [5,9]$)}, \qquad\text{if $k \in $ steady-state.}
              \end{cases}
 \end{align}
 By using such a $\gammabar(k)$, one obtains the best features of the high and low values of $\gammabar$ discussed in the 
 first paragraph of this subsection. 
 In addition, if the noise bound $B$ is known, then one should set $\gammabar(k)\geq 2B$ for all $k$ during the steady-state, 
 as explained in Subsection~\ref{sub:sm-nlms-bounded-noise}.
 It is worth mentioning that~\eqref{eq:gammabar-timevar} provides a general expression for $\tau(k)$ that allows it to vary smoothly 
 along the iterations even within a single period (i.e., transient period or steady-state).
 
 In order to apply the $\gammabar(k)$ defined above, the algorithm should be able to monitor the environment to determine 
 when there is a transition between transient and steady-state periods.
 An intuitive way to do this is to monitor the values of $|e(k)|$. 
 In this case, one should form a window with the $E \in \mathbb{N}$ most recent values of the error, compute the average 
 of these $|e(k)|$ within the window, and compare it against a threshold parameter to make the decision.
 An even more intuitive and efficient way to monitor the iterations relies on how often the algorithm is updating. 
 In this case, one should form a window of length $E$ containing Boolean variables (flags, i.e., 1-bit information) indicating the iterations 
 in which an update was performed considering the $E$ most recent iterations.
 If many updates were performed within the window, then the algorithm must be in the transient period; otherwise, the algorithm is likely to be 
 in steady-state.
 

\section{Robustness of the SM-AP Algorithm} \label{sec:robustness-sm-ap}

In this section, we address the robustness of the set-membership affine projection (SM-AP) algorithm. We study its 
robustness properties in Subsection~\ref{sub:robustness-sm-ap}, whereas in Subsection~\ref{sub:divergence_sm_ap}, we demonstrate that the SM-AP\abbrev{SM-AP}{Set-Membership Affine Projection} algorithm does not diverge.


\subsection{Robustness of the SM-AP algorithm}\label{sub:robustness-sm-ap}

Suppose that in a system identification problem the unknown system is denoted by $\wbf_o\in\mathbb{R}^{N+1}$ 
and the desired (reference) vector is given by
\begin{align}
 \dbf(k) \triangleq \Xbf^T(k) \wbf_o + \nbf(k) .         \label{eq:desiredSignalModel-SM-AP} 
\end{align}
By defining the coefficient mismatch $\wbftilde(k)\triangleq\wbf_o-\wbf(k)$, 
the error vector can be written as
\begin{align}
\ebf(k)=\Xbf^T(k)\wbf_o+\nbf(k)-\Xbf^T(k)\wbf(k)=\underbrace{\Xbf^T(k)\wbftilde(k)}_{\triangleq \ebftilde(k)}+\nbf(k) \ , \label{eq:def_noiselessError-SM-AP}
\end{align}
where $\ebftilde(k)$ denotes the noiseless error vector (i.e., the error due to a nonzero $\wbftilde(k)$). \symbl{$\ebftilde(k)$}{Noiseless error signal vector}
By defining the indicator function $f:\mathbb{R}\times\mathbb{R}_+ \rightarrow \{ 0,1 \}$ as in~\eqref{eq:def_indicatorFunc} 
and using it in (\ref{eq:sm-ap}), the update rule of the SM-AP\abbrev{SM-AP}{Set-Membership Affine Projection} algorithm can be written as follows:
\begin{align}
\hspace{-0.1cm}\wbf(k+1)=\wbf(k)+\Xbf(k)\Abf(k) 
(\ebf(k)-\gammabf(k))f(e(k),\gammabar), \label{eq:r_update}
\end{align}
where $\Abf(k)=[\Xbf^T(k)\Xbf(k)]^{-1}$. After subtracting $\wbf_o$ from both sides of (\ref{eq:r_update}), we obtain
\begin{align}
\wbftilde(k+1)=\wbftilde(k)-\Xbf(k)\Abf(k)(\ebf(k)-\gammabf(k))f(e(k),\gammabar).
\end{align}
Notice that $\Abf(k)$ is a symmetric positive definite matrix. 
To simplify our notation, we will omit the index $k$ and the arguments of function $f$ that appear on the 
right-hand side (RHS)\abbrev{RHS}{Right-Hand Side} of the previous equation, then by decomposing $\ebf(k)$ as in~\eqref{eq:def_noiselessError-SM-AP} 
we obtain
\begin{align}
\wbftilde(k+1)=\wbftilde-\Xbf\Abf\ebftilde f-\Xbf\Abf\nbf f+\Xbf\Abf \gammabf f \label{eq:update-SM-AP} ,
\end{align}
from which Theorem~\ref{thm:local_robustness-SM-AP} can be derived.

\begin{thm}[Local Robustness of SM-AP]\label{thm:local_robustness-SM-AP}
 For the SM-AP\abbrev{SM-AP}{Set-Membership Affine Projection} algorithm, at every iteration we have
 \begin{align}
  \| \wbftilde (k+1) \|^2 = \| \wbftilde (k) \|^2 , \text{ if } f(e(k),\gammabar) = 0     \label{eq:local_robustness_f0-SM-AP}
 \end{align}
otherwise
 \begin{align}
 \left\{\begin{array}{ll}
 \frac{\|\wbftilde(k+1)\|^2+\ebftilde^T\Abf\ebftilde}{\|\wbftilde(k)\|^2+\nbf^T\Abf\nbf}<1,&\text{if}~\gammabf^T\Abf\gammabf<2\gammabf^T\Abf\nbf,\\
 \frac{\|\wbftilde(k+1)\|^2+\ebftilde^T\Abf\ebftilde}{\|\wbftilde(k)\|^2+\nbf^T\Abf\nbf}=1,&\text{if}~\gammabf^T\Abf\gammabf=2\gammabf^T\Abf\nbf,\\
 \frac{\|\wbftilde(k+1)\|^2+\ebftilde^T\Abf\ebftilde}{\|\wbftilde(k)\|^2+\nbf^T\Abf\nbf}>1,&\text{if}~\gammabf^T\Abf\gammabf>2\gammabf^T\Abf\nbf,\end{array}\right.  \label{eq:local_robustness_f1-SM-AP}
 \end{align}
 where the iteration index $k$ has been dropped for the sake of clarity, and we assume that 
 $\|\wbftilde(k)\|^2+\nbf^T\Abf\nbf\neq0$ just to allow us to write the theorem in a compact form.
\end{thm}
\begin{proof}
By computing the Euclidean norm of Equation~\eqref{eq:update-SM-AP} and rearranging the terms we get
\begin{align}
\|\wbftilde(k+1)\|^2=&\wbftilde^T\wbftilde-\wbftilde^T\Xbf\Abf\ebftilde f -\wbftilde^T\Xbf\Abf\nbf f +\wbftilde^T\Xbf\Abf\gammabf f -\ebftilde^T\Abf^T\Xbf^T\wbftilde f \nonumber\nonumber\\
&+\ebftilde^T\Abf^T\Abf^{-1}\Abf\ebftilde f^2 +\ebftilde^T\Abf^T\Abf^{-1}\Abf\nbf f^2 -\ebftilde^T\Abf^T\Abf^{-1}\Abf\gammabf f^2 \nonumber\\ &-\nbf^T\Abf^T\Xbf^T\wbftilde f
+\nbf^T\Abf^T\Abf^{-1}\Abf\ebftilde f^2 +\nbf^T\Abf^T\Abf^{-1}\Abf\nbf f^2 \nonumber\\ &-\nbf^T\Abf^T\Abf^{-1}\Abf\gammabf f^2 +\gammabf^T\Abf^T\Xbf^T\wbftilde f
-\gammabf^T\Abf^T\Abf^{-1}\Abf\ebftilde f^2 \nonumber\\ &-\gammabf^T\Abf^T\Abf^{-1}\Abf\nbf f^2
+\gammabf^T\Abf^T\Abf^{-1}\Abf\gammabf f^2 \nonumber\\
=&\|\wbftilde\|^2-\ebftilde^T\Abf\ebftilde f -\ebftilde^T\Abf\nbf f +\ebftilde^T\Abf\gammabf f -\ebftilde^T\Abf\ebftilde f+\ebftilde^T\Abf\ebftilde f^2 \nonumber\\
&+\ebftilde^T\Abf\nbf f^2 -\ebftilde^T\Abf\gammabf f^2 -\nbf^T\Abf\ebftilde f +\nbf^T\Abf\ebftilde f^2 +\nbf^T\Abf\nbf f^2 \nonumber\\
&-\nbf^T\Abf\gammabf f^2 +\gammabf^T\Abf\ebftilde f -\gammabf^T\Abf\ebftilde f^2 -\gammabf^T\Abf\nbf f^2 +\gammabf^T\Abf\gammabf f^2 \ , \label{eq:norm2-sm-ap-robustness}
\end{align}
where it was used that $\Abf^{-1} = \Xbf^T(k)\Xbf(k)$ and $\ebftilde(k) = \Xbf^T(k) \wbftilde(k)$.
From the above equation we observe that when $f=0$ we have
\begin{align}
\|\wbftilde(k+1)\|^2=\|\wbftilde(k)\|^2 \label{eq:equality}
\end{align}
as expected, since $f=0$ means that the algorithm does not update its coefficients. 
However, when $f=1$ the following equality is achieved from \eqref{eq:norm2-sm-ap-robustness}:
\begin{align}
\|\wbftilde(k+1)\|^2=\|\wbftilde\|^2-\ebftilde^T\Abf\ebftilde +\nbf^T\Abf\nbf-2\gammabf^T\Abf\nbf +\gammabf^T\Abf\gammabf \ . \label{eq:main_equation}
\end{align}
After rearranging the terms of the previous equation we obtain 
\begin{align}
\|\wbftilde(k+1)\|^2+\ebftilde^T\Abf\ebftilde=\|\wbftilde\|^2+\nbf^T\Abf\nbf-2\gammabf^T\Abf\nbf+\gammabf^T\Abf\gammabf \ . \label{eq:NiceIdentity}
\end{align}
Therefore,  
$\|\wbftilde(k+1)\|^2+\ebftilde^T\Abf\ebftilde<\|\wbftilde\|^2+\nbf^T\Abf\nbf$ if $\gammabf^T\Abf\gammabf<2\gammabf^T\Abf\nbf$,  
$\|\wbftilde(k+1)\|^2+\ebftilde^T\Abf\ebftilde=\|\wbftilde\|^2+\nbf^T\Abf\nbf$ if $\gammabf^T\Abf\gammabf=2\gammabf^T\Abf\nbf$, and 
$\|\wbftilde(k+1)\|^2+\ebftilde^T\Abf\ebftilde>\|\wbftilde\|^2+\nbf^T\Abf\nbf$ if $\gammabf^T\Abf\gammabf>2\gammabf^T\Abf\nbf$.

Assuming $\|\wbftilde\|^2+\nbf^T\Abf\nbf\neq0$ we can summarize the discussion above in a compact form as follows:
\begin{align}
\left\{\begin{array}{ll}\frac{\|\wbftilde(k+1)\|^2+\ebftilde^T\Abf\ebftilde}{\|\wbftilde(k)\|^2+\nbf^T\Abf\nbf}<1,&\text{if}~\gammabf^T\Abf\gammabf<2\gammabf^T\Abf\nbf,\\
\frac{\|\wbftilde(k+1)\|^2+\ebftilde^T\Abf\ebftilde}{\|\wbftilde(k)\|^2+\nbf^T\Abf\nbf}=1,&\text{if}~\gammabf^T\Abf\gammabf=2\gammabf^T\Abf\nbf,\\
\frac{\|\wbftilde(k+1)\|^2+\ebftilde^T\Abf\ebftilde}{\|\wbftilde(k)\|^2+\nbf^T\Abf\nbf}>1,&\text{if}~\gammabf^T\Abf\gammabf>2\gammabf^T\Abf\nbf.\end{array}\right. 
\end{align}
\end{proof}

The combination of the first two inequalities in~\eqref{eq:local_robustness_f1-SM-AP}, which corresponds to the 
case $\gammabf^T\Abf\gammabf \leq 2\gammabf^T\Abf\nbf$, has an interesting interpretation. 
It describes that for any constraint vector $\gammabf$ satisfying this condition we have
\begin{align}
\|\wbftilde(k+1)\|^2+\ebftilde^T\Abf\ebftilde \leq \|\wbftilde(k)\|^2+\nbf^T\Abf\nbf,  \label{eq:first_eq_loccal_SM-AP}
\end{align}
no matter what the noise vector $\nbf(k)$ is. 
In this way, we can derive the global robustness property of the SM-AP\abbrev{SM-AP}{Set-Membership Affine Projection} algorithm.

\begin{cor}[Global Robustness of SM-AP]\label{cor:global_robustness-SM-AP}
Suppose that the SM-AP\abbrev{SM-AP}{Set-Membership Affine Projection} algorithm running from 0 (initialization) to a given iteration $K$ 
employs a constraint vector $\gammabf$ satisfying $\gammabf^T\Abf\gammabf \leq 2\gammabf^T\Abf\nbf$ 
at every iteration in which an update occurs. 
Then, it always holds that
\begin{align}
\frac{\|\wbftilde(K)\|^2+\sum\limits_{k\in{\cal K}_{\rm up}}\ebftilde^T\Abf\ebftilde}{\|\wbftilde(0)\|^2
+\sum\limits_{k\in{\cal K}_{\rm up}}\nbf^T\Abf\nbf} \leq 1,
\end{align}
where ${\cal K}_{\rm up} \neq \emptyset$ is the set comprised of the iteration indexes $k$ in which $\wbf(k)$ is indeed updated and the equality 
holds when $\gammabf^T\Abf\gammabf=2\gammabf^T\Abf\nbf$ for every $k \in {\cal K}_{\rm up}$. 
If ${\cal K}_{\rm up} = \emptyset$, then $\|\wbftilde(K)\|^2 = \|\wbftilde(0)\|^2$, a case that has no practical interest 
since no update is performed.
\end{cor}
\begin{proof}
Denote by ${\cal K} \triangleq \{ 0, 1, 2, \ldots,$ $K-1 \}$ the set of all iterations. 
Let ${\cal K}_{\rm up}\subseteq{\cal K}$ be the subset containing only the iterations in which an update 
occurs, whereas ${\cal K}_{\rm up}^c \triangleq {\cal K} \setminus {\cal K}_{\rm up}$ is comprised of the 
iterations in which the filter coefficients are not updated.

As a consequence of Theorem~\ref{thm:local_robustness-SM-AP}, when an update occurs the inequality given in~\eqref{eq:first_eq_loccal_SM-AP}
is valid provided $\gammabf$ is chosen such that $\gammabf^T\Abf\gammabf \leq 2\gammabf^T\Abf\nbf$ is respected. 
In this way, by summing such inequality for all $k \in {\cal K}_{\rm up}$ we obtain 
\begin{align}
\sum_{k \in {\cal K}_{\rm up}} \Big(\|\wbftilde(k+1)\|^2+\ebftilde^T\Abf\ebftilde \Big)
\leq \sum_{k \in {\cal K}_{\rm up}} \Big( \|\wbftilde(k)\|^2+\nbf^T\Abf\nbf\Big). \label{eq:robustness_accumulation_f1_SM-AP}
\end{align}
Observe that $\gammabf$, $\ebftilde$, $\nbf$, and $\Abf$ all depend on the independent variable $k$, which we have omitted for 
the sake of simplification.
In addition, for the iterations without coefficient update, we have~\eqref{eq:local_robustness_f0-SM-AP}, which 
can be summed for all $k \in {\cal K}_{\rm up}^c$ leading to
\begin{align}
 \sum_{k \in {\cal K}_{\rm up}^c} \| \wbftilde(k+1) \|^2 = \sum_{k \in {\cal K}_{\rm up}^c} \| \wbftilde(k) \|^2. \label{eq:robustness_accumulation_f0_SM-AP}
\end{align}
Summing~\eqref{eq:robustness_accumulation_f1_SM-AP} and~\eqref{eq:robustness_accumulation_f0_SM-AP} yields 
\begin{align}
\sum_{k \in {\cal K}} \| \wbftilde(k+1) \|^2 
+ \sum_{k \in {\cal K}_{\rm up}}\ebftilde^T\Abf\ebftilde
\leq \sum_{k \in {\cal K}} \| \wbftilde(k) \|^2 
+ \sum_{k \in {\cal K}_{\rm up}}\nbf^T\Abf\nbf. \label{eq:robustness_accumulation_SM-AP} 
\end{align}
Then, we can cancel several of the terms $\| \wbftilde(k) \|^2$ from both sides of the above inequality simplifying it as follows
\begin{align}
\| \wbftilde(K) \|^2 
+\sum_{k \in {\cal K}_{\rm up}}\ebftilde^T\Abf\ebftilde
\leq \| \wbftilde(0) \|^2 
+ \sum_{k \in {\cal K}_{\rm up}}\nbf^T\Abf\nbf.
\end{align}
Assuming a nonzero denominator, we can write the previous inequality in a compact form
\begin{align}
\frac{\| \wbftilde(K) \|^2 
+ \sum\limits_{k \in {\cal K}_{\rm up}}\ebftilde^T\Abf\ebftilde}{\| \wbftilde(0) \|^2 
+ \sum\limits_{k \in {\cal K}_{\rm up}}\nbf^T\Abf\nbf} \leq 1.
\end{align}
This relation holds for all $K$, provided $\gammabf^T\Abf\gammabf \leq 2\gammabf^T\Abf\nbf$ is satisfied for every iteration  
in which an update occurs, i.e., for every $k \in {\cal K}_{\rm up}$. 
The only assumption used in the derivation is that ${\cal K}_{\rm up}\neq\emptyset$. 
Otherwise, we would have $\| \wbftilde(K) \|^2 = \| \wbftilde(0) \|^2$, which would occur only if $\wbf(k)$ 
is never updated, which is not of practical interest.
\end{proof}

Observe that, unlike the SM-NLMS\abbrev{SM-NLMS}{Set-Membership Normalized LMS} algorithm, the SM-AP\abbrev{SM-AP}{Set-Membership Affine Projection} algorithm requires the condition $\gammabf^T\Abf\gammabf \leq 2\gammabf^T\Abf\nbf$ 
to be satisfied in order to guarantee  $l_2$-stability from its uncertainties 
$\{ \wbftilde(0), \{ n(k) \}_{0\leq k\leq K} \}$ to its errors $\{ \wbftilde(K), \{ \etilde(k) \}_{0\leq k\leq K} \}$. 
The next question is: are there constraint vectors $\gammabf$ satisfying such a condition? 
This is a very interesting point because the left-hand side (LHS)\abbrev{LHS}{Left-Hand Side} of the condition is always positive, whereas the RHS\abbrev{RHS}{Right-Hand Side} is not. 
Corollary~\ref{cor:global_robustness-SM-AP-c*n(k)} answers this question and shows an example of such a constraint vector.

\begin{cor}\label{cor:global_robustness-SM-AP-c*n(k)}
Suppose the CV\abbrev{CV}{Constraint Vector} is chosen as $\gammabf(k) = c\nbf(k)$ in the SM-AP\abbrev{SM-AP}{Set-Membership Affine Projection} algorithm, where $\nbf(k)$ is the noise vector defined 
in~\eqref{eq:desiredSignalModel-SM-AP}.  
If $0 \leq c \leq 2$, then the condition $\gammabf^T\Abf\gammabf \leq 2\gammabf^T\Abf\nbf$ always holds, implying that 
the SM-AP\abbrev{SM-AP}{Set-Membership Affine Projection} algorithm is globally robust by Corollary~\ref{cor:global_robustness-SM-AP}.
\end{cor}
\begin{proof}
Substituting $\gammabf(k) =  c\nbf(k)$ in $\gammabf^T\Abf\gammabf \leq 2\gammabf^T\Abf\nbf$ leads to  
the following condition $(c^2-2c)\nbf^T(k)\Abf(k)\nbf(k)\leq 0$, which is satisfied for $c^2-2c \leq 0 \Rightarrow 0\leq c \leq 2$ since 
$\Abf(k)$ is positive definite.
Hence, due to Corollary~\ref{cor:global_robustness-SM-AP} the proposed $\gammabf(k)$ leads to a globally robust SM-AP\abbrev{SM-AP}{Set-Membership Affine Projection} algorithm.
\end{proof}

It is worth mentioning that the constraint vector $\gammabf(k)$ in Corollary~\ref{cor:global_robustness-SM-AP-c*n(k)} 
is not practical because $\nbf(k)$ is not observable. 
Therefore, Corollary~\ref{cor:global_robustness-SM-AP-c*n(k)} is actually related to the existence of $\gammabf(k)$ 
satisfying $\gammabf^T\Abf\gammabf<2\gammabf^T\Abf\nbf$.

Unlike the SM-NLMS\abbrev{SM-NLMS}{Set-Membership Normalized LMS} algorithm, the $l_2$-stability of the SM-AP\abbrev{SM-AP}{Set-Membership Affine Projection} algorithm is not guaranteed. 
Indeed, as demonstrated in Theorem~\ref{thm:local_robustness-SM-AP} and Corollary~\ref{cor:global_robustness-SM-AP}, 
a judicious choice of the CV\abbrev{CV}{Constraint Vector} is required for the SM-AP\abbrev{SM-AP}{Set-Membership Affine Projection} algorithm to be $l_2$-stable.
{\it It is worth mentioning that practical choices of $\gammabf (k)$ satisfying the robustness condition 
$\gammabf^T\Abf\gammabf \leq 2\gammabf^T\Abf\nbf$ for every iteration $k$ are not known yet!}
Even widely used CVs\abbrev{CV}{Constraint Vector}, like the simple choice CV (SC-CV)~\cite{Markus_optimalCV_sigpro2017}\abbrev{SC-CV}{Simple Choice CV}, sometimes violate this condition 
as will be shown in Section~\ref{sec:simulation-robustness}.
However, this does not mean that the SM-AP\abbrev{SM-AP}{Set-Membership Affine Projection} algorithm diverges. 
In fact, it does not diverge regardless the choice of $\gammabf(k)$, as demonstrated in the next subsection.


\subsection{The SM-AP algorithm does not diverge}\label{sub:divergence_sm_ap}

When the SM-AP\abbrev{SM-AP}{Set-Membership Affine Projection} algorithm updates (i.e., when $|e(k)| > \gammabar$), it generates $\wbf (k+1)$ as the solution to the following 
optimization problem~\cite{Werner_sm_ap_letter2001,Diniz_adaptiveFiltering_book2013}:
\begin{align}
 & \text{minimize }    \| \wbf(k+1) - \wbf(k) \|^2         \nonumber \\
 & \text{subject to }  \dbf(k) - \Xbf^T(k) \wbf(k+1) = \gammabf(k). 
\end{align}
The constraint essentially states that the a posteriori errors $\epsilon(k-l) \triangleq d(k-l) - \xbf^T(k-l) \wbf(k+1)$ are equal to 
their respective $\gamma_l(k)$, which in turn are bounded by $\gammabar$.
This leads to the following derivation:
\begin{align}
 | \epsilon(k-l) | = | d(k-l) - \xbf^T(k-l) \wbf(k+1) | & \leq \gammabar   , \nonumber \\
                     | \xbf^T(k-l) \wbftilde(k+1) + n(k-l) | & \leq \gammabar ,
\end{align}
which should be valid for all iterations and suitable values of the involved variables. Therefore, we have
\begin{align}
-\gammabar-n(k-l)&\leq \xbf^T(k-l)\wbftilde(k+1)\leq \gammabar-n(k-l).
\end{align}
Since the noise sequence is bounded and $\gammabar < \infty$, we have
\begin{align}
-\infty < \sum_{i=0}^N x_i(k-l){\tilde w}_i(k+1) < \infty,
\end{align}
where $x_i(k-l), {\tilde w}_i(k+1) \in \mathbb{R}$ denote the $i$th entry of vectors $\xbf(k-l), \wbftilde(k+1) \in \mathbb{R}^{N+1}$, respectively.
As a result, $|{\tilde w}_i(k+1)|$ is also bounded implying $\| \wbftilde(k+1) \|^2 < \infty$, which means that 
the SM-AP\abbrev{SM-AP}{Set-Membership Affine Projection} algorithm does not diverge even when its CV\abbrev{CV}{Constraint Vector} is not properly chosen. 
In Section~\ref{sec:simulation-robustness} we verify this fact experimentally by using a {\it general CV}\abbrev{CV}{Constraint Vector}, i.e., 
a CV\abbrev{CV}{Constraint Vector} whose entries are randomly chosen but satisfying $| \gamma_i (k) | \leq \gammabar$. 
Such general CV\abbrev{CV}{Constraint Vector} leads to poor performance, in comparison to the SM-AP\abbrev{SM-AP}{Set-Membership Affine Projection} algorithm using adequate CVs\abbrev{CV}{Constraint Vector}, but the algorithm 
does not diverge.

The same reasoning could be applied to demonstrate that the SM-NLMS\abbrev{SM-NLMS}{Set-Membership Normalized LMS} algorithm does not diverge as well. 
However, from Corollary~\ref{thm:global_robustness-sm-nlms} it is straightforward to verify that $\| \wbftilde(K) \|^2 < \infty$ for every $K$, 
as the denominator in~\eqref{eq:global_robustness-sm-nlms} is finite.


\section{Simulations} \label{sec:simulation-robustness}

In this section, we provide simulation results for the SM-NLMS\abbrev{SM-NLMS}{Set-Membership Normalized LMS} and SM-AP\abbrev{SM-AP}{Set-Membership Affine Projection} algorithms in order to verify their robustness properties addressed in the previous sections.  
These results are obtained by applying the aforementioned algorithms to a system identification problem. 
The unknown system $\wbf_o$ is comprised of $10$ coefficients drawn from a standard Gaussian distribution. 
The noise $n(k)$ is a zero-mean white Gaussian noise with variance $\sigma_n^2=0.01$ yielding a signal-to-noise ratio (SNR)\abbrev{SNR}{Signal-to-Noise Ratio} 
equal to $20$~dB.
The regularization factor and the initialization for the adaptive filtering coefficient vector are $\delta = 10^{-12}$ and   
$\wbf(0)=[0~\cdots~0]^T \in \mathbb{R}^{10}$, respectively. 
The error bound parameter is usually set as $\gammabar = \sqrt{5 \sigma_n^2}=0.2236$, unless otherwise stated.


\subsection{Confirming the results for the SM-NLMS algorithm} \label{subsec:simulation-robustness-sm-nlms}

Here, the input signal $x(k)$ is a zero-mean white Gaussian noise with variance equal to $1$. 
Fig.~\ref{fig:sim1-sm-nlms-robustness} aims at verifying Theorem~\ref{thm:local_robustness-sm-nlms}. 
Thus, for the iterations $k$ with coefficient update, let us denote the left-hand side (LHS)\abbrev{LHS}{Left-Hand Side} and the right-hand side (RHS)\abbrev{RHS}{Right-Hand Side} 
of~\eqref{eq:local_robustness_f1} as $g_1(k)$ and $g_2(k)$, respectively. 
In addition, to simultaneously account for~\eqref{eq:local_robustness_f0}, we define 
$g_1(k) = \| \wbftilde(k+1) \|^2$ and $g_2(k) = \| \wbftilde(k) \|^2$ for the iterations 
without coefficient update. 
Fig.~\ref{fig:sim1-sm-nlms-robustness} depicts $g_1(k)$ and $g_2(k)$ considering the system identification scenario 
described in the beginning of Section~\ref{sec:simulation-robustness}. 
In this figure, we can observe that $g_1(k) \leq g_2(k)$ for all $k$.
Indeed, we verified that $g_1(k) = g_2(k)$ (i.e., curves are overlaid) only in the iterations without update, i.e., 
$\wbf(k+1) = \wbf(k)$. 
In the remaining iterations we have $g_1(k) < g_2(k)$, corroborating Theorem~\ref{thm:local_robustness-sm-nlms}.

\begin{figure}[t!]
\centering
\includegraphics[width=1\linewidth]{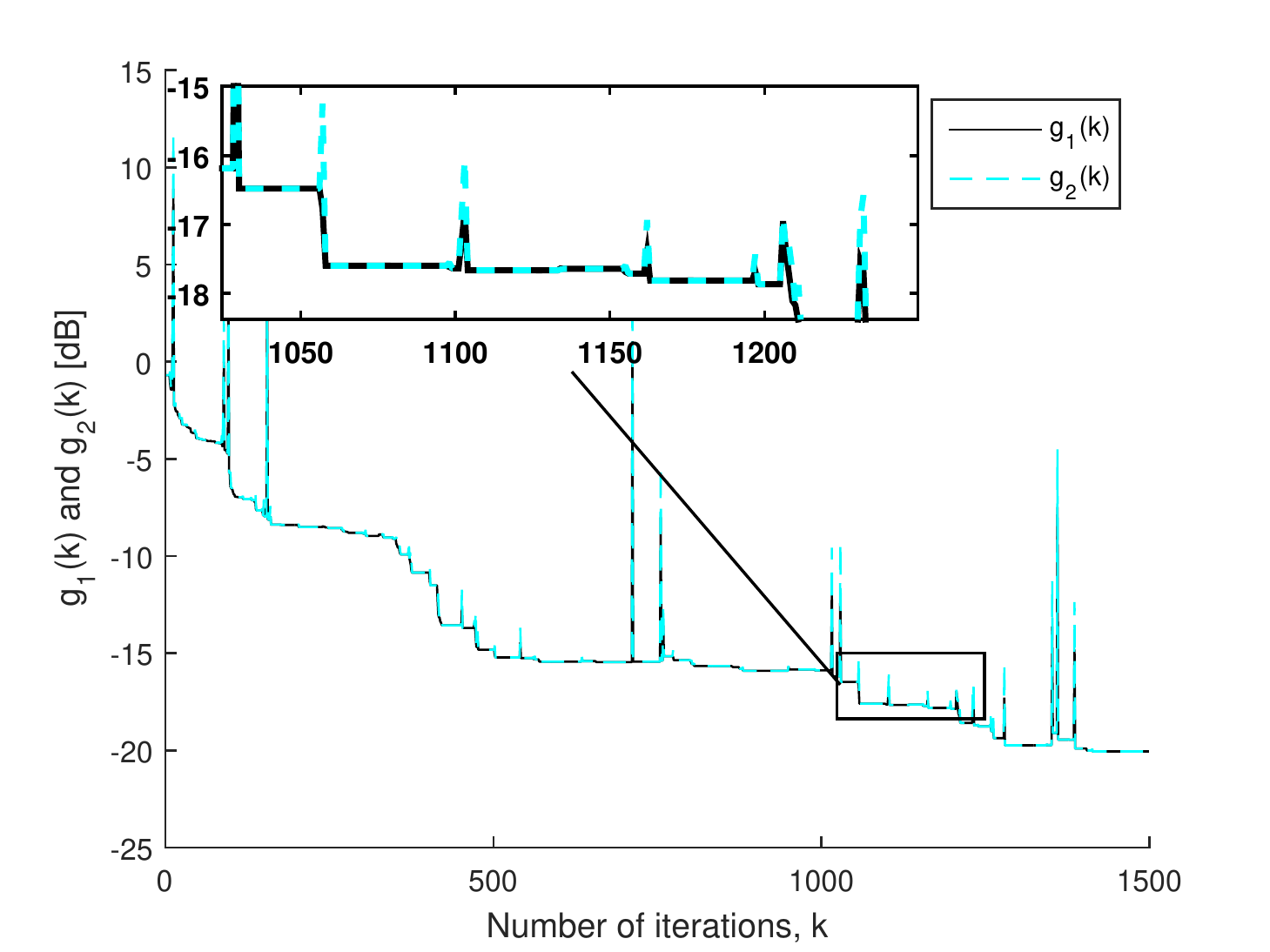}
\caption{Values of $g_1(k)$ and $g_2(k)$ over the iterations for the SM-NLMS algorithm corroborating  Theorem~\ref{thm:local_robustness-sm-nlms}.  \label{fig:sim1-sm-nlms-robustness}}
\end{figure}

Fig.~\ref{fig:sim2-sm-nlms-robustness} depicts the sequence  $\{\|\wbftilde(k)\|^2\}$ for the SM-NLMS\abbrev{SM-NLMS}{Set-Membership Normalized LMS} algorithm and its classical 
counterpart, the NLMS\abbrev{NLMS}{Normalized LMS} algorithm. 
For the SM-NLMS\abbrev{SM-NLMS}{Set-Membership Normalized LMS} algorithm, we consider three cases: fixed $\gammabar$ with unknown noise bound (blue solid line), 
fixed $\gammabar$ with known noise bound $B=0.11$ (cyan solid line), and time-varying $\gammabar(k)$, 
defined as $\sqrt{5\sigma_n^2}$ during the transient period and $\sqrt{9\sigma_n^2}$ during the steady-state, 
with unknown noise bound (green solid line). 
For the results using the time-varying $\gammabar(k)$, the window length is $E=20$, and when the number of updates in the window 
is less than 4, we assume the algorithm is in the steady-state period. 
For the NLMS\abbrev{NLMS}{Normalized LMS} algorithm, two different step-sizes are used: $\mu=0.9$, which leads to fast convergence but high misadjustment, 
and $\mu=0.05$, which leads to slow convergence but low misadjustment.

In Fig.~\ref{fig:sim2-sm-nlms-robustness}, the blue curve confirms the discussion in Subsection~\ref{sub:sm-nlms-unbounded-noise}. 
Indeed, we can observe that the sequence $\{\|\wbftilde(k)\|^2\}$ represented by this blue curve increases only 
$30$ times along the $2500$ iterations, meaning that the SM-NLMS\abbrev{SM-NLMS}{Set-Membership Normalized LMS} algorithm did not improve its estimate $\wbf(k+1)$ only 
in $30$ iterations.
Thus, in this experiment we have $\mathbb{P}[\|\wbftilde(k+1)\|^2>\|\wbftilde(k)\|^2] = 0.012$, whose value is lower than
its corresponding upper bound given by ${\rm erfc}(\sqrt{2.5})=0.0253$, as explained in Subsection~\ref{sub:sm-nlms-unbounded-noise}.
Also, we can observe that the event $\|\wbftilde(k+1)\|^2>\|\wbftilde(k)\|^2$ did not occur in the early iterations because 
in these iterations $\etilde^2(k)$ is usually large due to a significant mismatch between $\wbf(k)$ and $\wbf_o$, i.e., 
the condition specified in Corollary~\ref{cor:sm_nlms_decreasing} is frequently satisfied.

Also in Fig.~\ref{fig:sim2-sm-nlms-robustness}, the cyan curve shows that when the noise bound is known we can obtain a 
monotonic decreasing sequence $\{\|\wbftilde(k)\|^2\}$ by selecting $\gammabar \geq 2B$, corroborating 
Theorem~\ref{thm:strong-local-robustness-sm-nlms} and Corollary~\ref{cor:strong-global-robustness-sm-nlms}. 
The sequence $\{\|\wbftilde(k)\|^2\}$ represented by the green curve in Fig.~\ref{fig:sim2-sm-nlms-robustness} increases only $3$ times, 
thus confirming the advantage 
of using a time-varying $\gammabar(k)$ when the noise bound is unknown, as explained in 
Subsection~\ref{sub:sm-nlms-time-varying-gammabar}. 
As compared to the SM-NLMS\abbrev{SM-NLMS}{Set-Membership Normalized LMS} algorithm, the behavior of the sequence $\{\|\wbftilde(k)\|^2\}$ for the NLMS\abbrev{NLMS}{Normalized LMS} algorithm 
is very irregular. 
Indeed, for the NLMS\abbrev{NLMS}{Normalized LMS} algorithm there are many iterations in which $\|\wbftilde(k+1)\|^2>\|\wbftilde(k)\|^2$, even 
when using a small step-size $\mu$.
Hence, the NLMS\abbrev{NLMS}{Normalized LMS} algorithm does not use the input data as efficiently as the SM-NLMS\abbrev{SM-NLMS}{Set-Membership Normalized LMS} algorithm does, 
given that the NLMS\abbrev{NLMS}{Normalized LMS} performs many ``useless updates''.

In conclusion, an interesting advantage of the SM-NLMS\abbrev{SM-NLMS}{Set-Membership Normalized LMS} algorithm over the NLMS\abbrev{NLMS}{Normalized LMS} algorithm is that the former can achieve 
fast convergence and has a well-behaved sequence $\{\|\wbftilde(k)\|^2\}$ (which rarely increases) at the same time. 
In addition, the SM-NLMS\abbrev{SM-NLMS}{Set-Membership Normalized LMS} algorithm also saves computational resources by not updating the filter coefficients at 
every iteration. 
In Fig.~\ref{fig:sim2-sm-nlms-robustness}, the update rates of the blue, cyan, and green curves are 4.6$\%$, 1.5$\%$, and 1.9$\%$, respectively. 
They confirm that the computational cost of the SM-NLMS\abbrev{SM-NLMS}{Set-Membership Normalized LMS} algorithm is significantly lower than that of the NLMS \abbrev{NLMS}{Normalized LMS}
algorithm.\footnote{In comparison to the NLMS\abbrev{NLMS}{Normalized LMS} algorithm, whenever the SM-NLMS\abbrev{SM-NLMS}{Set-Membership Normalized LMS} algorithm updates it performs two additional operations: 
One division and one subtraction due to the computation of $\mu(k)$. However, for most of the iterations the SM-NLMS\abbrev{SM-NLMS}{Set-Membership Normalized LMS} algorithm requires fewer 
operations because it does not update often.}

\begin{figure}[t!]
\centering
\includegraphics[width=1\linewidth]{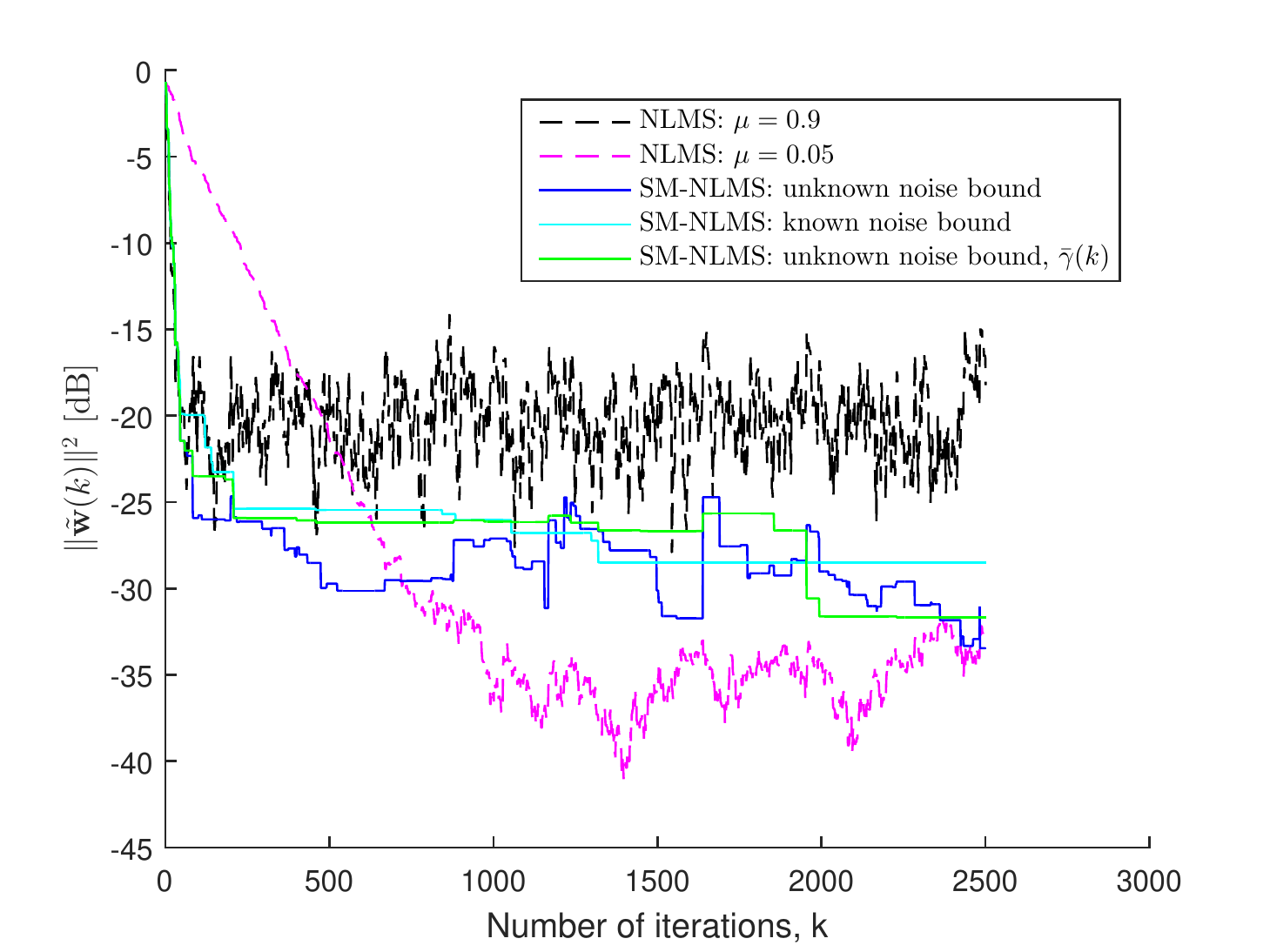}
\caption{$\|\wbftilde(k)\|^2 \triangleq \| \wbf_o - \wbf(k) \|^2$ for the NLMS and the SM-NLMS algorithms. \label{fig:sim2-sm-nlms-robustness}}
\end{figure}


\subsection{Confirming the results for the SM-AP algorithm} \label{subsec:simulation-robustness-sm-ap}

For the case of the SM-AP\abbrev{SM-AP}{Set-Membership Affine Projection} algorithm, the input is a first-order autoregressive signal generated as $x(k)=0.95x(k-1)+n(k-1)$. 
We test the SM-AP\abbrev{SM-AP}{Set-Membership Affine Projection} algorithm employing $L=2$ (i.e., reuse of two previous input data) and 
three different constraint vectors (CVs)\abbrev{CV}{Constraint Vector} $\gammabf(k)$: a general CV\abbrev{CV}{Constraint Vector}, the SC-CV,\abbrev{SC-CV}{Simple Choice CV} 
and the noise vector CV.\abbrev{CV}{Constraint Vector} 
The general CV\abbrev{CV}{Constraint Vector} $\gammabf(k)$, in which the entries are set as $\gamma_l(k) = \gammabar$ for $0\leq l \leq L$, illustrates a case 
where the CV\abbrev{CV}{Constraint Vector} is not properly chosen~\cite{Markus_edcv_eusipco2013,Markus_optimalCV_sigpro2017}. 
The SC-CV~\cite{Markus_edcv_eusipco2013,Markus_optimalCV_sigpro2017}\abbrev{SC-CV}{Simple Choice CV} is defined as 
$\gamma_0(k) = \gammabar\frac{e(k)}{|e(k)|}$ and $\gamma_l(k) = \epsilon(k-l)$ for $1 \leq l \leq L$. 
The noise vector CV\abbrev{CV}{Constraint Vector} is given by $\gammabf(k) = \nbf(k)$.

The results depicted in Figs.~\ref{fig:sim-robustness-sm-ap}, \ref{fig:sim-robustness-sm-ap-simp}, \ref{fig:sm-ap-noise-robustness}, and \ref{fig:sm-ap-simp-boounded-robustness} aim at verifying Theorem~\ref{thm:local_robustness-SM-AP} and 
Corollary~\ref{cor:global_robustness-SM-AP-c*n(k)}. 
We define $g_1(k)$ and $g_2(k)$ as the numerator and the denominator of~\eqref{eq:local_robustness_f1-SM-AP} 
in Theorem~\ref{thm:local_robustness-SM-AP}, respectively, when an update occurs; otherwise, we define 
$g_1(k) = \| \wbftilde(k+1) \|^2$ and $g_2(k) = \| \wbftilde(k) \|^2$.

\begin{figure}[t!]
\centering
\includegraphics[width=1\linewidth]{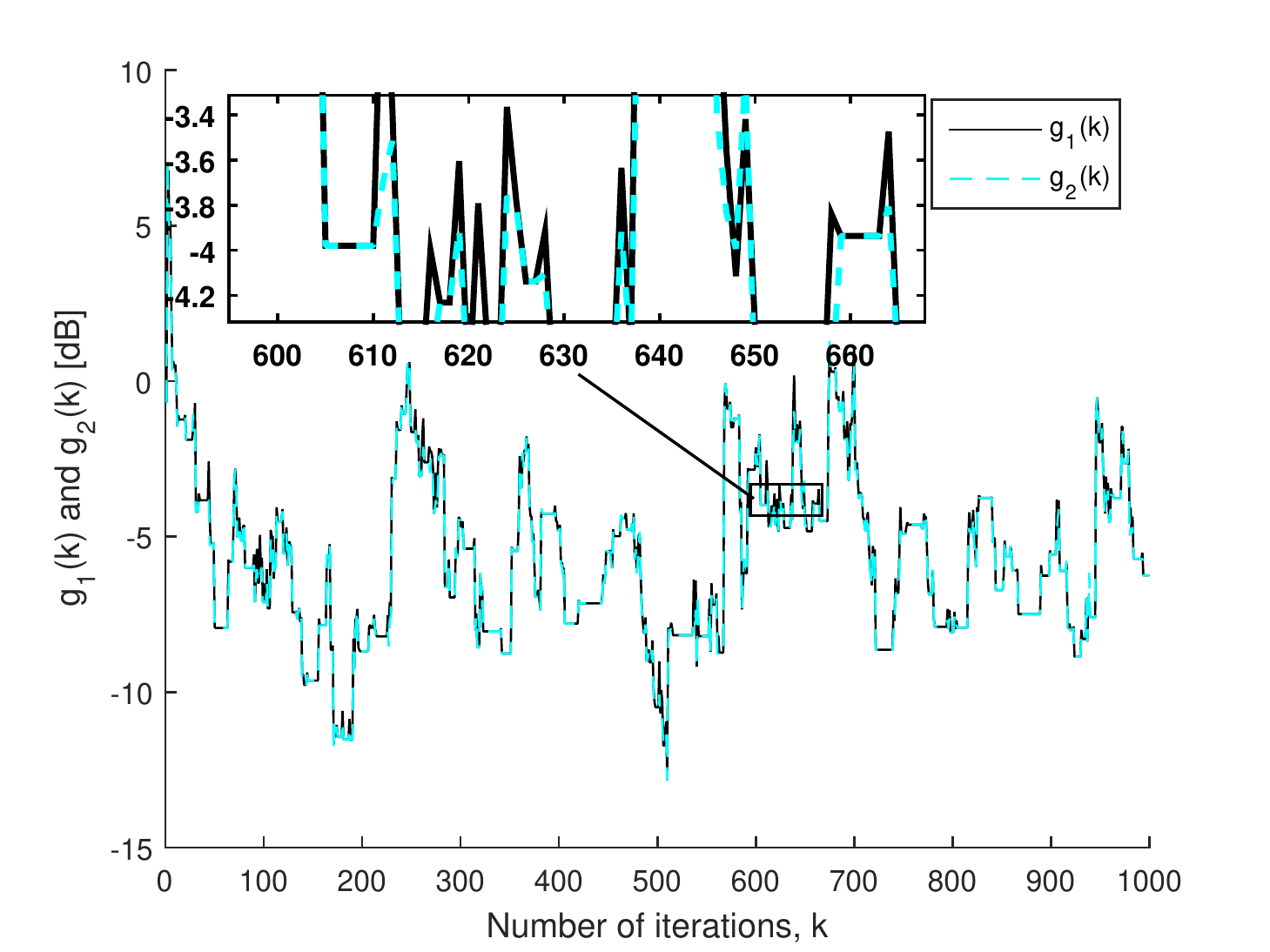}
\caption{Values of $g_1(k)$ and $g_2(k)$ over the iterations for the SM-AP algorithm with $\gammabf(k)$ as the general CV,
where $g_1(k)$ and $g_2(k)$ are the numerator and denominator of~\eqref{eq:local_robustness_f1-SM-AP} in 
Theorem~\ref{thm:local_robustness-SM-AP}, when an update occurs; otherwise, $g_1(k)=\|\wbftilde(k+1)\|^2$ and $g_2(k)=\|\wbftilde(k)\|^2$.  \label{fig:sim-robustness-sm-ap}}
\end{figure}
\begin{figure}[t!]
\centering
\includegraphics[width=1\linewidth]{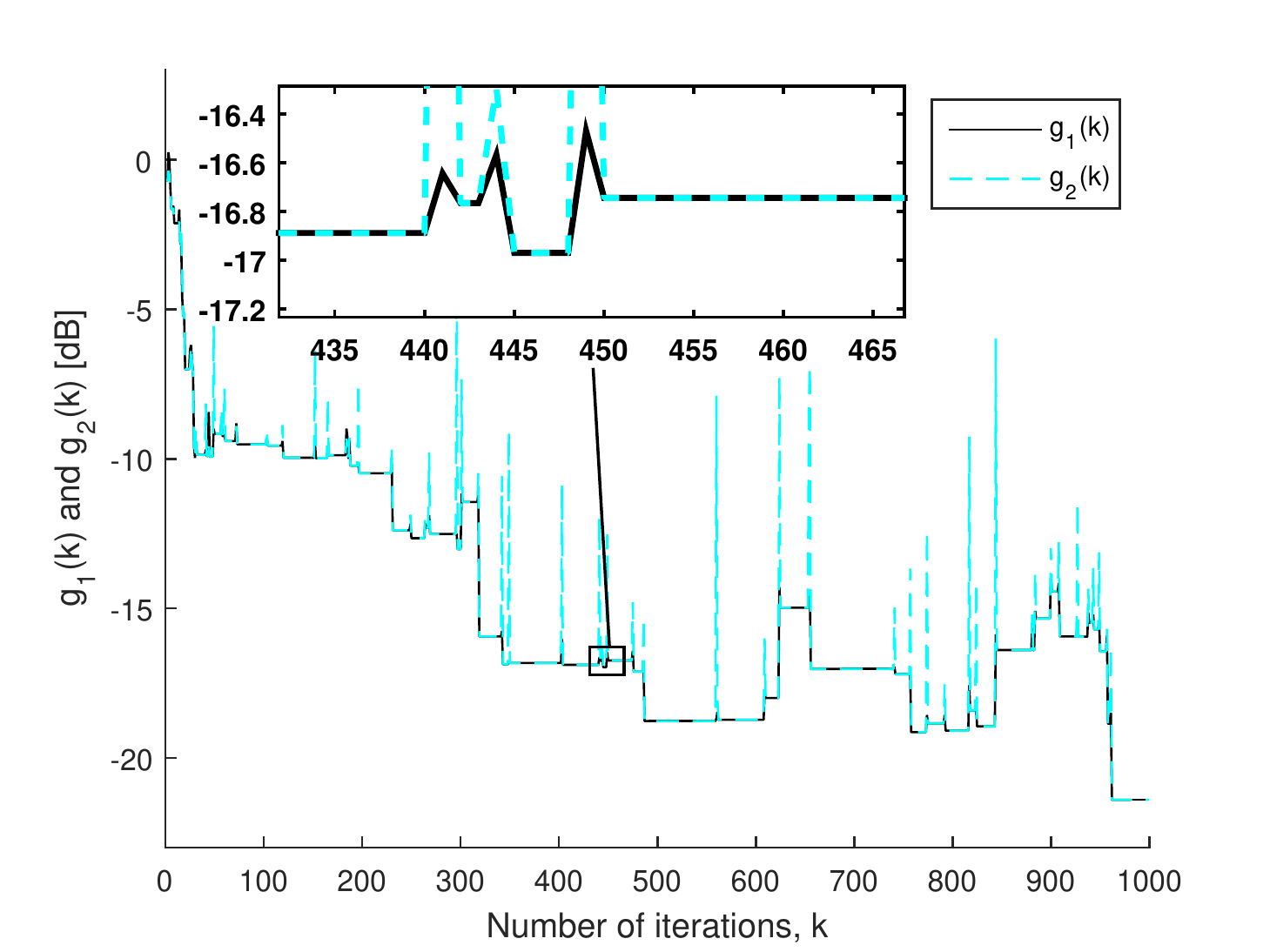}
\caption{Values of $g_1(k)$ and $g_2(k)$ over the iterations for the SM-AP algorithm with $\gammabf(k)$ as the SC-CV, 
where $g_1(k)$ and $g_2(k)$ are the numerator and denominator of~\eqref{eq:local_robustness_f1-SM-AP} in 
Theorem~\ref{thm:local_robustness-SM-AP}, when an update occurs; otherwise, $g_1(k)=\|\wbftilde(k+1)\|^2$ and $g_2(k)=\|\wbftilde(k)\|^2$.  \label{fig:sim-robustness-sm-ap-simp}}
\end{figure}
\begin{figure}[t!]
\centering
\includegraphics[width=1\linewidth]{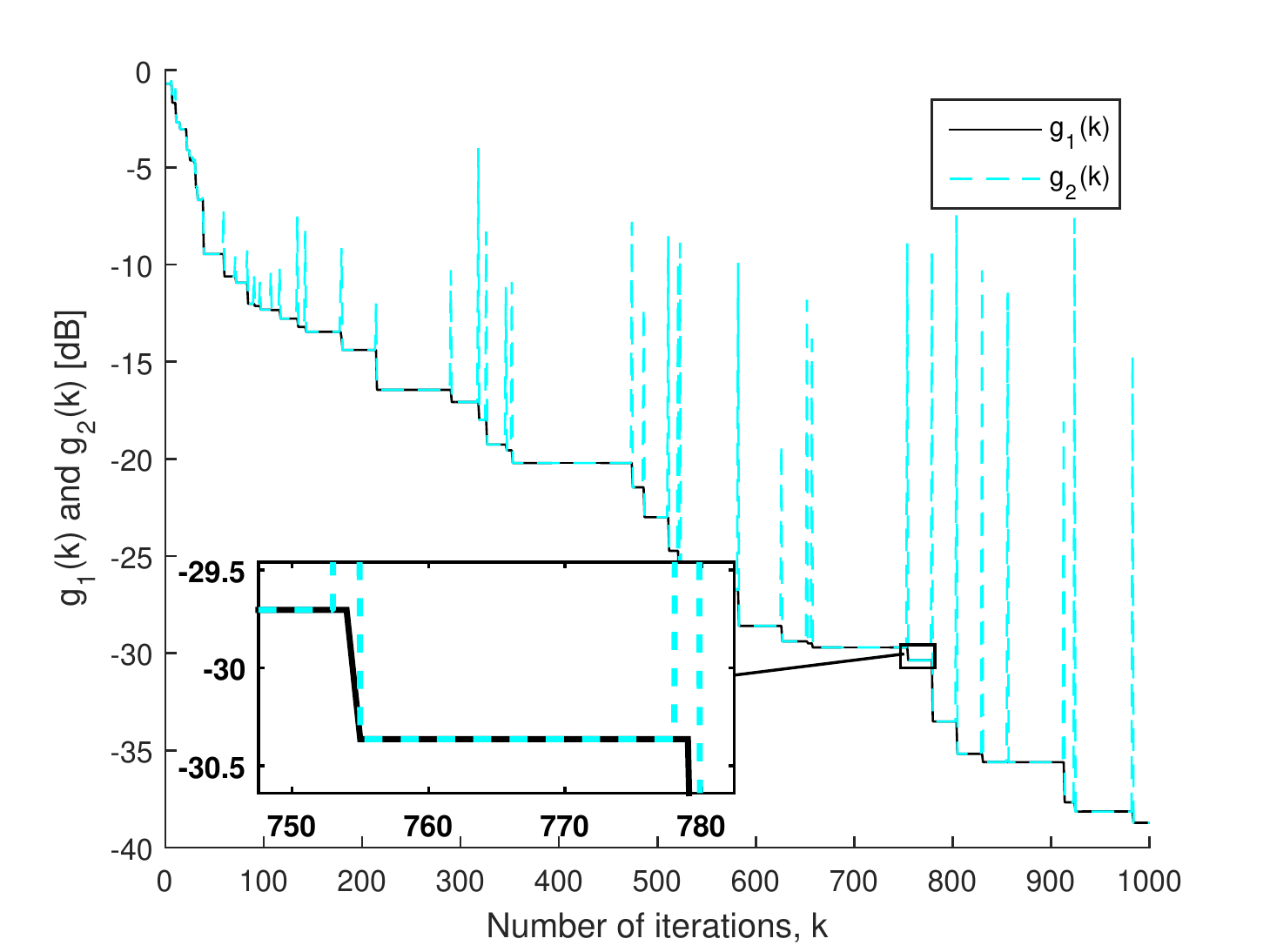}
\caption{Values of $g_1(k)$ and $g_2(k)$ over the iterations for the SM-AP algorithm with $\gammabf(k) = \nbf(k)$, 
where $g_1(k)$ and $g_2(k)$ are the numerator and denominator of~\eqref{eq:local_robustness_f1-SM-AP} in 
Theorem~\ref{thm:local_robustness-SM-AP}, when an update occurs; otherwise, $g_1(k)=\|\wbftilde(k+1)\|^2$ and $g_2(k)=\|\wbftilde(k)\|^2$.  \label{fig:sm-ap-noise-robustness}}
\end{figure}
\begin{figure}[t!]
\centering
\includegraphics[width=1\linewidth]{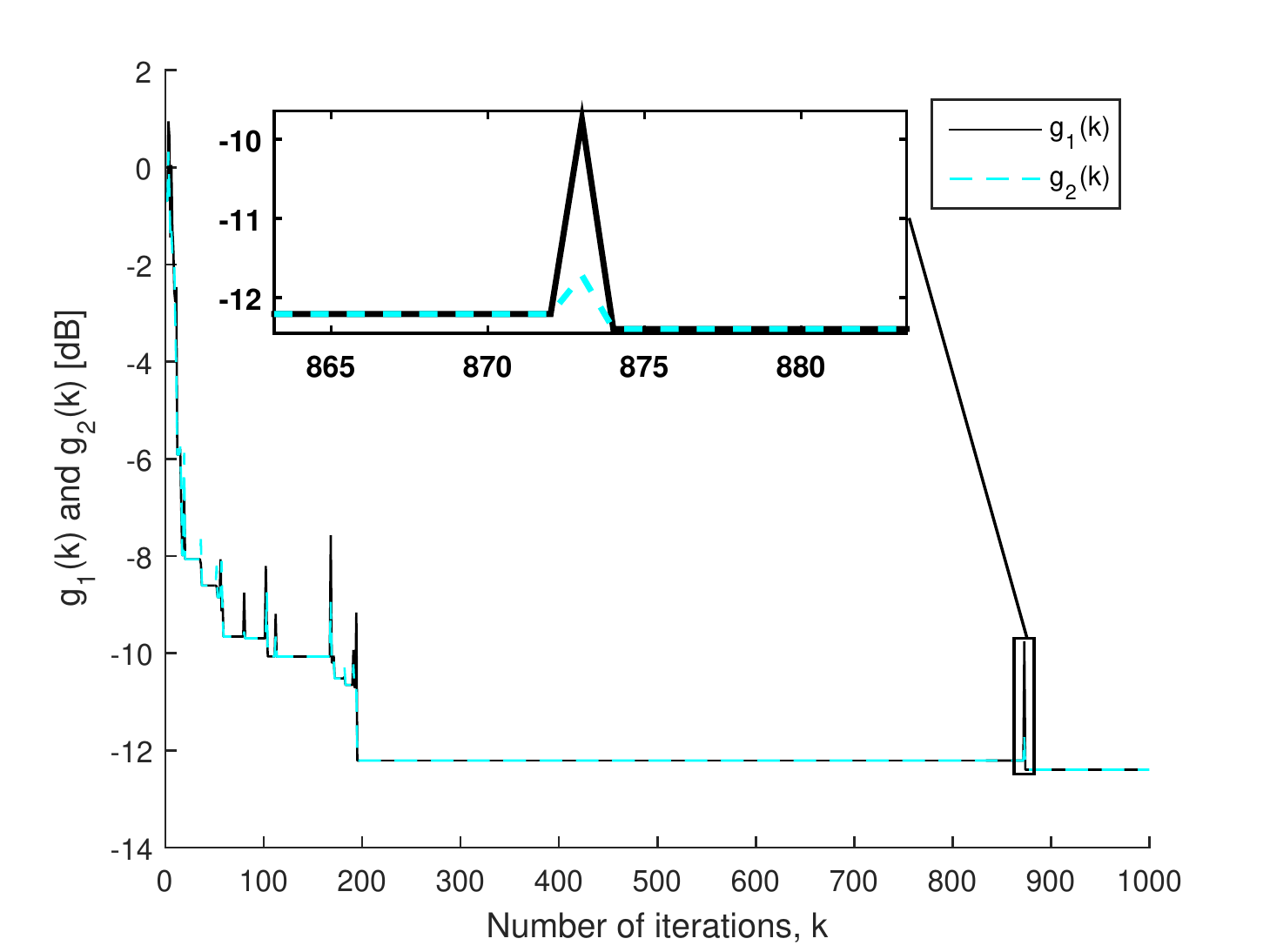}
\caption{Values of $g_1(k)$ and $g_2(k)$ over the iterations for the SM-AP algorithm with $\gammabf(k)$ as the SC-CV when the noise bound is known, where $g_1(k)$ and $g_2(k)$ are the numerator and denominator of~\eqref{eq:local_robustness_f1-SM-AP} in 
Theorem~\ref{thm:local_robustness-SM-AP}, when an update occurs; otherwise, $g_1(k)=\|\wbftilde(k+1)\|^2$ and $g_2(k)=\|\wbftilde(k)\|^2$.  \label{fig:sm-ap-simp-boounded-robustness}}
\end{figure}

The results depicted in Fig.~\ref{fig:sim-robustness-sm-ap} illustrate that, for the general CV\abbrev{CV}{Constraint Vector}, there are many iterations 
in which $g_1(k)>g_2(k)$ (about $293$ out of $1000$ iterations).
This is an expected behavior since the general CV\abbrev{CV}{Constraint Vector} does not take into account (directly or indirectly) the value of $n(k)$ and, 
therefore, it does not consider the robustness condition $\gammabf^T(k) \Abf(k) \gammabf(k) \leq 2 \gammabf^T(k) \Abf(k) \nbf(k)$.

For the SM-AP\abbrev{SM-AP}{Set-Membership Affine Projection} algorithm employing the SC-CV\abbrev{SC-CV}{Simple Choice CV}, however, there are very few iterations in which 
$g_1(k)> g_2(k)$ (only $19$ out of $1000$ iterations), as shown in Fig.~\ref{fig:sim-robustness-sm-ap-simp}. 
This means that even the widely used SC-CV\abbrev{SC-CV}{Simple Choice CV} does not lead to global robustness. 

Fig.~\ref{fig:sm-ap-noise-robustness} depicts the results for the SM-AP\abbrev{SM-AP}{Set-Membership Affine Projection} algorithm with $\gammabf(k)=\nbf(k)$. 
In this case, we can observe that $g_1(k)\leq g_2(k)$ for all $k$, corroborating Corollary~\ref{cor:global_robustness-SM-AP-c*n(k)}.
In other words, this CV\abbrev{CV}{Constraint Vector} guarantees the global robustness of the SM-AP\abbrev{SM-AP}{Set-Membership Affine Projection} algorithm. 

Fig.~\ref{fig:sm-ap-simp-boounded-robustness} illustrates $g_1(k)$ and $g_2(k)$ for the SM-AP\abbrev{SM-AP}{Set-Membership Affine Projection} algorithm with SC-CV\abbrev{SC-CV}{Simple Choice CV} when 
the noise bound is known and 10 times smaller than $\gammabar$. 
In contrast with the SM-NLMS\abbrev{SM-NLMS}{Set-Membership Normalized LMS} algorithm, for the SM-AP\abbrev{SM-AP}{Set-Membership Affine Projection} algorithm even when the noise bound is known and much smaller than $\gammabar$,  
we cannot guarantee that $g_1(k)\leq g_2(k)$ for all $k$. 
In Fig.~\ref{fig:sm-ap-simp-boounded-robustness}, for example, we observe $g_1(k)>g_2(k)$ in 15 iterations.

Fig.~\ref{fig:sm-ap-all-wtilde-robustness} depicts the sequence $\{\|\wbftilde(k)\|^2\}$ for the AP\abbrev{AP}{Affine Projection} and the SM-AP\abbrev{SM-AP}{Set-Membership Affine Projection} algorithms.
For the AP\abbrev{AP}{Affine Projection} algorithm, the step-size $\mu$ is set as 0.9 and 0.05, whereas for the SM-AP\abbrev{SM-AP}{Set-Membership Affine Projection} algorithm the three previously 
defined CVs\abbrev{CV}{Constraint Vector} are tested. 
For the AP\abbrev{AP}{Affine Projection} algorithm, we can observe an irregular behavior of $\{\|\wbftilde(k)\|^2\}$, i.e., this sequence increases and decreases 
very often. 
Even when a low value of $\mu$ is applied we still observe many iterations in which $\|\wbftilde(k+1)\|^2 > \|\wbftilde(k)\|^2$ (425 out of 1000 iterations). 
The SM-AP\abbrev{SM-AP}{Set-Membership Affine Projection} algorithm using the general CV\abbrev{CV}{Constraint Vector} performs similar to the AP\abbrev{AP}{Affine Projection} algorithm with high $\mu$. 
But when the CV\abbrev{CV}{Constraint Vector} is properly chosen, like the SC-CV\abbrev{SC-CV}{Simple Choice CV} for example, we observe that the number of iterations 
in which $\|\wbftilde(k+1)\|^2 > \|\wbftilde(k)\|^2$ is dramatically reduced (26 out of 1000 iterations), which means that the SM-AP\abbrev{SM-AP}{Set-Membership Affine Projection} with an adequate CV \abbrev{CV}{Constraint Vector}
performs fewer ``useless updates'' than the AP\abbrev{AP}{Affine Projection} algorithm. 
Another interesting, although not practical, choice of CV\abbrev{CV}{Constraint Vector} is $\gammabf(k) = \nbf(k)$, which leads to a monotonic decreasing 
sequence $\{\|\wbftilde(k)\|^2\}$.

\begin{figure}[t!]
\centering
\includegraphics[width=1\linewidth]{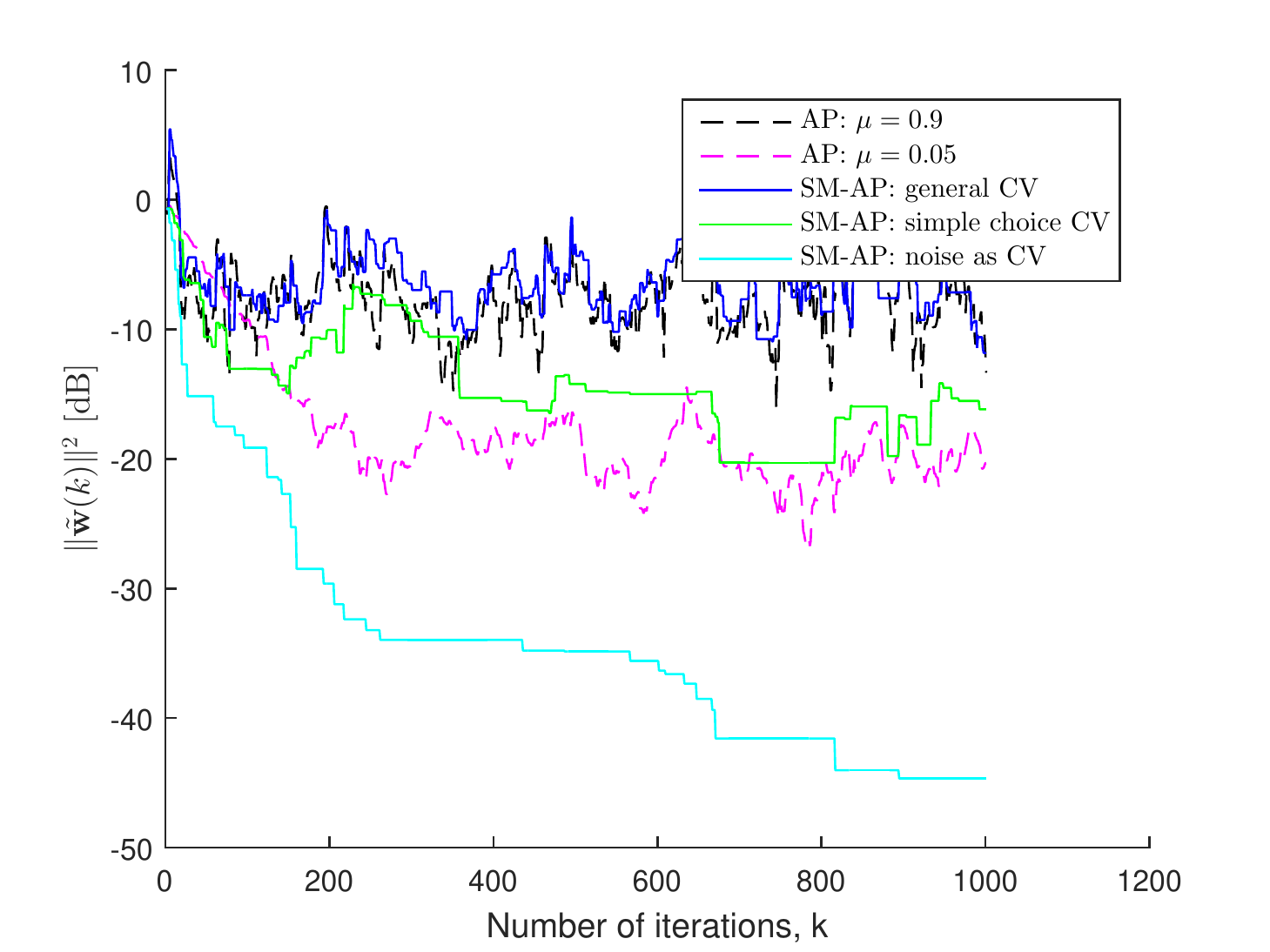}
\caption{$\|\wbftilde(k)\|^2 \triangleq \| \wbf(k) - \wbf_o \|^2$ for the AP and the SM-AP algorithms.  \label{fig:sm-ap-all-wtilde-robustness}}
\end{figure}

The MSE\abbrev{MSE}{Mean-Squared Error} learning curves for the AP\abbrev{AP}{Affine Projection} and the SM-AP\abbrev{SM-AP}{Set-Membership Affine Projection} algorithms are depicted in Fig.~\ref{fig:Robustness_learning_curves}.
These results were computed by averaging the squared error over 1000 trials for each curve. 
Observing the results of the AP\abbrev{AP}{Affine Projection} algorithm, the trade-off between convergence rate and steady-state MSE\abbrev{MSE}{Mean-Squared Error} is evident.
Indeed, excluding the SM-AP\abbrev{SM-AP}{Set-Membership Affine Projection} with general CV\abbrev{CV}{Constraint Vector} (which is not an adequate choice for the CV)\abbrev{CV}{Constraint Vector}, the AP\abbrev{AP}{Affine Projection} algorithm could not achieve fast convergence 
and low MSE\abbrev{MSE}{Mean-Squared Error} simultaneously, as the SM-AP\abbrev{SM-AP}{Set-Membership Affine Projection} algorithm did. In addition, observe that $\gammabf(k)=\nbf(k)$ leads to the best results in terms of convergence rate and steady-state MSE\abbrev{MSE}{Mean-Squared Error}, but the 
performance of the SM-AP\abbrev{SM-AP}{Set-Membership Affine Projection} with SC-CV\abbrev{SC-CV}{Simple Choice CV} is quite close.  
The average number of updates required by the SM-AP\abbrev{SM-AP}{Set-Membership Affine Projection} algorithm using the general CV\abbrev{CV}{Constraint Vector}, the SC-CV\abbrev{SC-CV}{Simple Choice CV}, and the noise CV\abbrev{CV}{Constraint Vector} are
35$\%$, 9.7$\%$, and 3.6$\%$, respectively, implying that the last two CVs\abbrev{CV}{Constraint Vector} also have lower computational cost.     
It is worth noticing that even when using the general CV\abbrev{CV}{Constraint Vector}, the SM-AP\abbrev{SM-AP}{Set-Membership Affine Projection} algorithm still converges although it presents poor performance, 
as explained in Subsection~\ref{sub:divergence_sm_ap}.

\begin{figure}[t!]
\centering
\includegraphics[width=1\linewidth]{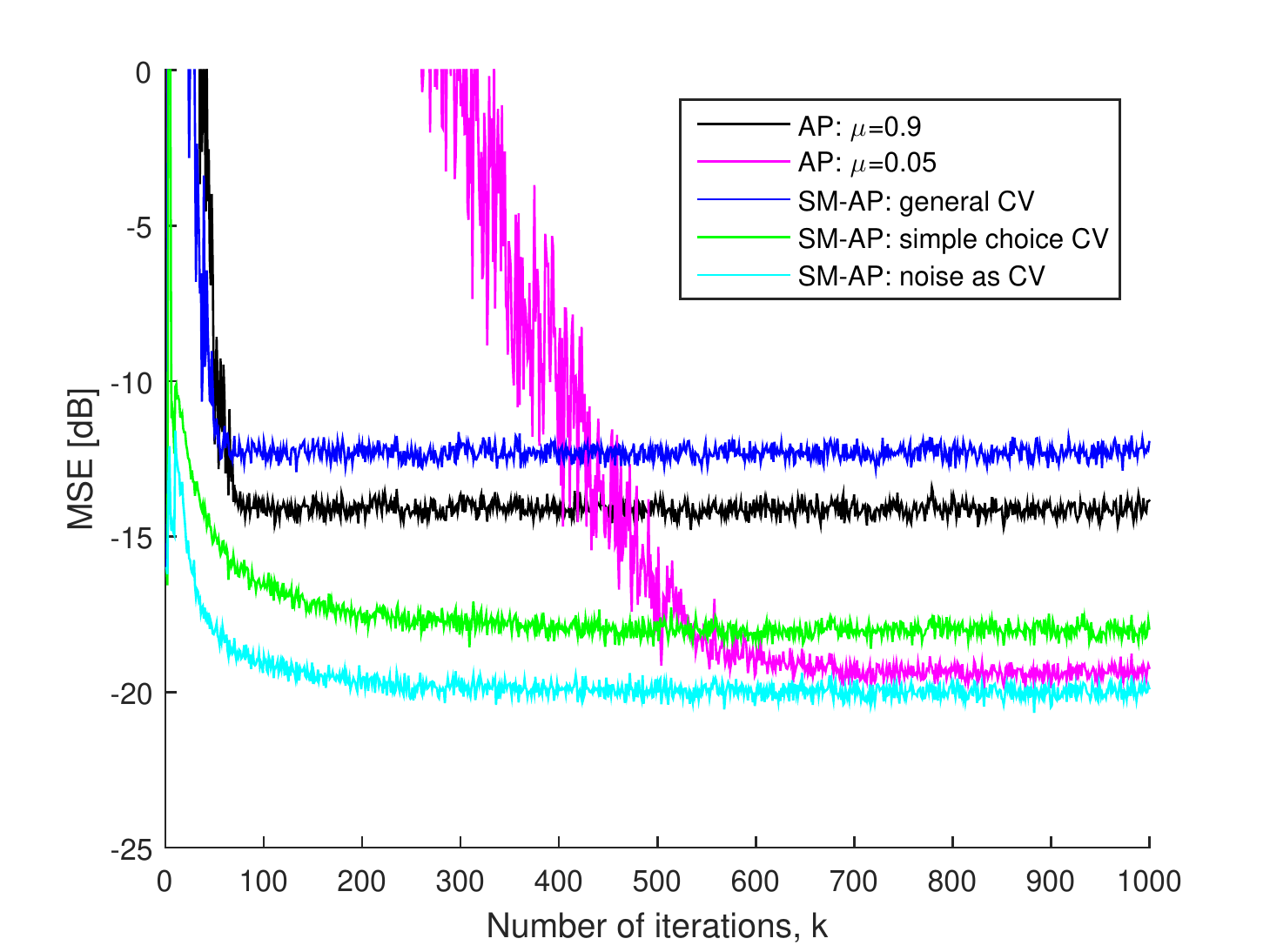}
\caption{Learning curves for the AP and SM-AP algorithm using different constraint vectors.  \label{fig:Robustness_learning_curves}}
\end{figure}


\section{Conclusion} \label{sec:conclusion-robustness}

In this chapter, we addressed the robustness (in the sense of $l_2$-stability) of the SM-NLMS\abbrev{SM-NLMS}{Set-Membership Normalized LMS} and the SM-AP\abbrev{SM-AP}{Set-Membership Affine Projection} algorithms.
In addition to the already known advantages of the SM-NLMS\abbrev{SM-NLMS}{Set-Membership Normalized LMS} algorithm over the NLMS\abbrev{NLMS}{Normalized LMS} algorithm, regarding accuracy and 
computational cost, in this chapter we demonstrated that: 
(i) the SM-NLMS\abbrev{SM-NLMS}{Set-Membership Normalized LMS} algorithm is robust regardless the choice of its parameters and 
(ii) the SM-NLMS\abbrev{SM-NLMS}{Set-Membership Normalized LMS} algorithm uses the input data very efficiently, i.e., it rarely produces a worse estimate $\wbf(k+1)$ 
during its update process. 
For the case where the noise bound is known, we explained how to set appropriately the parameter $\gammabar$ so that 
the SM-NLMS\abbrev{SM-NLMS}{Set-Membership Normalized LMS} algorithm {\it never generates a worse estimate}, i.e., the sequence $\{ \| \wbftilde(k) \|^2 \}$ (the squared Euclidean norm of the parameters deviation) becomes  
monotonously decreasing.  
For the case where the noise bound is unknown, we designed a time-varying parameter $\gammabar(k)$ that achieves simultaneously 
fast convergence and efficient use of the input data.

Unlike the SM-NLMS\abbrev{SM-NLMS}{Set-Membership Normalized LMS} algorithm, we demonstrated that there exists a condition to guarantee the $l_2$-stability 
of the SM-AP\abbrev{SM-AP}{Set-Membership Affine Projection} algorithm. 
This robustness condition depends on a parameter known as the  constraint vector (CV)\abbrev{CV}{Constraint Vector} $\gammabf(k)$. 
We proved the existence of vectors $\gammabf(k)$ satisfying such a condition, but practical choices remain unknown.
In addition, it was shown that the SM-AP\abbrev{SM-AP}{Set-Membership Affine Projection} with an adequate CV\abbrev{CV}{Constraint Vector} uses the input data more efficiently than the AP\abbrev{AP}{Affine Projection} algorithm.

We also demonstrated that both the SM-AP\abbrev{SM-AP}{Set-Membership Affine Projection} and SM-NLMS\abbrev{SM-NLMS}{Set-Membership Normalized LMS} algorithms do not diverge, even when their parameters are not properly 
selected, provided the noise is bounded. 
Finally, numerical results that corroborate our study were presented.
  \chapter{Trinion and Quaternion Set-Membership Affine Projection Algorithms}

The quaternions are a number system that extends the complex numbers. They were introduced by William Rowan Hamilton in 1843 for the first time~\cite{Hamilton_quaternion_PM1844}. Quaternions have several applications in multivariate signal processing problems, such as color image processing~\cite{Pei_crb_tsp2004,Guo_rbct_sp2011}, wind profile prediction~\cite{Took_qvstjf_RE2011,Barth_aqd_letter2014,Jiang_gqvgo_DSP2014}, and adaptive beamforming~\cite{Zhang_qvrab_SP2014}. A wide family of quaternion based algorithms have been introduced in adaptive filtering literatures~\cite{Ujang_qvnaf_TNN2011,Took_qlaafhp_TSP2009,Took_sqlfcla_icassp2009,Neto_nrcwlqa_SSP2011}. 

As a generalization of the complex domain, the quaternion domain provides a useful way to process 3- and 4-dimensional signals. Recently, several quaternion based adaptive filtering algorithms have appeared and they take benefit from the fact that the quaternion domain is a division algebra and it has a suitable data representation~\cite{Pei_quaternion_TIP1999,Bihan_quaternion_ICIP2003,Campa_quaternion_CDC2006}. Therefore, the quaternion algorithms allow a coupling between the components of 3- and 4-dimensional processes. Also, the quaternion-valued algorithm results in better performance compared to the real-valued algorithms, since it accounts for the coupling of the wind measurements and can be developed to exploit the augmented quaternion statistics~\cite{Took_qvstjf_RE2011}. As a by-product, in comparison with the real-valued algorithms in $\mathbb{R}^3$ and $\mathbb{R}^4$, they show enhanced stability and more degrees of freedom in the control of the adaptation mechanism.

However, when the signals involved in the adaptation process have only three dimensions, i.e., one real and two imaginary components, we can apply the trinion based algorithms. Using a data set for wind profile prediction, the trinion-valued least mean square (TLMS) algorithm is proposed~\cite{Guo_tdwpp_DSP2015} and its learning speed is compared with the quaternion least mean square (QLMS) algorithm~\cite{Barth_aqd_letter2014}. In the TLMS\abbrev{TLMS}{Trinion-Valued LMS} algorithm, the computational complexity is lower than QLMS\abbrev{QLMS}{Quaternion-Valued LMS} algorithm, since the implementation of a full quaternion-valued multiplication requires 16 and 12 real-valued multiplications and additions, respectively. In the trinion case, to multiply two 3-D numbers we only need 9 and 6 real-valued multiplications and additions, respectively. The quaternion affine projection (QAP)\abbrev{QAP}{Quaternion-Valued Affine Projection} algorithm~\cite{Jahanchahil_cqvapa_SP2013} has been applied to predict noncircular real-world 4-D wind, but it can also be used to 3-D profile wind prediction.

Here we consider a powerful approach to decrease the computational complexity of an adaptive filter by employing set-membership filtering (SMF)\abbrev{SMF}{Set-Membership Filtering} approach~\cite{Diniz_adaptiveFiltering_book2013,Gollamudi_smf_letter1998}. For real numbers, the set-membership NLMS~\cite{Gollamudi_smf_letter1998,Diniz_adaptiveFiltering_book2013} \abbrev{NLMS}{Normalized LMS}and AP~\cite{Werner_sm_ap_letter2001,Diniz_adaptiveFiltering_book2013,Diniz_sm_bnlms_tsp2003} algorithms were reviewed in Chapter 2. This chapter aims to generalize these algorithms to operate with  trinion and quaternion numbers. The trinion number system is not a mathematical field since there are elements which are not invertible. Therefore, to address this drawback, we replace the non-invertible element with an invertible one. In the quaternion number system, each nonzero element has inverse while the product operation is not commutative. The proposed algorithms get around these drawbacks.

Finally, we apply the trinion based algorithms to predicting the wind profile and compare their competitive performance with the quaternion based algorithms. However, the quaternion algorithms require remarkably higher computational complexity compared to their trinion counterparts. Also, we study the quaternion adaptive beamforming as an application of the quaternion-valued algorithms. In this manner, we will reduce the number of involved sensors in the adaptation mechanism. As a result, we can decrease the computational complexity and the energy consumption of the system.

Part of the content of this chapter was published in~\cite{Hamed_smtrinion-tcssII2016}. This chapter introduces new data selective adaptive filtering algorithms for trinion and quaternion number systems  $\mathbb{T}$ and $\mathbb{H}$. The work advances the set-membership trinion and quaternion-valued normalized least mean square (SMTNLMS\abbrev{SMTNLMS}{Set-Membership Trinion-Valued NLMS} and SMQNLMS)\abbrev{SMQNLMS}{Set-Membership Quaternion-Valued NLMS} and the set-membership trinion and quaternion-valued affine projection (SMTAP\abbrev{SMTAP}{Set-Membership Trinion-Valued AP} and SMQAP)\abbrev{SMQAP}{Set-Membership Quaternion-Valued AP} algorithms. Also, as individual cases, we obtain trinion and quaternion algorithms not employing the set-membership strategy. 

This chapter is organized as follows. Short introductions to quaternions and trinions are provided in Sections~\ref{sec:quaternions} and~\ref{sec:Trinions}, respectively. Section~\ref{sec:set-membership-trinion} briefly reviews the concept of SMF\abbrev{SMF}{Set-Membership Filtering} but instead of real numbers we use trinions and quaternions. The new trinion based SMTAP\abbrev{SMTAP}{Set-Membership Trinion-Valued AP} algorithm is derived in Section~\ref{sec:smtap_smnlms}. Section~\ref{sec:smqap_smqnlms} introduces the quaternion based SMQAP\abbrev{SMQAP}{Set-Membership Quaternion-Valued AP} algorithm. Section~\ref{sec:adaptive-beamforming-tr} reviews the application of quaternion-valued adaptive algorithms to adaptive beamforming. Simulations are presented in Section~\ref{sec:simulations-trinion} and Section~\ref{sec:conclusion-trinion} contains the conclusions.


\section{Quaternions} \label{sec:quaternions}

The quaternion number system is a non-commutative extension of complex numbers, denoted by $\mathbb{H}$. A quaternion $q\in\mathbb{H}$ is defined as~\cite{Hamilton_quaternion_PM1844} \symbl{$q_a$}{The real component of a quaternion $q$} \symbl{$q_b$}{The first imaginary component of a quaternion $q$} \symbl{$q_c$}{The second imaginary component of a quaternion $q$} \symbl{$q_d$}{The third imaginary component of a quaternion $q$}
\symbl{$\imath$}{The first orthogonal unit imaginary axis vector in quaternion numbers} \symbl{$\jmath$}{The second orthogonal unit imaginary axis vector in quaternion numbers} \symbl{$\kappa$}{The third orthogonal unit imaginary axis vector in quaternion numbers}
\begin{align}
q=q_a+q_b\imath+q_c\jmath+q_d\kappa,
\end{align}
where $q_a$, $q_b$, $q_c$, and $q_d$ are in $\mathbb{R}$. $q_a$ is the real component, while $q_b$, $q_c$, and $q_d$ are the three imaginary components. The orthogonal unit imaginary axis vectors $\imath$, $\jmath$, and $\kappa$ obey the following rules
\begin{align}
\imath\jmath=\kappa\qquad \jmath\kappa=\imath\qquad \kappa\imath=\jmath,\nonumber\\
\imath^2=\jmath^2=\kappa^2=\imath\jmath\kappa=-1.
\end{align}
Note that due to non-commutativity of the quaternion multiplication, we have $\jmath\imath=-\kappa\neq \imath\jmath$ for example. The element 1 is the identity element of $\mathbb{H}$, i.e., multiplication by 1 does nothing. The conjugate of a quaternion, denoted by $q^*$, is defined as \symbl{$(\cdot)^*$}{Conjugation operator}
\begin{align}
q^*=q_a-q_b\imath-q_c\jmath-q_d\kappa,
\end{align}
and the norm $|q|$ is given by
\begin{align}
|q|=\sqrt{qq^*}=\sqrt{q_a^2+q_b^2+q_c^2+q_d^2}.
\end{align}
The inverse of $q$ is introduced as 
\begin{align}
q^{-1}=\frac{q^*}{|q|^2}.
\end{align}
Observe that $q$ can be reformulated into the Cayley-Dickson~\cite{Zhang_qvrab_SP2014} form as
\begin{align}
q=\underbrace{(q_a+q_c\jmath)}_{z_1}+\imath\underbrace{(q_b+q_d\jmath)}_{z_2}, \label{eq:Cayley_Dickson}
\end{align}
where $z_1$ and $z_2$ are complex numbers.

The quaternion involutions are defined as follows~\cite{Ell_quaternion_involution_CMA2011,Mandic_quaternion_gradient_SPL2011}
\begin{align}
q^\imath=&-\imath q\imath=q_a+q_b\imath-q_c\jmath-q_d\kappa,\nonumber\\
q^\jmath=&-\jmath q\jmath=q_a-q_b\imath+q_c\jmath-q_d\kappa,\nonumber\\
q^\kappa=&-\kappa q\kappa=q_a-q_b\imath-q_c\jmath+q_d\kappa.
\end{align}
Therefore, we can present the four real components of a quaternion $q$ by the convolutions of $q$
\begin{align}
q_a=&\frac{1}{4}(q+q^\imath+q^\jmath+q^\kappa),\nonumber\\
q_b=&\frac{1}{4\imath}(q+q^\imath-q^\jmath-q^\kappa),\nonumber\\
q_c=&\frac{1}{4\jmath}(q-q^\imath+q^\jmath-q^\kappa),\nonumber\\
q_d=&\frac{1}{4\kappa}(q-q^\imath-q^\jmath+q^\kappa). \label{eq:convolution-relations-tr}
\end{align}
These expressions allow us presenting any quadrivariate or quaternion-valued function $f(q)$ as~\cite{Ell_quaternion_involution_CMA2011}
\begin{align}
f(q)=f(q_a,q_b,q_c,q_d)=f(q,q^\imath,q^\jmath,q^\kappa). \label{eq:convolution_expression_function_tr}
\end{align}

We know that the quaternion ring and $\mathbb{R}^4$ are isomorphic. Hence, by the same argument in the $\mathbb{C}\mathbb{R}$ calculus~\cite{Brandwood_gradient_FCRSP1983}, to introduce the duality between the derivatives of $f(q)\in\mathbb{H}$ and the derivatives of the corresponding quadrivariate real function $g(q_a,q_b,q_c,q_d)\in\mathbb{R}^4$, we begin with~\cite{Mandic_quaternion_gradient_SPL2011}
\begin{align}
f(q)=f_a(q_a,q_b,q_c,q_d)+&f_b(q_a,q_b,q_c,q_d)\imath+f_c(q_a,q_b,q_c,q_d)\jmath\nonumber\\+&f_d(q_a,q_b,q_c,q_d)\kappa=g(q_a,q_b,q_c,q_d).
\end{align}
The real variable function $g(q_a,q_b,q_c,q_d)$ has the following differential
\begin{align}
dg=&\frac{\partial g}{\partial q_a}dq_a+\frac{\partial g}{\partial q_b}dq_b+\frac{\partial g}{\partial q_c}dq_c+\frac{\partial g}{\partial q_d}dq_d\nonumber\\
=&\frac{\partial f(q)}{\partial q_a}dq_a+\frac{\partial f(q)}{\partial q_b}dq_b\imath+\frac{\partial f(q)}{\partial q_c}dq_c\jmath+\frac{\partial f(q)}{\partial q_d}dq_d\kappa. \label{eq:dg_tr}
\end{align}
By using the relations in~\eqref{eq:convolution-relations-tr}, the derivatives of the components of a quaternion $q$ are given by
\begin{align}
dq_a&=\frac{1}{4}(dq+dq^\imath+dq^\jmath+dq^\kappa),\nonumber\\
dq_b&=\frac{-\imath}{4}(dq+dq^\imath-dq^\jmath-dq^\kappa),\nonumber\\
dq_c&=\frac{-\jmath}{4}(dq-dq^\imath+dq^\jmath-dq^\kappa),\nonumber\\
dq_d&=\frac{-\kappa}{4}(dq-dq^\imath-dq^\jmath+dq^\kappa). \label{eq:dq_components-tr}
\end{align}
Also, using~\eqref{eq:convolution_expression_function_tr} we obtain
\begin{align}
df(q)=&\frac{\partial f(q,q^\imath,q^\jmath,q^\kappa)}{\partial q}dq+\frac{\partial f(q,q^\imath,q^\jmath,q^\kappa)}{\partial q^\imath}dq^\imath\nonumber\\
&+\frac{\partial f(q,q^\imath,q^\jmath,q^\kappa)}{\partial q^\jmath}dq^\jmath+\frac{\partial f(q,q^\imath,q^\jmath,q^\kappa)}{\partial q^\kappa}dq^\kappa. \label{eq:df-tr}
\end{align}
Therefore, by replacing the components of $dq$ from~\eqref{eq:dq_components-tr} in Equation~\eqref{eq:dg_tr}, and solving for the coefficients of $dq$, $dq^\imath$, $dq^\jmath$, $dq^\kappa$ from~\eqref{eq:dg_tr} and~\eqref{eq:df-tr}, we will obtain the $\mathbb{H}\mathbb{R}$-derivatives identities as follows
\begin{align}
\left[\begin{array}{c}\frac{\partial f(q,q^\imath,q^\jmath,q^\kappa)}{\partial q}\\\frac{\partial f(q,q^\imath,q^\jmath,q^\kappa)}{\partial q^\imath}\\\frac{\partial f(q,q^\imath,q^\jmath,q^\kappa)}{\partial q^\jmath}\\\frac{\partial f(q,q^\imath,q^\jmath,q^\kappa)}{\partial q^\kappa}\end{array}\right]=\frac{1}{4}\left[\begin{array}{cccc}1&-\imath&-\jmath&-\kappa\\1&-\imath&\jmath&\kappa\\1&\imath&-\jmath&\kappa\\1&\imath&\jmath&-\kappa\end{array}\right]\left[\begin{array}{c}\frac{\partial f}{\partial q_a}\\\frac{\partial f}{\partial q_b}\\\frac{\partial f}{\partial q_c}\\\frac{\partial f}{\partial q_d}\end{array}\right].
\end{align}
Our interest is in the derivative $\frac{\partial f(q,q^\imath,q^\jmath,q^\kappa)}{\partial q}$, thus the gradient of $f(q)$ with respect to $q$ is given by~\cite{Mandic_quaternion_gradient_SPL2011}
\begin{align}
\nabla_qf=\frac{1}{4}(\frac{\partial f}{\partial q_a}-\frac{\partial f}{\partial q_b}\imath-\frac{\partial f}{\partial q_c}\jmath-\frac{\partial f}{\partial q_d}\kappa)=\frac{1}{4}(\nabla_{q_a}f-\nabla_{q_b}f\imath-\nabla_{q_c}f\jmath-\nabla_{q_d}f\kappa).
\end{align}

The real values elements $q_a$, $q_b$, $q_c$, $q_d$ of a quaternion $q$ can be presented in terms of $q^*$, $q^{\imath^*}$, $q^{\jmath^*}$, $q^{\kappa^*}$ as follows~\cite{Mandic_quaternion_gradient_SPL2011}
\begin{align}
q_a&=\frac{1}{4}(q^*+q^{\imath^*}+q^{\jmath^*}+q^{\kappa^*}),\nonumber\\
q_b&=\frac{1}{4\imath}(-q-q^{\imath^*}+q^{\jmath^*}+q^{\kappa^*}),\nonumber\\
q_c&=\frac{1}{4\jmath}(-q+q^{\imath^*}-q^{\jmath^*}+q^{\kappa^*}),\nonumber\\
q_d&=\frac{1}{4\kappa}(-q^*+q^{\imath^*}+q^{\jmath^*}-q^{\kappa^*}).
\end{align}
Then the derivative of the function $f(q)=f(q^*,q^{\imath^*},q^{\jmath^*},q^{\kappa^*})$ can be expressed as
\begin{align}
df(q)=&\frac{\partial f(q^*,q^{\imath^*},q^{\jmath^*},q^{\kappa^*})}{\partial q^*}dq^*+\frac{\partial f(q^*,q^{\imath^*},q^{\jmath^*},q^{\kappa^*})}{\partial q^{\imath^*}}dq^{\imath^*}\nonumber\\
&+\frac{\partial f(q^*,q^{\imath^*},q^{\jmath^*},q^{\kappa^*})}{\partial q^{\jmath^*}}dq^{\jmath^*}+\frac{\partial f(q^*,q^{\imath^*},q^{\jmath^*},q^{\kappa^*})}{\partial q^{\kappa^*}}dq^{\kappa^*}.
\end{align}
Also, the derivative of the quadrivariate $g(q_a,q_b,q_c,q_d)$ is given by
\begin{align}
dg(q_a,q_b,q_c,q_d)=Adq^*+Bdq^{\imath^*}+Cdq^{\jmath^*}+Ddq^{\kappa^*}.
\end{align}
By the same argument above, if we solve for the coefficients of $dq^*$, $dq^{\imath^*}$, $dq^{\jmath^*}$, $dq^{\kappa^*}$ then we will obtain the $\mathbb{H}\mathbb{R}^*$-derivatives identities,
\begin{align}
\left[\begin{array}{c}\frac{\partial f(q^*,q^{\imath^*},q^{\jmath^*},q^{\kappa^*})}{\partial q^*}\\\frac{\partial f(q^*,q^{\imath^*},q^{\jmath^*},q^{\kappa^*})}{\partial q^{\imath^*}}\\\frac{\partial f(q^*,q^{\imath^*},q^{\jmath^*},q^{\kappa^*})}{\partial q^{\jmath^*}}\\\frac{\partial f(q^*,q^{\imath^*},q^{\jmath^*},q^{\kappa^*})}{\partial q^{\kappa^*}}\end{array}\right]=\frac{1}{4}\left[\begin{array}{cccc}1&\imath&\jmath&\kappa\\1&\imath&-\jmath&-\kappa\\1&-\imath&\jmath&-\kappa\\1&-\imath&-\jmath&\kappa\end{array}\right]\left[\begin{array}{c}\frac{\partial f}{\partial q_a}\\\frac{\partial f}{\partial q_b}\\\frac{\partial f}{\partial q_c}\\\frac{\partial f}{\partial q_d}\end{array}\right].
\end{align}
The derivative $\frac{\partial f(q^*,q^{\imath^*},q^{\jmath^*},q^{\kappa^*})}{\partial q^*}$ is of particular interest, thus the gradient of $f(q)$ with respect to $q^*$ is given by~\cite{Mandic_quaternion_gradient_SPL2011}
\begin{align}
\nabla_{q^*}f&=\frac{1}{4}(\frac{\partial f}{\partial q_a}+\frac{\partial f}{\partial q_b}\imath+\frac{\partial f}{\partial q_c}\jmath+\frac{\partial f}{\partial q_d}\kappa)=\frac{1}{4}(\nabla_{q_a}f+\nabla_{q_b}f\imath+\nabla_{q_c}f\jmath+\nabla_{q_d}f\kappa).
\end{align}


\section{Trinions} \label{sec:Trinions}   

As a group, the trinion number system $\mathbb{T}$ is isomorphic to $\mathbb{R}^3$. A number $v$ in $\mathbb{T}$ is composed of one real part, $v_a$, and two imaginary parts, $v_b$ and $v_c$, \symbl{$v_a$}{The real component of a trinion $v$} \symbl{$v_b$}{The first imaginary component of a trinion $v$} \symbl{$v_c$}{The third imaginary component of a trinion $v$} \symbl{$\bar{\imath}$}{The first orthogonal unit imaginary axis vector in trinion numbers} \symbl{$\bar{\jmath}$}{The second orthogonal unit imaginary axis vector in trinion numbers}
\begin{align}
v=v_a+v_b\bar\imath+v_c\bar\jmath.
\end{align}
The number system $\mathbb{T}$ has three operations: addition, scalar multiplication, and trinion multiplication. The sum of two elements of $\mathbb{T}$ is defined to be their sum as elements of $\mathbb{R}^3$. Similarly the product of an element of $\mathbb{T}$ by a real number is defined to be the same as the product by a scalar in $\mathbb{R}^3$. To make a commutative algebraic group of the basis elements 1, $\bar\imath$, and $\bar\jmath$ the following rules apply \cite{Assefa_tftci_SP2011}
\begin{align}
\bar\imath^2=\bar\jmath,~\bar\imath\bar\jmath=\bar\jmath\bar\imath=-1,~\bar\jmath^2=-\bar\imath.
\end{align}
Trinions with these rules set a commutative mathematical ring, i.e., $vw=wv$ for $v,w\in\mathbb{T}$. The basis element 1 will be the identity element of $\mathbb{T}$, meaning that multiplication by 1 does nothing. The conjugate of $v$ is given by~\cite{Guo_tdwpp_DSP2015}
\begin{align}
v^*=v_a-v_b\bar\jmath-v_c\bar\imath,
\end{align}
and the norm by~\cite{Guo_tdwpp_DSP2015} \symbl{$\Re(\cdot)$}{The real part of $(\cdot)$}
\begin{align}
|v|=\sqrt{\Re(vv^*)}=\sqrt{v_a^2+v_b^2+v_c^2}.
\end{align}

The inverse of $v$, if exists, is $w=(w_a+w_b\bar\imath+w_c\bar\jmath)\in\mathbb{T}$ such that $vw=wv=1$. To solve this equation we consider $v=[v_a~v_b~v_c]^T$ and $w=[w_a~w_b~w_c]^T$ then we get 
\begin{align}
\left\{\begin{array}{l}v_aw_a-v_cw_b-v_bw_c=1,\\
v_bw_a+v_aw_b-v_cw_c=0,\\
v_cw_a+v_bw_b+v_aw_c=0,\end{array}\right.
\end{align}
or in the matrix form $\Abf w=[1~0~0]^T$ where $\Abf$ is given by
\begin{align}
\Abf=\left[\begin{array}{ccc}v_a&-v_c&-v_b\\ v_b&v_a&-v_c\\ v_c&v_b&v_a\end{array}\right],
\end{align}
thus $w=\Abf^{-1}[1~0~0]^T$. When the determinant of $\Abf$ is zero, the inverse of $v$ does not exist. In order to get around this problem when the determinant of $\Abf$ is zero, we define $\Abf=\delta \Ibf$ where $\delta$ is a small positive constant and $\Ibf$ is a $3\times3$ identity matrix. Note that $\Abf$ is replaced by the identity matrix multiplied by a small constant in order to avoid numerical problems in the matrix inversion. This strategy avoids division by zero in the trinion-valued algorithms. We will now define $v^{-1}=\Abf^{-1}[1~0~0]^T$.

In the field of complex numbers, a variable $z$ and its conjugate $z^*$ can be considered as two independent variables, so that the complex-valued gradient can be defined \cite{Bos_cgh_PVISP1994}. As far as we know, the trinion involutions, $v^{\bar{\imath}}$ and $v^{\bar{\jmath}}$, are not available in general. In this chapter, we use the following formulas for the gradients of a function $f(v)$ with respect to the trinion-valued variable $v$ and its conjugate \cite{Guo_tdwpp_DSP2015}
\begin{equation}
\begin{aligned}
\nabla_vf&=\frac{1}{3}(\nabla_{v_a}f-\nabla_{v_b}f\bar\jmath-\nabla_{v_c}f\bar\imath),\\
\nabla_{v^*}f&=\frac{1}{3}(\nabla_{v_a}f+\nabla_{v_b}f\bar\imath+\nabla_{v_c}f\bar\jmath),
\end{aligned}
\end{equation}
where $v=v_a+v_b\bar\imath+v_c\bar\jmath$.


\section{Set-Membership Filtering (SMF) in $\mathbb{T}$ and $\mathbb{H}$} \label{sec:set-membership-trinion} 

The target of the SMF\abbrev{SMF}{Set-Membership Filtering} is to design $\wbf$ such that the magnitude of the estimation error is upper bounded by a predetermined parameter $\gammabar$. The value of $\gammabar$ can change with the specific application. If the value of $\gammabar$ is suitably selected, there are many valid estimates for $\wbf$. Suppose that ${\cal S}$ denotes the set of all possible input-desired data pairs $(\xbf,d)$ of interest and define $\Theta$ as the set of all vectors $\wbf$ whose magnitudes of their estimation errors are upper bounded by $\gammabar$ whenever $(\xbf,d)\in{\cal S}$. The set $\Theta$ is named feasibility set and is given by
\begin{align}
\Theta\triangleq\bigcap_{(\xbf,d)\in{\cal S}}\{\wbf\in\mathbb{F}^{N+1}:|d-\wbf^H\xbf|\leq\gammabar\},
\end{align}
where $\mathbb{F}$ is $\mathbb{T}$ or $\mathbb{H}$. 
Let's define the constraint set ${\cal H}(k)$ consisting of all vectors $\wbf$ such that their estimation errors at time instant $k$ are upper bounded in magnitude by $\gammabar$,
\begin{align}
{\cal H}(k)\triangleq\{\wbf\in\mathbb{F}^{N+1}:|d(k)-\wbf^H\xbf(k)|\leq\gammabar\}.
\end{align}
The membership set $\psi(k)$ defined as
\begin{align}
\psi(k)\triangleq\bigcap_{i=0}^k{\cal H}(i)\label{eq:set-membership_set-trinion}
\end{align}
will include $\Theta$ and will coincide with $\Theta$ if all data pairs in ${\cal S}$ are traversed up to time instant $k$. Owing to difficulties to compute $\psi(k)$, adaptive approaches are required~\cite{Gollamudi_smf_letter1998}. The easiest route is to compute a point estimate using, for example, the information provided by the constraint set ${\cal H}(k)$ like in the set-membership NLMS\abbrev{NLMS}{Normalized LMS} algorithm~\cite{Gollamudi_smf_letter1998}, or several previous constraint sets as is done in the set-membership affine projection algorithm~\cite{Werner_sm_ap_letter2001}.


\section{SMTAP Algorithm} \label{sec:smtap_smnlms} 

In this section, we propose the SMTAP\abbrev{SMTAP}{Set-Membership Trinion-Valued AP} algorithm. This trinion-valued algorithm is the counterpart of the real-valued SM-AP\abbrev{SM-AP}{Set-Membership Affine Projection} algorithm. Then we derive the update equations for the simpler algorithms related to the normalized LMS \abbrev{LMS}{Least-Mean-Square}algorithm.

The membership set $\psi(k)$ defined in (\ref{eq:set-membership_set-trinion}) encourages the use of more constraint sets in the update. Therefore, we elaborate an algorithm whose updates belong to a set composed of $L+1$ constraint sets. 

For this purpose, we express $\psi(k)$ as
\begin{align}
\psi(k)=\bigcap_{i=0}^{k-L-1}{\cal H}(i)\bigcap_{j=k-L}^k{\cal H}(j)=\psi^{k-L-1}(k)\bigcap\psi^{L+1}(k), \label{eq:divide_intersection-trinion}
\end{align}
where $\psi^{L+1}(k)$ indicates the intersection of the $L+1$ last constraint sets, and $\psi^{k-L-1}(k)$ represents the intersection of the first $k-L$ constraint sets. Our goal is to formulate an algorithm whose coefficient update belongs to the last $L+1$ constraint sets, i.e., $\wbf(k+1)\in\psi^{L+1}(k)$. \symbl{$\psi^{L+1}(k)$}{The intersection of the $L+1$ last constraint sets}

Assume that ${\cal S}(k-i)$ denotes the set which includes all vectors $\wbf$ such that $d(k-i)-\wbf^H\xbf(k-i)=\gamma_i(k)$, for $i=0,\cdots,L$. All choices for $\gamma_i(k)$ satisfying the bound constraint are valid. That is, if all $\gamma_i(k)$ are selected such that $|\gamma_i(k)|\leq\gammabar$, then ${\cal S}(k-i)\in{\cal H}(k-i)$, for $i=0,\cdots,L$.

The objective function which we ought to minimize can now be stated. A coefficient update is implemented whenever $\wbf(k)\not\in\psi^{L+1}(k)$ as follows
\begin{align}
&\min\frac{1}{2}\|\wbf(k+1)-\wbf(k)\|^2\nonumber\\
&\text{subject to:}\nonumber\\
&\dbf(k)-(\wbf^H(k+1)\Xbf(k))^T=\gammabf(k),\label{eq:constraint-trinion}
\end{align} 
where 

\begin{tabular}{ll}
$\dbf(k)\in\mathbb{T}^{(L+1)\times1}$& contains the desired output from the $L+1$ last \\&time instants;\\
$\gammabf(k)\in\mathbb{T}^{(L+1)\times1}$&specifies the point in $\psi^{L+1}(k)$;\\
$\Xbf(k)\in\mathbb{T}^{(N+1)\times(L+1)}$&contains the corresponding input vectors, i.e.,
\end{tabular}
\begin{equation}
\begin{aligned}
\dbf(k)&=[d(k)~d(k-1)~\cdots~d(k-L)]^T,\\
\gammabf(k)&=[\gamma_0(k)~\gamma_1(k)~\cdots~\gamma_L(k)]^T,\\
\Xbf(k)&=[\xbf(k)~\xbf(k-1)~\cdots~\xbf(k-L)], \label{eq:pack-trinion}
\end{aligned}
\end{equation}
with $\xbf(k)$ being the input-signal vector
\begin{align}
\xbf(k)=[x(k)~x(k-1)~\cdots~x(k-N)]^T. \label{eq:x(k)-trinion}
\end{align}
If we use the method of Lagrange multipliers to transform a  constrained minimization into an unconstrained one, then we have to minimize
\begin{align}
F[\wbf(k+1)]=&\frac{1}{2}\|\wbf(k+1)-\wbf(k)\|^2\nonumber\\
&+\Re\{\lambdabf^T(k)[\dbf(k)-(\wbf^H(k+1)\Xbf(k))^T-\gammabf(k)]\}, \label{eq:objective-trinion}
\end{align}
where $\lambdabf(k)\in\mathbb{T}^{(L+1)\times1}$ is a vector of Lagrange multipliers. To find the minimum solution, we must calculate the following gradient
\begin{align}
\nabla_{\wbf^*(k+1)}F[\wbf(k+1)]=&\frac{1}{3}\Big[\nabla_{\wbf_a(k+1)}F[\wbf(k+1)]+\nabla_{\wbf_b(k+1)}F[\wbf(k+1)]\bar\imath\nonumber\\
&+\nabla_{\wbf_c(k+1)}F[\wbf(k+1)]\bar\jmath\Big].\label{eq:gradient-trinion}
\end{align}
In order to find the above gradient, we ought to calculate the cost function $F[\wbf(k+1)]$ as a function of real-valued variables. As a result we have,
\begin{align}
\|\wbf(k+1)-\wbf(k)\|^2=&\|\wbf_a(k+1)-\wbf_a(k)\|^2+\|\wbf_b(k+1)-\wbf_b(k)\|^2\nonumber\\
&+\|\wbf_c(k+1)-\wbf_c(k)\|^2.\label{eq:first_part_cost-trinion}
\end{align}
We drop the time index '$k$' for the sake of compact notation. In order to find the second term in (\ref{eq:objective-trinion}) as a real-valued term we perform the following calculations,
\begin{align}
\Re&\{\lambdabf^T[\dbf-\Xbf^T\wbf^*(k+1)-\gammabf]\}
=\Re\{(\lambdabf_a^T+\lambdabf_b^T\bar\imath+\lambdabf_c^T\bar\jmath)[(\dbf_a+\dbf_b\bar\imath+\dbf_c\bar\jmath)\nonumber\\
&-(\Xbf_a^T+\Xbf_b^T\bar\imath+\Xbf_c^T\bar\jmath)(\wbf_a(k+1)-\wbf_b(k+1)\bar\jmath-\wbf_c(k+1)\bar\imath)-(\gammabf_a+\gammabf_b\bar\imath+\gammabf_c\bar\jmath)]\}\nonumber\\
=&\Re\{(\lambdabf_a^T+\lambdabf_b^T\bar\imath+\lambdabf_c^T\bar\jmath)[(\dbf_a-\Xbf_a^T\wbf_a(k+1)-\Xbf_b^T\wbf_b(k+1)-\Xbf_c^T\wbf_c(k+1)-\gammabf_a)\nonumber\\
&+(\dbf_b-\Xbf_b^T\wbf_a(k+1)+\Xbf_a^T\wbf_c(k+1)-\Xbf_c^T\wbf_b(k+1)-\gammabf_b)\bar\imath\nonumber\\
&+(\dbf_c-\Xbf_c^T\wbf_a(k+1)+\Xbf_a^T\wbf_b(k+1)+\Xbf_b^T\wbf_c(k+1)-\gammabf_c)\bar\jmath]\}\nonumber\\
=&\lambdabf_a^T(\dbf_a-\Xbf_a^T\wbf_a(k+1)-\Xbf_b^T\wbf_b(k+1)-\Xbf_c^T\wbf_c(k+1)-\gammabf_a)\nonumber\\
&-\lambdabf_b^T(\dbf_c-\Xbf_c^T\wbf_a(k+1)+\Xbf_a^T\wbf_b(k+1)+\Xbf_b^T\wbf_c(k+1)-\gammabf_c)\nonumber\\
&-\lambdabf_c^T(\dbf_b-\Xbf_b^T\wbf_a(k+1)+\Xbf_a^T\wbf_c(k+1)-\Xbf_c^T\wbf_b(k+1)-\gammabf_b).\label{eq:second_part_cost-trinion}
\end{align}
Therefore, by (\ref{eq:objective-trinion}), (\ref{eq:first_part_cost-trinion}), and (\ref{eq:second_part_cost-trinion}) we obtain
\begin{align}
F[\wbf(k+1)]=\frac{1}{2}\text{Eq.}\eqref{eq:first_part_cost-trinion}+\text{Eq.}\eqref{eq:second_part_cost-trinion}.
\end{align}
Thus, the three component-wise gradients can be attained as
\begin{align}
\nabla_{\wbf_a(k+1)}F[\wbf(k+1)]=&(\wbf_a(k+1)-\wbf_a(k))-\lambdabf_a^T\Xbf_a^T+\lambdabf_b^T\Xbf_c^T+\lambdabf_c^T\Xbf_b^T,\label{eq:component_gradient1-trinion}\\
\nabla_{\wbf_b(k+1)}F[\wbf(k+1)]=&(\wbf_b(k+1)-\wbf_b(k))-\lambdabf_a^T\Xbf_b^T-\lambdabf_b^T\Xbf_a^T+\lambdabf_c^T\Xbf_c^T,\label{eq:component_gradient2-trinion}\\
\nabla_{\wbf_c(k+1)}F[\wbf(k+1)]=&(\wbf_c(k+1)-\wbf_c(k))-\lambdabf_a^T\Xbf_c^T-\lambdabf_b^T\Xbf_b^T-\lambdabf_c^T\Xbf_a^T.\label{eq:component_gradient3-trinion}
\end{align}
On the other hand, we have
\begin{align}
\Xbf\lambdabf=&(\Xbf_a+\Xbf_b\bar\imath+\Xbf_c\bar\jmath)(\lambdabf_a+\lambdabf_b\bar\imath+\lambdabf_c\bar\jmath)\nonumber\\
=&(\Xbf_a\lambdabf_a-\Xbf_b\lambdabf_c-\Xbf_c\lambdabf_b)+(\Xbf_a\lambdabf_b+\Xbf_b\lambdabf_a-\Xbf_c\lambdabf_c)\bar\imath\nonumber\\
&+(\Xbf_a\lambdabf_c+\Xbf_c\lambdabf_a+\Xbf_b\lambdabf_b)\bar\jmath.\label{eq:xlambda-trinion}
\end{align}
Overall, by employing Equations (\ref{eq:gradient-trinion}) and (\ref{eq:component_gradient1-trinion})-(\ref{eq:xlambda-trinion}), we get,

\begin{align}
\nabla_{\wbf^*(k+1)}F[\wbf(k+1)]=&\frac{1}{3}\{[(\wbf_a(k+1)-\wbf_a(k))-(\Xbf(k)\lambdabf(k))_a]\nonumber\\
&+[(\wbf_b(k+1)-\wbf_b(k))-(\Xbf(k)\lambdabf(k))_b]\bar\imath\nonumber\\
&+[(\wbf_c(k+1)-\wbf_c(k))-(\Xbf(k)\lambdabf(k))_c]\bar\jmath\}\nonumber\\
=&\frac{1}{3}[\wbf(k+1)-\wbf(k)-\Xbf(k)\lambdabf(k)].
\end{align}
After setting the above equation equal to zero, we obtain
\begin{align}
\wbf(k+1)=\wbf(k)+\Xbf(k)\lambdabf(k).\label{(eq:update_with_lambda-trinion)}
\end{align}
If we substitute (\ref{(eq:update_with_lambda-trinion)}) in the constraint relation (\ref{eq:constraint-trinion}) the following expression results,
\begin{align}
\Xbf^T(k)\Xbf^*(k)\lambdabf^*(k)=\dbf(k)-\Xbf^T(k)\wbf^*(k)-\gammabf(k)=(\ebf(k)-\gammabf(k)).
\end{align}
From the above equation we get $\lambdabf(k)$ as
\begin{align}
\lambdabf(k)=(\Xbf^H(k)\Xbf(k))^{-1}(\ebf(k)-\gammabf(k))^*,\label{eq:lambda-trinion}
\end{align}
where
\begin{align}
\ebf(k)&=[e(k)~\epsilon(k-1)~\cdots~\epsilon(k-L)]^T, \label{eq:e_ap-trinion}
\end{align}
with $e(k)=d(k)-\wbf^H(k)\xbf(k)$, and $\epsilon(k-i)=d(k-i)-\wbf^H(k)\xbf(k-i)$ for $i=1,\cdots,L$. We can now conclude the SMTAP\abbrev{SMTAP}{Set-Membership Trinion-Valued AP} algorithm by starting from (\ref{(eq:update_with_lambda-trinion)}) with $\lambdabf(k)$ being given by (\ref{eq:lambda-trinion}), i.e.,
\begin{align}
\wbf(k+1)=\left\{\begin{array}{ll}\wbf(k)+\pbf_{\rm ap}(k)&\text{if}~|e(k)|>\gammabar,\\\wbf(k)&\text{otherwise},\end{array}\right.\label{eq:update_SMTAP}
\end{align}
where
\begin{align}
\pbf_{\rm ap}(k)&=\Xbf(k)(\Xbf^H(k)\Xbf(k))^{-1}(\ebf(k)-\gammabf(k))^*\label{eq:P(k)-trinion}.
\end{align}

{\it Remark 1:} In order to check if an update $\wbf(k+1)$ is required, we only have to test if $\wbf(k)\not\in{\cal H}(k)$ since in the previous updates $\wbf(k)\in{\cal H}(k-i+1)$ is guaranteed for $i=2,\cdots,L+1$.

{\it Remark 2:} For the initial time instants $k<L+1$, i.e., during initialization, only the knowledge of ${\cal H}(i)$ for $i=0,1,\cdots,k$ is available. As a consequence, if an update is required for $k<L+1$, the algorithm is implemented with the available $k+1$ accessible constraint sets.

{\it Remark 3:} By adopting the bound $\gammabar=0$, the algorithm will convert to the trinion affine projection (TAP)\abbrev{TAP}{Trinion-Valued Affine Projection} algorithm with unity step size which is the generalization of the conventional real-valued AP\abbrev{AP}{Affine Projection} algorithm in $\mathbb{T}$. Therefore, the TAP\abbrev{TAP}{Trinion-Valued Affine Projection} algorithm can be described as
\begin{align}
\wbf(k+1)=\wbf(k)+\mu \pbf'_{\rm ap}(k),\label{eq:TAP}
\end{align} 
where $\mu$ is the convergence factor and 
\begin{align}
\pbf'_{\rm ap}(k)=\Xbf(k)(\Xbf^H(k)\Xbf(k))^{-1}\ebf^*(k).
\end{align}

Note that we can utilize (\ref{eq:update_SMTAP}) and derive the update equation of the SMTNLMS\abbrev{SMTNLMS}{Set-Membership Trinion-Valued NLMS} algorithm. In this case we have to evade data-reusing in (\ref{eq:update_SMTAP}), $L=0$, so that the updating equation becomes,
\begin{align}
\wbf(k+1)=\left\{\begin{array}{ll}\wbf(k)+\pbf(k)&\text{if}~|e(k)|>\gammabar,\\\wbf(k)&\text{otherwise},\end{array}\right.\label{eq:smtnlms_without_norm}
\end{align}
where
\begin{align}
\pbf(k)&=\xbf(k)(\xbf^H(k)\xbf(k))^{-1}(e(k)-\gamma(k))^*,\\
e(k)&=d(k)-\wbf^H(k)\xbf(k). \label{eq:e-trinion}
\end{align}
We will now choose $\gamma(k)=\frac{\gammabar e(k)}{|e(k)|}$, hence from (\ref{eq:smtnlms_without_norm}) we attain the SMTNLMS\abbrev{SMTNLMS}{Set-Membership Trinion-Valued NLMS} update equation as
\begin{align}
\wbf(k+1)=\wbf(k)+\mu(k)\xbf(k)(\xbf^H(k)\xbf(k))^{-1}e^*(k),\label{eq:smtnlms}
\end{align}
where
\begin{align}
\mu(k)&=\left\{\begin{array}{ll}1-\frac{\gammabar}{|e(k)|}&\text{if}~|e(k)|>\gammabar,\\0&\text{otherwise}.\end{array}\right. \label{eq:mu-trinion}
\end{align}
Recalling that the normalized LMS \abbrev{LMS}{Least-Mean-Square}algorithm can be derived as a particular case of AP\abbrev{AP}{Affine Projection} algorithm for $L=0$.

{\it Remark 4:} By choosing the bound $\gammabar=0$ in (\ref{eq:smtnlms}), the algorithm will reduce to the TNLMS\abbrev{TNLMS}{Trinion-Valued Normalized LMS} algorithm with unity step size which is the generalization of the popular real-valued NLMS\abbrev{NLMS}{Normalized LMS} algorithm in $\mathbb{T}$. As a result, TNLMS\abbrev{TNLMS}{Trinion-Valued Normalized LMS} algorithm can be described as

\begin{align}
\wbf(k+1)=\wbf(k)+\mu\xbf(k)(\xbf^H(k)\xbf(k))^{-1}e^*(k),
\end{align}
where $\mu$ is the convergence factor.


\section{SMQAP Algorithm} \label{sec:smqap_smqnlms}
 
This section outlines the derivation of the SMQAP\abbrev{SMQAP}{Set-Membership Quaternion-Valued AP} algorithm. Then we obtain an update equation for the SMQNLMS\abbrev{SMQNLMS}{Set-Membership Quaternion-Valued NLMS} algorithm that follows the same steps as the derivation of the SMTNLMS\abbrev{SMTNLMS}{Set-Membership Trinion-Valued NLMS} algorithm. The SMQAP\abbrev{SMQAP}{Set-Membership Quaternion-Valued AP} and the SMQNLMS\abbrev{SMQNLMS}{Set-Membership Quaternion-Valued NLMS} algorithms are the quaternion versions of the real-valued SM-AP\abbrev{SM-AP}{Set-Membership Affine Projection} and SM-NLMS\abbrev{SM-NLMS}{Set-Membership Normalized LMS} algorithms, respectively.

The membership set $\psi(k)$ introduced in (\ref{eq:set-membership_set-trinion}) suggests the use of more constraint sets in the update. Let us express $\psi(k)$ as in (\ref{eq:divide_intersection-trinion}), our purpose is to derive an algorithm whose coefficient update belongs to the last $L+1$ constraint set, i.e., $\wbf(k+1)\in\psi^{L+1}(k)$. Suppose that ${\cal S}(k-i)$ describes the set which contains all vectors $\wbf$ such that $d(k-i)-\wbf^H\xbf(k-i)=\gamma_i(k)$, for $i=0,\cdots,L$. All choices for $\gamma_i(k)$ satisfying the bound constraint are valid. That is, if all $\gamma_i(k)$ are chosen such that $|\gamma_i(k)|\leq\gammabar$, then ${\cal S}(k-i)\in{\cal H}(k-i)$, for $i=0,\cdots,L$.

The objective function to be minimized in case of the SMQAP\abbrev{SMQAP}{Set-Membership Quaternion-Valued AP} algorithm can be stated as follows: perform a coefficient update whenever $\wbf(k)\not\in\psi^{L+1}(k)$ as in Equation (\ref{eq:constraint-trinion}). Note that $\dbf(k),\gammabf(k)\in\mathbb{H}^{(L+1)\times1}$, $\Xbf(k)\in\mathbb{H}^{(N+1)\times(L+1)}$, and $\xbf(k)$ are defined as in (\ref{eq:pack-trinion}) and (\ref{eq:x(k)-trinion}).

By employing the method of Lagrange multipliers, the unconstrained function to be minimized becomes as in Equation (\ref{eq:objective-trinion}), where $\lambdabf(k)\in\mathbb{H}^{(L+1)\times1}$ is a vector of Lagrange multipliers. After setting the gradient of $F[\wbf(k+1)]$ with respect to $\wbf^*(k+1)$ equal to zero, we will get the equation
\begin{align}
\wbf(k+1)=\wbf(k)+\Xbf(k)\lambdabf(k). \label{eq:update_smqap_with_lambda}
\end{align}
Then, by invoking the constraints in (\ref{eq:constraint-trinion}), the expression of $\lambdabf(k)$ is as
\begin{align}
\lambdabf(k)=(\Xbf^H(k)\Xbf(k))^{-1}(\ebf(k)-\gammabf(k))^*, \label{eq:lambda_q-trinion}
\end{align}
where $\ebf(k)$ is defined as in (\ref{eq:e_ap-trinion}). Finally, the update equation for the SMQAP\abbrev{SMQAP}{Set-Membership Quaternion-Valued AP} algorithm is given by
\begin{align}
\wbf(k+1)=\left\{\begin{array}{ll}\wbf(k)+\qbf_{\rm ap}(k)&\text{if}~|e(k)|>\gammabar,\\\wbf(k)&\text{otherwise},\end{array}\right.\label{eq:update_SMQAP}
\end{align}
where
\begin{align}
\qbf_{\rm ap}(k)&=\Xbf(k)(\Xbf^H(k)\Xbf(k))^{-1}(\ebf(k)-\gammabf(k))^*\label{eq:Q(k)-trinion}.
\end{align}

Note that the {\it Remarks 1} and {\it 2} of Subsection \ref{sec:smtap_smnlms} also apply to the SMQAP\abbrev{SMQAP}{Set-Membership Quaternion-Valued AP} algorithm.

{\it Remark 5:} We can quickly verify that adopting the bound $\gammabar=0$, the algorithm will reduce to QAP\abbrev{QAP}{Quaternion-Valued Affine Projection} algorithm \cite{Jahanchahil_cqvapa_SP2013} with unity step size. Therefore, the QAP\abbrev{QAP}{Quaternion-Valued Affine Projection} algorithm cab be expressed as
\begin{align}
\wbf(k+1)=\wbf(k)+\mu \Xbf(k)(\Xbf^H(k)\Xbf(k))^{-1}\ebf^*(k),\label{eq:QAP}
\end{align} 
where $\mu$ is the convergence factor.

Note that we can use the SMQAP\abbrev{SMQAP}{Set-Membership Quaternion-Valued AP} algorithm to derive the update equation of the SMQNLMS\abbrev{SMQNLMS}{Set-Membership Quaternion-Valued NLMS} algorithm. In fact, the SMQNLMS\abbrev{SMQNLMS}{Set-Membership Quaternion-Valued NLMS} does not require data-reusing as the SMQAP\abbrev{SMQAP}{Set-Membership Quaternion-Valued AP} algorithm \cite{Gollamudi_smf_letter1998}, thus by taking $L=0$ and  $\gamma(k)=\frac{\gammabar e(k)}{|e(k)|}$ we obtain the update equation of the SMQNLMS\abbrev{SMQNLMS}{Set-Membership Quaternion-Valued NLMS} algorithm as
\begin{align}
\wbf(k+1)=\wbf(k)+\mu(k)\|\xbf(k)\|^{-2}\xbf(k)e^*(k),\label{eq:smqnlms_without_norm}
\end{align}
where $e(k)$ and $\mu(k)$ are defined as in (\ref{eq:e-trinion}) and (\ref{eq:mu-trinion}), respectively.

{\it Remark 6:} By adopting the bound $\gammabar=0$ in (\ref{eq:smqnlms_without_norm}), the algorithm will reduce to the QNLMS\abbrev{QNLMS}{Quaternion-Valued Normalized LMS} algorithm with unity step size. Therefore, the QNLMS\abbrev{QNLMS}{Quaternion-Valued Normalized LMS} algorithm can be described as
\begin{align}
\wbf(k+1)=\wbf(k)+\mu\|\xbf(k)\|^{-2}\xbf(k)e^*(k),\label{eq:qnlms_without_regularization}
\end{align}
where $\mu$ is the convergence factor.

The computational complexity for each update of the weight vector of the trinion based and quaternion based adaptive filtering algorithms are listed in Table \ref{tab:complexity}. The filter length and the memory length are $N$ and $L$, respectively. Also, Figures \ref{fig:Complexity_L_trinion} and \ref{fig:Complexity_N_trinion} show a comparison between the total number of real multiplications and additions required by the TAP\abbrev{TAP}{Trinion-Valued Affine Projection} and the QAP\abbrev{QAP}{Quaternion-Valued Affine Projection} algorithms for two cases: $N=15$, variable $L$ and $L=3$, variable $N$. As can be seen, the trinion model can efficiently decrease the computational complexity in comparison with the quaternion model, whenever the problem at hand suits both the quaternion and trinion solutions.

\begin{table*} [!t]
\caption{{COMPUTATIONAL COMPLEXITY PER UPDATE OF THE WEIGHT VECTOR}}\label{tab:complexity}
\begin{center}
 \begin{tabular}{|c|c|c|}
 \hline
 Algorithm & Real Multiplications & Real additions \\
 \hline 
 QNLMS & $20N+4$ & $20N-1$ \\\hline
  QAP & $32L^3+16NL^2+16L^2$&$32L^3+16NL^2+4L^2$\\&$+19NL+26L$ & $+16NL+8L$ \\\hline
   TNLMS & $12N+3$ & $12N-1$ \\\hline
   TAP & $18L^3+9NL^2+9L^2$&$18L^3+9NL^2$\\&$+11NL+50L$ & $+9NL+39L$ \\
 \hline
 \end{tabular}
 \end{center}
 \end{table*}
 
 \begin{figure}[t!]
 \centering
 \subfigure[b][]{\includegraphics[width=.48\linewidth,height=7cm]{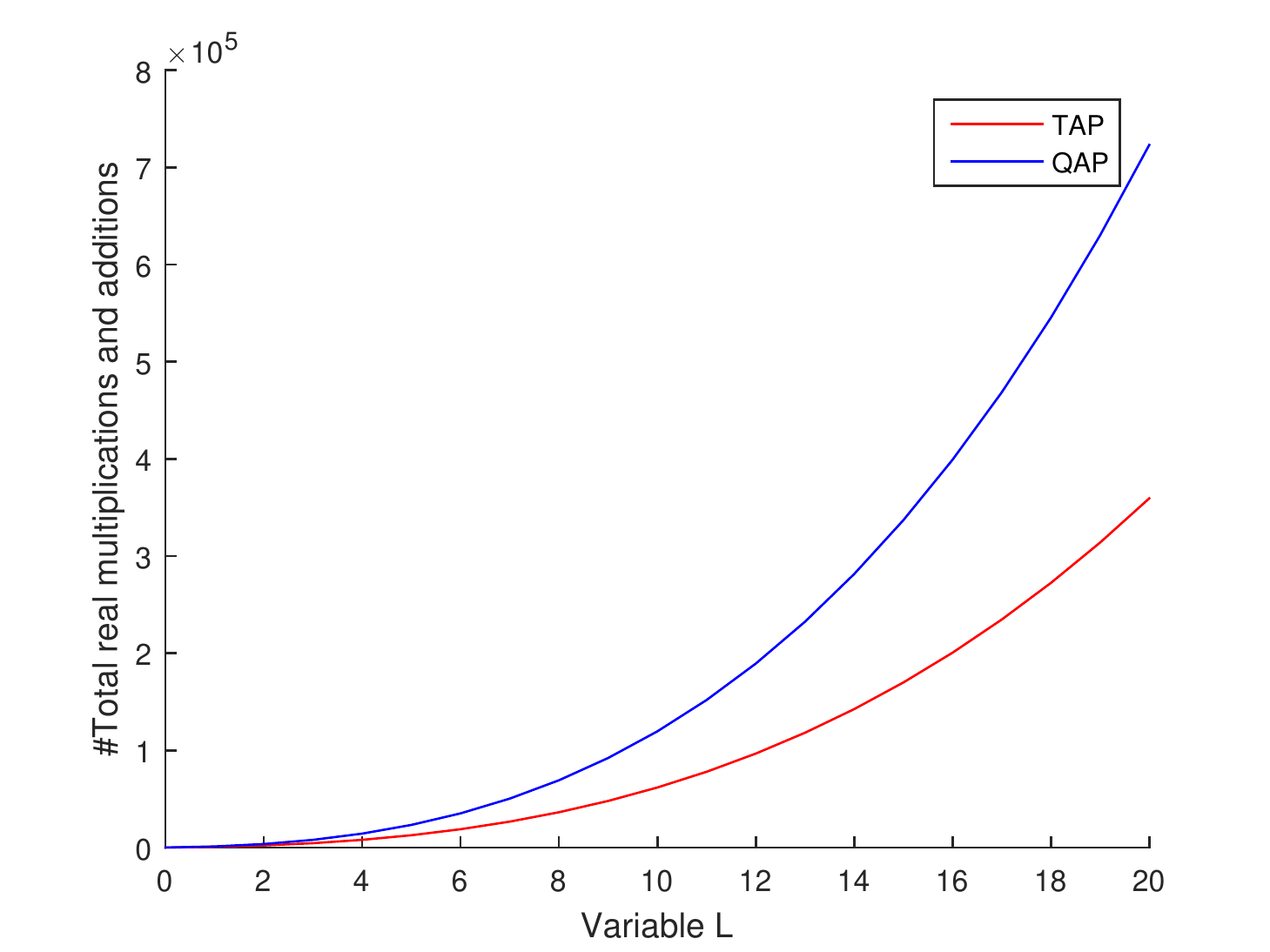}
 \label{fig:Complexity_L_trinion}}
 \subfigure[b][]{\includegraphics[width=.48\linewidth,height=7cm]{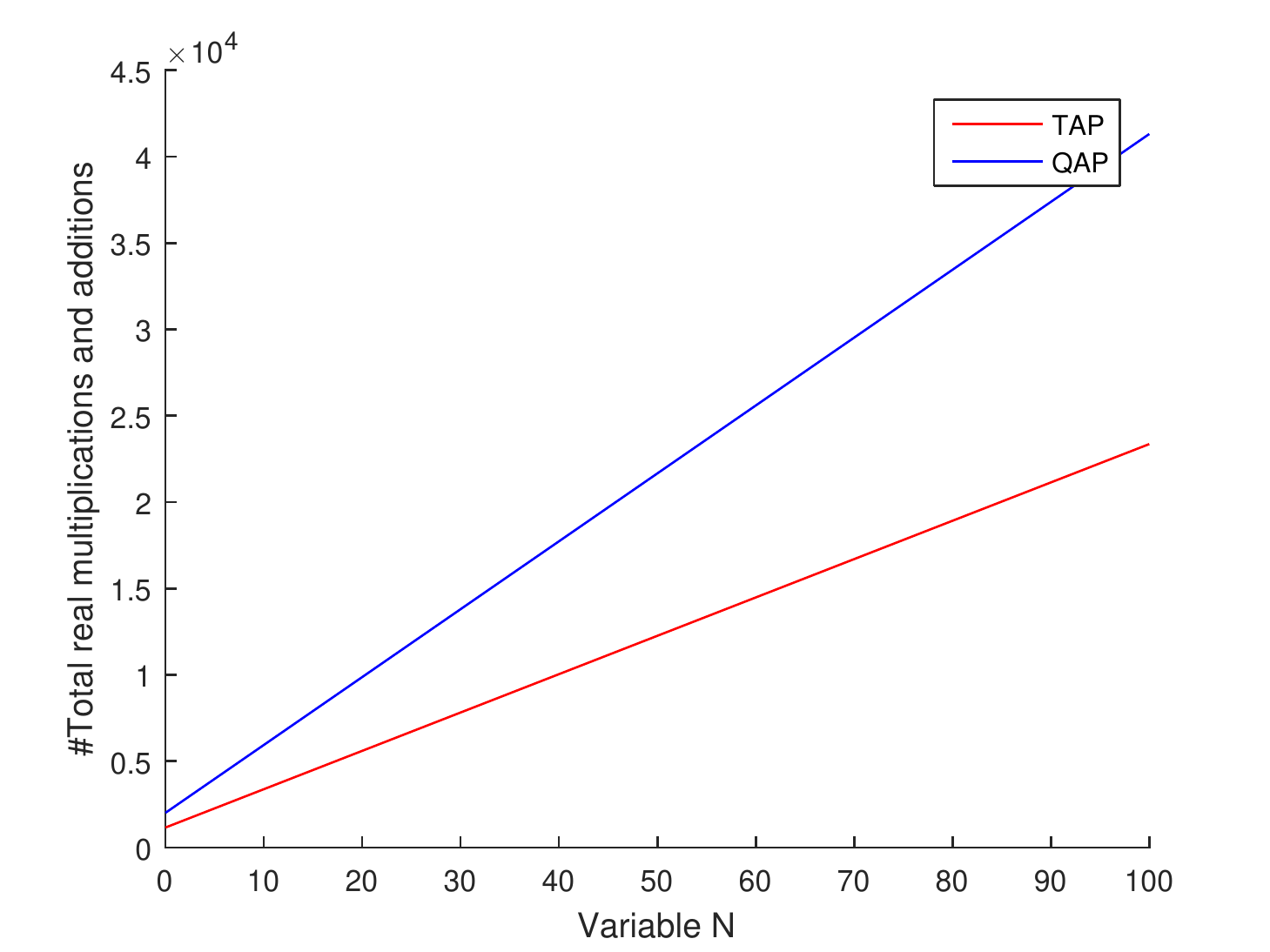}
 \label{fig:Complexity_N_trinion}}
 \caption{The numerical complexity of the TAP and the QAP algorithms for two cases: (a) $N=15$, variable $L$; (b) $L=3$, variable $N$.  \label{fig:Complexity-trinion}}
 \end{figure}
 
 
\section{Application of quaternion-valued adaptive \\algorithms to adaptive beamforming} \label{sec:adaptive-beamforming-tr}
 
As an illustration for the use of quaternions, we can study its application to adaptive beamforming. By utilizing the crossed-dipole array and quaternions, we can decrease the number of engaged sensors in the adaptive beamforming process. Therefore, the computational complexity and the energy consumption of the system will reduce without losing the quality of the performance~\cite{Jiang-phdthesis,Gou_beamformer_MAPE2011,Tao_beamformer_TAES2013,Tao_beamformer_MPE2014}.
   
A uniform linear array (ULA)\abbrev{ULA}{Uniform Linear Array} is illustrated in Figure~\ref{fig:beamformer-tr}~\cite{Jiang-phdthesis,Jiang_gqvgo_DSP2014}. It contains $M$ crossed-dipole pairs, they are placed on $y$-axis and the distance between neighboring antennas is $d$. At each position, the two crossed components are parallel to $x$-axis and $y$-axis, respectively. The direction of arrival (DOA)\abbrev{DOA}{Direction of Arrival} of a far-field incident signal is defined by the angles $\theta$ and $\phi$. Assume that this signal impinges upon the array from the $y$-$z$ plane. Thus, $\phi=\frac{\pi}{2}$ or $-\frac{\pi}{2}$, and $0\leq\theta\leq\frac{\pi}{2}$. As a consequence, the spatial steering vector for this far-field incident signal is given by \symbl{$\sbf_c(\theta,\phi)$}{The spatial steering vector for a far-field incident signal in adaptive beamforming}
\begin{align}
\sbf_c(\theta,\phi)=[1,e^{-\jmath2\pi d\sin\theta\sin\phi/\lambda},\cdots,e^{-\jmath2\pi(M-1) d\sin\theta\sin\phi/\lambda}]^T,
\end{align}
where $\lambda$ stands for the wavelength of the incident signal. For a crossed-dipole the spatial-polarization coherent vector can be expressed by~\cite{Compton_beamformer_TAP1981,Li_beamformer_TAP1991}
\begin{align}
\sbf_p(\theta,\phi,\gamma,\eta)=\left\{\begin{array}{ll}[-\cos\gamma,\cos\theta\sin\gamma e^{\jmath\eta}]&{\rm for~}\phi=\frac{\pi}{2},\\ $[$\cos\gamma,-\cos\theta\sin\gamma e^{\jmath\eta}]&{\rm for~}\phi=-\frac{\pi}{2},\end{array}\right.
\end{align}
where $\gamma\in[0,\frac{\pi}{2}]$ and $\eta\in[-\pi,\pi]$ are the auxiliary polarization angle and the polarization phase difference, respectively.
 
\begin{figure}[t!]
\begin{center}
\includegraphics[width=0.7\linewidth]{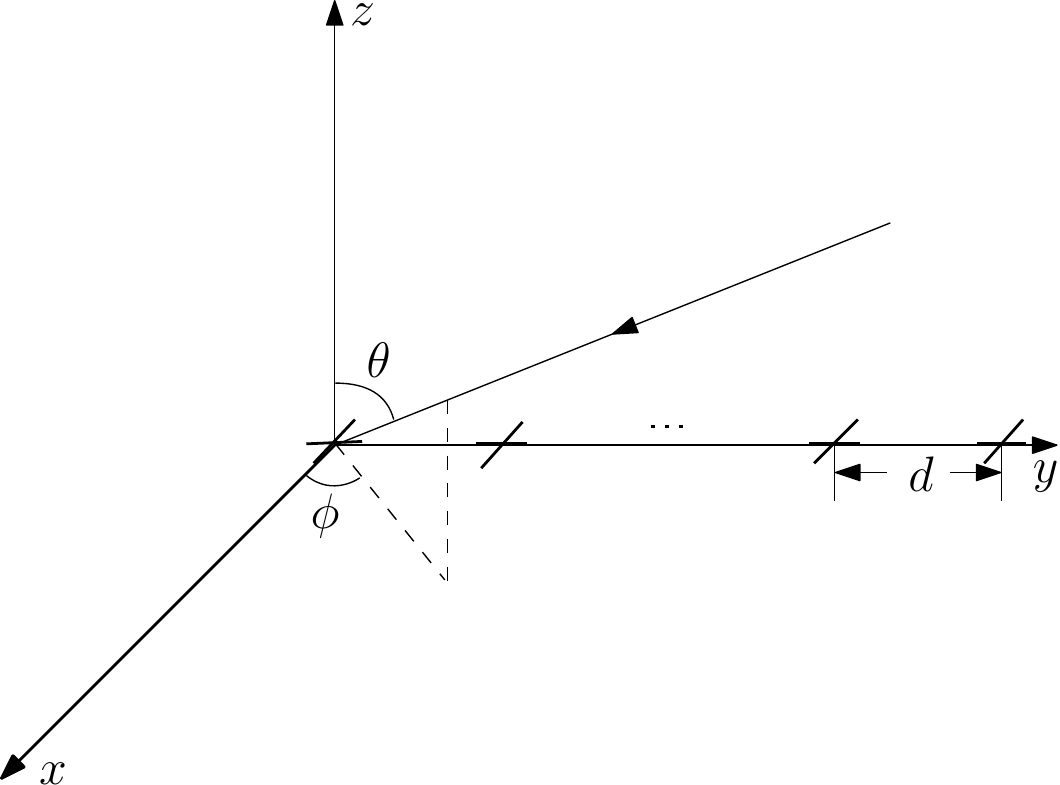}
\caption{A ULA with crossed-dipole~\cite{Jiang-phdthesis}.}
\label{fig:beamformer-tr}
\end{center}
\end{figure}
 
We can divide the array structure into two sub-arrays so that one of them is parallel to the $x$-axis and the other one is parallel to the $y$-axis. Then the complex-valued steering vector parallel to the $x$-axis is presented as
\begin{align}
\sbf_x(\theta,\phi,\gamma,\eta)=\left\{\begin{array}{ll}-\cos\gamma\sbf_c(\theta,\phi)&{\rm for~}\phi=\frac{\pi}{2},\\ \cos\gamma\sbf_c(\theta,\phi)&{\rm for~}\phi=-\frac{\pi}{2},\end{array}\right.
\end{align}
and the one parallel to the $y$-axis is given by
\begin{align}
\sbf_y(\theta,\phi,\gamma,\eta)=\left\{\begin{array}{ll}\cos\theta\sin\gamma e^{\jmath\eta}\sbf_c(\theta,\phi)&{\rm for~}\phi=\frac{\pi}{2},\\-\cos\theta\sin\gamma e^{\jmath\eta}\sbf_c(\theta,\phi)&{\rm for~}\phi=-\frac{\pi}{2}.\end{array}\right.
\end{align}
 
Using the Cayley-Dickson formula~\eqref{eq:Cayley_Dickson}, we can combine $\sbf_x(\theta,\phi,\gamma,\eta)$ and $\sbf_y(\theta,\phi,\gamma,\eta)$ together, we obtain a quaternion-valued steering vector as follows
\begin{align}
\sbf_q(\theta,\phi,\gamma,\eta)=\sbf_x(\theta,\phi,\gamma,\eta)+\imath \sbf_y(\theta,\phi,\gamma,\eta).
\end{align}
The response of the array for the quaternion-valued weight vector $\wbf$ is given as below
\begin{align}
r(\theta,\phi,\gamma,\eta)=\wbf^H\sbf_q(\theta,\phi,\gamma,\eta).
\end{align}
 
In the case of reference signal based quaternion-valued adaptive beamforming, the reference signal $d(k)$ is available. Therefore, the response of the array is the quaternion-valued beamformer output and it is defined as $y(k)=\wbf^H(k)\xbf(k)$, where $\xbf(k)$ is the received quaternion-valued vector sensor signals and $\wbf(k)$ is the quaternion-valued weigh vector. Also, the quaternion-valued error signal can be defined as $e(k)=d(k)-y(k)$. 
 
 
\section{Simulations} \label{sec:simulations-trinion} 
 
In this section, we apply the proposed algorithms to two scenarios. Scenario 1 verifies the performance of the trinion based and the quaternion based algorithms when they are used to wind profile prediction. In Scenario 2, we implement quaternionic adaptive beamforming by quaternion-valued algorithms.

\subsection{Scenario 1} \label{sub:wind-trinion}
 
In this scenario, all the proposed algorithms in this chapter are applied to anemometer readings provided by Google's RE$<$C Initiative \cite{Google_wind}. The wind speed recorded on May 25, 2011, is utilized for the algorithms comparisons. The step size, $\mu$, is selected to be $10^{-8}$ for the TLMS\abbrev{TLMS}{Trinion-Valued LMS} and the QLMS\abbrev{QLMS}{Quaternion-Valued LMS} algorithms and 0.9 for the TNLMS\abbrev{TNLMS}{Trinion-Valued Normalized LMS}, the TAP\abbrev{TAP}{Trinion-Valued Affine Projection}, the QNLMS\abbrev{QNLMS}{Quaternion-Valued Normalized LMS}, and the QAP\abbrev{QAP}{Quaternion-Valued Affine Projection} algorithms, and $\gammabar$ is set to be 5. Also, the threshold bound vector $\gammabf(k)$ is selected as {\it simple choice constraint vector}~\cite{Markus_sparseSMAP_tsp2014} which is defined as $\gamma_0(k)=\frac{\gammabar e(k)}{|e(k)|}$ and $\gamma_i(k)=d(k-i)-\wbf^T(k)\xbf(k-i)$, for $i=1,\cdots,L$. The filter length is 8, the memory length, $L$, and the prediction step are chosen equal to 1. All algorithms are initialized with zeros. 

The predicted results provided by trinion and quaternion based algorithms are shown in Figures~\ref{fig:Trinion_wind} and~\ref{fig:Quaternion_wind}, respectively. The learning curves using the TNLMS\abbrev{TNLMS}{Trinion-Valued Normalized LMS}, the SMTNLMS\abbrev{SMTNLMS}{Set-Membership Trinion-Valued NLMS}, the TAP\abbrev{TAP}{Trinion-Valued Affine Projection}, and the SMTAP\abbrev{SMTAP}{Set-Membership Trinion-Valued AP} algorithms are shown in Figures~\ref{fig:error_nlms-trinion} and~\ref{fig:error_ap-trinion}. Also, for comparison between the trinion and the quaternion based algorithms, the learning curves related to the TNLMS\abbrev{TNLMS}{Trinion-Valued Normalized LMS}, the QNLMS\abbrev{QNLMS}{Quaternion-Valued Normalized LMS}, the TAP\abbrev{TAP}{Trinion-Valued Affine Projection}, and the QAP\abbrev{QAP}{Quaternion-Valued Affine Projection} algorithms are depicted in Figures~\ref{fig:q_vs_t_nlms-trinion} and~\ref{fig:q_vs_t_ap-trinion}. 

The average of implementation times and the number of updates performed by the trinion and the quaternion based algorithms are presented in Table~\ref{tab:update_rate}. From the results, we can observe that all algorithms can track the wind data efficiently; however, the trinion based algorithms need a shorter time for implementation compared to their corresponding quaternion based algorithms. Also, we can observe that the set-membership based versions of the TNLMS\abbrev{TNLMS}{Trinion-Valued Normalized LMS}, the QNLMS\abbrev{QNLMS}{Quaternion-Valued Normalized LMS}, the TAP\abbrev{TAP}{Trinion-Valued Affine Projection}, and the QAP\abbrev{QAP}{Quaternion-Valued Affine Projection} algorithms have a low number of updates. Therefore, the set-membership algorithms can save energy effectively.
 
\begin{table} [!t]
\caption{The Average of implementation times and the number of updates for the trinion and the quaternion based algorithms using MATLAB software}\label{tab:update_rate}
\begin{center}
\begin{tabular}{|c|c|c|c|c|c|}
\hline
Algorithm & Time & Update & Algorithm & Time & Update \\           & (second) & rate   &       & (second)     & rate   \\
\hline 
TLMS & 2.45 & 100$\%$ & QLMS & 7.2 & 100$\%$ \\\hline
TNLMS & 8 & 100$\%$ & QNLMS & 9.4 & 100$\%$ \\\hline
TAP & 67 & 100$\%$ & QAP & 142 & 100$\%$ \\\hline
\hspace{-0.1cm}SMTNLMS\hspace{-0.1cm} & 3.8 & 17.92$\%$ & \hspace{-0.1cm}SMQNLMS\hspace{-0.1cm} & 9.2 & 17.87$\%$ \\\hline
SMTAP & 13 & 6.52$\%$ & SMQAP & 20.1 & 6.34$\%$ \\\hline
\end{tabular}
\end{center}
\end{table}
  
  \begin{figure}[t!]
  \begin{center}
  \includegraphics[width=1\linewidth] {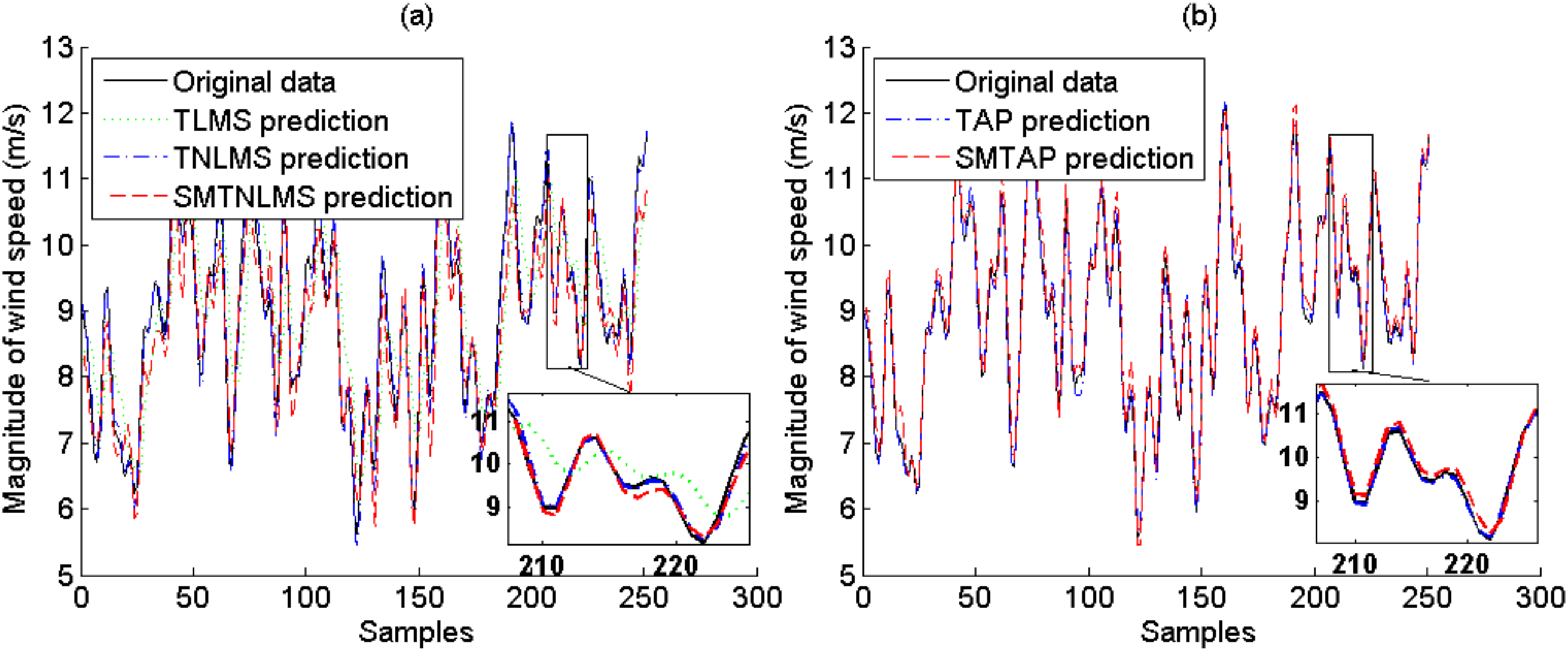} 
  \caption{Predicted results from the trinion based algorithms.}
  \label{fig:Trinion_wind}
  \end{center}
  \end{figure}
  
  \begin{figure}[t!]
  \begin{center}
  \includegraphics[width=1\linewidth] {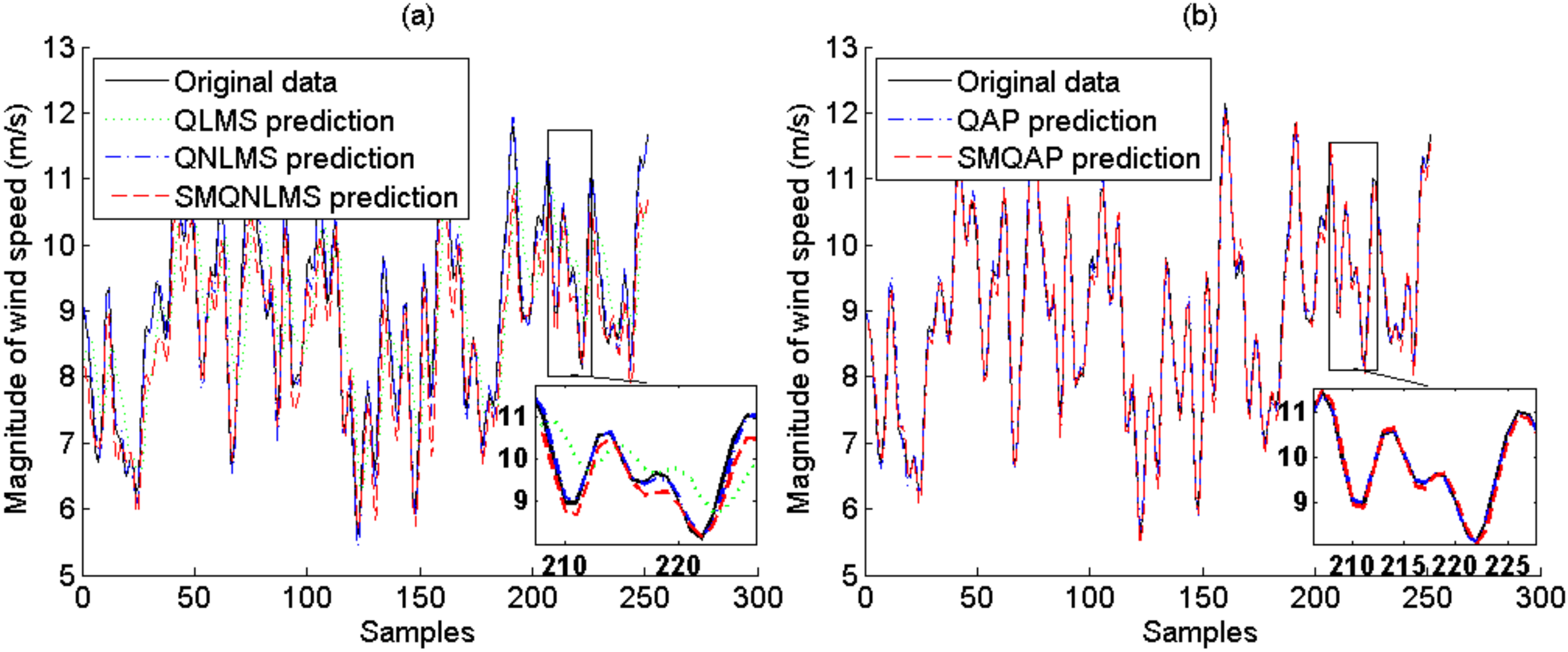} 
  \caption{Predicted results from the quaternion based algorithms.}
  \label{fig:Quaternion_wind}
  \end{center}
  \end{figure}
  
  \begin{figure}[t!]
  \centering
  \subfigure[b][]{\includegraphics[width=.48\linewidth,height=7cm]{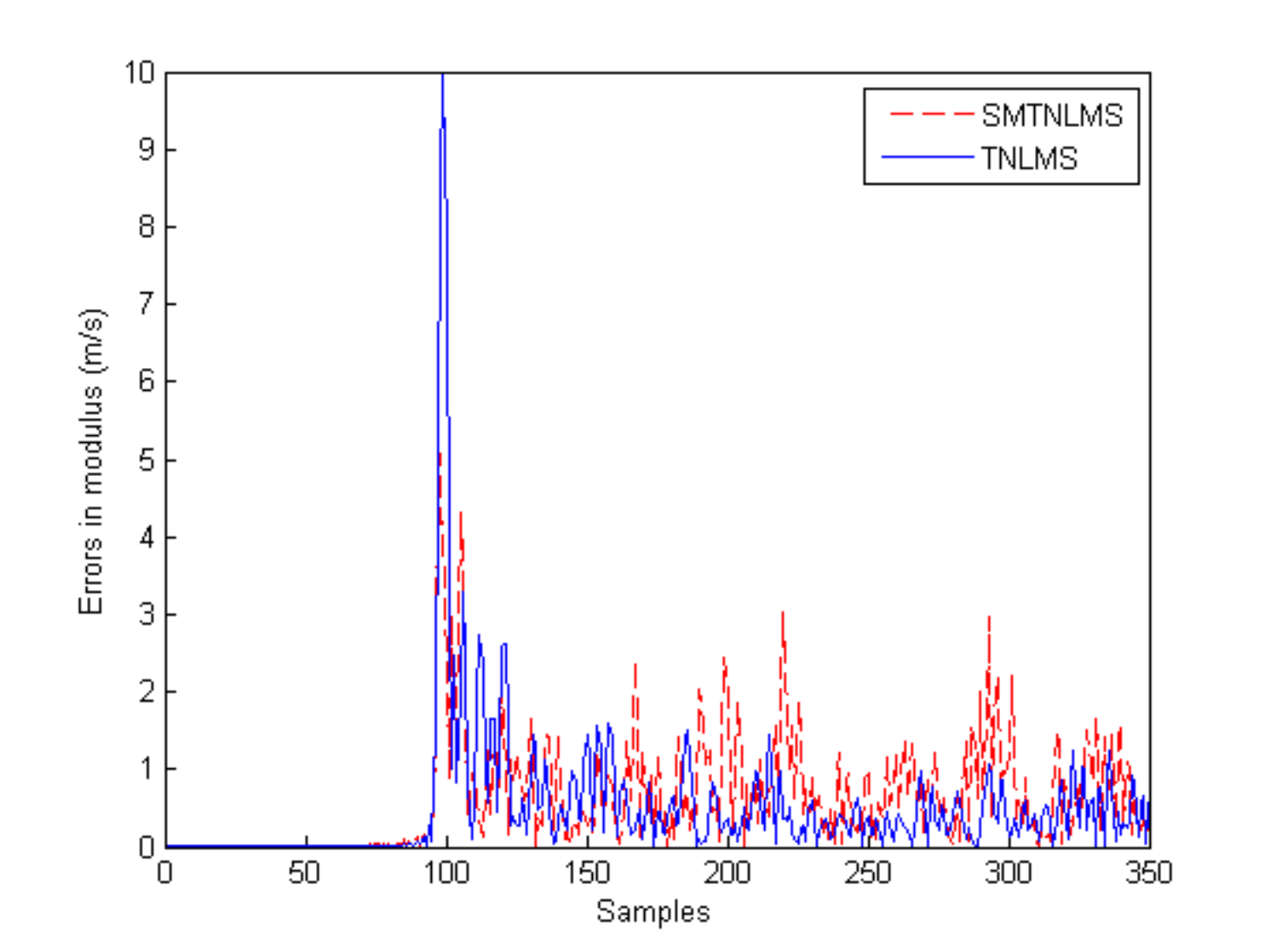}
  \label{fig:error_nlms-trinion}}
  \subfigure[b][]{\includegraphics[width=.48\linewidth,height=7cm]{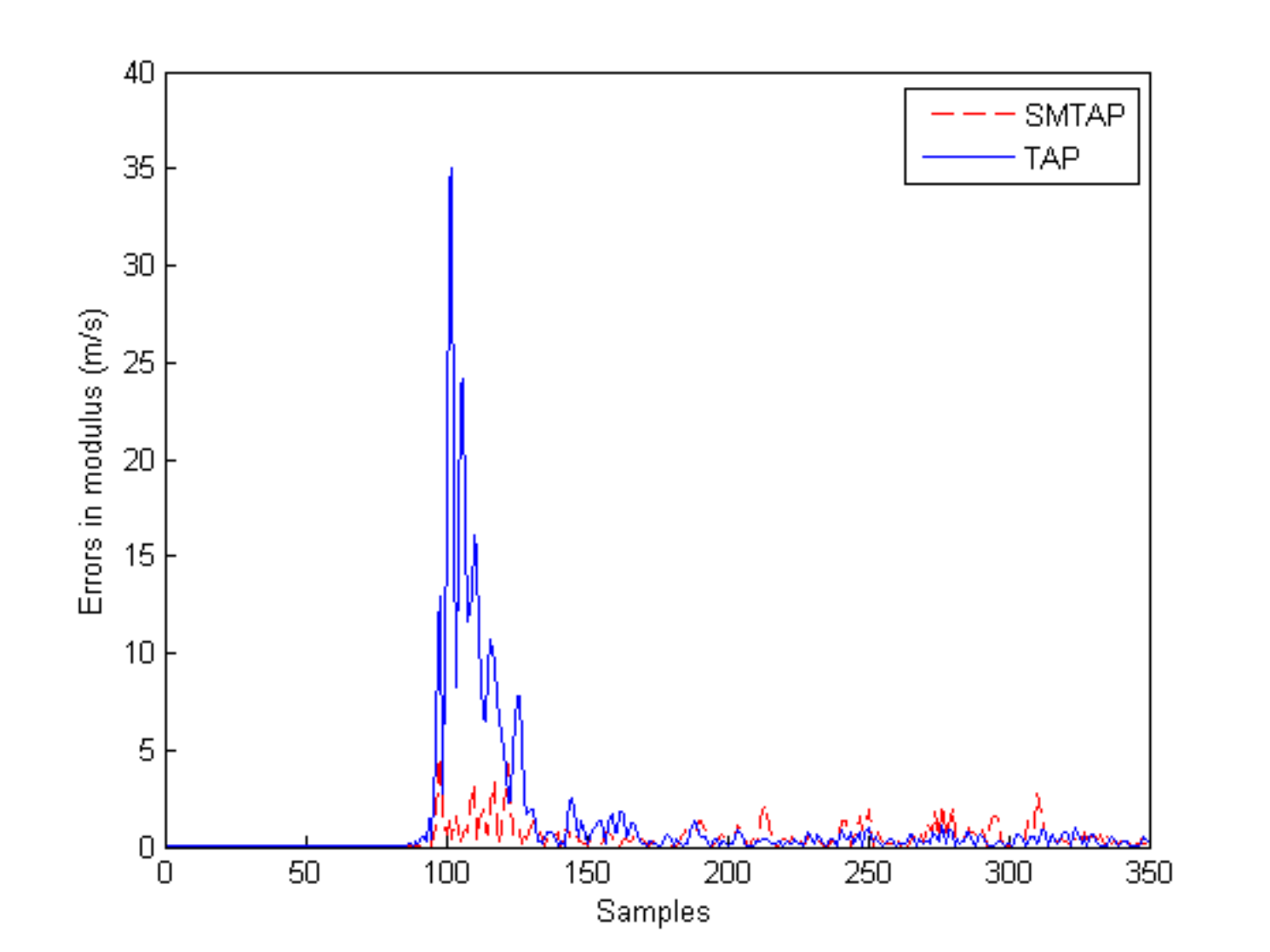}
  \label{fig:error_ap-trinion}}
  \caption{Learning curves of (a) the TNLMS and the SMTNLMS algorithms; (b) the TAP and the SMTAP algorithms.  \label{fig:sim-error-trinion}}
  \end{figure}
  
\begin{figure}[t!]
\centering
\subfigure[b][]{\includegraphics[width=.48\linewidth,height=7cm]{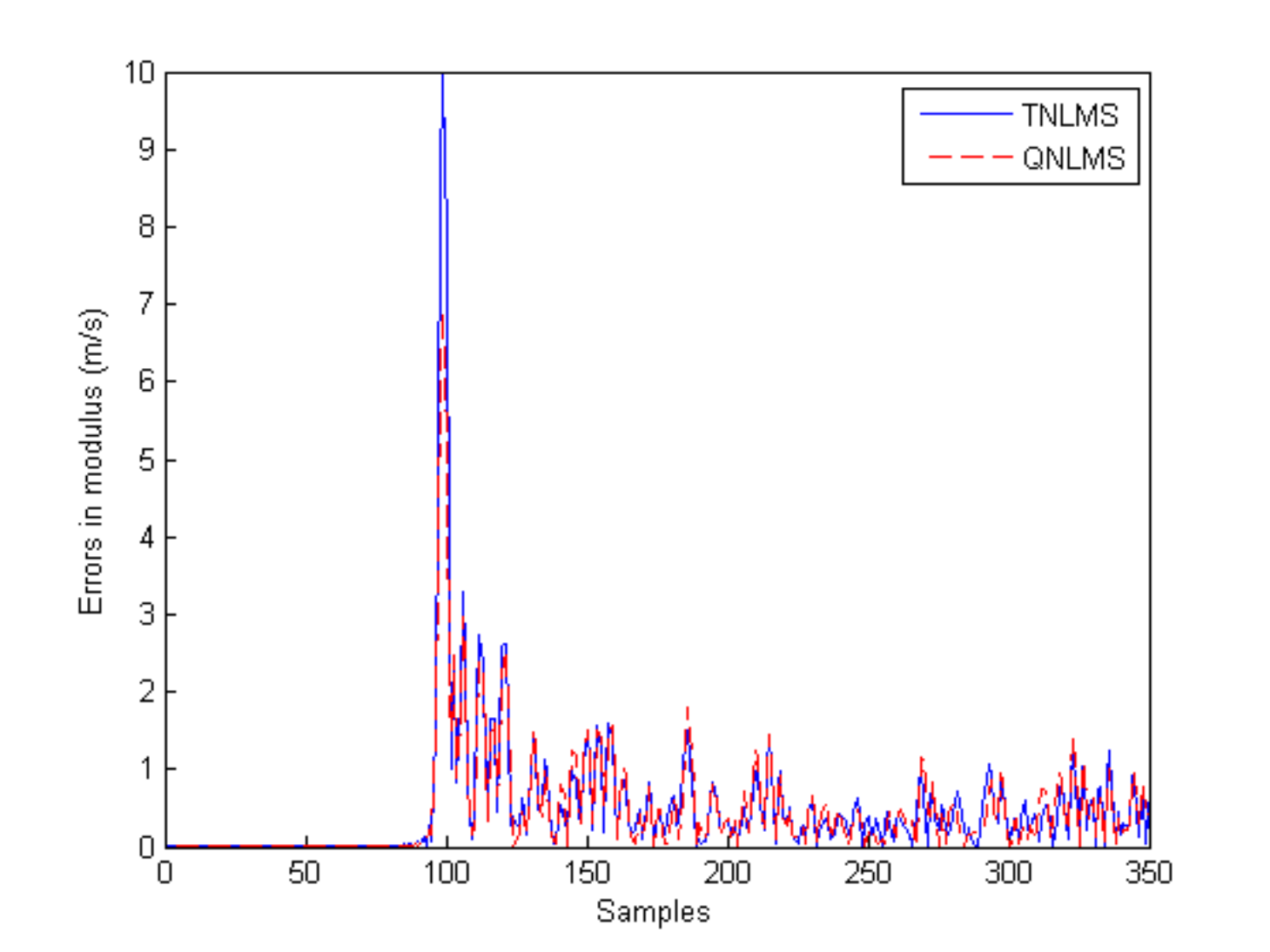}
\label{fig:q_vs_t_nlms-trinion}}
\subfigure[b][]{\includegraphics[width=.48\linewidth,height=7cm]{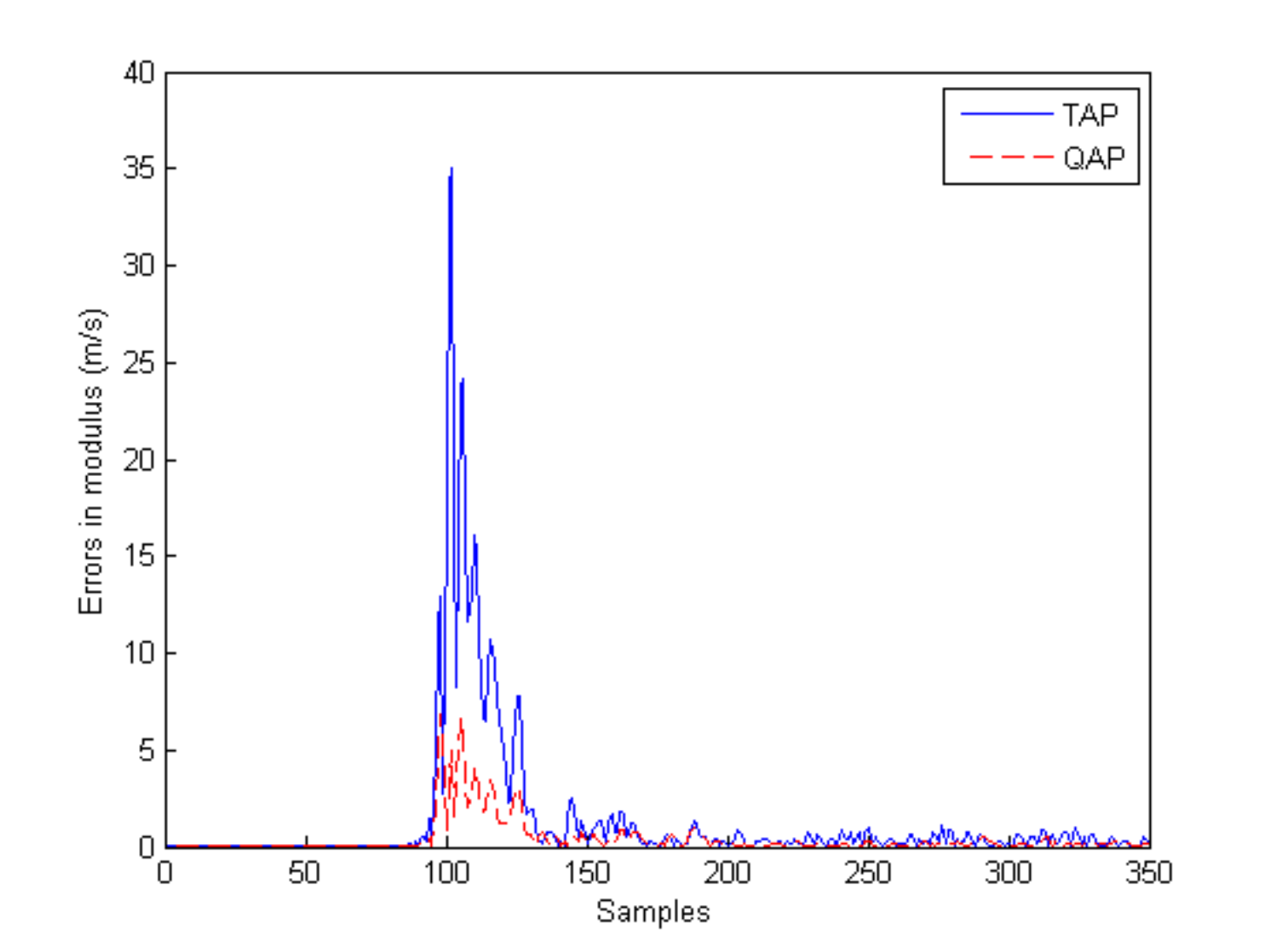}
\label{fig:q_vs_t_ap-trinion}}
\caption{Learning curves of (a) the TNLMS and the QNLMS algorithms; (b) the TAP and the QAP algorithms.  \label{fig:sim-error-q_vs_t-trinion}}
\end{figure}

Moreover, we implemented the same scenario using a real-valued algorithm. Indeed, we used three affine projection (AP) algorithms whose parameters are chosen similar to the TAP algorithm to compare the tracking results between the AP and the TAP algorithms. We did not notify a significant difference between the tracking results of the AP and the TAP algorithms, thus we avoid presenting an additional figure since the results were similar to Figure~\ref{fig:Trinion_wind}(b). However, in wind profile prediction, It would be preferable to employ trinion-valued algorithms since there is some structure between the three components of data.
  
  
\subsection{Scenario 2} \label{sub:beamforming}

In this scenario, we simulate the quaternionic adaptive beamforming~\cite{Jiang_gqvgo_DSP2014} using the QLMS\abbrev{QLMS}{Quaternion-Valued LMS}, the QNLMS\abbrev{QNLMS}{Quaternion-Valued Normalized LMS}, the SMQNLMS\abbrev{SMQNLMS}{Set-Membership Quaternion-Valued NLMS}, the QAP\abbrev{QAP}{Quaternion-Valued Affine Projection}, and the SMQAP\abbrev{SMQAP}{Set-Membership Quaternion-Valued AP} algorithms. We assume a sensor array with 10 crossed-dipoles and half-wavelength spacing. The step size, $\mu$, for the QLMS\abbrev{QLMS}{Quaternion-Valued LMS}, the QNLMS\abbrev{QNLMS}{Quaternion-Valued Normalized LMS}, and the QAP\abbrev{QAP}{Quaternion-Valued Affine Projection} algorithms are $4\times10^{-5}$, 0.009, and 0.005, respectively. For the QAP\abbrev{QAP}{Quaternion-Valued Affine Projection} and the SMQAP\abbrev{SMQAP}{Set-Membership Quaternion-Valued AP} algorithms, the memory length, $L$, is set to 1. A desired signal with 20 dB SNR\abbrev{SNR}{Signal-to-Noise Ratio} ($\sigma_n^2=0.01$) impinges  from broadside, $\theta=0$ and $\phi=\frac{\pi}{2}$, and two interfering signals with signal-to-interference ratio (SIR)\abbrev{SIR}{Signal-to-Interference Ratio} of -10 dB arrive from $(\theta,\phi)=(\frac{\pi}{9},\frac{\pi}{2})$ and $(\theta,\phi)=(\frac{\pi}{6},-\frac{\pi}{2})$, respectively. All the signals have the same polarization of $(\gamma,\eta)=(0,0)$. $\gammabar$ is set to be $\sqrt{2\sigma_n^2}$, and the vector $\gammabf(k)$ is selected as simple choice constraint vector defined in Scenario 1. 

The learning curves of quaternion algorithms over 100 trials are shown in Figure~\ref{fig:ball}. The average number of updates performed by the SMQNLMS\abbrev{SMQNLMS}{Set-Membership Quaternion-Valued NLMS} and the SMQAP\abbrev{SMQAP}{Set-Membership Quaternion-Valued AP} algorithms are 1408 and 1815 in a total of 10000 iterations (about 14.08$\%$ and 18.15$\%$), respectively. As can be seen, the set-membership quaternion algorithms converge faster while having a lower number of updates. Also, the convergence rate of the QAP\abbrev{QAP}{Quaternion-Valued Affine Projection} algorithm is higher than the SMQNLMS\abbrev{SMQNLMS}{Set-Membership Quaternion-Valued NLMS} algorithm. 

The response of a beamformer to the impinging signals as a \symbl{$B(\theta)$}{The beam pattern of a beamformer} function of $\theta$ is called beam pattern and is defined as $B(\theta)=\wbf^H\sbf(\theta)$, where $\sbf(\theta)$ is the steering vector. The magnitude of beam pattern explains the variation of a beamformer concerning the signal arriving from different DOA\abbrev{DOA}{Direction of Arrival} angles. Figure~\ref{fig:beam_pattern} illustrates the magnitude of beam pattern of the quaternion algorithms with $\theta=0$. In this figure, the positive values of $\theta$ show the value range $\theta\in[0,\frac{\pi}{2}]$ for $\phi=\frac{\pi}{2}$ and the negative values, $\theta\in[-\frac{\pi}{2},0]$, indicate the same range of $\theta\in[0,\frac{\pi}{2}]$ but $\phi=-\frac{\pi}{2}$. We can observe that all the quaternion algorithms attained an acceptable beamforming result since the two nulls at the directions of the interfering signals are clearly visible.

The output signal to desired plus noise ratio (OSDR)\abbrev{OSDR}{Output Signal to Desired Plus Noise Ratio} and the output signal to interference plus noise ratio (OSIR)\abbrev{OSIR}{Output Signal to Interference Plus Noise Ratio} for the quaternion algorithms are presented in Table~\ref{tab:output_to_interference}. The OSDR\abbrev{OSDR}{Output Signal to Desired Plus Noise Ratio} is achieved by calculating the power of the output signal and the total power of desired plus one-third of the noise signal, then we compute the ratio between these two values. Also, the OSIR\abbrev{OSIR}{Output Signal to Interference Plus Noise Ratio} is obtained by computing the power of the output signal and the total power of interference plus one-third of the noise signal, then we find the ratio between the two. As can be seen, the best results are obtained by the SMQNLMS\abbrev{SMQNLMS}{Set-Membership Quaternion-Valued NLMS} and the SMQAP\abbrev{SMQAP}{Set-Membership Quaternion-Valued AP} algorithms.
  
\begin{figure}[t!]
\begin{center}
\includegraphics[width=1\linewidth]{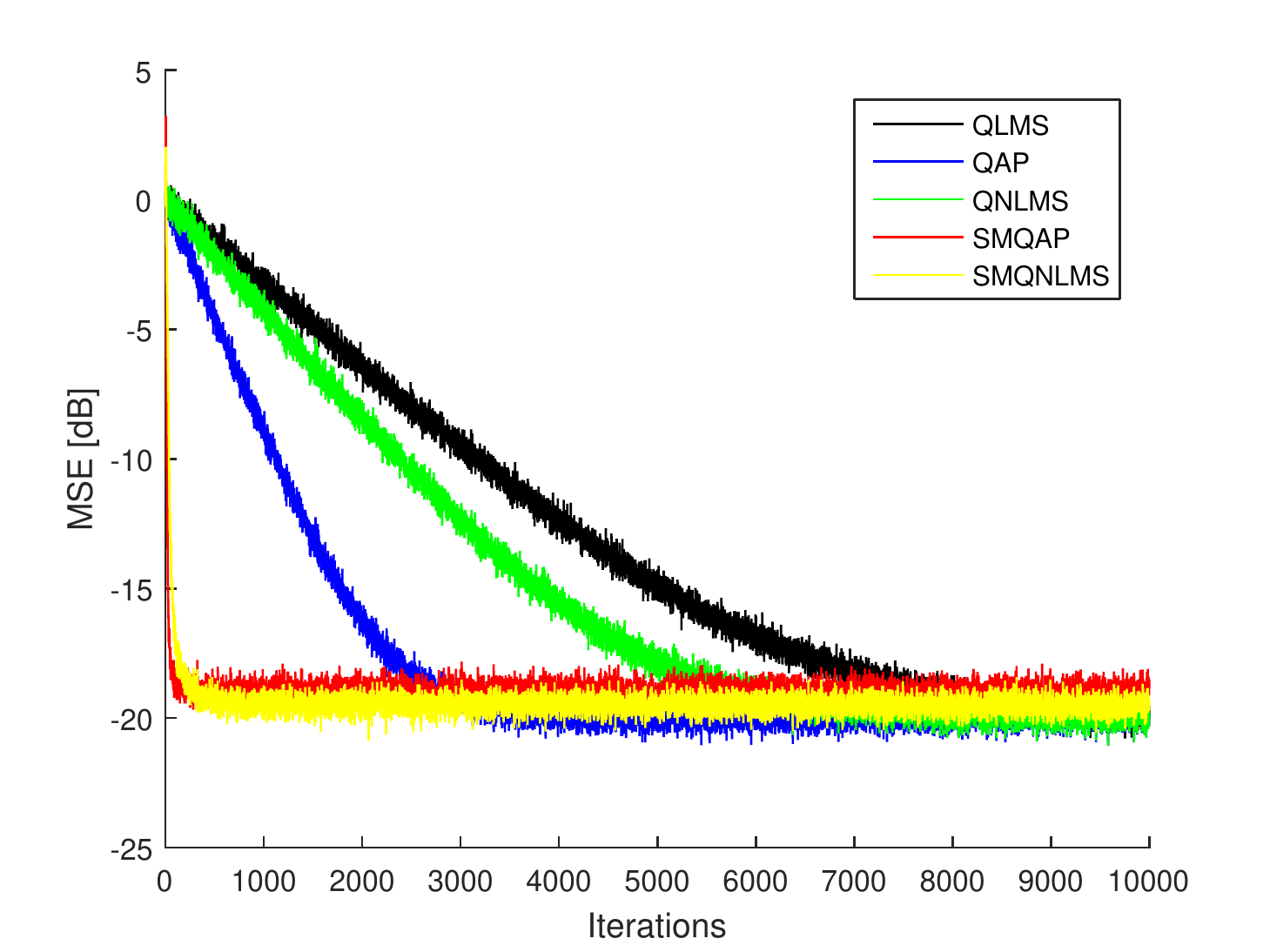}
\caption{Learning curves of the QLMS, the QNLMS, the QAP, the SMQNLMS, and the SMQAP algorithms.}
\label{fig:ball}
\end{center}
\end{figure}

\begin{figure}[t!]
\begin{center}
\includegraphics[width=1\linewidth]{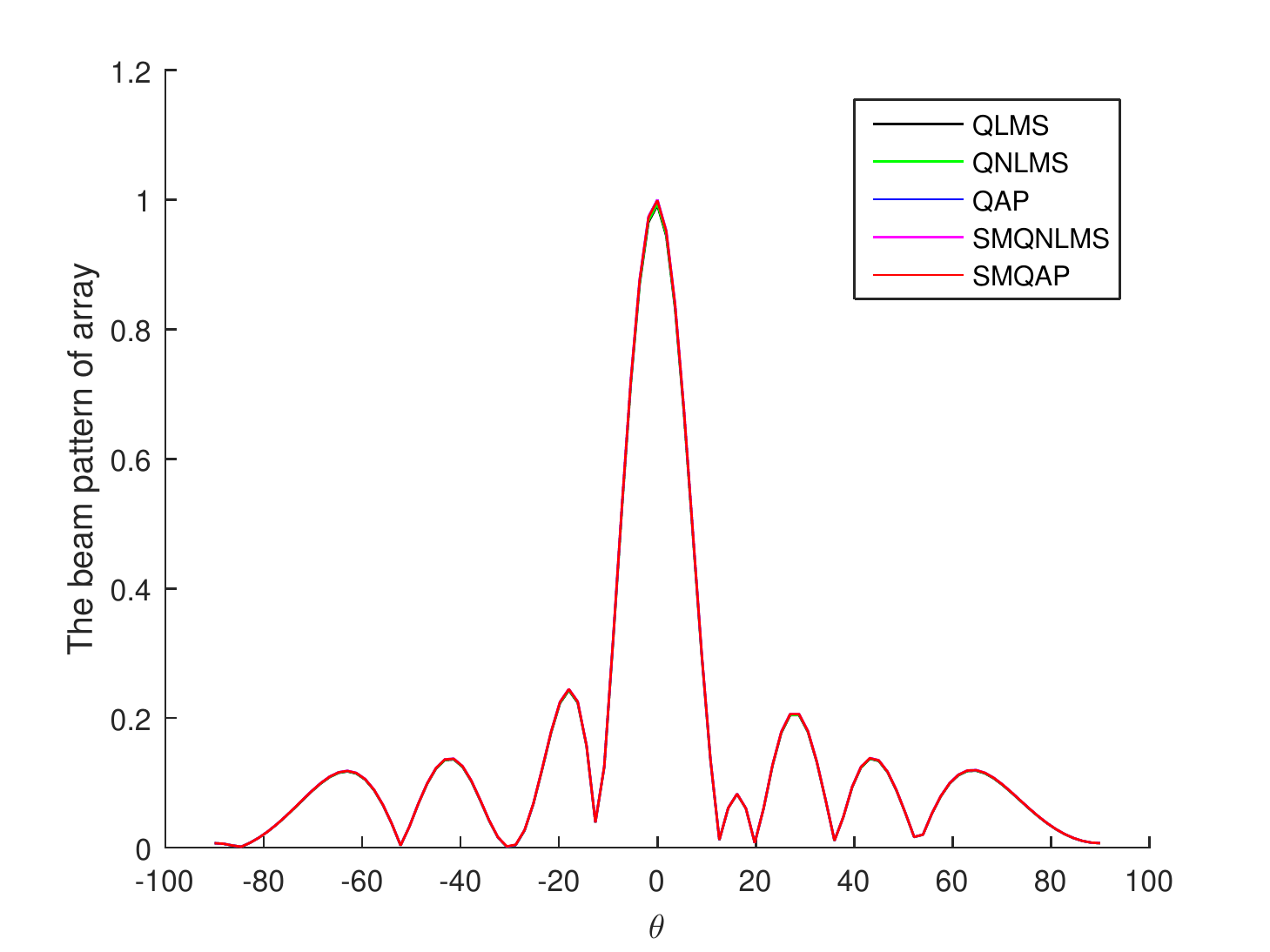}
\caption{Beam patterns of the QLMS, the QNLMS, the QAP, the SMQNLMS, and the SMQAP algorithms when DOA of desired signal is $(\theta,\phi)=(0,\frac{\pi}{2})$.}
\label{fig:beam_pattern}
\end{center}
\end{figure}

\begin{table} [!t]
\caption{The OSDR and the OSIR for the quaternion algorithms}\label{tab:output_to_interference}
\begin{center}
\begin{tabular}{|c|c|c|c|c|c|}
\hline
Algorithms & QLMS & QNLMS & QAP & SMQNLMS & SMQAP \\\hline 
OSDR (dB) & -1.645 & -1.502 & -0.647 & -0.024 & 0.004 \\\hline
OSIR (dB) & -11.699 & -11.557 & -10.701 & -10.079 & -10.050 \\\hline
\end{tabular}
\end{center}
\end{table}
  
  
\section{Conclusions} \label{sec:conclusion-trinion} 
  
In this chapter, we have generalized the set-membership model for the trinion and the quaternion number systems. First, we have reviewed some properties of the quaternion and the trinion systems. Then we have derived the set-membership trinion based algorithms and, by the same argument, the quaternion based adaptive filtering algorithms have been introduced. Also, we have presented the counterparts of the proposed algorithms without employing the set-membership approach. Moreover, we have reviewed the application of quaternion algorithms to adaptive beamforming. Numerical simulations for the recorded wind data and the adaptive beamforming have proven that the set-membership based algorithms have significantly lower update rates, while the penalty to be paid for that is not noteworthy. Also, we have observed that the trinion based algorithms have comparable performance to the quaternion based ones, however with striking lower computational complexity.
  \chapter{Improved Set-Membership Partial-Update Affine Projection Algorithm}

Adaptive filters have applications in a wide range of areas such as noise cancellation, signal prediction, echo cancellation, communications, radar, and speech processing. In several applications, a large number of coefficients to be updated leads to high computational complexity, turning the adaptation of the filter coefficients prohibitive regarding hardware requirements. In some cases, like acoustic echo cancellation, the adaptive filter might use a few thousand coefficients in order to model the underlying physical system with sufficient accuracy. In these applications, the convergence would entail a large number of iterations, calling for a more sophisticated updating rule which is inherently more computationally intensive. For a given adaptive filter, the computational complexity can be reduced by updating only part of the filter coefficients at each iteration, forming a  family of algorithms called partial-update (PU)\abbrev{PU}{Partial-Update} algorithms. In the literature, several variants of adaptive filtering algorithms with partial-update have been proposed \cite{PUbook,Diniz_adaptiveFiltering_book2013,Douglas-PU-1997,Aboulnasr-PU-1999,Dogancay-PU-2001,Werner-PU-2003,Werner-PU-2004,Godavarti,Grira,Arablouei,Pinto,Tandon,Bhotto,Deng}.

Another powerful approach to decrease the computational complexity of an adaptive filter is to employ set-membership filtering (SMF)\abbrev{SMF}{Set-Membership Filtering} approach \cite{Diniz_adaptiveFiltering_book2013,Gollamudi_smf_letter1998}. Algorithms developed from the SMF\abbrev{SMF}{Set-Membership Filtering} framework employ a deterministic objective function related to a bounded error constraint on the filter output, such that the updates belong to a set of feasible solutions. Implementation of SMF\abbrev{SMF}{Set-Membership Filtering} algorithms involves two main steps: 1) information evaluation, 2) parameter update. As compared with the standard  normalized least mean square (NLMS)\abbrev{NLMS}{Normalized LMS} and affine projection (AP)\abbrev{AP}{Affine Projection} algorithms, the set-membership normalized least mean square (SM-NLMS\abbrev{SM-NLMS}{Set-Membership Normalized LMS}) and the set-membership affine projection (SM-AP\abbrev{SM-AP}{Set-Membership Affine Projection}) algorithms lead to reduced computational complexity chiefly due to data-selective updates \cite{Gollamudi_smf_letter1998,Diniz_sm_bnlms_tsp2003,Werner_sm_ap_letter2001,Arablouei_2012_ICASSP,Bhotto_2012_ISCCSP,Yamada_sm-nlmsAnalysis_tsp2009,Bhotto_2012_TSP,Abadi_2008_ISCCSP}.

The use of PU\abbrev{PU}{Partial-Update} strategy decreases the computational complexity while reducing convergence speed. We employ SMF\abbrev{SMF}{Set-Membership Filtering} technique to reduce further the computational load due to a lower number of updates. However applying the SMF\abbrev{SMF}{Set-Membership Filtering} and PU\abbrev{PU}{Partial-Update} strategies together might result in slow convergence speed. One approach to accelerate the convergence speed is choosing a smaller error estimation bound, but it might increase the number of updates. Also, if we adopt a higher error estimation threshold to reduce the number of updates, the convergence rate will decrease. Therefore, convergence speed and computational complexity are conflicting requirements. 

In this chapter, we introduce an interesting algorithm which can accelerate the convergence speed and simultaneously reduce the 
number of updates (and as a result decrease the computational complexity) in the set-membership partial-update affine projection (SM-PUAP)\abbrev{SM-PUAP}{Set-Membership Partial-Update AP} algorithm. In the SM-PUAP\abbrev{SM-PUAP}{Set-Membership Partial-Update AP} algorithm, some updates move too far from their SM-AP\abbrev{SM-AP}{Set-Membership Affine Projection} update; especially when the angle between the updating direction and the threshold hyperplane is small. In this case, we might have a significant disturbance in the coefficient update while attempting to
reach the feasibility set. Therefore, to limit the distance between two consecutive updates, first, we will construct a 
hypersphere centered at the present weight vector whose radius equals the distance between the current weight vector and 
the weight vector that would be obtained with the SM-AP\abbrev{SM-AP}{Set-Membership Affine Projection} algorithm. This radius is an upper bound on the Euclidean norm
of the coefficient disturbance that is allowed in the proposed improved set-membership partial-update affine projection (I-SM-PUAP)\abbrev{I-SM-PUAP}{Improved SM-PUAP} algorithm.

The content of this chapter was published in~\cite{Hamed_I_SM-PUAP_ICASSP2016}. In this chapter, first of all, we review the SM-PUAP\abbrev{SM-PUAP}{Set-Membership Partial-Update AP} algorithm in Section \ref{sec:SM-PUAP-icassp}. Then, in Section \ref{sec:M-SM-PUAP-icassp}, we derive the I-SM-PUAP\abbrev{I-SM-PUAP}{Improved SM-PUAP} algorithm. Section \ref{sec:simulation-icassp} presents simulations of the algorithms. Finally, Section \ref{sec:conclusion-icassp} contains the conclusions.


\section{Set-Membership Partial-Update Affine Projection Algorithm} \label{sec:SM-PUAP-icassp} 

In this section, we present the SM-PUAP\abbrev{SM-PUAP}{Set-Membership Partial-Update AP} algorithm \cite{Diniz_adaptiveFiltering_book2013}. The main objective of the partial-update adaptation is to perform updates in $M$ out of $N+1$ adaptive filter coefficients, where $N$ is the order of the adaptive filter.  The $M$ coefficients to be updated at time instant $k$ are specified by an index set ${\cal I}_M(k)=\{i_0(k),\cdots,i_{M-1}(k)\}$ with $\{i_j(k)\}_{j=0}^{M-1}$ chosen from the set $\{0,\cdots,N\}$.\symbl{${\cal I}_M(k)$}{The set of $M$ coefficients to be updated at time instant $k$} The subset of coefficients with indices in ${\cal I}_M(k)$ plays an essential role in the performance and the effectiveness of the partial-update strategy.  Note that ${\cal I}_M(k)$ varies with the time instant $k$. As a result, the $M$ coefficients to be updated can change according to the time instant. The choice of which $M$ coefficients should be updated is related to the optimization criterion chosen for algorithm derivation. The SM-PUAP\abbrev{SM-PUAP}{Set-Membership Partial-Update AP} algorithm \cite{Diniz_adaptiveFiltering_book2013} takes the update vector $\wbf(k+1)$ as the vector minimizing the Euclidean distance $\|\wbf(k+1)-\wbf(k)\|^2$ subject to the constraint $\wbf(k+1)\in{\cal H}(k)$ in such a way that only $M$ coefficients are updated. 

The optimization criterion in the SM-PUAP\abbrev{SM-PUAP}{Set-Membership Partial-Update AP} algorithm is described as follows. Let $\psi^{L+1}(k)$ indicate the intersection of the last $L+1$ constraint sets. A coefficient update is implemented whenever $\wbf(k)\not\in\psi^{L+1}(k)$ as follows
\begin{equation}
\begin{aligned}
&\min \|\wbf(k+1)-\wbf(k)\|^2\\
&{\rm subject~to:}\\
&\dbf(k)-\Xbf^T(k)\wbf(k+1)=\gammabf(k)\\
&\tilde{\Cbf}_{{\cal I}_M(k)}[\wbf(k+1)-\wbf(k)]=0
\end{aligned}
\end{equation}
where

\begin{tabular}{ll}
$\dbf(k)\in\mathbb{R}^{(L+1)\times1}$& contains the desired output from the \\&$L+1$ last time instants;\\
$\gammabf(k)\in\mathbb{R}^{(L+1)\times1}$&specifies the point in $\psi^{L+1}(k)$;\\
$\Xbf(k)\in\mathbb{R}^{(N+1)\times (L+1)}$&contains the corresponding input vectors, i.e.,
\end{tabular}
\begin{equation}
\begin{aligned}
\dbf(k)&=[d(k)~d(k-1)~\cdots~d(k-L)]^T,\\
\gammabf(k)&=[\gamma_0(k)~\gamma_1(k)~\cdots~\gamma_L(k)]^T,\\
\Xbf(k)&=[\xbf(k)~\xbf(k-1)~\cdots~\xbf(k-L)], \label{eq:pack-icassp}
\end{aligned}
\end{equation}
with $\xbf(k)$ being the input-signal vector
\begin{align}
\xbf(k)=[x(k)~x(k-1)~\cdots~x(k-N)]^T. \label{eq:x(k)-icassp}
\end{align}
Moreover, the matrix $\tilde{\Cbf}_{{\cal I}_M(k)}=\Ibf-\Cbf_{{\cal I}_M(k)}$ is a complementary matrix that gives $\tilde{\Cbf}_{{\cal I}_M(k)}\wbf(k+1)=\tilde{\Cbf}_{{\cal I}_M(k)}\wbf(k)$, which means that only $M$ coefficients are updated. The threshold vector elements are such that $|\gamma_i(k)|\leq\gammabar$, for $i=0,\cdots,L$. The matrix $\Cbf_{{\cal I}_M(k)}$ is a diagonal matrix that identifies the coefficients to be updated at instant $k$, if an update is required.\symbl{$\Cbf_{{\cal I}_M(k)}$}{The diagonal matrix that identifies the coefficients to be updated at instant time $k$, if an update is required} This matrix has $M$ nonzero elements equal to one located at positions declared by ${\cal I}_M(k)$.

Using the method of Lagrange multipliers we obtain the following updating rule
\begin{align}
\wbf(k+1)=\wbf(k)+\Cbf_{{\cal I}_M(k)}\Xbf(k)[\Xbf^T(k)\Cbf_{{\cal I}_M(k)}\Xbf(k)]^{-1}[\ebf(k)-\gammabf(k)]
\end{align}
The updating equation of the SM-PUAP\abbrev{SM-PUAP}{Set-Membership Partial-Update AP} algorithm is given by
\begin{align}
&\wbf(k+1)=\left\{\begin{array}{ll}\wbf(k)+\Cbf_{{\cal I}_M(k)}\Xbf(k)\Pbf(k)(\ebf(k)-\gammabf(k))&\text{if}~|e(k)|>\gammabar\\\wbf(k)&\text{otherwise}\end{array}\right.,\label{eq:update_SM-PUAP-icassp}
\end{align}
where \symbl{$\Pbf(k)$}{The auxiliary matrix $\Pbf(k)\triangleq(\Xbf^T(k)\Cbf_{{\cal I}_M(k)}\Xbf(k)+\delta\Ibf)^{-1}$}
\begin{align}
\Pbf(k)&=(\Xbf^T(k)\Cbf_{{\cal I}_M(k)}\Xbf(k)+\delta\Ibf)^{-1}, \label{eq:P'(k)-icassp}\\
\ebf(k)&=[e(k)~\epsilon(k-1)~\cdots~\epsilon(k-L)]^T,
\end{align}
with $e(k)=d(k)-\wbf^T(k)\xbf(k)$, and $\epsilon(k-i)=d(k-i)-\wbf^T(k)\xbf(k-i)$ for $i=1,\cdots,L$. In the Equation (\ref{eq:P'(k)-icassp}), $\delta$ and $\Ibf$ are a small positive constant and an $(L+1)\times (L+1)$ identity matrix, respectively. The diagonal matrix $\delta\Ibf$ is added to the matrix to be inverted in order to avoid numerical problems in the inversion operation in the cases $\Xbf^T(k)\Cbf_{{\cal I}_M(k)}\Xbf(k)$ is ill conditioned.

A natural choice for the $M$ nonzero diagonal elements of $\Cbf_{{\cal I}_M(k)}$ is those corresponding to the coefficients of $\wbf(k)$ with the most significant norms. In fact, by this selection, the $M$ coefficients with the largest norms will be updated, and the rest of the parameters will remain unchanged.  

Figure \ref{fig:SM-PUAP-icassp} illustrates a possible update in SM-PUAP\abbrev{SM-PUAP}{Set-Membership Partial-Update AP} algorithm in $\mathbb{R}^3$ for $L=0$. As can be seen, $\wbf(k+1)$ is far from the $\wbf_{{\rm SM-AP}}(k)$, and it will reduce the convergence rate of the SM-PUAP\abbrev{SM-PUAP}{Set-Membership Partial-Update AP} algorithm. In the next section, we will address this issue by presenting the I-SM-PUAP\abbrev{I-SM-PUAP}{Improved SM-PUAP} algorithm.

\begin{figure}[t!]
\begin{center}
\includegraphics[width=0.8\linewidth] {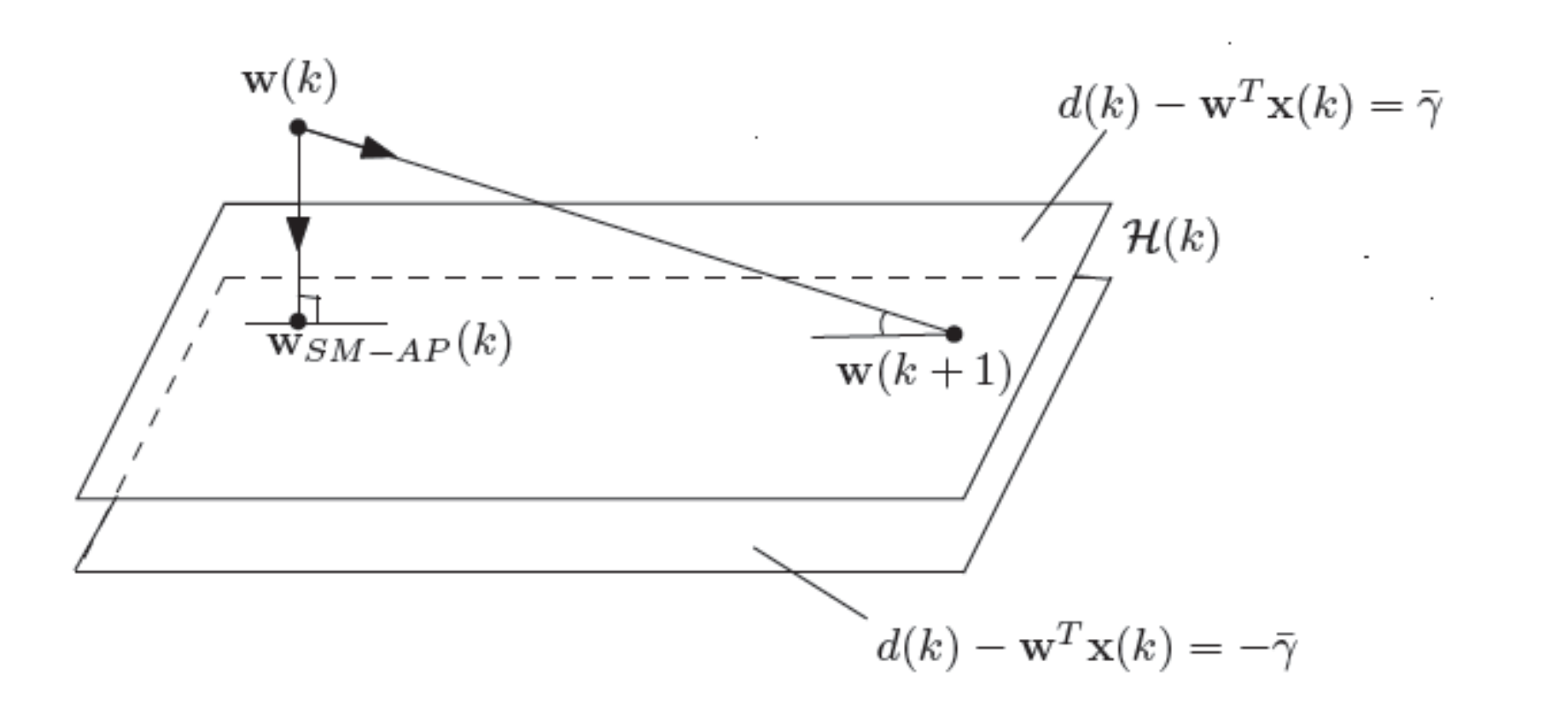} 
\caption{Update in SM-PUAP algorithm in $\mathbb{R}^3$ for $L=0$.}
\label{fig:SM-PUAP-icassp}
\end{center}
\end{figure}


\section{Improved  Set-membership  Partial-Update Affine Projection Algorithm} \label{sec:M-SM-PUAP-icassp}

In this section, we propose the I-SM-PUAP\abbrev{I-SM-PUAP}{Improved SM-PUAP} algorithm aiming at accelerating the convergence speed of SM-PUAP\abbrev{SM-PUAP}{Set-Membership Partial-Update AP} algorithm and 
decreasing the number of updates.

Since the partial update strategy deviates the updating direction from the one determined by the input signal vector $\xbf(k)$ utilized by the SM-PUAP\abbrev{SM-PUAP}{Set-Membership Partial-Update AP} algorithm, it is natural that the size of the step for a partial update algorithm should be smaller than the corresponding algorithm that updates all coefficients. A solution to this problem is to constrain the Euclidean norm of the coefficient disturbance of the partial update algorithm to the disturbance implemented by the originating nonpartial updating algorithm, in our case the SM-AP\abbrev{SM-AP}{Set-Membership Affine Projection} algorithm. For that, we build hypersphere, $S(k)$, whose radius is the distance between $\wbf(k)$ and the SM-AP\abbrev{SM-AP}{Set-Membership Affine Projection} update. The SM-AP\abbrev{SM-AP}{Set-Membership Affine Projection} update takes a step towards 
the hyperplanes $d(k)-\wbf^T\xbf(k)=\pm\gammabar$ with the minimum disturbance, i.e., when the step in the direction $\xbf(k)$ touches the hyperplane perpendicularly. Therefore, the radius of the hypersphere $S(k)$ is given by
\begin{align}
\mu(k)=\min\Big(\frac{|\wbf^T(k)\xbf(k)-d(k)\pm\gammabar|}{\|\xbf(k)\|_2}\Big),
\end{align}
where $\|\cdot\|_2$ is the Euclidean norm in $\mathbb{R}^{N+1}$. The equation describing the hypersphere $S(k)$ with the radius $\mu(k)$ and centered at $\wbf(k)$ is as follows \symbl{$S(k)$}{The hypersphere in $\mathbb{R}^{N+1}$ centered at $\wbf(k)$ with the radius $\mu(k)$}
\begin{align}
(w_0-w_0(k))^2+\cdots+(w_N-w_N(k))^2=\mu^2(k). \label{eq:sphere-icassp}
\end{align}

As can be observed in Figure \ref{fig:SM-PUAP-icassp}, $\wbf(k+1)$ is the point where, starting from $\wbf(k)$, the vector representing the $\wbf(k+1)$ direction touches the hyperplane
$d(k)-\wbf^T\xbf(k)=\gammabar$. Unlike the SM-PUAP\abbrev{SM-PUAP}{Set-Membership Partial-Update AP} algorithm, in the I-SM-PUAP\abbrev{I-SM-PUAP}{Improved SM-PUAP} algorithm  $\wbf(k+1)$ is the point where, starting from $\wbf(k)$, the vector representing the partial direction touches the defined  $N$ dimensional hypersphere $S(k)$ and points at a sparse version of $\xbf(k)$. A visual interpretation of the I-SM-PUAP\abbrev{I-SM-PUAP}{Improved SM-PUAP} algorithm is described in Figure \ref{fig:M-SM-PUAP-icassp}.

\begin{figure}[t!]
\begin{center}
\includegraphics[width=0.8\linewidth] {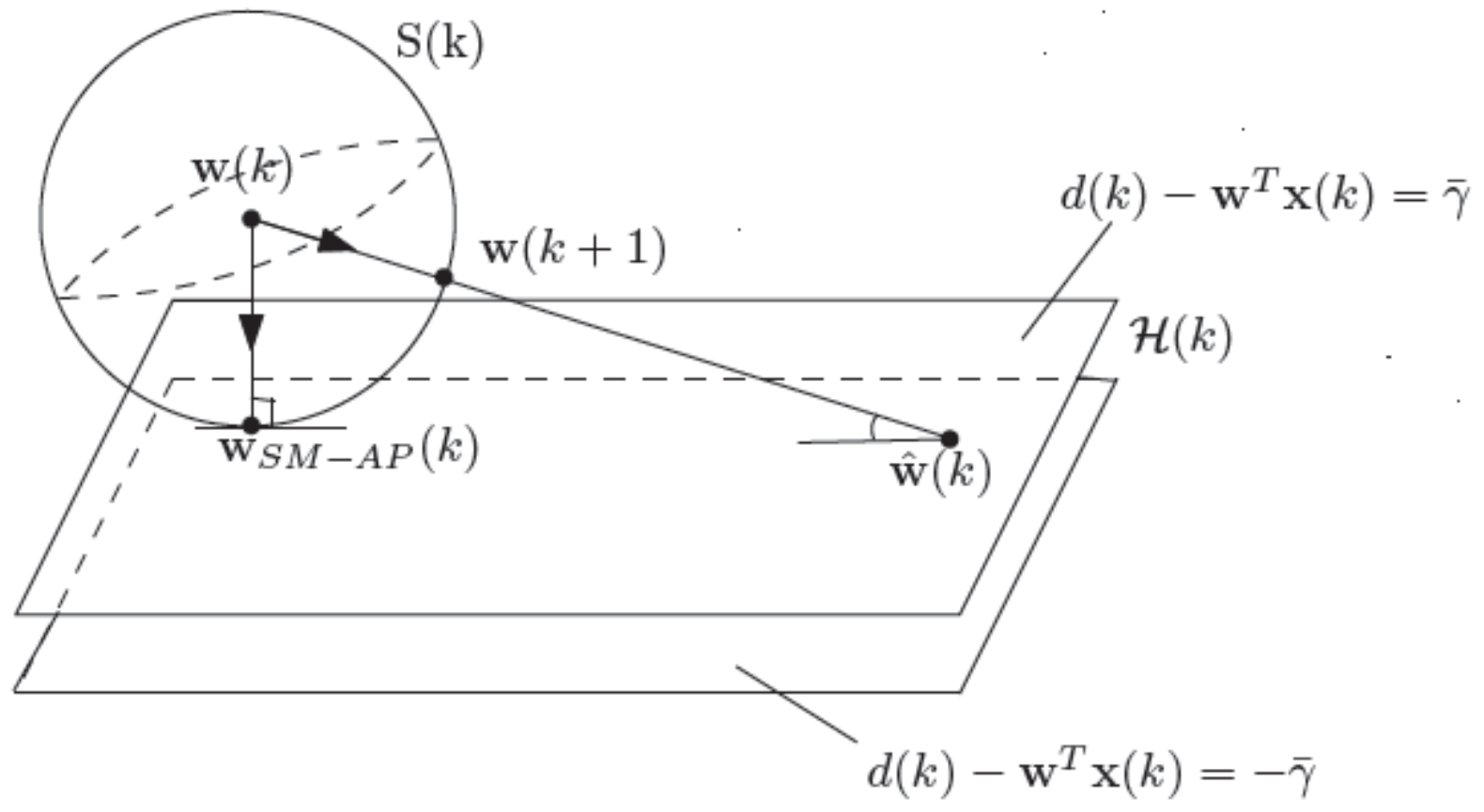} 
\caption{Update in I-SM-PUAP algorithm in $\mathbb{R}^3$ for $L=0$.}
\label{fig:M-SM-PUAP-icassp}
\end{center}
\end{figure}

Define $\hat{\wbf}(k)$ as the update result of Equation (\ref{eq:update_SM-PUAP-icassp}) with $\gammabf(k)=[0~\cdots~0]^T$.
In order to find the update of $\wbf(k)$ to the boundary of hypersphere $S(k)$ such that 
$\tilde{\Cbf}_{{\cal I}_{{M}}(k)}\wbf(k+1)=\tilde{\Cbf}_{{\cal I}_{{M}}(k)}\wbf(k)$ 
we have to find the intersection of the hypersphere $S(k)$ with the line $l(k)$ passing through $\wbf(k)$ and $\hat{\wbf}(k)$. 
This line is parallel to the vector $\ubf(k)=\frac{\abf(k)}{\|\abf(k)\|_2}$, where $\abf(k)=[\hat{w}_0(k)-w_0(k)~\cdots~\hat{w}_N(k)-w_N(k)]^T$. Hence, the equation of the line $l(k)$ is given as follows
\begin{align}
\left\{\begin{array}{ll}\frac{w_0-w_0(k)}{u_0(k)}=\cdots=\frac{w_i-w_i(k)}{u_i(k)}=\cdots=\frac{w_N-w_N(k)}{u_N(k)},&\text{for}~i\in{\cal I}_{{M}}(k),\\
w_i=w_i(k),&\text{for}~i\not\in{\cal I}_{{M}}(k).\end{array}\right. \label{eq:line-icassp}
\end{align}

In order to find the intersection of the line $l(k)$ with the hypersphere $S(k)$, we should replace Equation (\ref{eq:line-icassp}) in  Equation (\ref{eq:sphere-icassp}). 
Thus, we will attain $w_i=w_i(k)$ for $i\not\in{\cal I}_{{M}}(k)$, and for $i\in{\cal I}_{{M}}(k)$ we have
\begin{align}
&\frac{u_0^2(k)}{u_i^2(k)}(w_i-w_i(k))^2+\cdots+(w_i-w_i(k))^2+\cdots+\frac{u_N^2(k)}{u_i^2(k)}(w_i-w_i(k))^2=\mu^2(k).
\end{align}
Then,
\begin{align}
(w_i-w_i(k))^2=u_i^2(k)\mu^2(k),
\end{align}
where we obtained  the last equality owing to $\|\ubf(k)\|_2=1$. Therefore, the intersections of the line $l(k)$ and the hypersphere $S(k)$ are given by
\begin{align}
w_i=w_i(k)\pm u_i(k)\mu(k). \label{eq:intersection-icassp}
\end{align}
We will choose the positive sign in Equation (\ref{eq:intersection-icassp}) since the direction of the vector $\abf(k)$ is from $\wbf(k)$ to $\hat{\wbf}(k)$. As a result,  vector $\wbf(k+1)$ becomes as below
\begin{align}
\wbf(k+1)=\wbf(k)+\mu(k)\ubf(k).
\end{align}

Also, as an alternative method, we can get $\wbf(k+1)$ through an elegant geometrical view. Denote $\wbf(k+1)$ in Equation (\ref{eq:update_SM-PUAP-icassp}) as $\hat{\wbf}(k)$ while taking $\gammabf(k)=[0~\cdots~0]^T$. Define $\abf(k)$ as
\begin{align}
\abf(k)=\hat{\wbf}(k)-\wbf(k)=\Cbf_{{\cal I}_{{M}}(k)}\Xbf(k)\Pbf(k)\ebf(k). \label{eq:geometric-icassp}
\end{align}
If we take the step size equal to $\|\abf(k)\|_2$ and do the update in the direction of $\frac{\abf(k)}{\|\abf(k)\|_2}$, then the parameters will reach $\hat{\wbf}(k)$. However, our objective is to reach the boundary of hypersphere $S(k)$ centered at $\wbf(k)$ with radius $\mu(k)$ in the direction of $\frac{\abf(k)}{\|\abf(k)\|_2}$, thus the step size must be equal to the radius of $S(k)$ so that the update equation becomes
\begin{align}
\wbf(k+1)&=\wbf(k)+\mu(k)\frac{\abf(k)}{\|\abf(k)\|_2}=\wbf(k)+\mu(k)\ubf(k).
\end{align}
Table \ref{tb:m_sm_puap-icassp} summarizes the I-SM-PUAP algorithm.

\begin{table}[t!]
\caption{Improved Set-Membership Partial-Update Affine Projection(I-SM-PUAP) Algorithm}
\begin{center}
\begin{footnotesize}
\begin {tabular}{|l|} \hline\\ \hspace{2.2cm}{\bf I-SM-PUAP Algorithm}\\ \\
\hline\\
Initialization
\\
$\xbf(-1)=\wbf(0)=[0~\cdots~0]^T$\\
$\delta=$ small positive constant\\
choose $\gammabar$\\
Do for $k\geq0$\\
\hspace*{0.15cm} $\ebf(k)=\dbf(k)-\Xbf^T(k)\wbf(k)$\\
\hspace*{0.15cm} ${\rm if}~|e(k)|>\gammabar$\\
\hspace*{0.3cm} $\mu(k)=\min\Big(\frac{|-e(k)\pm\gammabar|}{\|\xbf(k)\|_2}\Big)$\\
\hspace*{0.3cm} $\abf(k)=\Cbf_{{\cal I}_{{M}}(k)}\Xbf(k)[\Xbf^T(k)\Cbf_{{\cal I}_{{M}}(k)}\Xbf(k)+\delta\Ibf]^{-1}\ebf(k)$ \\
\hspace*{0.3cm} $\wbf(k+1)=\wbf(k)+\frac{\mu(k)}{\|\abf(k)\|_2}\abf(k)$
\\
\hspace*{0.15cm} else\\
\hspace*{0.3cm} $\wbf(k+1)=\wbf(k)$\\
\hspace*{0.15cm} end\\
end  \\
\\
\hline
\end {tabular}
\end{footnotesize}
\end{center}
\label{tb:m_sm_puap-icassp}
\end{table}


\section{Simulations} \label{sec:simulation-icassp} 

\subsection{Scenario 1}

In this section, the SM-PUAP\abbrev{SM-PUAP}{Set-Membership Partial-Update AP} algorithm \cite{Diniz_adaptiveFiltering_book2013} and the proposed I-SM-PUAP\abbrev{I-SM-PUAP}{Improved SM-PUAP} algorithm are applied to a system identification problem. The unknown system has order $N=79$ and its coefficients are random scalars drawn from the standard normal distribution. The input signal is a binary phase-shift keying (BPSK)\abbrev{BPSK}{Binary Phase-Shift Keying} signal with $\sigma_x^2=1$. The signal-to-noise ratio (SNR)\abbrev{SNR}{Signal-to-Noise Ratio} is set to 20 dB, i.e., $\sigma_n^2=0.01$. The bound on the output estimation error is chosen as $\gammabar=\sqrt{25\sigma_n^2}$. Also, we adopt the threshold bound vector $\gammabf(k)$ as $\gamma_0(k)=\frac{\gammabar e(k)}{|e(k)|}$ and $\gamma_i(k)=d(k-i)-\wbf^T(k)\xbf(k-i)$, for $i=1,\cdots,L$ \cite{Diniz_adaptiveFiltering_book2013,Markus_edcv_eusipco2013}. The regularization constant, $\delta$, is $10^{-12}$ and $\wbf(0)=[1~\cdots~1]^T$ which is not close to the unknown system. All learning curves averaged over 200 trials. We are updating 50 percent of the components randomly chosen of the filter to illustrate the partial updating, i.e., half of the elements of ${\cal I}_M(k)$ are nonzero at each time instant $k$. 

Figure \ref{fig:sim1-icassp} shows the learning curves for the I-SM-PUAP\abbrev{I-SM-PUAP}{Improved SM-PUAP} algorithm with $L=1,4$, and it illustrates the learning curves for the SM-PUAP\abbrev{SM-PUAP}{Set-Membership Partial-Update AP} algorithm with $L=64$ and 69. Also, in Figure \ref{fig:sim1-icassp} a blue curve is depicted using correlated inputs and $L=1$. In fact, for the blue curve all of the specifications of the system are the same as explained above and the only difference is the input signal. The correlated input signal is chosen as $x(k)=0.95x(k-1)+0.19x(k-2)+0.09x(k-3)-0.5x(k-1)+m(k-4)$, where $m(k)$ is a zero-mean Gaussian noise with unit variance.

\begin{figure}[t!]
\begin{center}
\includegraphics[width=1\linewidth] {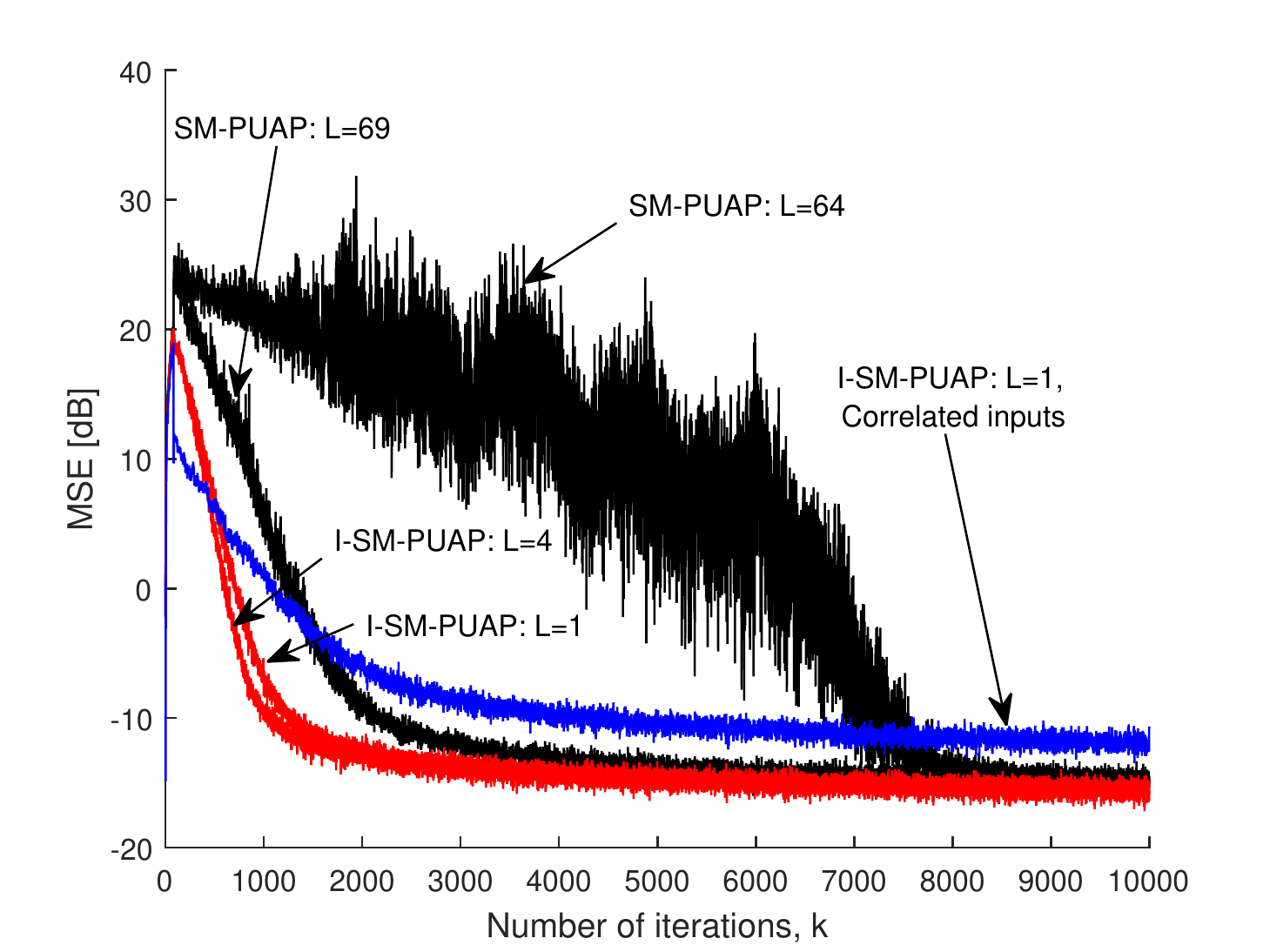} 
\caption{Learning curves of the I-SM-PUAP  and the SM-PUAP algorithms applied on system identification problem.}
\label{fig:sim1-icassp}
\end{center}
\end{figure}

The average number of updates performed by the I-SM-PUAP\abbrev{I-SM-PUAP}{Improved SM-PUAP} algorithm are 8.3$\%$ and 6.5$\%$ for $L=1$ and 4, respectively, and 20$\%$ in the case of the correlated input signal. The average number of updates implemented by the SM-PUAP\abbrev{SM-PUAP}{Set-Membership Partial-Update AP} algorithm are 14$\%$ and 25$\%$ for $L=69$ and 64, respectively. Note that in both algorithms we have to find the inverse of an $(L+1)\times (L+1)$ matrix,  thus large $L$ implies high computational complexity. Therefore, the I-SM-PUAP\abbrev{I-SM-PUAP}{Improved SM-PUAP} algorithm requires lower implementation time since it presents fast convergence even for a small value of $L$. Also, it is worth mentioning that for $L<64$ the SM-PUAP\abbrev{SM-PUAP}{Set-Membership Partial-Update AP} algorithm does not reach its steady-state in 10000 iterations. From the results, we can observe that the proposed algorithm, I-SM-PUAP,\abbrev{I-SM-PUAP}{Improved SM-PUAP} has faster convergence speed and lower number of updates as compared to the SM-PUAP\abbrev{SM-PUAP}{Set-Membership Partial-Update AP} algorithm.

\subsection{Scenario 2}
In this section, we perform the equalization of a channel with the following impulse response
\begin{align}
\hbf=[1~2~3~4~4~3~2~1]^T.
\end{align} 
We use a known training signal that consists of independent binary samples $(-1,1)$ and an additional Gaussian white noise with variance 0.01 is present at the channel output. The I-SM-PUAP\abbrev{I-SM-PUAP}{Improved SM-PUAP} and the SM-PUAP\abbrev{SM-PUAP}{Set-Membership Partial-Update AP} algorithms are applied to find the impulse response of an equalizer of order 80. The delay in the reference signal is selected as 45. The parameters $\gammabar$ and $\gammabf(k)$ are chosen as $\sqrt{25\sigma_n^2}$ and the simple choice constraint vector is utilized as Scenario 1, respectively. The regularization constant, $\delta$, is $10^{-12}$ and $\wbf(0)=[1~\cdots~1]^T$. All learning curves are averaged over 100 trials. At each iteration, half of the elements of ${\cal I}_M(k)$ are set nonzero randomly. The memory-length, $L$, is 3.

Figure~\ref{fig:Equalization-icassp} shows the learning curves  for the I-SM-PUAP\abbrev{I-SM-PUAP}{Improved SM-PUAP} and the SM-PUAP\abbrev{SM-PUAP}{Set-Membership Partial-Update AP} algorithms. The convolution of the equalizer impulse response at a given iteration after convergence with the channel impulse response is shown in Figure~\ref{fig:Convolution-icassp}. The average number of updates implemented by the I-SM-PUAP\abbrev{I-SM-PUAP}{Improved SM-PUAP} and the SM-PUAP\abbrev{SM-PUAP}{Set-Membership Partial-Update AP} algorithms are 61$\%$ and 82$\%$, respectively. As can be seen, the I-SM-PUAP\abbrev{I-SM-PUAP}{Improved SM-PUAP} algorithm has lower MSE\abbrev{MSE}{Mean-Squared Error} and lower number of updates compared to the SM-PUAP\abbrev{SM-PUAP}{Set-Membership Partial-Update AP} algorithm.

\begin{figure}[t!]
\centering
\subfigure[b][]{\includegraphics[width=.48\linewidth,height=7cm]{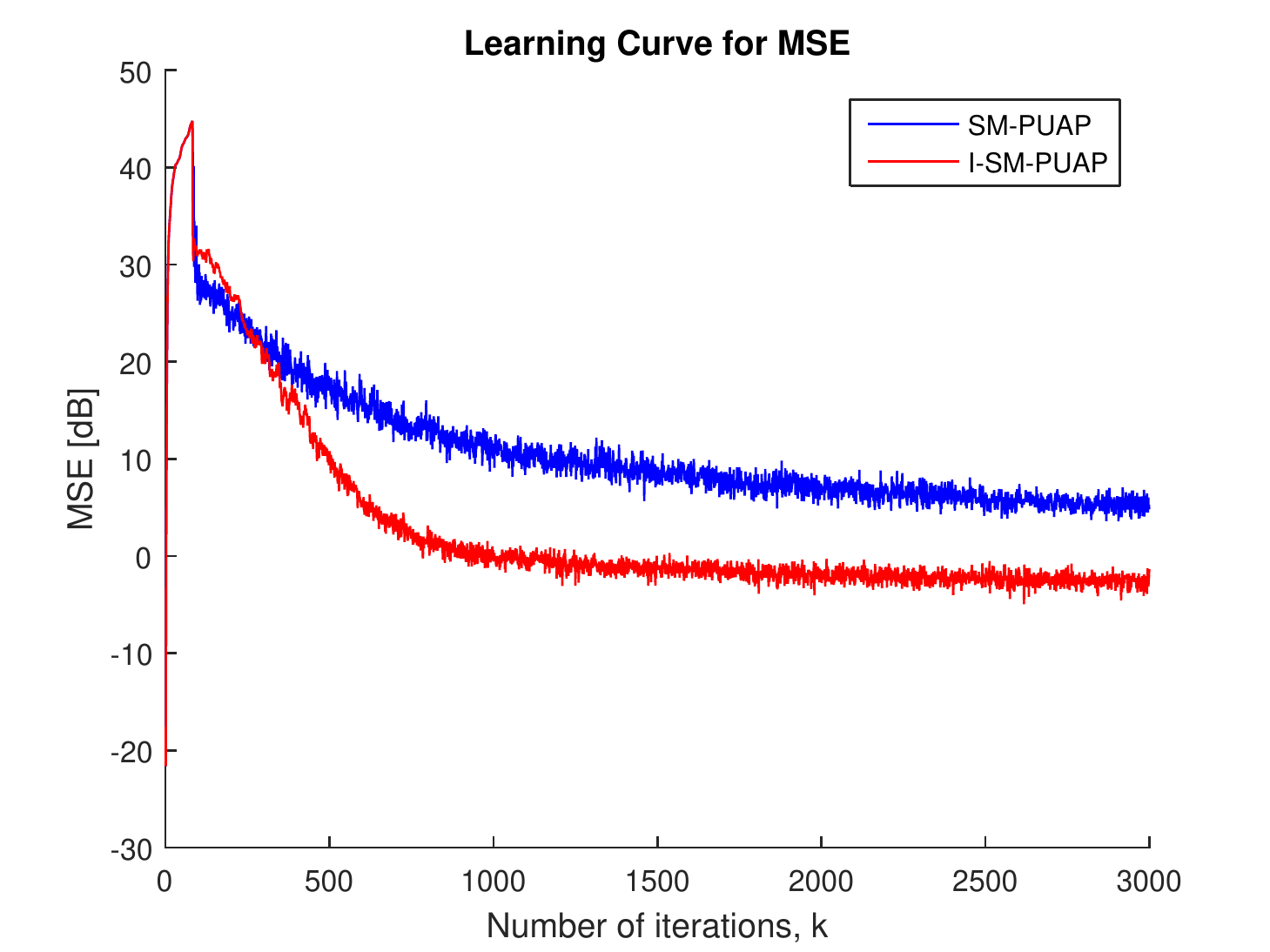}
\label{fig:Equalization-icassp}}
\subfigure[b][]{\includegraphics[width=.48\linewidth,height=7cm]{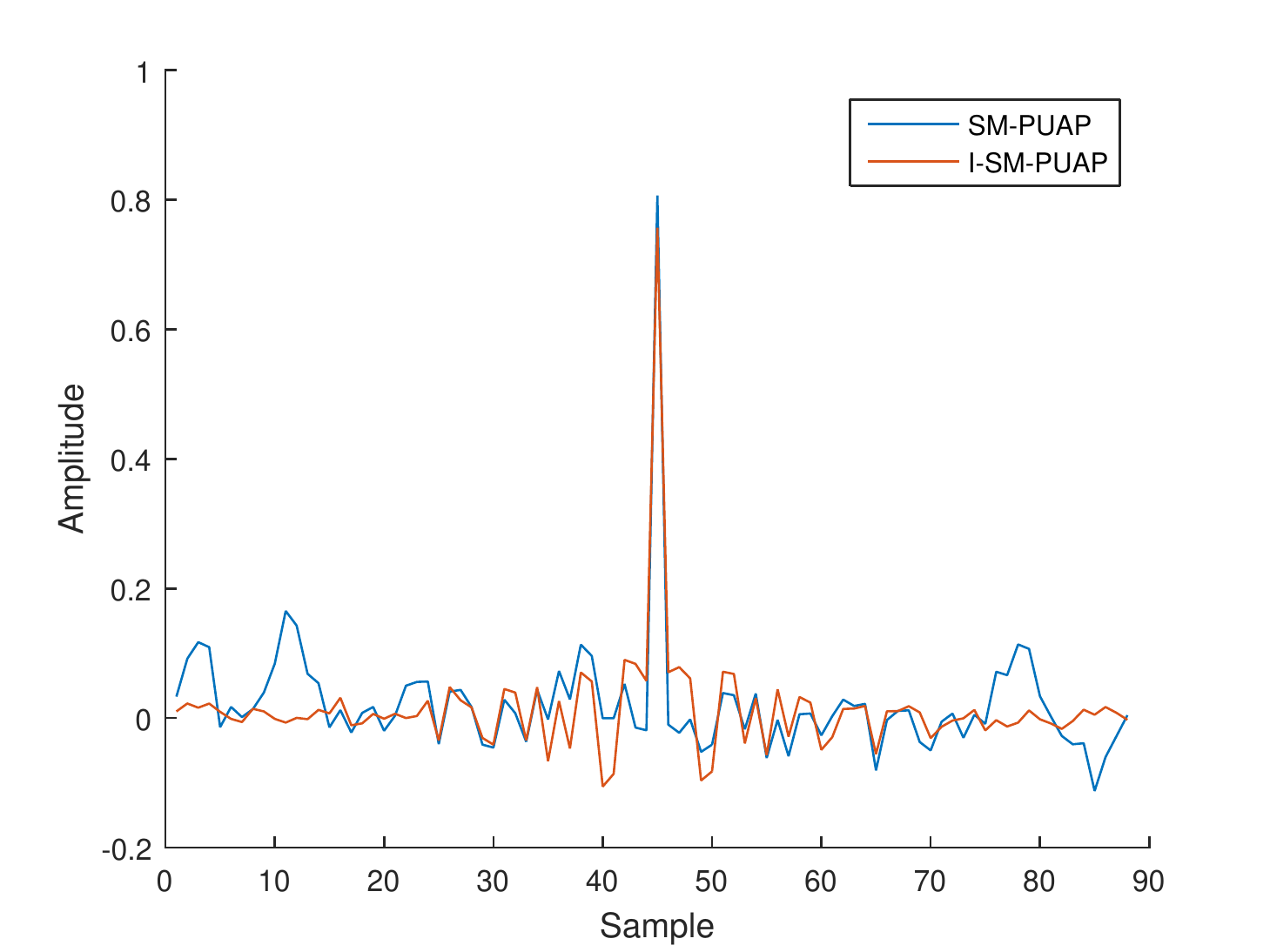}
\label{fig:Convolution-icassp}}
\caption{(a) Learning curves of the I-SM-PUAP and the SM-PUAP algorithms performing the equalization of a channel; (b) convolution results. \label{fig:Equalization}}
\end{figure}


\section{Conclusions} \label{sec:conclusion-icassp}

In this chapter, we have introduced the improved set-membership partial-update affine projection (I-SM-PUAP) algorithm aiming at accelerating the convergence rate of the set-membership partial-update affine projection (SM-PUAP) algorithm, with lower computational complexity and reduced number of updates. To achieve this goal, we use the distance between the present weight vector and the
one obtained with the SM-AP\abbrev{SM-AP}{Set-Membership Affine Projection} update, in order to provide a hypersphere that upper bounds the coefficient disturbance.  
Numerical simulations for the system identification and the channel equalization problems have confirmed that the I-SM-PUAP\abbrev{I-SM-PUAP}{Improved SM-PUAP} algorithm has not only faster convergence rate, but also it requires a lower number of updates as compared to the SM-PUAP\abbrev{SM-PUAP}{Set-Membership Partial-Update AP} algorithm.
  \chapter{Adaptive Filtering Algorithms for Sparse System Modeling}

Adaptive filtering applied to signals originating from time-varying systems find applications in a wide diversity of areas such as communications, control, radar, acoustics, and speech processing. 
Nowadays, it is well known that many types of signal or system parameters admit sparse representation in a certain domain. However, classical adaptive algorithms such as the least-mean-square (LMS)\abbrev{LMS}{Least-Mean-Square}, the normalized LMS (NLMS)\abbrev{NLMS}{Normalized LMS}, the affine projection (AP)\abbrev{AP}{Affine Projection}, and the recursive least-squares (RLS)\abbrev{RLS}{Recursive Least-Squares} do not take into consideration the sparsity in the signal or system models.

Recently, it has been understood that by exploiting appropriately signal sparsity, significant improvement in convergence rate and steady-state performance can be achieved. As a consequence, many extensions of the classical algorithms were proposed aiming at exploiting sparsity. One of the most widely used approaches consists in updating each filter coefficient using a step-size proportional to its magnitude in order to speed up the convergence rate of the coefficients with large magnitudes. This approach led to the development of a family of algorithms known as {\it proportionate}~\cite{Duttweiler_PNLMS_tsap2000,Benesty_IPNLMS_icassp2002,Gay_pnlmsPlusPlus_acssc1998,Diniz_sm_pap_jasmp2007,Paleologu_papaEcho_spl2010}. Another interesting approach to exploit sparsity is to include a {\it sparsity-promoting penalty} (sometimes called regularization) function into the original optimization problem of classical algorithms~\cite{Markus-phdthesis}. Within this approach, most algorithms employ the $l_1$ norm as the sparsity-promoting 
penalty~\cite{Vitor_SparsityAwareAPA_sspd2011,Theodoridis_l1ball_tsp2011,Chen_sparseLMS_icassp2009,Babadi_Sparse_RLS_tsp2010}, but recently an approximation to the $l_0$ norm has shown some 
advantages~\cite{Markus_sparseSMAP_tsp2014,Markus_apssiAnalysis_icassp2014,Markus_apssi_icassp2013,Gu_l0_LMS_SPletter2009}. In addition, these two approaches were combined and tested 
in~\cite{Pelekanakis2012,Markus_proportionatePlusPenalty_iscas2016} yielding interesting results. Observe that in all of the aforementioned approaches something is being included/added to the classical algorithms, thus entailing an increase in their computational complexity. 

In this chapter, we use a different strategy to exploit sparsity. 
Instead of including additional features in the algorithm, as the techniques described in the previous paragraph, we actually discard some coefficients, thus reducing the computational burden.
This idea is motivated by the existence of some uncertainty in the coefficients in practical applications. Indeed, a measured sparse impulse response of a system presents a few coefficients concentrating most of the energy, whereas the other coefficients are close to zero, but not precisely equal to zero~\cite{Markus_sparseSMAP_tsp2014}~\footnote{A system whose impulse response presents this characteristic is formally known as a {\it compressible system}~\cite{Markus-phdthesis}.}. Thus, if we have some prior information about the uncertainty in those parameters, then we can replace the parameters which are ``lower than'' this uncertainty with zero (i.e., discard the coefficients) in order to save computational resources. 

In addition to this new way of exploiting sparsity, we also employ the set-membership filtering (SMF) approach \cite{Gollamudi_smf_letter1998,Diniz_adaptiveFiltering_book2013} in order to generate the {Simple Set-Membership Affine Projection} (S-SM-AP)\abbrev{S-SM-AP}{Simple SM-AP} algorithm, which is mostly the combination of the 
set-membership affine projection  algorithm~\cite{Werner_sm_ap_letter2001} with our strategy to exploit sparsity. The SMF\abbrev{SMF}{Set-Membership Filtering} approach is used just to reduce the computational burden even further since the filter coefficients are updated only when the estimation error is greater than a predetermined threshold. 

Moreover, we derive the improved S-SM-AP\abbrev{S-SM-AP}{Simple SM-AP} (IS-SM-AP)\abbrev{IS-SM-AP}{Improved S-SM-AP} algorithm to reduce the overall number of computations required by the S-SM-AP\abbrev{S-SM-AP}{Simple SM-AP} algorithm even further by replacing small coefficients with zero. Also, we obtain the simple affine projection (S-AP)\abbrev{S-AP}{Simple AP} and the improved S-AP (IS-AP)\abbrev{IS-AP}{Improved S-AP} algorithms as special cases of the S-SM-AP\abbrev{S-SM-AP}{Simple SM-AP} and the IS-SM-AP\abbrev{IS-SM-AP}{Improved S-SM-AP} algorithms, respectively. The S-AP\abbrev{S-AP}{Simple AP} and the IS-AP\abbrev{IS-AP}{Improved S-AP} algorithms do not resort to the SMF concept and can be regarded as affine projection algorithms for sparse systems. 

Finally, we introduce some sparsity-aware RLS\abbrev{RLS}{Recursive Least-Squares} algorithms employing the discard function and the $l_0$ norm approximation. The first proposed algorithm, the RLS for sparse systems (S-RLS)\abbrev{S-RLS}{RLS Algorithm for Sparse System}, sets low weights to the coefficients close to zero and exploits system sparsity with low computational complexity. On the other hand, the second algorithm, the $l_0$ norm RLS ($l_0$-RLS)\abbrev{$l_0$-RLS}{$l_0$ Norm RLS}, has higher computational complexity in comparison with the S-RLS\abbrev{S-RLS}{RLS Algorithm for Sparse System} algorithm. For both algorithms, in order to reduce the computational load further, we apply a data-selective strategy~\cite{Gollamudi_smf_letter1998} leading to the data-selective S-RLS (DS-S-RLS)\abbrev{DS-S-RLS}{Data-Selective S-RLS} and the data-selective $l_0$-RLS (DS-$l_0$-RLS)\abbrev{DS-$l_0$-RLS}{Data-Selective $l_0$-RLS} algorithms. That is, the proposed algorithms update the weight vector if the output estimation error is larger than a prescribed value. By applying the data-selective strategy, both algorithms attain lower computational complexity compared to the RLS\abbrev{RLS}{Recursive Least-Squares} algorithm.

The content of this chapter was published in~\cite{Hamed_eusipco2016,Hamed_S_RLS_ICASSP2017}. In Sections~\ref{sec:SSM-AP-eusipco} and \ref{sec:SM-PAPA-eusipco}, we review the sparsity-aware SM-AP (SSM-AP)\abbrev{SSM-AP}{Sparsity-Aware SM-AP} algorithm and the set-membership proportionate AP algorithm (SM-PAPA)\abbrev{SM-PAPA}{Set-Membership Proportionate AP Algorithm}, respectively. The proposed {S-SM-AP}\abbrev{S-SM-AP}{Simple SM-AP} algorithm is derived in Section~\ref{sec:ss-sm-ap-eusipco}. Sections~\ref{sec:s-rls-sparse} and~\ref{sec:l0-rls-sparse} propose the S-RLS\abbrev{S-RLS}{RLS Algorithm for Sparse System} and the $l_0$-RLS\abbrev{$l_0$-RLS}{$l_0$ Norm RLS} algorithms, respectively. Simulations are presented in Section~\ref{sec:simulations-eusipco} and 
Section~\ref{sec:conclusions-eusipco} contains the conclusions.


\section{Sparsity-Aware SM-AP Algorithm}\label{sec:SSM-AP-eusipco}

In literature, a method to deal with the sparsity has been obtained by adding a penalty function to the original objective function \cite{Vitor_SparsityAwareAPA_sspd2011,Markus-phdthesis,Markus_sparseSMAP_tsp2014,Markus_apssiAnalysis_icassp2014,Markus_apssi_icassp2013}. This penalty function is generally related to the $l_0$ or $l_1$ norms. Utilizing $l_0$ norm has some difficulties since it leads to an NP-hard problem. Therefore, we must try to approximate the $l_0$ norm by {\it almost everywhere} differentiable functions, for then we can apply stochastic gradient methods to solve the optimization problem. In other words, the $l_0$ norm can be estimated by a continuous function $G_\beta:\mathbb{R}^{N+1}\rightarrow\mathbb{R}_+$, where $\beta\in\mathbb{R}_+$ is a parameter responsible for controlling the agreement between quality of the estimation and smoothness of $G_\beta$. This function must satisfy the following condition~\cite{Markus-phdthesis,Markus_sparseSMAP_tsp2014} \symbl{$G_\beta$}{Continuous and almost everywhere differentiable function that approximates the $l_0$ norm; $\beta$ controls the quality of the approximation}
\begin{align}
\lim_{\beta\rightarrow\infty}G_\beta(\wbf)=\|\wbf\|_0,
\end{align}
where $\|\cdot\|_0$ denotes the $l_0$ norm which, for $\wbf\in\mathbb{R}^{N+1}$, is defined as  $\|\wbf\|_0\triangleq\#\{i\in\mathbb{N}:~w_i\neq0\}$, in which $\#$ stands for the cardinality of a finite set. 
Here we present four examples of function $G_\beta$~\cite{Markus-phdthesis,Markus_sparseSMAP_tsp2014}
\begin{subequations}
\begin{align}
{\rm LF:~}G_\beta(\wbf)&=\sum_{i=0}^N(1-e^{-\beta|w_i|}), \label{eq:mult_Laplace}\\
{\rm MLF:~}G_\beta(\wbf)&=\sum_{i=0}^N(1-e^{-0.5\beta^2w^2_i}), \label{eq:modf_mult_Laplace}\\
{\rm GMF:~}G_\beta(\wbf)&=\sum_{i=0}^N(1-\frac{1}{1+\beta|w_i|}), \label{eq:mult_Geman}\\
{\rm MGMF:~}G_\beta(\wbf)&=\sum_{i=0}^N(1-\frac{1}{1+\beta^2w^2_i}). \label{eq:modf_mult_Geman}
\end{align}
\end{subequations}
The functions expressed in Equations~\eqref{eq:mult_Laplace} and~\eqref{eq:mult_Geman} are called the multivariate Laplace function (LF)\abbrev{LF}{Laplace Function} and the multivariate Geman-McClure function (GMF)\abbrev{GMF}{Geman-McClure Function}, respectively. Equations~\eqref{eq:modf_mult_Laplace} and~\eqref{eq:modf_mult_Geman} are modifications of the LF\abbrev{LF}{Laplace Function} and the GMF\abbrev{GMF}{Geman-McClure Function}, respectively, so that they have continuous derivatives too. Figure~\ref{fig:sparsity_Functions} shows the univariate Laplace and Geman-McClure functions for $\beta=5$.

\begin{figure}[t!]
\centering
\subfigure[b][]{\includegraphics[width=.48\linewidth,height=7cm]{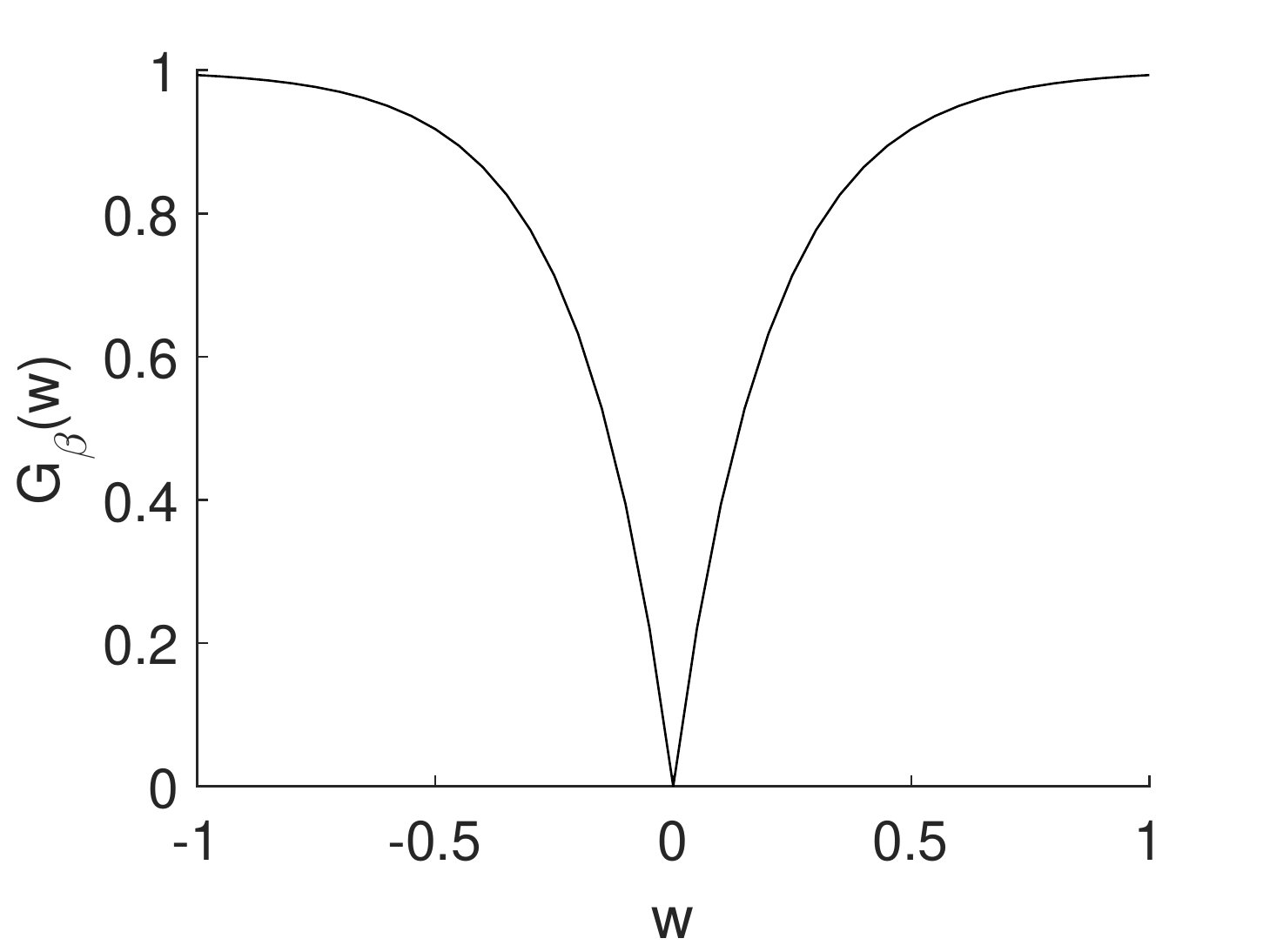}
\label{fig:Laplace_sparse}}
\subfigure[b][]{\includegraphics[width=.48\linewidth,height=7cm]{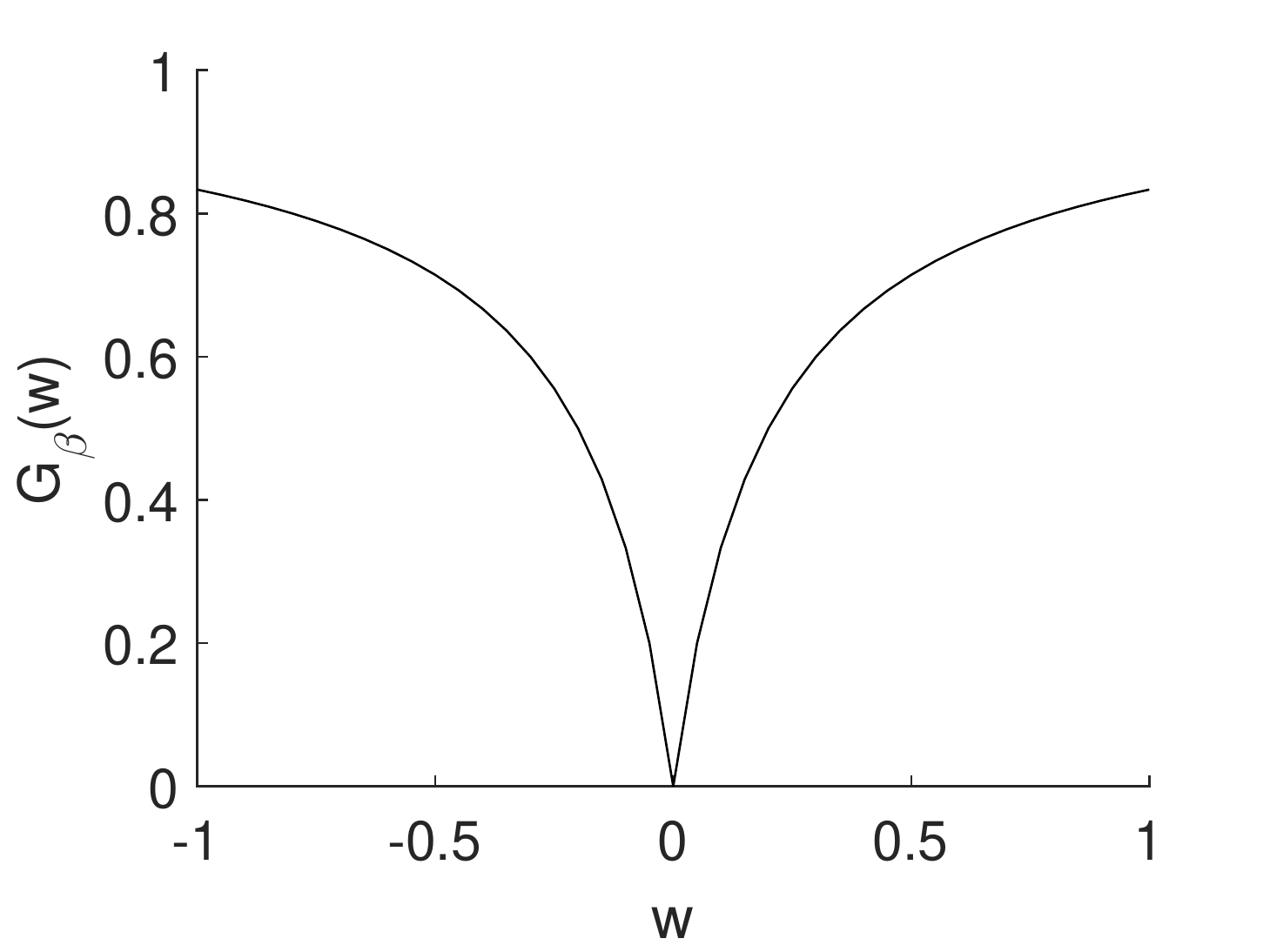}
\label{fig:Geman_sparse}}
\caption{Univariate functions $G_\beta(w)$, with $w\in[-1,1]$ and $\beta=5$: (a) LF; (b) GMF.  \label{fig:sparsity_Functions}}
\end{figure}

The gradient of $G_\beta$ is defined as follows \symbl{$\gbf_\beta(\wbf)$}{Gradient of $G_\beta(\wbf)$ with respect to $\wbf$}
\begin{align}
\nabla G_\beta(\wbf)\triangleq\gbf_\beta(\wbf)\triangleq[g_\beta(w_0)~\cdots~g_\beta(w_N)]^T,
\end{align}
where $g_\beta(w_i)=\frac{\partial G_\beta(\wbf)}{\partial w_i}$. Note that~\eqref{eq:mult_Laplace} and~\eqref{eq:mult_Geman} are not differentiable at the origin, thus we define their derivatives at the origin equal to zero. The derivatives corresponding to~\eqref{eq:mult_Laplace}-\eqref{eq:modf_mult_Geman} are, respectively,
\begin{subequations}
\begin{align}
g_\beta(w_i)&=\beta {\rm sgn}(w_i){\rm e}^{-\beta|w_i|},\\ g_\beta(w_i)&=\beta^2w_i{\rm e}^{-0.5\beta^2w_i^2},\\
g_\beta(w_i)&=\frac{\beta{\rm sgn}(w_i)}{(1+\beta|w_i|)^2},\\ g_\beta(w_i)&=\frac{2\beta^2w_i}{(1+\beta^2w_i^2)^2},
\end{align}
\end{subequations}
where ${\rm sgn}(\cdot)$ denotes the sign function. \symbl{${\rm sgn}(\cdot)$}{The sign function} The interested reader can find the details of approximating the $l_0$ norm in~\cite{Markus_sparseSMAP_tsp2014}.

The SSM-AP\abbrev{SSM-AP}{Sparsity-Aware SM-AP} algorithm performs an update whenever $|e(k)|=|d(k)-\wbf^T(k)\xbf(k)|>\gammabar$, following an update recursion that is an approximation of the solution to the optimization problem \cite{Markus_sparseSMAP_tsp2014}
\begin{align}
&\min \|\wbf(k+1)-\wbf(k)\|_2^2+\alpha\|\wbf(k+1)\|_0\nonumber\\
&{\rm subject~to}\nonumber\\
&\dbf(k)-\Xbf^T(k)\wbf(k+1)=\gammabf(k),\label{eq:ssm_ap_optimization-eusipco}
\end{align}
where $\alpha\in\mathbb{R}_+$ denotes the weight given to the $l_0$ norm.

After replacing the $l_0$ norm with its approximation $G_\beta$, and using the method of Lagrange multipliers, the updating equation of the SSM-AP\abbrev{SSM-AP}{Sparsity-Aware SM-AP} algorithm is reached as follows \cite{Markus_sparseSMAP_tsp2014}
\begin{align}
\wbf(k+1)=
\left\{\begin{array}{ll}\wbf(k)+\Xbf(k)\Abf(k)[\ebf(k)-\gammabf(k)]&\\+\frac{\alpha}{2}[\Xbf(k)\Abf(k)\Xbf^T(k)-\Ibf]\gbf_\beta(\wbf(k))&\text{if }|e(k)|>\gammabar,\\\wbf(k)&\text{otherwise},\end{array}\right. \label{eq:update_rule_ssm-ap}
\end{align} 
where $\Abf(k)=(\Xbf^T(k)\Xbf(k))^{-1}$.

\section{Set-Membership Proportionate AP Algorithm}\label{sec:SM-PAPA-eusipco}

The sparsity of the signals in some applications motivates
us to update each coefficient of the model independently
of the others. Therefore, in adaptive filtering, one of the most widely used methods to exploit sparsity is by implementing coefficient updates that are proportional to the magnitude of the related coefficients. Thus, the coefficients with large magnitude will update with higher convergence rate and, as a result, we have faster overall convergence speed~\cite{Benesty_IPNLMS_icassp2002}. This approach leads to a well known family of algorithms called proportionate. A noticeable number of algorithms utilizing the proportionate approach have been already introduced in the literature. Some of them are the proportionate NLMS\abbrev{NLMS}{Normalized LMS} (PNLMS)~\cite{Duttweiler_PNLMS_tsap2000}\abbrev{PNLMS}{Proportionate Normalized LMS}, the proportionate AP algorithm (PAPA)~\cite{Paleologu_papaEcho_spl2010},\abbrev{PAPA}{Proportionate Affine Projection Algorithm} and their set-membership counterparts~\cite{Diniz_sm_pap_jasmp2007}. In this section, we review the set-membership PAPA (SM-PAPA)\abbrev{SM-PAPA}{Set-Membership Proportionate AP Algorithm}. The optimization criterion of the SM-PAPA\abbrev{SM-PAPA}{Set-Membership Proportionate AP Algorithm} when it implements an update (i.e., when $|e(k)|>\gammabar$) is given by
\begin{align}
&\min \|\wbf(k+1)-\wbf(k)\|^2_{\Mbf^{-1}(k)}\nonumber\\
&{\rm subject~to}\nonumber\\
&\dbf(k)-\Xbf^T(k)\wbf(k+1)=\gammabf(k).\label{eq:sm-papa_optimization-eusipco}
\end{align}
The norm in this optimization criterion is defined as $\|\wbf\|^2_{\Mbf}\triangleq\wbf^T\Mbf\wbf$ and matrix $\Mbf(k)$ is a diagonal weighting matrix of the form
\begin{align}
\Mbf(k)\triangleq{\rm diag}[m_0(k)~\cdots~m_N(k)],
\end{align} 
where 
\begin{align}
m_i(k)\triangleq\frac{1-r\mu(k)}{N}+\frac{r\mu(k)|w_i(k)|}{\|\wbf(k)\|_1},
\end{align}
with
\begin{align}
\mu(k)=\left\{\begin{array}{ll}1-\frac{\gammabar}{|e(k)|}&\text{if }|e(k)|>\gammabar,\\0&\text{otherwise},\end{array}\right.
\end{align}
and $r\in[0,1]$. Also, $\|\cdot\|_1$ stands for the $l_1$ norm and for $\wbf\in\mathbb{R}^{N+1}$ it is defined as $\|\wbf\|_1=\sum_{i=0}^N|w_i|$. Utilizing the method of Lagrange multipliers to solve (\ref{eq:sm-papa_optimization-eusipco}), the update equation of the SM-PAPA\abbrev{SM-PAPA}{Set-Membership Proportionate AP Algorithm} is obtained as follows \cite{Diniz_sm_pap_jasmp2007}
\begin{align}
&\wbf(k+1)=\nonumber\\
&\left\{\begin{array}{ll}\wbf(k)+\Mbf(k)\Xbf(k)[\Xbf^T(k)\Mbf(k)\Xbf(k)]^{-1}[\ebf(k)-\gammabf(k)]&\text{if }|e(k)|>\gammabar,\\\wbf(k)&\text{otherwise}.\end{array}\right.
\end{align}

\section{{A Simple Set-Membership Affine Projection Algorithm}} \label{sec:ss-sm-ap-eusipco}

In the previous sections, we have observed that to exploit sparsity, we require a higher number of arithmetic operations compared to the SM-AP algorithm, which cannot exploit sparsity. Here we introduce a new algorithm to exploit sparsity with low computational complexity. In this algorithm, instead of including/adding something to the classical algorithms, we discard the coefficients close to zero. 

In Subsection~\ref{sub:derivation-eusipco}, we propose {a Simple  Set-Membership Affine Projection (S-SM-AP)}\abbrev{S-SM-AP}{Simple SM-AP} algorithm that exploits the sparsity of the involved system with low computational complexity. For this purpose, the strategy consists in not updating the coefficients of the sparse filter which are close to zero. Then, in Subsection~\ref{sub:discussion-eusipco}, we include a discussion of some characteristics of the proposed algorithm. In Subsection~\ref{sub:Modified-eusipco}, we introduce an improved version of the proposed algorithm aiming at reducing the computational burden even further. Finally, in Subsection~\ref{sub:non-SM-eusipco}, we derive the S-AP\abbrev{S-AP}{Simple AP} and IS-AP\abbrev{IS-AP}{Improved S-AP} algorithms by not employing the SMF\abbrev{SMF}{Set-Membership Filtering} technique.


\subsection{Derivation of the {S-SM-AP} algorithm \label{sub:derivation-eusipco}}

Let us define the {\it discard function} $f_\epsilon:\mathbb{R}\rightarrow\mathbb{R}$ for the 
positive constant $\epsilon$ as follows \symbl{$f_\epsilon(\cdot)$}{Discard function; $\epsilon$ defines what is considered as close to zero}
\begin{align}
f_\epsilon(w)=\left\{\begin{array}{ll}w&{\rm if~} |w|> \epsilon,  \\0&{\rm if~} |w|\leq \epsilon. \end{array}\right.\label{eq:f_epsilon-eusipco}
\end{align}
That is, function $f_\epsilon$ discards the values of $w$ which are close to zero. 
The parameter $\epsilon$ defines what is considered as close to zero and, therefore, should be chosen 
based on some {\it a priori} information about the relative importance of a coefficient to the sparse system.
Figure \ref{fig:f-a-eusipco} depicts the function $f_\epsilon(w)$ for $\epsilon=10^{-4}$. Note that the function $f_\epsilon(w)$ is not differentiable at $\pm \epsilon$, however, we 
need to differentiate this function in order to derive the {S-SM-AP}\abbrev{S-SM-AP}{Simple SM-AP} algorithm. 
To address this issue, we define the derivative of $f_\epsilon(w)$ at $+\epsilon$ and $-\epsilon$ 
as equal to the left and the right derivatives, respectively. 
Thus, the derivative of $f_\epsilon(w)$ at $\pm \epsilon$ is zero. 
Define the {\it discard vector function} 
$\fbf_\epsilon:\mathbb{R}^{N+1}\rightarrow\mathbb{R}^{N+1}$ as \symbl{$\fbf_\epsilon(\cdot)$}{Discard vector function}
\begin{align}
\fbf_\epsilon(\wbf)=[f_\epsilon(w_0)~\cdots~f_\epsilon(w_N)]^T. \label{eq:discard vector function}
\end{align}

\begin{figure}[t!]
\centering
\includegraphics[width=0.5\linewidth]{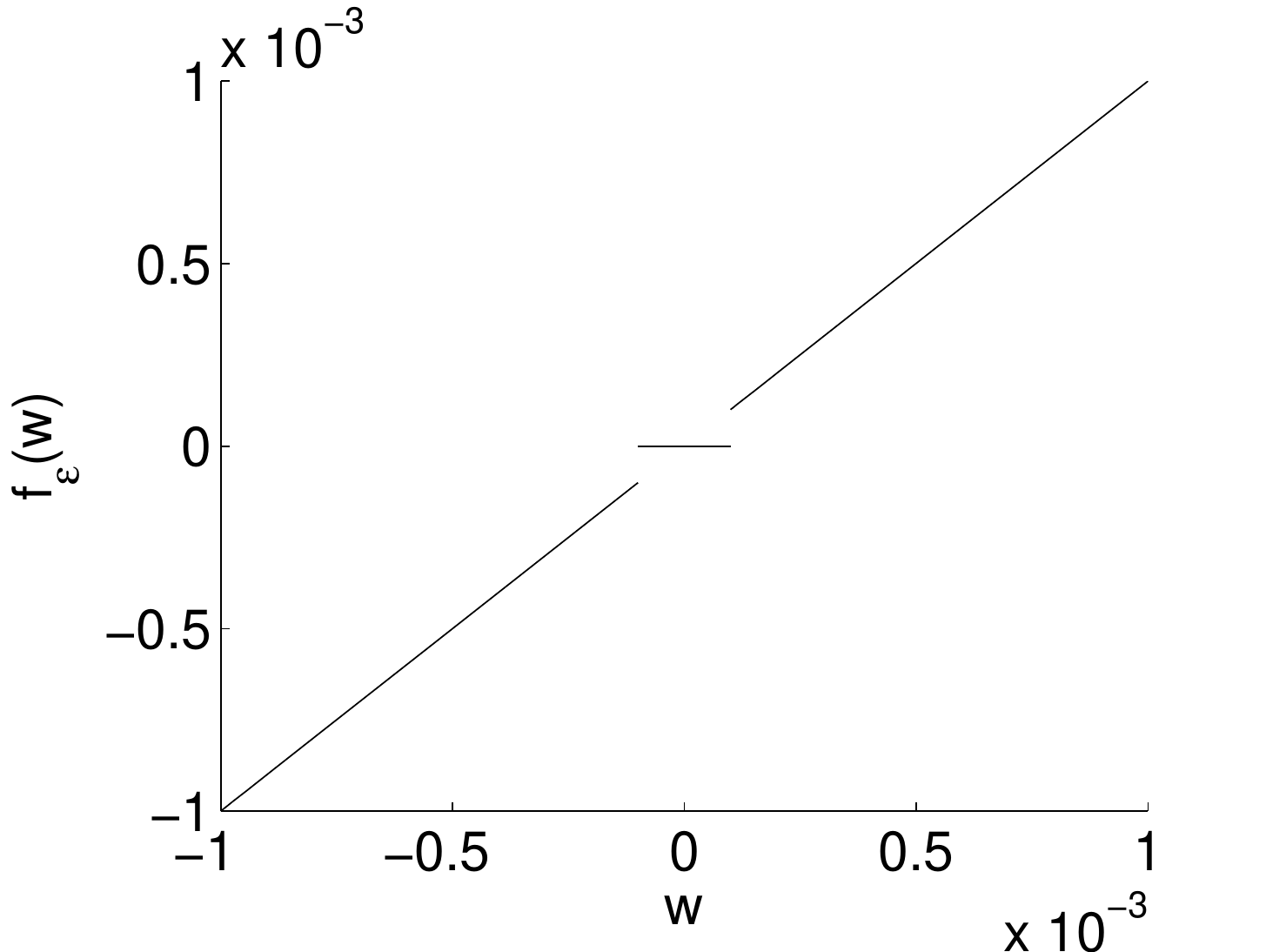}
\caption{Discard function $f_\epsilon(w)$ for $\epsilon=10^{-4}$.\label{fig:f-a-eusipco}}
\end{figure}

The {S-SM-AP}\abbrev{S-SM-AP}{Simple SM-AP} algorithm updates the coefficients whose absolute values are larger than 
$\epsilon$ whenever the error is such that $| e(k) | = |d(k) - \wbf^T(k) \xbf(k)|>\gammabar$. 
Let $\psi^{L+1}(k)$ denote the intersection of the last $L+1$ constraint sets and state the following optimization criterion for the vector update whenever $\wbf(k)\not\in\psi^{L+1}(k)$
\begin{align}
&\min \frac{1}{2}\|\fbf_\epsilon(\wbf(k+1))-\wbf(k)\|^2 \nonumber\\
&{\rm subject~to}\nonumber\\
&\dbf(k)-\Xbf^T(k)\wbf(k+1)=\gammabf(k).\label{eq:ssm_optimization-eusipco}
\end{align}
In order to solve this optimization problem, we construct the Lagrangian $\mathbb{L}$ as
\begin{align}
\mathbb{L}=\frac{1}{2}\|\fbf_\epsilon(\wbf(k+1))-\wbf(k)\|^2
+\lambdabf^T(k)[\dbf(k)-\Xbf^T(k)\wbf(k+1)-\gammabf(k)],
\end{align}
where $\lambdabf(k)\in\mathbb{R}^{L+1}$ is a vector of Lagrange multipliers. After differentiating the above equation with respect to $\wbf(k+1)$ and setting the result 
equal to zero, we obtain
\begin{align}
\fbf_\epsilon(\wbf(k+1))=\wbf(k)+\Fbf_\epsilon^{-1}(\wbf(k+1))\Xbf(k)\lambdabf(k),\label{eq:F_a(w(k+1))-eusipco}
\end{align}
where $\Fbf_\epsilon(\wbf(k+1))$ is the Jacobian matrix of $\fbf_\epsilon(\wbf(k+1))$. \symbl{$\Fbf_\epsilon(\wbf)$}{The Jacobian matrix of $\fbf_\epsilon(\wbf)$} In Equation (\ref{eq:F_a(w(k+1))-eusipco}), by employing a similar strategy as the PASTd\abbrev{PASTd}{Projection Approximation Subspace Tracking with Deflation} (projection approximation subspace tracking with deflation)~\cite{Wang_WirelessCommunicationSystems_book2004}, we replace $\fbf_\epsilon(\wbf(k+1))$ and $\Fbf_\epsilon^{-1}(\wbf(k+1))$ with $\wbf(k+1)$ and $\Fbf_\epsilon^{-1}(\wbf(k))$, respectively, in order to form the recursion, then we obtain
\begin{align}
\wbf(k+1)=\wbf(k)+\Fbf_\epsilon^{-1}(\wbf(k))\Xbf(k)\lambdabf(k).\label{eq:w(k+1)-ssm-eusipco}
\end{align}
If we substitute the above equation in the constraint relation (\ref{eq:ssm_optimization-eusipco}), then we will find $\lambdabf(k)$ as follows
\begin{align}
\lambdabf(k)=(\Xbf^T(k)\Fbf_\epsilon^{-1}(\wbf(k))\Xbf(k))^{-1}(\ebf(k)-\gammabf(k)).\label{eq:lambda-ssm-eusipco}
\end{align}
Replacing (\ref{eq:lambda-ssm-eusipco}) into (\ref{eq:w(k+1)-ssm-eusipco}) leads to the following updating equation 
\begin{align}
\wbf(k+1)&=\wbf(k)\nonumber\\
&+\Fbf_\epsilon^{-1}(\wbf(k))\Xbf(k)(\Xbf^T(k)\Fbf_\epsilon^{-1}(\wbf(k))\Xbf(k))^{-1}(\ebf(k)-\gammabf(k)).
\end{align}
Note that $\Fbf_\epsilon(\wbf(k))$ is not an invertible matrix and, therefore, 
we apply the Moore-Penrose pseudoinverse (generalization of the inverse matrix) 
instead of the standard inverse. However, $\Fbf_\epsilon(\wbf(k))$ is a diagonal matrix with diagonal entries 
equal to zero or one. 
Indeed, for the components of $\wbf(k)$ whose absolute values are larger 
than $\epsilon$, their corresponding entries on the diagonal matrix $\Fbf_\epsilon(\wbf(k))$ 
are equal to one, whereas the remaining entries are zero. 
Hence, the pseudoinverse of $\Fbf_\epsilon(\wbf(k))$ is again $\Fbf_\epsilon(\wbf(k))$. 
As a result, the update equation of the {S-SM-AP}\abbrev{S-SM-AP}{Simple SM-AP} algorithm is as follows
\begin{align}
\wbf(k+1)=\left\{\begin{array}{ll}\wbf(k)+\qbf(k)&\text{if }|e(k)|>\gammabar,\\\wbf(k)&\text{otherwise},\end{array}\right. \label{eq:update_equation-eusipco}
\end{align} 
where
\begin{align}
\qbf(k)=\Fbf_\epsilon(\wbf(k))\Xbf(k)[\Xbf^T(k)\Fbf_\epsilon(\wbf(k))\Xbf(k)+\delta\Ibf]^{-1}(\ebf(k)-\gammabf(k)). \label{eq:q(k)-eusipco}
\end{align}
Note that, we applied a regularization factor $\delta\Ibf$ in (\ref{eq:q(k)-eusipco}) in order to avoid numerical problems in the matrix inversion. The S-SM-AP algorithm is described in Table~\ref{tb:S-SM-AP-chap6}.

\begin{table}[t!]
\caption{Simple set-membership affine projection algorithm (S-SM-AP)}
\begin{center}
\begin{footnotesize}
\begin {tabular}{|l|} \hline\\ \hspace{3.4cm}{\bf S-SM-AP Algorithm}\\
\\
\hline\\
Initialization
\\
$\wbf(0)=[1~1~\cdots~1]^T$\\
choose $\gammabar$ around $\sqrt{5\sigma_n^2}$ and small constant $\delta>0$\\
Do for $k>0$\\
\hspace*{0.3cm} $\ebf(k)=\dbf(k)-\Xbf^T(k)\wbf(k)$\\
\hspace*{0.3cm} if $|e(k)|>\gammabar$\\
\hspace*{0.45cm} $\qbf(k)=\Fbf_\epsilon(\wbf(k))\Xbf(k)[\Xbf^T(k)\Fbf_\epsilon(\wbf(k))\Xbf(k)+\delta\Ibf]^{-1}(\ebf(k)-\gammabf(k))$\\
\hspace*{0.45cm} $\wbf(k+1)= \wbf(k)+\qbf(k)$\\
\hspace*{0.3cm} else\\
\hspace*{0.45cm} $\wbf(k+1)= \wbf(k)$\\
\hspace*{0.3cm} end\\
end\\
 \\
\hline
\end {tabular}
\end{footnotesize}
\end{center}
\label{tb:S-SM-AP-chap6}
\end{table}


\subsection{Discussion of the {S-SM-AP} algorithm \label{sub:discussion-eusipco}}

\subsubsection{Computational Complexity}

The update equation of the {S-SM-AP}\abbrev{S-SM-AP}{Simple SM-AP} algorithm is similar to the update equation of 
the SM-AP\abbrev{SM-AP}{Set-Membership Affine Projection} algorithm, but the former one updates only the subset of coefficients of $\wbf(k)$ 
whose absolute values are larger than $\epsilon$. As a result, the role of matrix $\Fbf_\epsilon(\wbf(k))$ is to discard some coefficients 
of $\wbf(k)$, thus reducing the computational complexity when compared to the SM-AP\abbrev{SM-AP}{Set-Membership Affine Projection} algorithm. 

The computational complexity for each update of the weight vector of the 
SM-PAPA~\cite{Diniz_sm_pap_jasmp2007}\abbrev{SM-PAPA}{Set-Membership Proportionate AP Algorithm}, 
the SSM-AP~\cite{Markus_sparseSMAP_tsp2014}\abbrev{SSM-AP}{Sparsity-Aware SM-AP}, 
and the proposed {S-SM-AP}\abbrev{S-SM-AP}{Simple SM-AP} algorithms are listed in Table~\ref{tab1-eusipco}. The filter order and the memory length factors are $N$ and $L$, respectively. 
It should be noted that the number of operations in Table~\ref{tab1-eusipco} is presented for 
the full update of all coefficients. 
In other words, for the {S-SM-AP}\abbrev{S-SM-AP}{Simple SM-AP} algorithm we have presented the worst case scenario 
which is equivalent to setting $\epsilon=0$,\footnote{In this case, the complexity of the 
{S-SM-AP}\abbrev{S-SM-AP}{Simple SM-AP} and SM-AP\abbrev{SM-AP}{Set-Membership Affine Projection} algorithms are the same.} 
while in practice we are updating only the coefficients with absolute values larger than 
a predetermined positive constant. Also, it is notable that the number of divisions in 
the {S-SM-AP}\abbrev{S-SM-AP}{Simple SM-AP} algorithm is less than the SM-PAPA\abbrev{SM-PAPA}{Set-Membership Proportionate AP Algorithm} and SSM-AP\abbrev{SSM-AP}{Sparsity-Aware SM-AP} algorithms. This is quite significant, as divisions are more complex than 
other operations. Figures \ref{fig:Complexity_L_eusipco} and \ref{fig:Complexity_N_eusipco} show a comparison of the total number of arithmetic operations required by the SM-PAPA\abbrev{SM-PAPA}{Set-Membership Proportionate AP Algorithm}, the SSM-AP\abbrev{SSM-AP}{Sparsity-Aware SM-AP}, and the S-SM-AP\abbrev{S-SM-AP}{Simple SM-AP} algorithms for two cases: $N=15$, variable $L$ and $L=3$, variable $N$. As can be seen, the S-SM-AP\abbrev{S-SM-AP}{Simple SM-AP} algorithm is much less complex than the other two algorithms, especially for high values of $N$ and $L$.

\begin{figure}[t!]
\centering
\subfigure[b][]{\includegraphics[width=.48\linewidth,height=7cm]{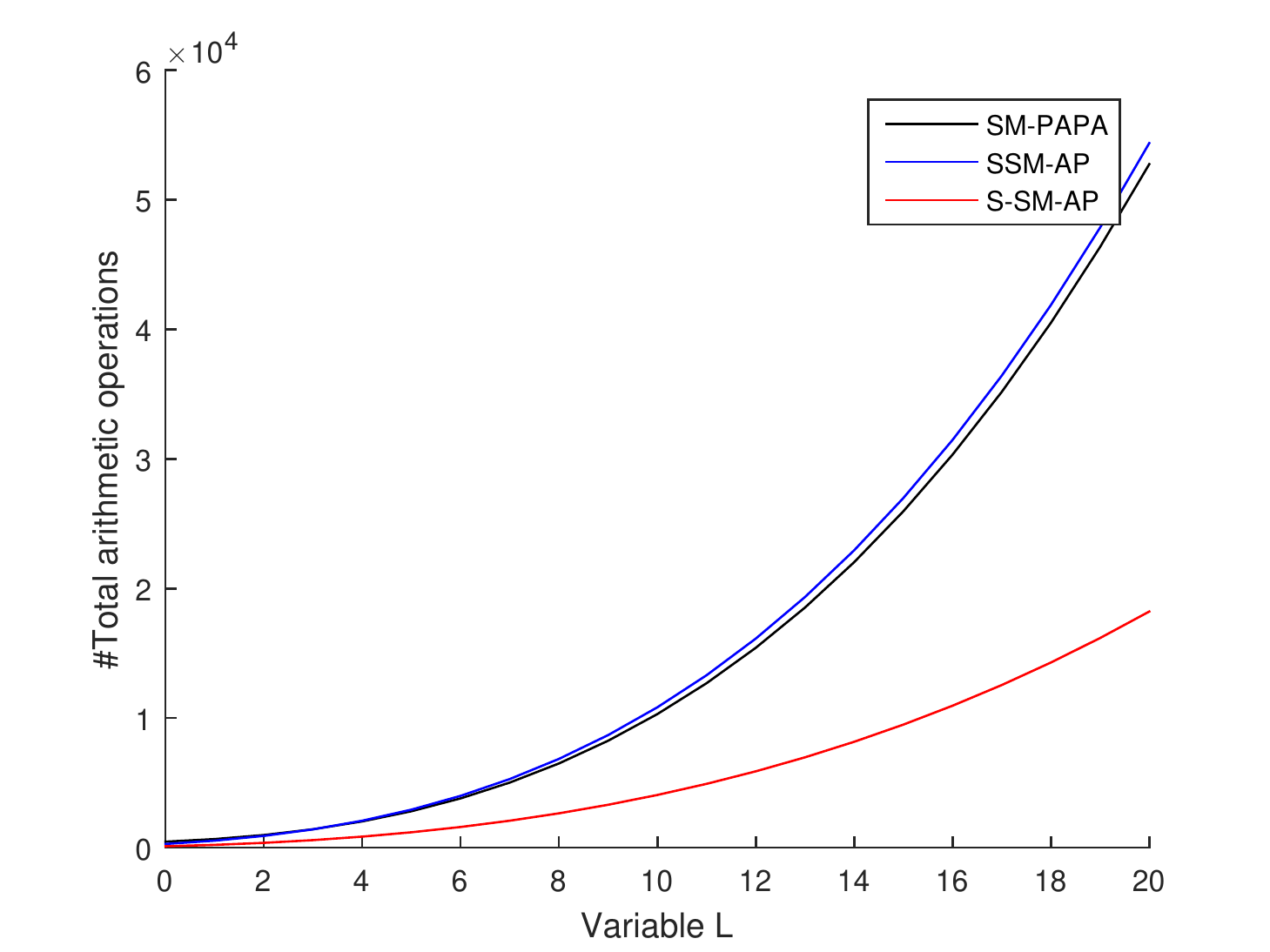}
\label{fig:Complexity_L_eusipco}}
\subfigure[b][]{\includegraphics[width=.48\linewidth,height=7cm]{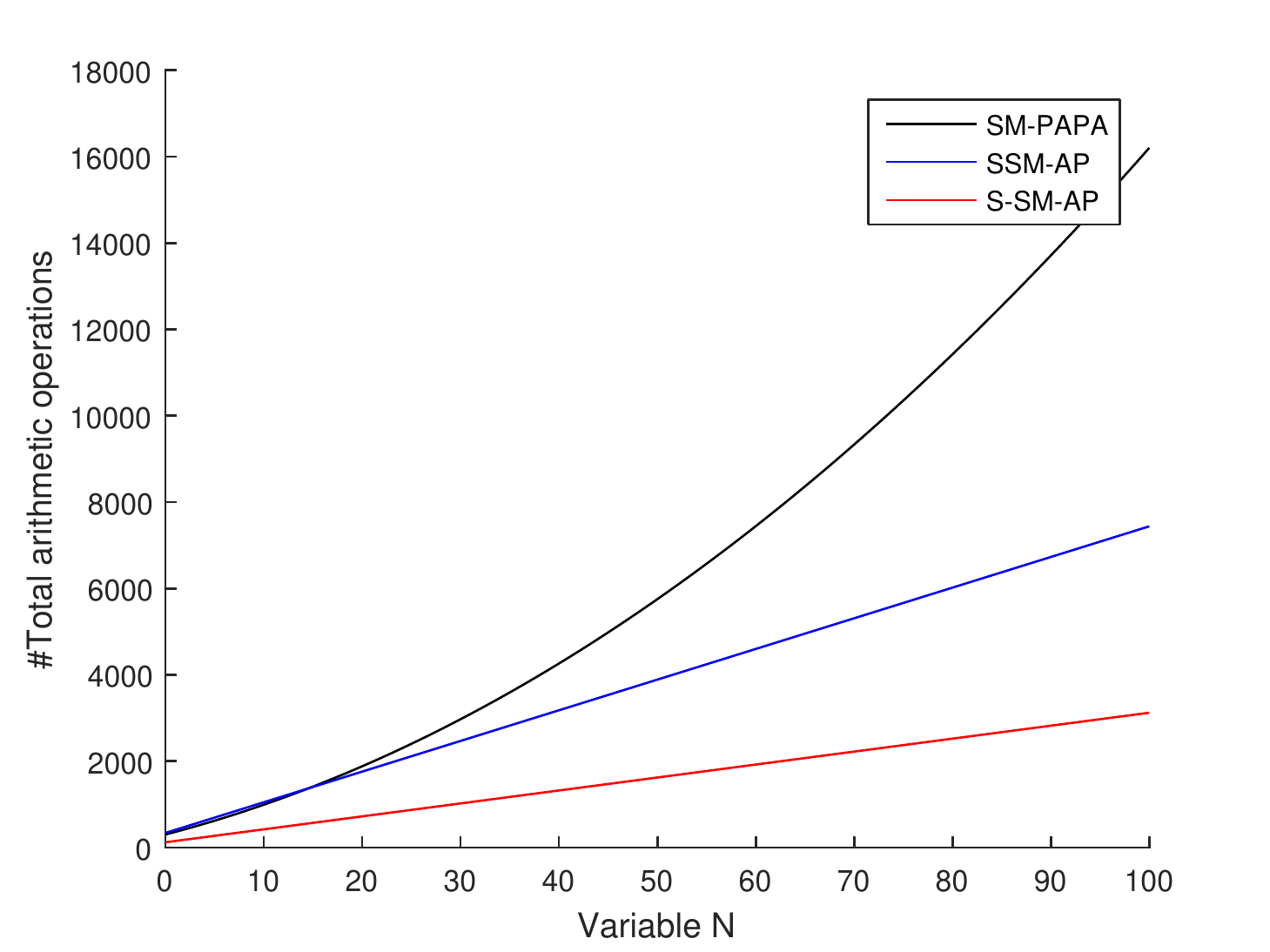}
\label{fig:Complexity_N_eusipco}}
\caption{The numerical complexity of the SM-PAPA, the SSM-AP, and the  IS-SM-AP algorithms for two cases: (a) $N=15$, variable $L$; (b) $L=3$, variable $N$.  \label{fig:Complexity-eusipco}}
\end{figure}

\begin{table*}[t]
\caption{Number of operations for SM-PAPA, SSM-AP, and {S-SM-AP} algorithms \label{tab1-eusipco}}
\begin{center}
\begin{tabular}{|l|c|c|c|} \hline
Algorithm & Addition $\&$ Subtraction & Multiplication & Division\\\hline
\parbox[t]{0mm}{\multirow{2}{*}{SM-PAPA}} & $N^2+(L^2+4L+5)N+$ & $(L^2+5L+7)N+$ & $2N+$ \\ & $(2L^3+5L^2+7L+5)$ & $(2L^3+6L^2+9L+8)$ & $(2L^2+4L+4)$ \\\hline
\parbox[t]{0mm}{\multirow{2}{*}{SSM-AP}} & $(L^2+6L+7)N+$ & $(L^2+6L+9)N+$ & $N+$ \\
 & $(2L^3+6L^2+9L+7)$ & $(2L^3+7L^2+12L+11)$ & $(2L^2+4L+3)$ \\\hline
 \parbox[t]{0mm}{\multirow{2}{*}{{S-SM-AP}}} & $\frac{1}{2}(L^2+5L+6)N$ & $\frac{1}{2}(L^2+5L+6)N$ & \parbox[t]{4mm}{\multirow{2}{*}{$L^2$}} \\ & $\frac{1}{2}(L^3+4L^2+11L+8)$ & $\frac{1}{2}(L^3+6L^2+11L+8)$ & \\\hline
\end{tabular}
\end{center}
\end{table*}

\subsubsection{Initialization}

Unlike classical algorithms in which the initialization of the weight vector is often chosen as 
$\wbf(0) = {\bf 0}$, this same procedure cannot be applied to the proposed algorithm. If the initial coefficients have absolute values lower than $\epsilon$, then the matrix $\Fbf_\epsilon$ is equal to the zero matrix, and it does not allow any update.
Indeed, for the {S-SM-AP}\abbrev{S-SM-AP}{Simple SM-AP} algorithm, each of the coefficients should be initialized as 
$|w_i (0)| > \epsilon$ for $i = 0, 1, \cdots, N$.

\subsubsection{Relation with other algorithms}

The similarities and differences between the proposed algorithm and the SM-AP\abbrev{SM-AP}{Set-Membership Affine Projection} algorithm 
were already addressed when we discussed the complexity of these algorithms. 
Now, one should observe that the update equation of the {S-SM-AP}\abbrev{S-SM-AP}{Simple SM-AP} algorithm is similar to 
the one of the set-membership partial update affine projection (SM-PUAP) 
algorithm \cite{Diniz_adaptiveFiltering_book2013}, in which our matrix $\Fbf_\epsilon(\wbf(k))$ is replaced by 
a diagonal matrix $\Cbf$ also with entries equal to 1 or 0, but there is no specific 
form to set/select $\Cbf$. Therefore, the proposed algorithm can be considered as a particular case of the SM-PUAP in 
which there is a mathematically defined way (based on the sparsity of the unknown system) 
to select the coefficients that are relevant and the ones that will be discarded. {Regarding the memory 
requirements of the proposed algorithm, they are the same as in the AP\abbrev{AP}{Affine Projection} algorithm, i.e., determined by the data-reuse factor $L$.}


\subsection{The Improved {S-SM-AP} (IS-SM-AP) algorithm} \label{sub:Modified-eusipco}

As we can observe in the update equation of the S-SM-AP\abbrev{S-SM-AP}{Simple SM-AP} algorithm, 
if a coefficient of the weight vector falls inside the interval 
$[-\epsilon,+\epsilon]$, then in the next update this coefficient 
does not update since it is eliminated by the discard function. On the other hand, the coefficients $w_i(k)$ inside the interval $[-\epsilon,+\epsilon]$ 
are close to zero, and the best intuitive approximation for them is zero (the center of the interval).
Besides, making these coefficients $w_i(k)$ equal to zero implies in a reduction of 
computational complexity, because it reduces the number of operations required to 
compute the output of the adaptive filter $y(k) = \xbf^T(k) \wbf(k)$.\footnote{This 
additional reduction in the number of operations becomes more important as the filter 
order increases. For instance, in acoustic echo cancellation systems, in which the 
adaptive filter has a few thousands of coefficients~\cite{Hansler_echo_book2004,Benesty_echo_book2010}, 
this simple strategy implies in significant computational savings.} For this purpose, we multiply $\wbf(k)$ by $\Fbf_\epsilon(\wbf(k))$, and 
obtain the Improved S-SM-AP (IS-SM-AP)\abbrev{IS-SM-AP}{Improved S-SM-AP} algorithm as follows
\begin{align}
\wbf(k+1)=\left\{\begin{array}{ll}\Fbf_\epsilon(\wbf(k))\wbf(k)+\qbf(k)&\text{if }|e(k)|>\gammabar,\\\wbf(k)&\text{otherwise}.\end{array}\right. \label{eq:is-sm-ap-update_equation-eusipco}
\end{align}
Table~\ref{tb:IS-SM-AP-chap6} illustrates the IS-SM-AP algorithm.

\begin{table}[t!]
\caption{Improved simple set-membership affine projection algorithm (IS-SM-AP)}
\begin{center}
\begin{footnotesize}
\begin {tabular}{|l|} \hline\\ \hspace{3.4cm}{\bf IS-SM-AP Algorithm}\\
\\
\hline\\
Initialization
\\
$\wbf(0)=[1~1~\cdots~1]^T$\\
choose $\gammabar$ around $\sqrt{5\sigma_n^2}$ and small constant $\delta>0$\\
Do for $k>0$\\
\hspace*{0.3cm} $\ebf(k)=\dbf(k)-\Xbf^T(k)\wbf(k)$\\
\hspace*{0.3cm} if $|e(k)|>\gammabar$\\
\hspace*{0.45cm} $\qbf(k)=\Fbf_\epsilon(\wbf(k))\Xbf(k)[\Xbf^T(k)\Fbf_\epsilon(\wbf(k))\Xbf(k)+\delta\Ibf]^{-1}(\ebf(k)-\gammabf(k))$\\
\hspace*{0.45cm} $\wbf(k+1)= \Fbf_\epsilon(\wbf(k))\wbf(k)+\qbf(k)$\\
\hspace*{0.3cm} else\\
\hspace*{0.45cm} $\wbf(k+1)= \wbf(k)$\\
\hspace*{0.3cm} end\\
end\\
 \\
\hline
\end {tabular}
\end{footnotesize}
\end{center}
\label{tb:IS-SM-AP-chap6}
\end{table} 


\subsection{The S-AP and the IS-AP algorithms} \label{sub:non-SM-eusipco}

By adopting the bound $\gammabar=0$, the S-SM-AP\abbrev{S-SM-AP}{Simple SM-AP} algorithm will convert to the S-AP\abbrev{S-AP}{Simple AP} algorithm with unity step size. Therefore, the S-AP\abbrev{S-AP}{Simple AP} algorithm can be described as follows
\begin{align}
\wbf(k+1)=\wbf(k)+\mu\Fbf_\epsilon(\wbf(k))\Xbf(k)[\Xbf^T(k)\Fbf_\epsilon(\wbf(k))\Xbf(k)+\delta\Ibf]^{-1}\ebf(k) \label{eq:update-s-ap-eusipco}
\end{align}
where $\mu$ is the convergence factor.

By the same argument, we can obtain the update equation of the IS-AP\abbrev{IS-AP}{Improved S-AP} algorithm as below
\begin{align}
\wbf(k+1)=&\Fbf_\epsilon(\wbf(k))\wbf(k)\nonumber\\
&+\mu\Fbf_\epsilon(\wbf(k))\Xbf(k)[\Xbf^T(k)\Fbf_\epsilon(\wbf(k))\Xbf(k)+\delta\Ibf]^{-1}\ebf(k) \label{eq:update-is-ap-eusipco}
\end{align}
where $\mu$ is the convergence factor. These algorithms are counterparts of the AP\abbrev{AP}{Affine Projection} algorithm, however they can exploit the sparsity in systems.

{\it Remark}: In the previous sections, we have focused on the AP\abbrev{AP}{Affine Projection} algorithms. However, the NLMS and the binormalized data-reusing  LMS\abbrev{LMS}{Least-Mean-Square} algorithms can be derived as special cases of the AP\abbrev{AP}{Affine Projection} algorithms. Indeed, by choosing $L=0$ and $1$, the AP\abbrev{AP}{Affine Projection} algorithms will be reduced to the NLMS and the binormalized data-reusing LMS\abbrev{LMS}{Least-Mean-Square} algorithms, respectively.


\section{Some issues of the S-SM-AP and the \\IS-SM-AP Algorithms} \label{sec:difficulties-is-sm-ap-chap6}

As we discussed in Subsection~\ref{sub:discussion-eusipco}, the proposed S-SM-AP\abbrev{S-SM-AP}{Simple SM-AP} and the IS-SM-AP\abbrev{IS-SM-AP}{Improved S-SM-AP} algorithms are sensitive to the initialization. In fact, the absolute value of parameters of $\wbf(0)$ have to be greater than $\epsilon$ and $w_i(0)w_{oi}>0$ for $i=0,\cdots,N$, i.e., $w_i(0)$ and $w_{oi}$ must have the same sign, where $w_{oi}$ is the $i$-th component of the unknown system. Moreover, when the system is time-varying, these algorithms cannot track the system. In other words, if a coefficient falls inside $[-\epsilon,\epsilon]$, then it cannot go out. Thus, in the case of time-varying systems, it means that the algorithm is unable to track the system.

To address this issue, we can use an auxiliary weight vector $\mbf(k)$ as in~\cite{Hu_shrink_sparse_icassp2014}. Through this technique, the discard function applies only to the auxiliary weight vector, and we can propose the discard SM-AP (D-SM-AP)\abbrev{D-SM-AP}{Discard SM-AP} algorithm. The D-SM-AP\abbrev{D-SM-AP}{Discard SM-AP} algorithm is presented in Table~\ref{tb:D-SM-AP-chap6}. Note that the computational burden of the D-SM-AP\abbrev{D-SM-AP}{Discard SM-AP} algorithm is higher than the IS-SM-AP\abbrev{IS-SM-AP}{Improved S-SM-AP} and the S-SM-AP\abbrev{S-SM-AP}{Simple SM-AP} algorithms. However, it can be utilized in time-varying systems, and we can adopt any initialization $\wbf(0)$.

\begin{table}[t!]
\caption{Discard set-membership affine projection algorithm (D-SM-AP)}
\begin{center}
\begin{footnotesize}
\begin {tabular}{|l|} \hline\\ \hspace{4.4cm}{\bf D-SM-AP Algorithm}\\
\\
\hline\\
Initialization
\\
$\wbf(0)={\bf 0}$ and $\mbf(0)={\bf 0}$\\
choose $\gammabar$ around $\sqrt{5\sigma_n^2}$ and small constant $\delta>0$\\
Do for $k>0$\\
\hspace*{0.3cm} $\ebf(k)=\dbf(k)-\Xbf^T(k)\wbf(k)$\\
\hspace*{0.3cm} $\mbf(k+1)= \left\{\begin{array}{ll}\mbf(k)+\Xbf(k)[\Xbf^T(k)\Xbf(k)+\delta\Ibf]^{-1}(\ebf(k)-\gammabf(k))&{\rm if~} |e(k)|>\gammabar\\\mbf(k)&{\rm otherwise}\end{array}\right.$\\
\hspace*{0.3cm} $\wbf(k+1)=\Fbf_{\epsilon}(\mbf(k+1))\mbf(k+1)$\\
end\\
 \\
\hline
\end {tabular}
\end{footnotesize}
\end{center}
\label{tb:D-SM-AP-chap6}
\end{table}


\section{Recursive Least-Squares Algorithm Exploiting Sparsity} \label{sec:s-rls-sparse}

In this section, we utilize the discard function to introduce an RLS\abbrev{RLS}{Recursive Least-Squares} algorithm for sparse systems. In Subsection~\ref{sub:derivation-s-rls-sparse}, we derive the S-RLS\abbrev{S-RLS}{RLS Algorithm for Sparse System} algorithm that exploits the sparsity of the estimated parameters by giving low weight to the small coefficients. 
For this purpose, the strategy consists in multiplying the coefficients of the sparse filter which 
are close to zero by a small constant.  
Then, in Subsection~\ref{sub:discussion-s-rls-sparse}, we include a discussion of some 
characteristics of the proposed algorithm. Subsection~\ref{sub:ds-s-rls-sparse} briefly describes the DS-S-RLS\abbrev{DS-S-RLS}{Data-Selective S-RLS} algorithm, the data-selective version of the S-RLS algorithm.


\subsection{Derivation of the S-RLS algorithm} \label{sub:derivation-s-rls-sparse}

We utilize the discard vector function defined in Equation~\eqref{eq:discard vector function} in order to introduce the objective function of the S-RLS\abbrev{S-RLS}{RLS Algorithm for Sparse System} algorithm as follows
\begin{align}
&\min \xi^d(k)=\sum_{i=0}^k\lambda^{k-i}[d(i)-\xbf^T(i)\fbf_\epsilon(\wbf(k))]^2,\label{eq:ssm_optimization-sparse}
\end{align}
where the parameter $\lambda$ is an exponential weighting factor that should be selected in the range $0\ll\lambda\leq1$.

By differentiating $\xi^d(k)$ with respect to $\wbf(k)$, we obtain
\begin{align}
\frac{\partial\xi^d(k)}{\partial\wbf(k)}=-2\sum_{i=0}^k\lambda^{k-i}\Fbf_\epsilon(\wbf(k))\xbf(i)[d(i)-\xbf^T(i)\fbf_\epsilon(\wbf(k))],
\end{align}
where $\Fbf_\epsilon(\wbf(k))$ is the Jacobian matrix of  $\fbf_\epsilon(\wbf(k))$ (see~\eqref{eq:discard vector function}). By equating the above equation to zero, we find the optimal vector $\wbf(k)$ that solves the least-square problem, as follows
\begin{align}
-\sum_{i=0}^k\lambda^{k-i}\Fbf_\epsilon(\wbf(k))\xbf(i)\xbf^T(i)\fbf_\epsilon(\wbf(k))+\sum_{i=0}^k\lambda^{k-i}\Fbf_\epsilon(\wbf(k))\xbf(i)d(i)=\left[\begin{array}{c}0\\\vdots\\0\end{array}\right].
\end{align}
Therefore,
\begin{align}
\fbf_\epsilon(\wbf(k))=&\Big[\sum_{i=0}^k\lambda^{k-i}\Fbf_\epsilon(\wbf(k))\xbf(i)\xbf^T(i)\Big]^{-1}\times\sum_{i=0}^k\lambda^{k-i}\Fbf_\epsilon(\wbf(k))\xbf(i)d(i). \label{eq:before_double_F-sparse}
\end{align}
Note that $\Fbf_\epsilon(\wbf(k))$ is a diagonal matrix with diagonal entries equal to zero or one. 
Indeed, for the components of $\wbf(k)$ whose absolute values are larger than $\epsilon$, their corresponding entries on the diagonal matrix $\Fbf_\epsilon(\wbf(k))$ are one, whereas the remaining entries are zero. Hence,
\begin{align}
\Fbf_\epsilon(\wbf(k))\xbf(i)\xbf^T(i)&=\Fbf_\epsilon^2(\wbf(k))\xbf(i)\xbf^T(i)=\Fbf_\epsilon(\wbf(k))(\xbf^T(i)\Fbf_\epsilon(\wbf(k)))^T\xbf^T(i)\nonumber\\
&=\Fbf_\epsilon(\wbf(k))\xbf(i)\xbf^T(i)\Fbf_\epsilon(\wbf(k)). \label{eq:F^2=F_indentity-sparse}
\end{align}
By utilizing (\ref{eq:F^2=F_indentity-sparse}) in (\ref{eq:before_double_F-sparse}) and replacing $\fbf_\epsilon(\wbf(k))$ by $\wbf(k+1)$, we get
\begin{align}
\wbf(k+1)&=\Big[\sum_{i=0}^k\lambda^{k-i}\Fbf_\epsilon(\wbf(k))\xbf(i)\xbf^T(i)\Fbf_\epsilon(\wbf(k))\Big]^{-1}\times\sum_{i=0}^k\lambda^{k-i}\Fbf_\epsilon(\wbf(k))\xbf(i)d(i)\nonumber\\
&=\Rbf_{D,\epsilon}^{-1}(k)\pbf_{D,\epsilon}(k), \label{eq:original_version-s-rls}
\end{align}
where $\Rbf_{D,\epsilon}(k)$ and $\pbf_{D,\epsilon}(k)$ are called the deterministic correlation matrix of the input signal and the
deterministic cross-correlation vector between the input and the desired signals, respectively. \symbl{$\Rbf_{D,\epsilon}(k)$}{The deterministic correlation matrix of the input signal involved $\Fbf_\epsilon(\wbf(k))$} \symbl{$\pbf_{D,\epsilon}(k)$}{The deterministic cross-correlation vector between the input and the desired signals involved $\Fbf_\epsilon(\wbf(k))$} Whenever the $i$-th diagonal entry of matrix $\Fbf_\epsilon(\wbf(k))$ is zero, it is replaced by a small power-of-two (e.g., $2^{-5}$) multiplied by the sign of the component $w_i(k)$ in order to avoid that matrix $\Rbf_{D,\epsilon}(k)$ becomes ill conditioned.

If we apply the direct method to calculate the inverse of $\Rbf_{D,\epsilon}(k)$, then the resulting algorithm has computational complexity of $O[N^3]$. Generally, in the traditional RLS\abbrev{RLS}{Recursive Least-Squares} algorithm, the inverse matrix is computed through the matrix inversion lemma~\cite{Goodwin_Dynamic_system_id_book1977}. In matrix inversion lemma, we have
\begin{align}
[\Abf+\Bbf\Cbf\Dbf]^{-1}=\Abf^{-1}-\Abf^{-1}\Bbf[\Dbf\Abf^{-1}\Bbf+\Cbf^{-1}]^{-1}\Dbf\Abf^{-1},
\end{align}
where $\Abf$, $\Bbf$, $\Cbf$, and $\Dbf$ are matrices of appropriate dimensions, and $\Abf$ and $\Cbf$ are invertible. If we choose $\Abf=\lambda\Rbf_{D,\epsilon}(k-1)$, $\Bbf=\Dbf^T=\Fbf_\epsilon(\wbf(k))\xbf(k)$, and $\Cbf=1$ then by
using the matrix inversion lemma, the inverse of the deterministic correlation matrix can be calculated in the form \symbl{$\Sbf_{D,\epsilon}(k)$}{The inverse of $\Rbf_{D,\epsilon}(k)$}
\begin{align}
\Sbf_{D,\epsilon}(k)=&\Rbf_{D,\epsilon}^{-1}(k)\nonumber\\=&\frac{1}{\lambda}\Big[\Sbf_{D,\epsilon}(k-1)
-\frac{\Sbf_{D,\epsilon}(k-1)\Fbf_\epsilon(\wbf(k))\xbf(k)\xbf^T(k)\Fbf_\epsilon(\wbf(k))\Sbf_{D,\epsilon}(k-1)}{\lambda+\xbf^T(k)\Fbf_\epsilon(\wbf(k))\Sbf_{D,\epsilon}(k-1)\Fbf_\epsilon(\wbf(k))\xbf(k)}\Big].\label{eq:S_D-sparse}
\end{align}
The resulting equation to compute $\Rbf_{D,\epsilon}^{-1}(k)$ has computational complexity of $O[N^2]$, whereas the computational resources for the direct inversion is of order $N^3$. Finally, 
\begin{align}
\wbf(k+1)=\Sbf_{D,\epsilon}(k)\pbf_{D,\epsilon}(k).
\end{align}
Table~\ref{tb:S-RLS-sparse} describes the S-RLS\abbrev{S-RLS}{RLS Algorithm for Sparse System} algorithm.

\begin{table}[t!]
\caption{Recursive least-squares algorithm for sparse systems (S-RLS)}
\begin{center}
\begin{footnotesize}
\begin {tabular}{|l|} \hline\\ \hspace{2.9cm}{\bf S-RLS Algorithm}\\
\\
\hline\\
Initialization
\\
$\Sbf_{D,\epsilon}(-1)=\delta\Ibf$\\
where $\delta$ can be the inverse of the input  signal power estimate\\
$\pbf_{D,\epsilon}(-1)=[0~0~\cdots~0]^T$\\
$\wbf(0)=[1~1~\cdots~1]^T$\\
Do for $k\geq0$\\
\hspace*{0.3cm} compute $\Sbf_{D,\epsilon}(k)$ through Equation (\ref{eq:S_D-sparse})\\
\hspace*{0.3cm} $\pbf_{D,\epsilon}(k)=\lambda \pbf_{D,\epsilon}(k-1)+\Fbf_\epsilon(\wbf(k))\xbf(k)d(k)$\\
\hspace*{0.3cm} $\wbf(k+1)=\Sbf_{D,\epsilon}(k)\pbf_{D,\epsilon}(k)$\\
end\\
 \\
\hline
\end {tabular}
\end{footnotesize}
\end{center}
\label{tb:S-RLS-sparse}
\end{table}

We can introduce the alternative S-RLS (AS-RLS)\abbrev{AS-RLS}{Alternative S-RLS} algorithm in order to decrease the computational load of the S-RLS\abbrev{S-RLS}{RLS Algorithm for Sparse System}. Assuming $\Fbf_\epsilon(\wbf(k)) \approx\Fbf_\epsilon(\wbf(k-1))$, we can rewrite Equation~\eqref{eq:original_version-s-rls} as
\begin{align}
&\Big[\sum_{i=0}^k\lambda^{k-i}\Fbf_\epsilon(\wbf(k))\xbf(i)\xbf^T(i)\Fbf_\epsilon(\wbf(k))\Big]\wbf(k+1)=\sum_{i=0}^k\lambda^{k-i}\Fbf_\epsilon(\wbf(k))\xbf(i)d(i)\nonumber\\
&=\lambda\Big[\sum_{i=0}^{k-1}\lambda^{k-i-1}\Fbf_\epsilon(\wbf(k))\xbf(i)d(i)\Big]+\Fbf_\epsilon(\wbf(k))\xbf(k)d(k)\nonumber\\
&\approx\lambda\Big[\sum_{i=0}^{k-1}\lambda^{k-i-1}\Fbf_\epsilon(\wbf(k-1))\xbf(i)d(i)\Big]+\Fbf_\epsilon(\wbf(k))\xbf(k)d(k)\nonumber\\
&=\lambda\pbf_{D,\epsilon}(k-1)+\Fbf_\epsilon(\wbf(k))\xbf(k)d(k)
\end{align}
By considering that $\Rbf_{D,\epsilon}(k-1)\wbf(k)=\pbf_{D,\epsilon}(k-1)$, we obtain
\begin{align}
&\Big[\sum_{i=0}^k\lambda^{k-i}\Fbf_\epsilon(\wbf(k))\xbf(i)\xbf^T(i)\Fbf_\epsilon(\wbf(k))\Big]\wbf(k+1)\nonumber\\
&\approx\lambda\Rbf_{D,\epsilon}(k-1)\wbf(k)+\Fbf_\epsilon(\wbf(k))\xbf(k)d(k)\nonumber\\
&=\Big[\sum_{i=0}^{k-1}\lambda^{k-i}\Fbf_\epsilon(\wbf(k-1))\xbf(i)\xbf^T(i)\Fbf_\epsilon(\wbf(k-1))\Big]\wbf(k)+\Fbf_\epsilon(\wbf(k))\xbf(k)d(k)\nonumber\\
&\approx\Big[\sum_{i=0}^{k}\lambda^{k-i}\Fbf_\epsilon(\wbf(k))\xbf(i)\xbf^T(i)\Fbf_\epsilon(\wbf(k))-\Fbf_\epsilon(\wbf(k))\xbf(k)\xbf^T(k)\Fbf_\epsilon(\wbf(k))\Big]\wbf(k)\nonumber\\
&+\Fbf_\epsilon(\wbf(k))\xbf(k)d(k).
\end{align}
Then, by using Equation~\eqref{eq:F^2=F_indentity-sparse} and a few manipulations, we get
\begin{align}
\wbf(k+1)\approx\wbf(k)+e(k)\Sbf_{D,\epsilon}(k)\Fbf_\epsilon(\wbf(k))\xbf(k),
\end{align}
where $e(k)=d(k)-\xbf^T(k)\wbf(k)$. Table~\ref{tb:alt-S-RLS-sparse} illustrates the AS-RLS\abbrev{AS-RLS}{Alternative S-RLS} algorithm.

\begin{table}[t!]
\caption{Alternative recursive least-squares algorithm for sparse systems}
\begin{center}
\begin{footnotesize}
\begin {tabular}{|l|} \hline\\ \hspace{2.8cm}{\bf AS-RLS Algorithm}\\
\\
\hline\\
Initialization
\\
$\Sbf_{D,\epsilon}(-1)=\delta\Ibf$\\
where $\delta$ can be inverse of the input  signal power estimate\\
$\wbf(0)=[1~1~\cdots~1]^T$\\
Do for $k\geq0$\\
\hspace*{0.3cm} $e(k)=d(k)-\xbf^T(k)\wbf(k)$\\
\hspace*{0.3cm} $\psi(k)=\Sbf_{D,\epsilon}(k-1)\Fbf_\epsilon(\wbf(k))\xbf(k)$\\
\hspace*{0.3cm} $\Sbf_{D,\epsilon}(k)=\frac{1}{\lambda}\Big[\Sbf_{D,\epsilon}(k-1)-\frac{\psi(k)\psi^T(k)}{\lambda+\psi^T(k)\Fbf_\epsilon(\wbf(k))\xbf(k)}\Big]$\\
\hspace*{0.3cm} $\wbf(k+1)=\wbf(k)+e(k)\Sbf_{D,\epsilon}(k)\Fbf_\epsilon(\wbf(k))\xbf(k)$\\
end\\
 \\
\hline
\end {tabular}
\end{footnotesize}
\end{center}
\label{tb:alt-S-RLS-sparse}
\end{table}


\subsection{Discussion of the {S-RLS} algorithm} \label{sub:discussion-s-rls-sparse}

The update equation of the {S-RLS}\abbrev{S-RLS}{RLS Algorithm for Sparse System} algorithm is similar to the update equation of the RLS\abbrev{RLS}{Recursive Least-Squares} algorithm, but the former gives importance only to the subset of coefficients of $\wbf(k)$ whose absolute values are larger than $\epsilon$.
The matrix $\Fbf_\epsilon(\wbf(k))$ defines the important coefficients of $\wbf(k)$.


\subsection{DS-S-RLS algorithm} \label{sub:ds-s-rls-sparse}

In this subsection, our goal is to reduce the update rate of the S-RLS\abbrev{S-RLS}{RLS Algorithm for Sparse System} algorithm. In fact, when the current weight vector is acceptable, i.e., the output estimation error is small, we can save computational resources by avoiding the new update. The data selective S-RLS (DS-S-RLS)\abbrev{DS-S-RLS}{Data-Selective S-RLS} algorithm updates whenever the output estimation error is larger than a prescribed value $\gammabar$, i.e., when $|e(k)|=|d(k)-\wbf^T(k)\xbf(k)|>\gammabar$. Therefore, the DS-S-RLS\abbrev{DS-S-RLS}{Data-Selective S-RLS} algorithm reduces the computational complexity by avoiding updates whenever the estimate is acceptable. Table~\ref{tb:DS-S-RLS-sparse} describes the DS-S-RLS algorithm.

\begin{table}[t!]
\caption{Data-selective recursive least-squares algorithm for sparse systems (DS-S-RLS)}
\begin{center}
\begin{footnotesize}
\begin {tabular}{|l|} \hline\\ \hspace{2.7cm}{\bf DS-S-RLS Algorithm}\\
\\
\hline\\
Initialization
\\
$\Sbf_{D,\epsilon}(-1)=\delta\Ibf$\\
where $\delta$ can be the inverse of the input  signal power estimate\\
choose $\gammabar$ around $\sqrt{5\sigma_n^2}$\\
$\pbf_{D,\epsilon}(-1)=[0~0~\cdots~0]^T$\\
$\wbf(0)=[1~1~\cdots~1]^T$\\
Do for $k\geq0$\\
\hspace*{0.3cm} $e(k)=d(k)-\wbf^T(k)\xbf(k)$\\
\hspace*{0.3cm} if $|e(k)|>\gammabar$\\
\hspace*{0.45cm} compute $\Sbf_{D,\epsilon}(k)$ through Equation (\ref{eq:S_D-sparse})\\
\hspace*{0.45cm} $\pbf_{D,\epsilon}(k)=\lambda \pbf_{D,\epsilon}(k-1)+\Fbf_\epsilon(\wbf(k))\xbf(k)d(k)$\\
\hspace*{0.45cm} $\wbf(k+1)=\Sbf_{D,\epsilon}(k)\pbf_{D,\epsilon}(k)$\\
\hspace*{0.3cm} else\\
\hspace*{0.45cm} $\wbf(k+1)=\wbf(k)$\\
\hspace*{0.3cm} end\\
end\\
 \\
\hline
\end {tabular}
\end{footnotesize}
\end{center}
\label{tb:DS-S-RLS-sparse}
\end{table}


\section{{$l_0$} Norm Recursive Least-Squares Algorithm} \label{sec:l0-rls-sparse}

In the previous section, we have introduced the S-RLS\abbrev{S-RLS}{RLS Algorithm for Sparse System} algorithm for sparse systems utilizing the discard function. Another interesting approach to exploit the system sparsity can be derived by using the $l_0$ norm~\cite{Markus_sparseSMAP_tsp2014} leading to the $l_0$-RLS\abbrev{$l_0$-RLS}{$l_0$ Norm RLS} algorithm. However, as mentioned earlier, the resulting optimization problem of $l_0$ norm has difficulties due to the discontinuity of the $l_0$ norm. Thus, we use Equations~\eqref{eq:mult_Laplace}-\eqref{eq:modf_mult_Geman} to approximate the $l_0$ norm.

Therefore, the objective function of the $l_0$-RLS\abbrev{$l_0$-RLS}{$l_0$ Norm RLS} algorithm is given by
\begin{align}
\min \sum_{i=0}^k\lambda^{k-i}[d(i)-\xbf^T(i)\wbf(k)]^2+\alpha\|\wbf(k)\|_0,
\end{align}
where $\alpha\in\mathbb{R}_+$ is the weight given to the $l_0$ norm penalty. Replacing $\|\wbf(k)\|_0$ by its approximation, we obtain
\begin{align}
\min \sum_{i=0}^k\lambda^{k-i}[d(i)-\xbf^T(i)\wbf(k)]^2+\alpha G_\beta(\wbf(k)).
\end{align}
By differentiating the above equation with respect to $\wbf(k)$, and equating the result to zero, we get
\begin{align}
\wbf(k)=&\Big[\sum_{i=0}^k\lambda^{k-i}\xbf(i)\xbf^T(i)\Big]^{-1}
\times\Big((\sum_{i=0}^k\lambda^{k-i}\xbf(i)d(i))-\frac{\alpha}{2}\gbf_\beta(\wbf(k))\Big)\nonumber\\
=&\Rbf_D^{-1}(k)\Big(\pbf_D(k)-\frac{\alpha}{2}\gbf_\beta(\wbf(k))\Big). \label{eq:pre_DS-l0-RLS-sparse}
\end{align}
If we adopt $\Abf=\lambda\Rbf_D(k-1)$, $\Bbf=\Dbf^T=\xbf(k)$, and $\Cbf=1$ then by using the matrix inversion lemma, the update equation of the $l_0$-RLS\abbrev{$l_0$-RLS}{$l_0$ Norm RLS} algorithm is given as follows
\begin{align}
\wbf(k)=\Sbf_D(k)\Big(\pbf_D(k)-\frac{\alpha}{2}\gbf_\beta(\wbf(k-1))\Big), \label{eq:DS-l0-RLS-sparse}
\end{align}
where the same strategy as the PASTd\abbrev{PASTd}{Projection Approximation Subspace Tracking with Deflation}
(projection approximation subspace tracking with deflation)~\cite{Wang_WirelessCommunicationSystems_book2004} is employed and $\gbf_\beta(\wbf(k))$ is replaced by $\gbf_\beta(\wbf(k-1))$ in order to form the recursion. Also, $\pbf_D(k)$ and $\Sbf_D(k)$ are given as follows
\begin{align}
\pbf_D(k)=&\lambda\pbf_D(k-1)+d(k)\xbf(k),\\
\Sbf_D(k)=&\frac{1}{\lambda}\Big[\Sbf_D(k-1)-\frac{\Sbf_D(k-1)\xbf(k)\xbf^T(k)\Sbf_D(k-1)}{\lambda+\xbf^T(k)\Sbf_D(k-1)\xbf(k)}\Big]. \label{eq:S_D-l0-sparse}
\end{align}
Table~\ref{tb:l0-RLS-sparse} presents the $l_0$-RLS\abbrev{$l_0$-RLS}{$l_0$ Norm RLS} algorithm.

\begin{table}[t!]
\caption{$l_0$ norm recursive least-squares algorithm for sparse systems ($l_0$-RLS)}
\begin{center}
\begin{footnotesize}
\begin {tabular}{|l|} \hline\\ \hspace{2.9cm}{\bf $l_0$-RLS Algorithm}\\
\\
\hline\\
Initialization
\\
$\Sbf_{D}(-1)=\delta\Ibf$\\
where $\delta$ can be inverse of the input  signal power estimate\\
$\pbf_{D}(-1)=[0~0~\cdots~0]^T$\\
$\wbf(-1)=[1~1~\cdots~1]^T$\\
Do for $k\geq0$\\
\hspace*{0.3cm} $\Sbf_{D}(k)$ as in Equation (\ref{eq:S_D-l0-sparse})\\
\hspace*{0.3cm} $\pbf_{D}(k)=\lambda \pbf_{D}(k-1)+d(k)\xbf(k)$\\
\hspace*{0.3cm} $\wbf(k)=\Sbf_D(k)\Big(\pbf_D(k)-\frac{\alpha}{2}\gbf_\beta(\wbf(k-1))\Big)$\\
end\\
 \\
\hline
\end {tabular}
\end{footnotesize}
\end{center}
\label{tb:l0-RLS-sparse}
\end{table}

Similarly to the AS-RLS\abbrev{AS-RLS}{Alternative S-RLS} algorithm, we can derive the alternative $l_0$-RLS (A-$l_0$-RLS)\abbrev{A-$l_0$-RLS}{Alternative $l_0$-RLS} algorithm. We can rewrite Equation~\eqref{eq:DS-l0-RLS-sparse} as
\begin{align}
\Rbf_D(k)\wbf(k)=\pbf_D(k)-\frac{\alpha}{2}\gbf_\beta(\wbf(k-1))=\lambda\pbf_D(k-1)+\xbf(k)d(k)-\frac{\alpha}{2}\gbf_\beta(\wbf(k-1)).
\end{align}

By Equation~\eqref{eq:pre_DS-l0-RLS-sparse}, we have $\Rbf_D(k-1)\wbf(k-1)=\pbf_D(k-1)-\frac{\alpha}{2}\gbf_\beta(\wbf(k-1))$, then we get
\begin{align}
\Rbf_D(k)\wbf(k)=&\lambda\Rbf_D(k-1)\wbf(k-1)+\frac{\lambda\alpha}{2}\gbf_\beta(\wbf(k-1))\nonumber\\
&-\frac{\alpha}{2}\gbf_\beta(\wbf(k-1))+\xbf(k)d(k)\nonumber\\
=&\Big[\sum_{j=0}^k\lambda^{k-i}\xbf(i)\xbf^T(i)-\xbf(k)\xbf^T(k)\Big]\wbf(k-1)\nonumber\\
&+\frac{(\lambda-1)\alpha}{2}\gbf_\beta(\wbf(k-1))+\xbf(k)d(k).
\end{align}
If we define the {\it a priori} error as
\begin{align}
e(k)=d(k)-\xbf^T(k)\wbf(k-1),
\end{align}
we obtain
\begin{align}
\Rbf_D(k)\wbf(k)=\Rbf_D(k)\wbf(k-1)+e(k)\xbf(k)+\frac{(\lambda-1)\alpha}{2}\gbf_\beta(\wbf(k-1)).
\end{align}
Therefore, the update equation of the A-$l_0$-RLS\abbrev{A-$l_0$-RLS}{Alternative $l_0$-RLS} algorithm is given by
\begin{align}
\wbf(k)=\wbf(k-1)+\Sbf_D(k)[e(k)\xbf(k)+\frac{(\lambda-1)\alpha}{2}\gbf_\beta(\wbf(k-1))].
\end{align}
Table~\ref{tb:alt-l0-RLS-sparse} presents the A-$l_0$-RLS\abbrev{A-$l_0$-RLS}{Alternative $l_0$-RLS} algorithm.

\begin{table}[t!]
\caption{Alternative $l_0$ norm recursive least-squares algorithm for sparse systems}
\begin{center}
\begin{footnotesize}
\begin {tabular}{|l|} \hline\\ \hspace{2.8cm}{\bf A-$l_0$-RLS Algorithm}\\
\\
\hline\\
Initialization
\\
$\Sbf_{D}(-1)=\delta\Ibf$\\
where $\delta$ can be inverse of the input  signal power estimate\\
$\wbf(-1)=[1~1~\cdots~1]^T$\\
Do for $k\geq0$\\
\hspace*{0.3cm} $e(k)=d(k)-\xbf^T(k)\wbf(k-1)$\\
\hspace*{0.3cm} $\psi(k)=\Sbf_{D}(k-1)\xbf(k)$\\
\hspace*{0.3cm} $\Sbf_{D}(k)=\frac{1}{\lambda}\Big[\Sbf_{D}(k-1)-\frac{\psi(k)\psi^T(k)}{\lambda+\psi^T(k)\xbf(k)}\Big]$\\
\hspace*{0.3cm} $\wbf(k)=\wbf(k-1)+\Sbf_D(k)[e(k)\xbf(k)+\frac{(\lambda-1)\alpha}{2}\gbf_\beta(\wbf(k-1))]$\\
end\\
 \\
\hline
\end {tabular}
\end{footnotesize}
\end{center}
\label{tb:alt-l0-RLS-sparse}
\end{table}

\subsection{DS-{$l_0$}-RLS algorithm} \label{sub:ds-l0-rls-sparse}

In this subsection, we propose the DS-$l_0$-RLS\abbrev{DS-$l_0$-RLS}{Data-Selective $l_0$-RLS} algorithm to decrease the update rate of the $l_0$-RLS\abbrev{$l_0$-RLS}{$l_0$ Norm RLS} algorithm. Similarly to the discussion in Subsection~\ref{sub:ds-s-rls-sparse}, the DS-$l_0$-RLS\abbrev{DS-$l_0$-RLS}{Data-Selective $l_0$-RLS} algorithm for sparse systems can be derived by implementing an update in the $l_0$-RLS\abbrev{$l_0$-RLS}{$l_0$ Norm RLS} algorithm whenever the output estimation error is larger than a predetermined value $\gammabar$, i.e., when $|e(k)|=|d(k)-\wbf^T(k)\xbf(k)|>\gammabar$. Hence, the computational resources of the DS-$l_0$-RLS\abbrev{DS-$l_0$-RLS}{Data-Selective $l_0$-RLS} algorithm is lower than the $l_0$-RLS\abbrev{$l_0$-RLS}{$l_0$ Norm RLS} algorithm since it prevents unnecessary updates. The DS-$l_0$-RLS algorithm is described in Table~\ref{tb:DS-l0-RLS-sparse}.

\begin{table}[t!]
\caption{Data-selective $l_0$ norm recursive least-squares algorithm for sparse systems (DS-$l_0$-RLS)}
\begin{center}
\begin{footnotesize}
\begin {tabular}{|l|} \hline\\ \hspace{2.7cm}{\bf DS-$l_0$-RLS Algorithm}\\
\\
\hline\\
Initialization
\\
$\Sbf_{D}(-1)=\delta\Ibf$\\
where $\delta$ can be inverse of the input  signal power estimate\\
choose $\gammabar$ around $\sqrt{5\sigma_n^2}$\\
$\pbf_{D}(-1)=[0~0~\cdots~0]^T$\\
$\wbf(-1)=[1~1~\cdots~1]^T$\\
Do for $k\geq0$\\
\hspace*{0.3cm} $e(k)=d(k)-\wbf^T(k-1)\xbf(k)$\\
\hspace*{0.3cm} if $|e(k)|>\gammabar$\\
\hspace*{0.45cm} $\Sbf_{D}(k)$ as in Equation (\ref{eq:S_D-l0-sparse})\\
\hspace*{0.45cm} $\pbf_{D}(k)=\lambda \pbf_{D}(k-1)+d(k)\xbf(k)$\\
\hspace*{0.45cm} $\wbf(k)=\Sbf_D(k)\Big(\pbf_D(k)-\frac{\alpha}{2}\gbf_\beta(\wbf(k-1))\Big)$\\
\hspace*{0.3cm} else\\
\hspace*{0.45cm} $\wbf(k)=\wbf(k-1)$\\
\hspace*{0.3cm} end\\
end\\
 \\
\hline
\end {tabular}
\end{footnotesize}
\end{center}
\label{tb:DS-l0-RLS-sparse}
\end{table}

In Subsection~\ref{sub:simulation_rls_based_sparse}, we compare the simulation results of the RLS-based\abbrev{RLS}{Recursive Least-Squares} algorithms with the Adaptive Sparse Variational Bayes iterative scheme based on Laplace prior (ASVB-L)\abbrev{ASVB-L}{Adaptive Sparse Variational Bayes Iterative Scheme Based on Laplace Prior} algorithm~\cite{Themelis_BayesianAP_tsp2014,Giampouras_Bayesian_LR_Subspace_eusipco2015,Themelis_Bayesian_GIGMC_eusipco2015}. Therefore, it is worthwhile to compare the computational complexity of these algorithms. Table~\ref{tab-complexity-rls} shows the number of real multiplications, real additions, and real divisions must be performed at each iteration by these algorithms.

\begin{table*}[t]
\caption{Number of operations for AS-RLS, $l_0$-RLS, and ASVB-L algorithms \label{tab-complexity-rls}}
\begin{center}
\begin{tabular}{|l|c|c|c|} \hline
Algorithm & Addition $\&$ Subtraction & Multiplication & Division\\\hline
AS-RLS & $N^2+3N$ & $N^2+5N+1$ & $1$ \\\hline
A-$l_0$-RLS & $N^2+5N$ & $N^2+9N+1$ & $N+1$ \\\hline
ASVB-L & $N^2+7N+6$ & $2N^2+10N+3$ & $6N+2$ \\\hline
\end{tabular}
\end{center}
\end{table*}


\section{Simulations} \label{sec:simulations-eusipco}

In this section, we present some numerical simulations for the proposed algorithms. In all scenarios, we deal with the system identification problem. In Subsection~\ref{sub:simulation_lms_based_sparse}, we apply the LMS-based algorithms. The numerical results of the RLS-based\abbrev{RLS}{Recursive Least-Squares} algorithms are illustrated in Subsection~\ref{sub:simulation_rls_based_sparse}.


\subsection{Simulation results of the LMS-based algorithms} \label{sub:simulation_lms_based_sparse}

Here, we have applied the algorithms described in this chapter, the NLMS,\abbrev{NLMS}{Normalized LMS} and the AP\abbrev{AP}{Affine Projection} algorithms to identify three unknown sparse systems of order 14.\footnote{The results 
for the S-SM-AP\abbrev{S-SM-AP}{Simple SM-AP} algorithm are not shown here because they are almost identical to the results of the 
IS-SM-AP\abbrev{IS-SM-AP}{Improved S-SM-AP} algorithm, but the latter has the advantage of requiring fewer computations.} The  first one is an arbitrary sparse system $\wbf_o$, the second one is a block sparse 
system $\wbf'_o$, and the third one is a symmetric-block sparse system $\wbf''_o$. 
The coefficients of these three systems are presented in Table \ref{tab2-eusipco}. The input is a binary phase-shift keying (BPSK)\abbrev{BPSK}{Binary Phase-Shift Keying} signal with variance $\sigma_x^2=1$. 
The signal-to-noise ratio (SNR)\abbrev{SNR}{Signal-to-Noise Ratio} is set to be 20 dB, i.e., the noise variance is $\sigma_n^2=0.01$. The data-reuse factor is $L=1$, the bound on the estimation error is set to be 
$\gammabar=\sqrt{5\sigma_n^2}$, and the threshold bound vector $\gammabf(k)$ is selected 
as the {\it simple-choice constraint vector}~\cite{Markus_sparseSMAP_tsp2014} which is defined as  
$\gamma_0(k)=\frac{\gammabar e(k)}{|e(k)|}$ and $\gamma_i(k)=d(k-i)-\wbf^T(k)\xbf(k-i)$, 
for $i=1,\cdots,L$. The initial vector $\wbf(0)$ and the regularization factor are $10^{-3}\times[1,\cdots,1]^T$ and 
$10^{-12}$, respectively. The learning curves are the results of averaging of the outcomes of 500 trials. 

\begin{table}[t!]
\caption{The coefficients of unknown systems $\wbf_o$, $\wbf'_o$, and $\wbf''_o$. \label{tab2-eusipco}}
\begin{center}
\begin{tabular}{|c|c|c|}\hline
$\wbf_o$&$\wbf'_o$&$\wbf''_o$\\\hline
{\bf 24e-2}&2e-7&2e-8\\
2e-8&-21e-10&-1e-9\\
{\bf -23e-2}&17e-8&1e-7\\
-3e-7&21e-8&-3e-7\\
{\bf 5e-1}&-3e-7&{\bf -64e-3}\\
-1e-9&{\bf 24e-2}&{\bf 2e-1}\\
{\bf 2e-1}&{\bf 7e-1}&{\bf 5e-1}\\
1e-7&{\bf 2e-1}&{\bf 2e-1}\\
-5e-8&{\bf 33e-2}&{\bf -64e-3}\\
12e-6&{\bf -6e-1}&-5e-5\\
1e-8&-5e-7&12e-6\\
-5e-6&18e-9&1e-8\\
4e-6&-5e-7&-5e-6\\
-1e-7&21e-8&4e-6\\
{\bf -2e-1}&-11e-8&-1e-5\\\hline
\end{tabular}
\end{center}
\end{table}

\subsubsection{Scenario 1}

In this scenario, we have implemented the IS-SM-AP\abbrev{IS-SM-AP}{Improved S-SM-AP}, the SSM-AP\abbrev{SSM-AP}{Sparsity-Aware SM-AP}, the SM-PAPA\abbrev{SM-PAPA}{Set-Membership Proportionate AP Algorithm}, and the NLMS\abbrev{NLMS}{Normalized LMS} algorithms to identify the three unknown sparse systems in Table \ref*{tab2-eusipco}. The convergence factor of the NLMS\abbrev{NLMS}{Normalized LMS} algorithm is $\mu =0.9$. The constant $\epsilon$ in the IS-SM-AP\abbrev{IS-SM-AP}{Improved S-SM-AP} algorithm is chosen as $2\times10^{-4}$; that is, on average, 5 out of 15 coefficients (boldface coefficients in $\wbf_o$, $\wbf'_o$, and $\wbf''_o$) are updated at each iteration. We have selected $\alpha=5\times10^{-3}$, $\beta=5$, and $\varepsilon=100$ for the SM-PAPA\abbrev{SM-PAPA}{Set-Membership Proportionate AP Algorithm} and the SSM-AP\abbrev{SSM-AP}{Sparsity-Aware SM-AP} algorithms. In the SSM-AP\abbrev{SSM-AP}{Sparsity-Aware SM-AP} algorithm, we have used the GMF\abbrev{GMF}{Geman-McClure Function} as the approximation of the $l_0$ norm. 

\begin{figure}[t!]
\centering
\subfigure[b][]{\includegraphics[width=.48\linewidth,height=7cm]{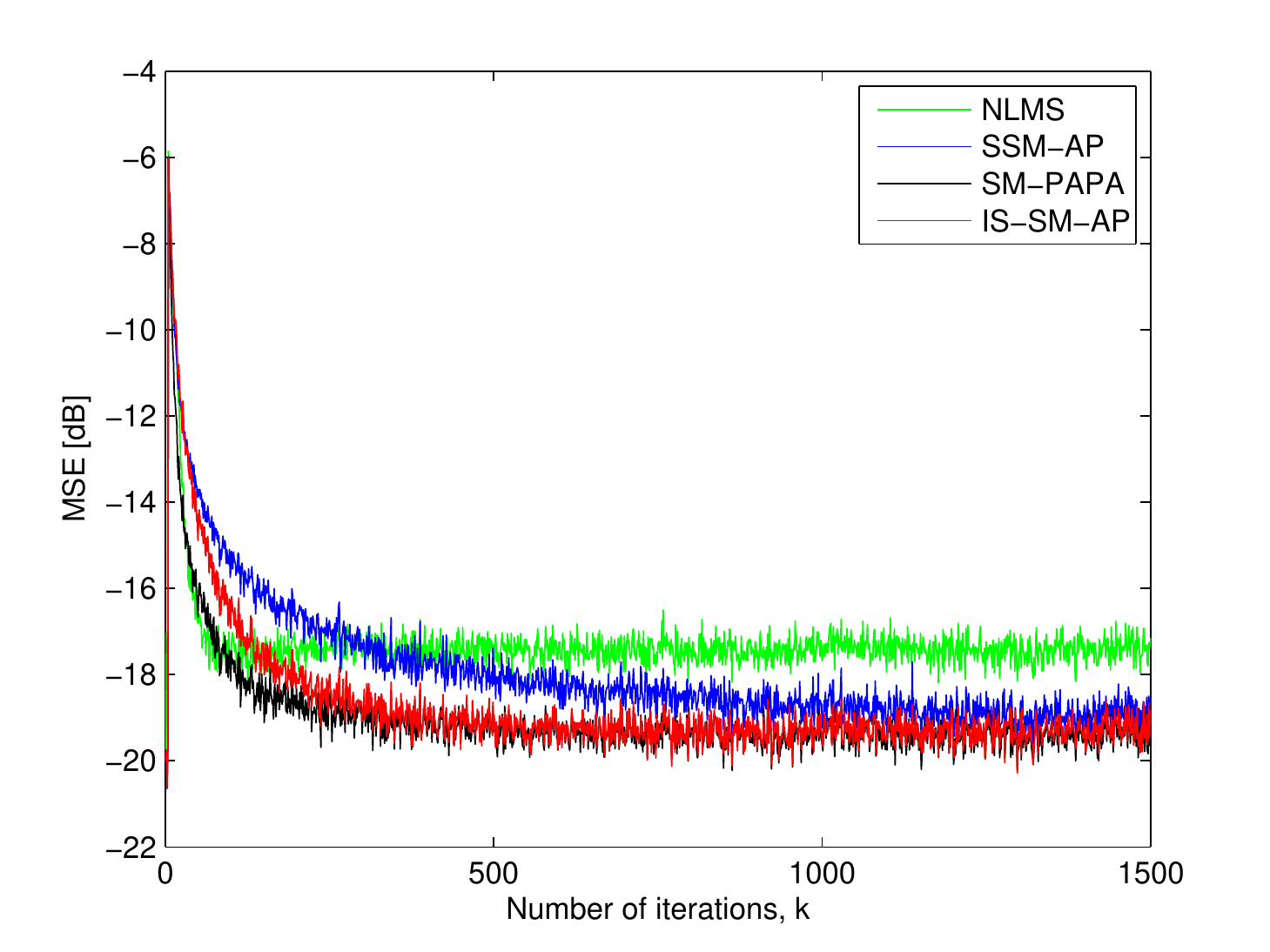}
\label{fig:sim1-eusipco}}
\subfigure[b][]{\includegraphics[width=.48\linewidth,height=7cm]{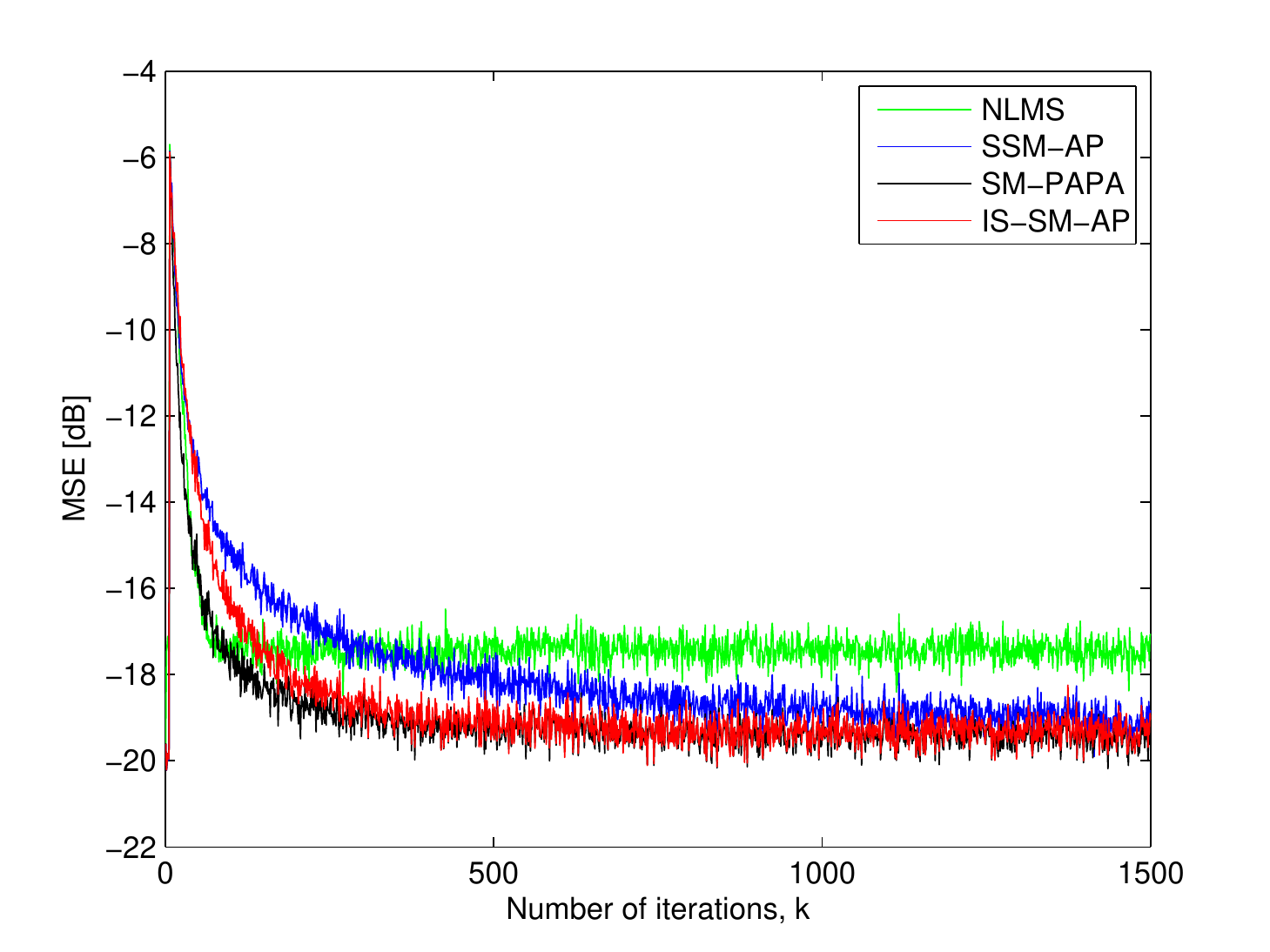}
\label{fig:sim2-eusipco}}
\subfigure[b][]{\includegraphics[width=.48\linewidth,height=7cm]{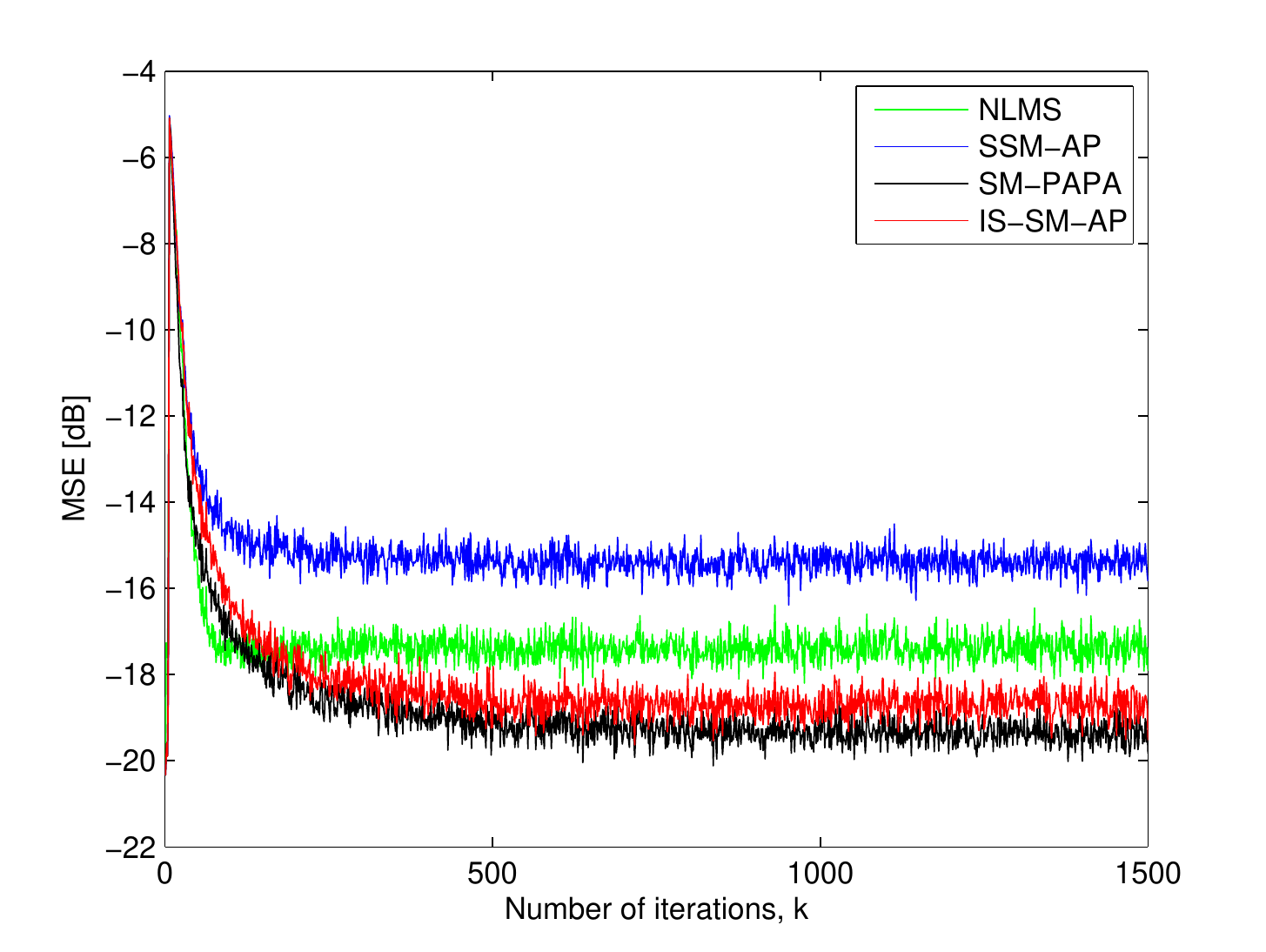}
\label{fig:sim3-eusipco}}
\caption{The learning curves of the SM-PAPA, the SSM-AP, the  IS-SM-AP, and the NLMS algorithms applied on: (a) $\wbf_o$; (b)  $\wbf'_o$; (c) $\wbf''_o$.  \label{fig:1-eusipco}}
\end{figure}

Figures \ref{fig:sim1-eusipco}, \ref{fig:sim2-eusipco}, and \ref{fig:sim3-eusipco} depict the learning curves for 
the IS-SM-AP\abbrev{IS-SM-AP}{Improved S-SM-AP}, the SM-PAPA\abbrev{SM-PAPA}{Set-Membership Proportionate AP Algorithm}, the SSM-AP\abbrev{SSM-AP}{Sparsity-Aware SM-AP}, and the NLMS\abbrev{NLMS}{Normalized LMS}  algorithms to identify the unknown systems $\wbf_o$, $\wbf'_o$, and $\wbf''_o$, respectively. The average number of updates implemented by the IS-SM-AP\abbrev{IS-SM-AP}{Improved S-SM-AP}, the SM-PAPA\abbrev{SM-PAPA}{Set-Membership Proportionate AP Algorithm}, and the SSM-AP\abbrev{SSM-AP}{Sparsity-Aware SM-AP} algorithms are given in columns 2 to 4 of Table \ref{tab:update-rate-eusipco}.

In addition, we have applied all the aforementioned algorithms in this scenario, using the parameters that 
were already defined in the previous paragraph, but changing the input signal model to an 
autoregressive (AR)\abbrev{AR}{Autoregressive} process in order to identify the unknown system $\wbf_o$. The new input signal is generated as a first-order AR\abbrev{AR}{Autoregressive} process defined as 
$x(k)=0.95x(k-1)+n(k)$. 
In this case, the learning curves of the algorithms are shown in Figure~\ref{fig:sim4-eusipco}, 
and the average number of updates performed by the IS-SM-AP\abbrev{IS-SM-AP}{Improved S-SM-AP}, the SM-PAPA\abbrev{SM-PAPA}{Set-Membership Proportionate AP Algorithm}, and the SSM-AP\abbrev{SSM-AP}{Sparsity-Aware SM-AP} algorithms 
are presented in the fifth column of Table~\ref{tab:update-rate-eusipco}. Also, the number of arithmetic operations required by the IS-SM-AP\abbrev{IS-SM-AP}{Improved S-SM-AP}, the SM-PAPA\abbrev{SM-PAPA}{Set-Membership Proportionate AP Algorithm}, and the SSM-AP\abbrev{SSM-AP}{Sparsity-Aware SM-AP} algorithms in whole iterations are 41635, 110835, and 84396, respectively.

Observe that, in every scenario we tested, the IS-SM-AP\abbrev{IS-SM-AP}{Improved S-SM-AP} algorithm performed as well as 
the other state-of-the-art sparsity-aware algorithms, but this  algorithm has the advantage of 
requiring fewer computations since at each iteration in which an update occurs only 
a subset (on average, one third) of the coefficients is updated. Another interesting observation is that the SM-PAPA\abbrev{SM-PAPA}{Set-Membership Proportionate AP Algorithm} algorithm works better with BPSK\abbrev{BPSK}{Binary Phase-Shift Keying} input signal, 
whereas the SSM-AP\abbrev{SSM-AP}{Sparsity-Aware SM-AP} algorithm is slightly better when a correlated input signal is used.

\begin{figure}[t!]
\centering
\includegraphics[width=1\linewidth]{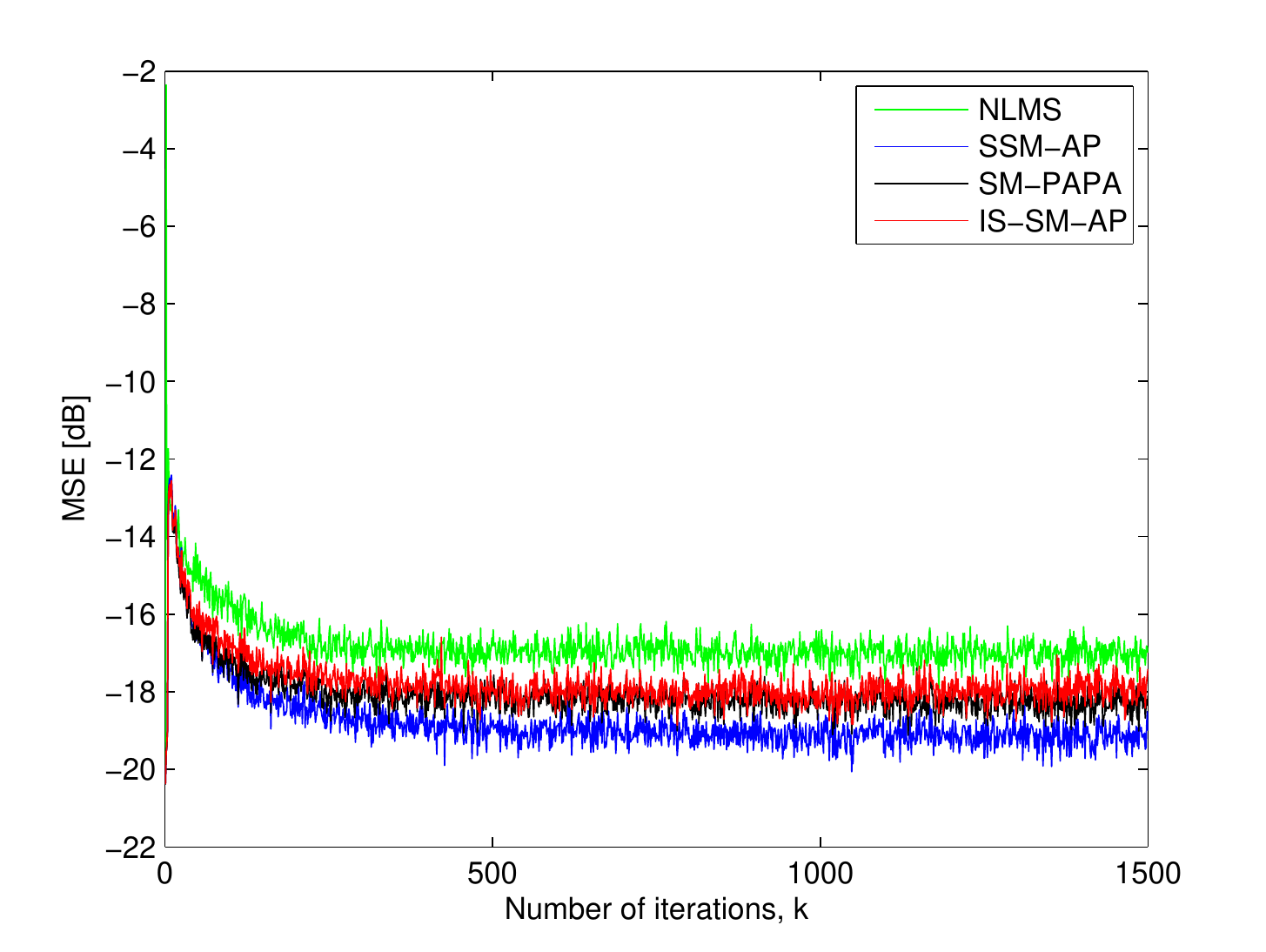}
\caption{The learning curves of the SM-PAPA, the SSM-AP, the  IS-SM-AP, and the NLMS algorithms applied on $\wbf_o$ using AR input signal.\label{fig:sim4-eusipco}}
\end{figure}

\begin{table*}
\caption{The average number of updates implemented by 
the IS-SM-AP, the SM-PAPA, and the SSM-AP algorithms \label{tab:update-rate-eusipco}}
\begin{center}
\begin{tabular}{|l|c|c|c|c|}\hline
Algorithm&$\wbf_o$ BPSK input&$\wbf'_o$ BPSK input&$\wbf''_o$ BPSK input& $\wbf''_o$ AR input\\\hline
IS-SM-AP&6.3$\%$&6.3$\%$&7.6$\%$&8.4$\%$\\
SM-PAPA&5.3$\%$&5.3$\%$&5.9$\%$&7.7$\%$\\
SSM-AP&8.9$\%$&8.9$\%$&20.5$\%$&5.6$\%$\\\hline
\end{tabular}
\end{center}
\end{table*}

\begin{figure}[t!]
\centering
\subfigure[b][]{\includegraphics[width=.48\linewidth,height=7cm]{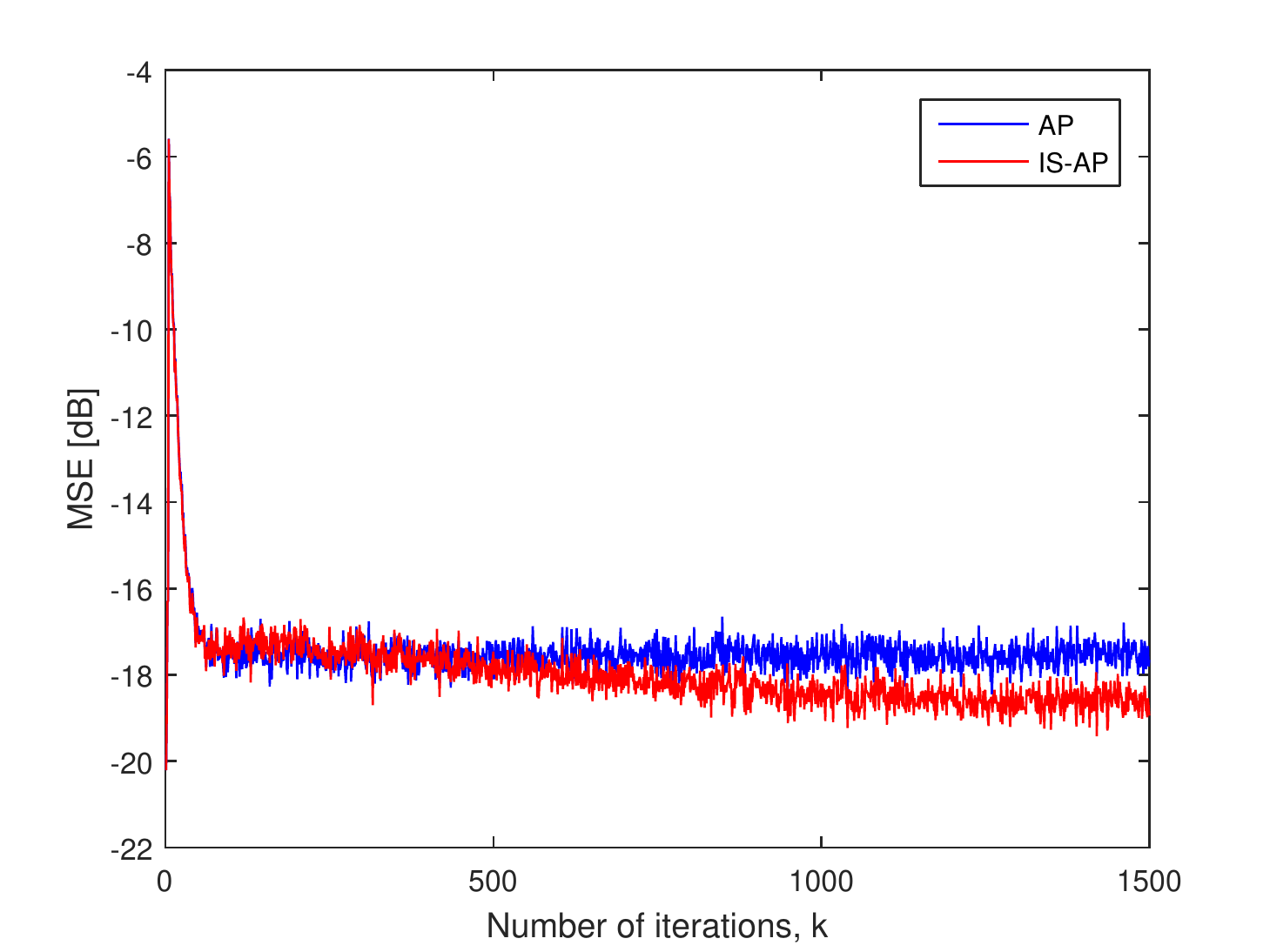}
\label{fig:sim5-eusipco}}
\subfigure[b][]{\includegraphics[width=.48\linewidth,height=7cm]{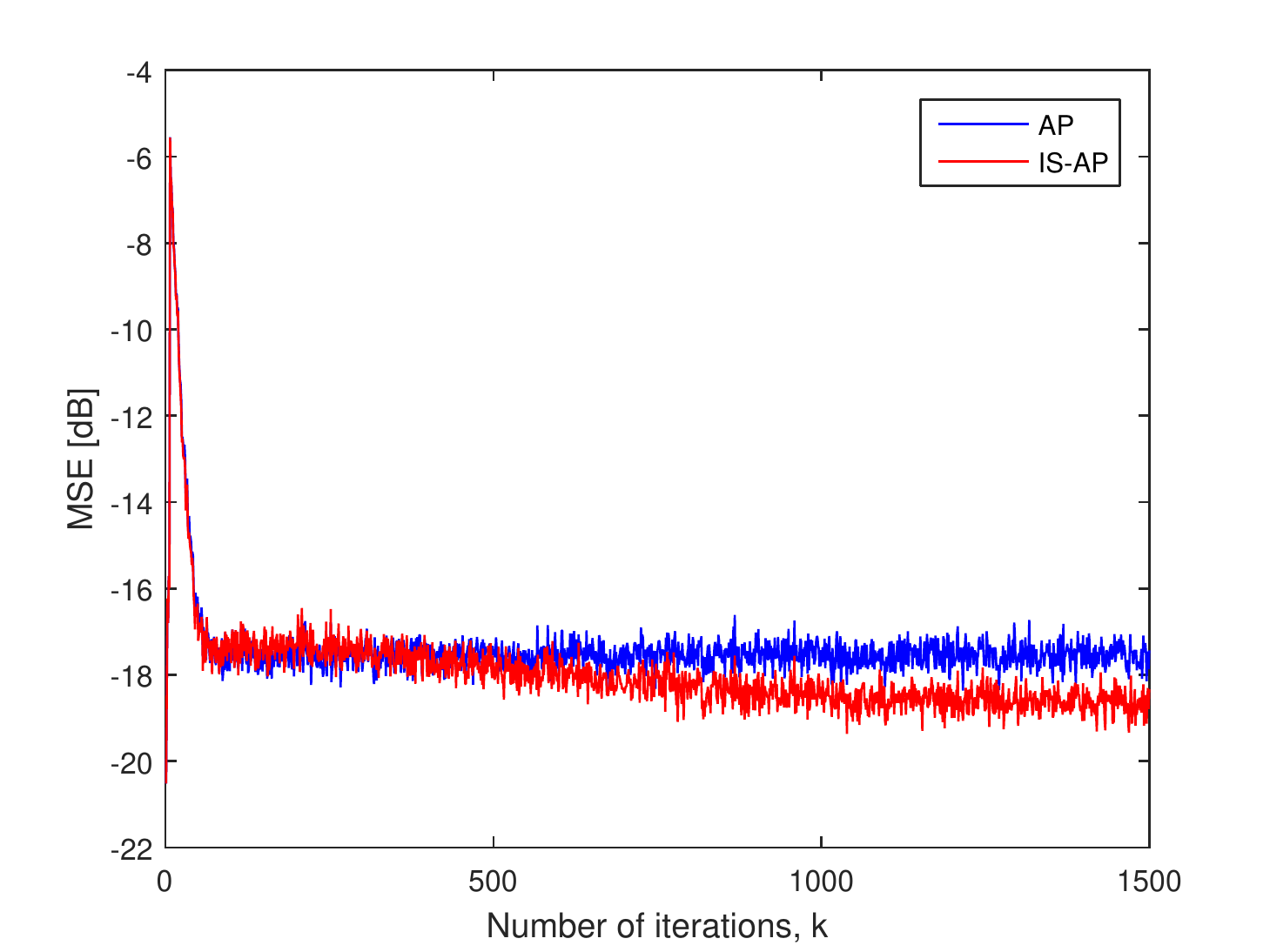}
\label{fig:sim6-eusipco}}
\subfigure[b][]{\includegraphics[width=.48\linewidth,height=7cm]{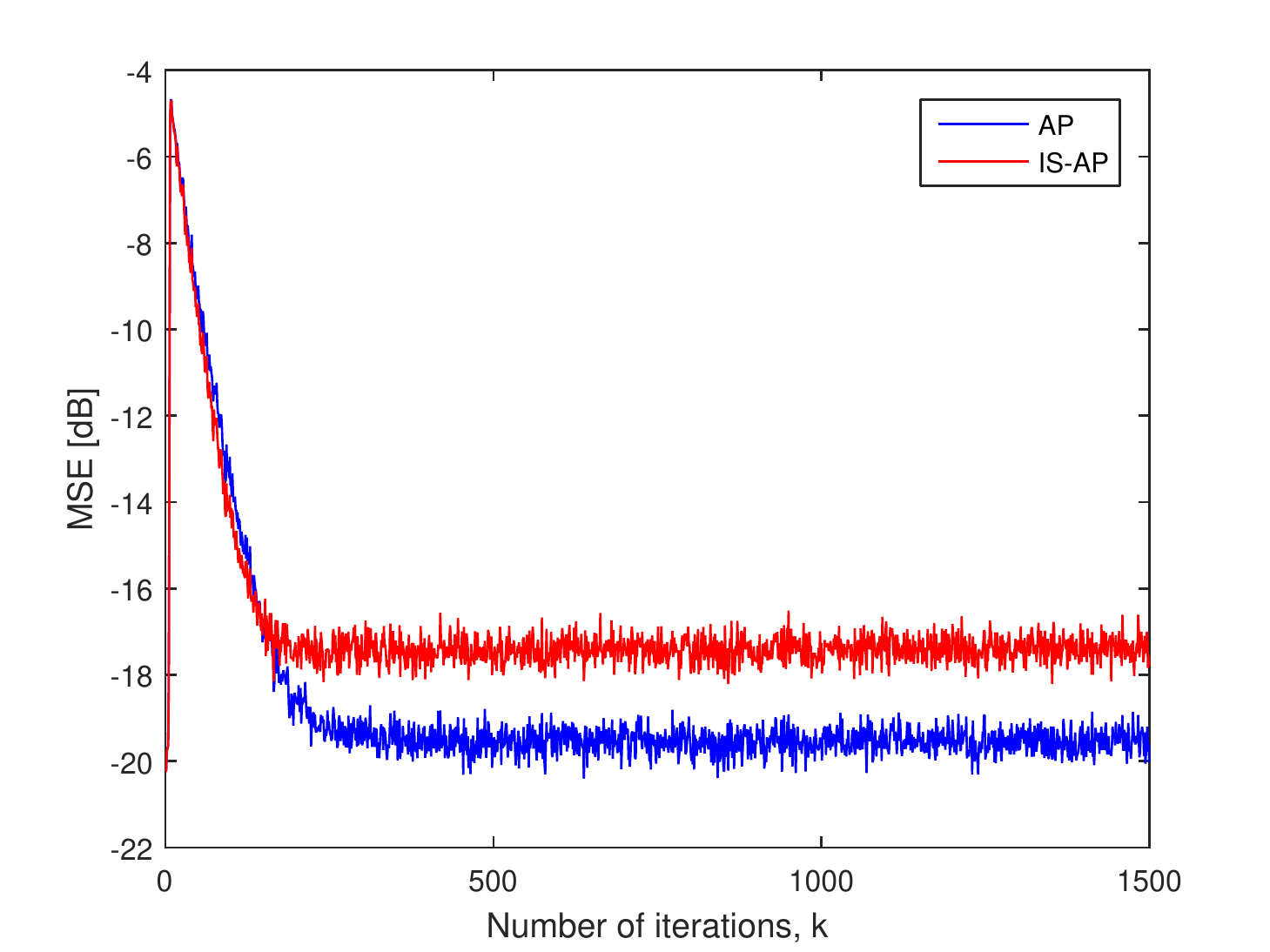}
\label{fig:sim7-eusipco}}
\caption{The learning curves of the AP and the IS-AP algorithms applied on: (a) $\wbf_o$; (b)  $\wbf'_o$; (c) $\wbf''_o$.  \label{fig:2-eusipco}}
\end{figure}

\subsubsection{Scenario 2}

In this scenario, we have applied the AP\abbrev{AP}{Affine Projection} and the IS-AP\abbrev{IS-AP}{Improved S-AP} algorithms to identify the three unknown sparse systems in Table \ref*{tab2-eusipco}. To identify $\wbf_o$ and $\wbf'_o$ we choose the convergence factor $\mu=0.6$ and to identify $\wbf''_o$ we adopt $\mu=0.1$. Figures \ref{fig:sim5-eusipco}, \ref{fig:sim6-eusipco}, and \ref{fig:sim7-eusipco} show the learning curves for 
the AP\abbrev{AP}{Affine Projection} and the IS-AP\abbrev{IS-AP}{Improved S-AP} algorithms to identify the 
unknown systems $\wbf_o$, $\wbf'_o$, and $\wbf''_o$, respectively.

Moreover, we have applied the AP\abbrev{AP}{Affine Projection} and the IS-AP\abbrev{IS-AP}{Improved S-AP} algorithms in this scenario, with same parameters, but changing the input signal model to an 
autoregressive (AR)\abbrev{AR}{Autoregressive} as Scenario 1 to identify the unknown system $\wbf_o$. The convergence factor $\mu$ is equal to 0.6. Their  learning curves are shown in Figure~\ref{fig:sim8-eusipco}. By comparing Figures \ref{fig:sim4-eusipco} and \ref{fig:sim8-eusipco} we can observe the value of set-membership filtering. In fact, by utilizing the SMF\abbrev{SMF}{Set-Membership Filtering} approach not only we have a lower number of arithmetic operations, but also we improve the steady state performance. Note that, we have obtained better MSE\abbrev{MSE}{Mean-Squared Error}  in all figures of Scenario 1 compared to their corresponding figures in Scenario 2.

\begin{figure}[t!]
\centering
\includegraphics[width=1\linewidth]{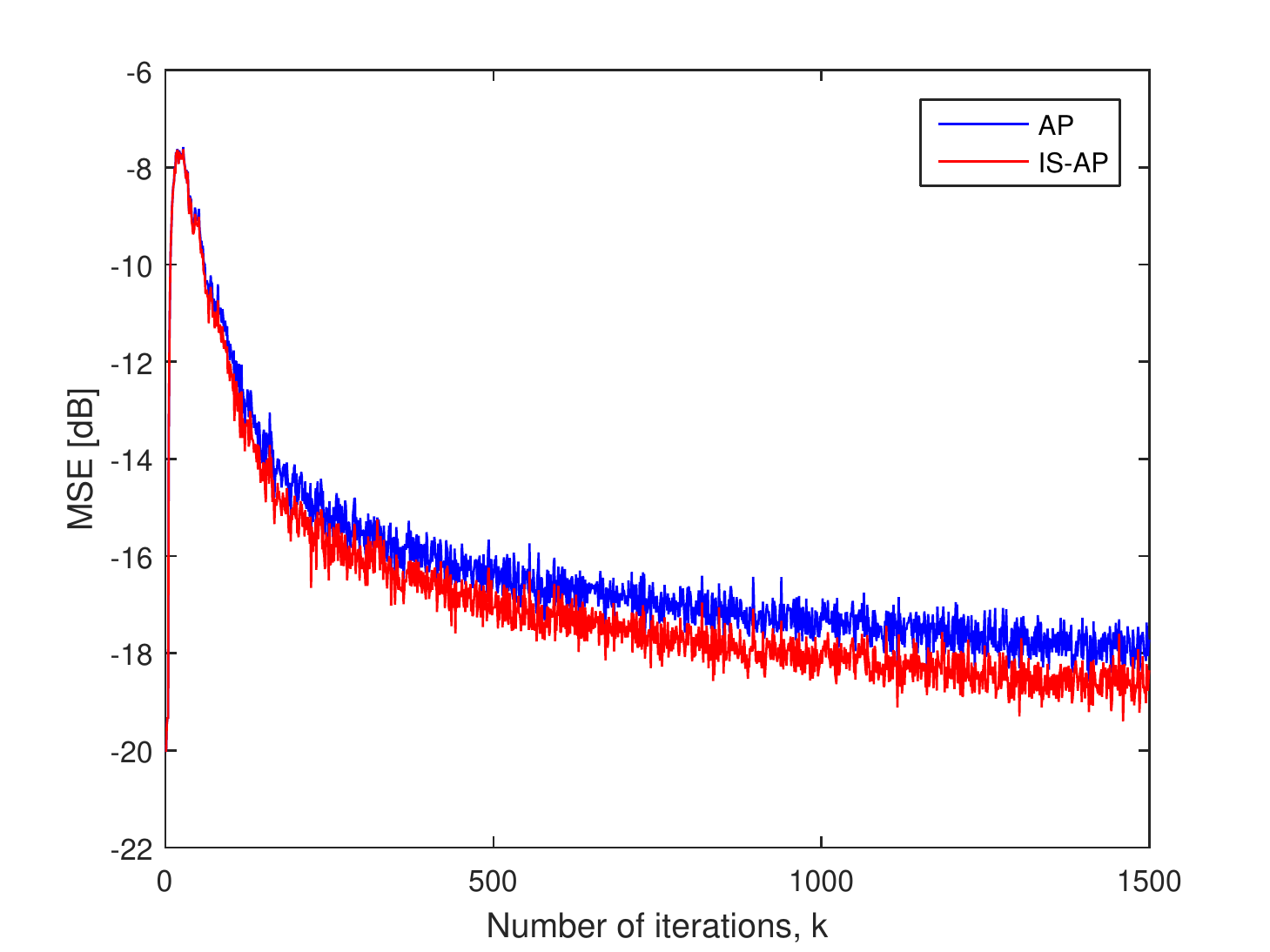}
\caption{The learning curves of the AP and the IS-AP algorithms applied on $\wbf_o$ using AR input signal.\label{fig:sim8-eusipco}}
\end{figure}


\begin{figure}[t!]
\centering
\subfigure[b][]{\includegraphics[width=.48\linewidth,height=7cm]{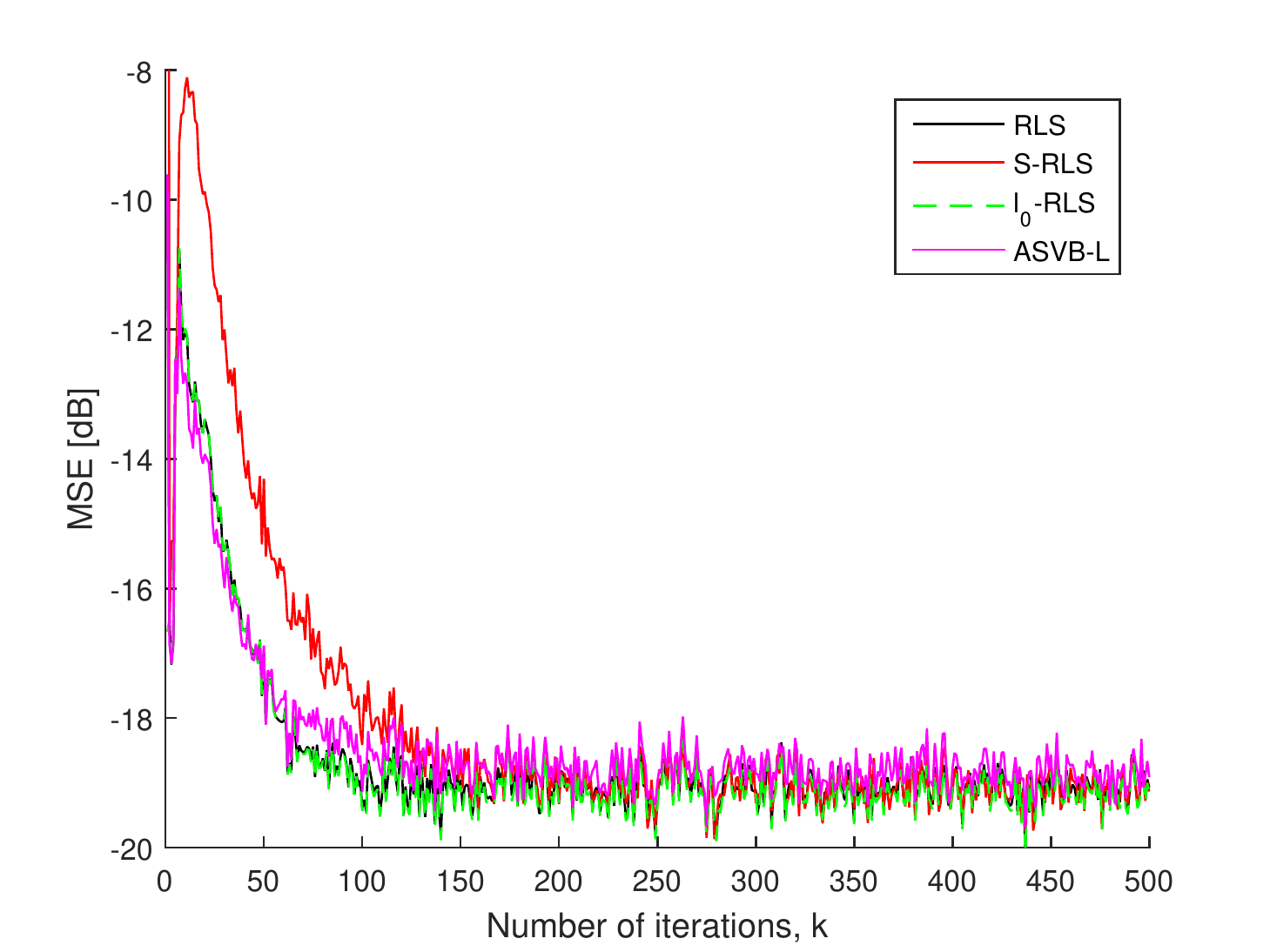}
\label{fig:RLS_sys1_sparse}}
\subfigure[b][]{\includegraphics[width=.48\linewidth,height=7cm]{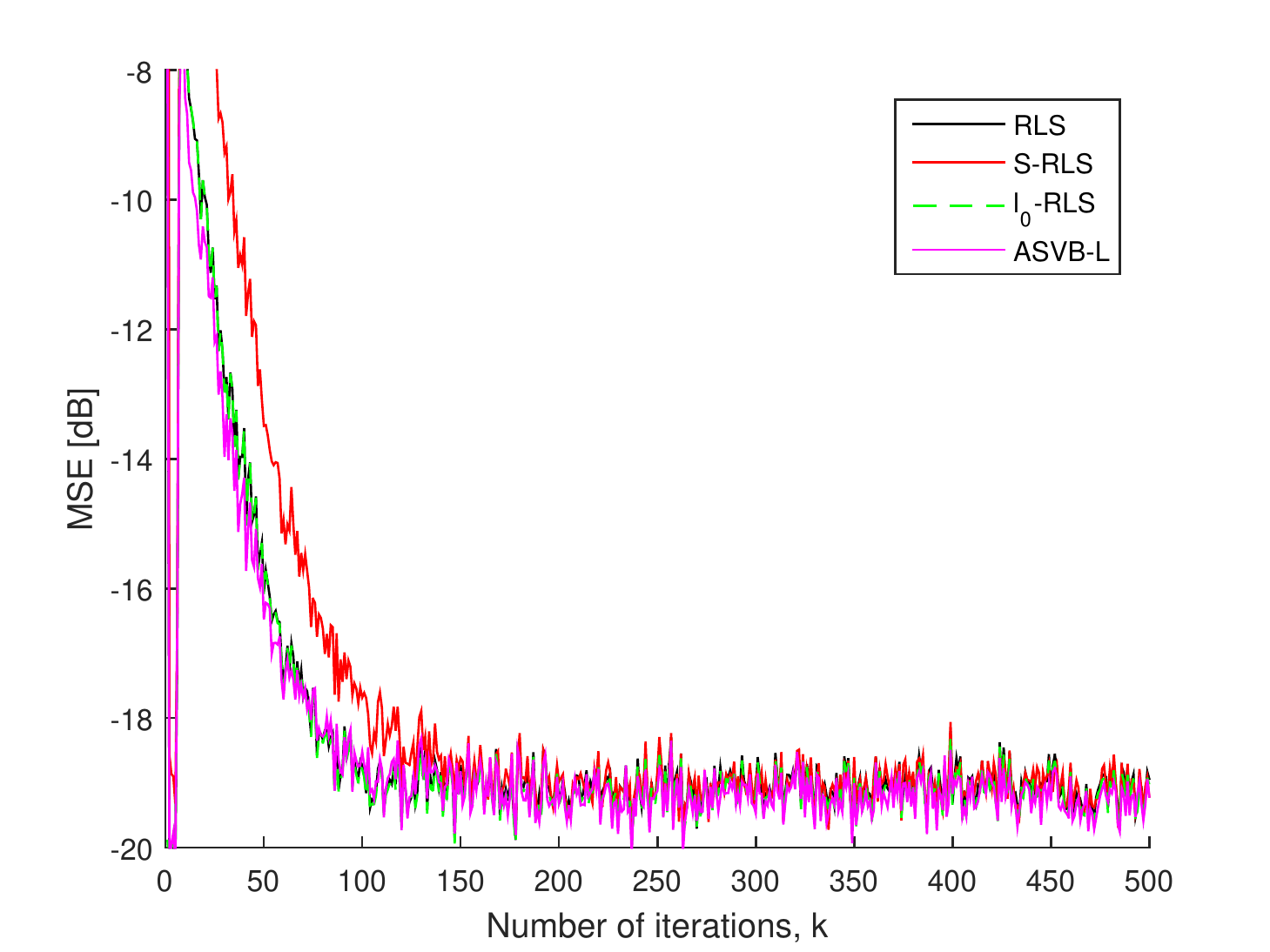}
\label{fig:RLS_sys2_sparse}}
\subfigure[b][]{\includegraphics[width=.48\linewidth,height=7cm]{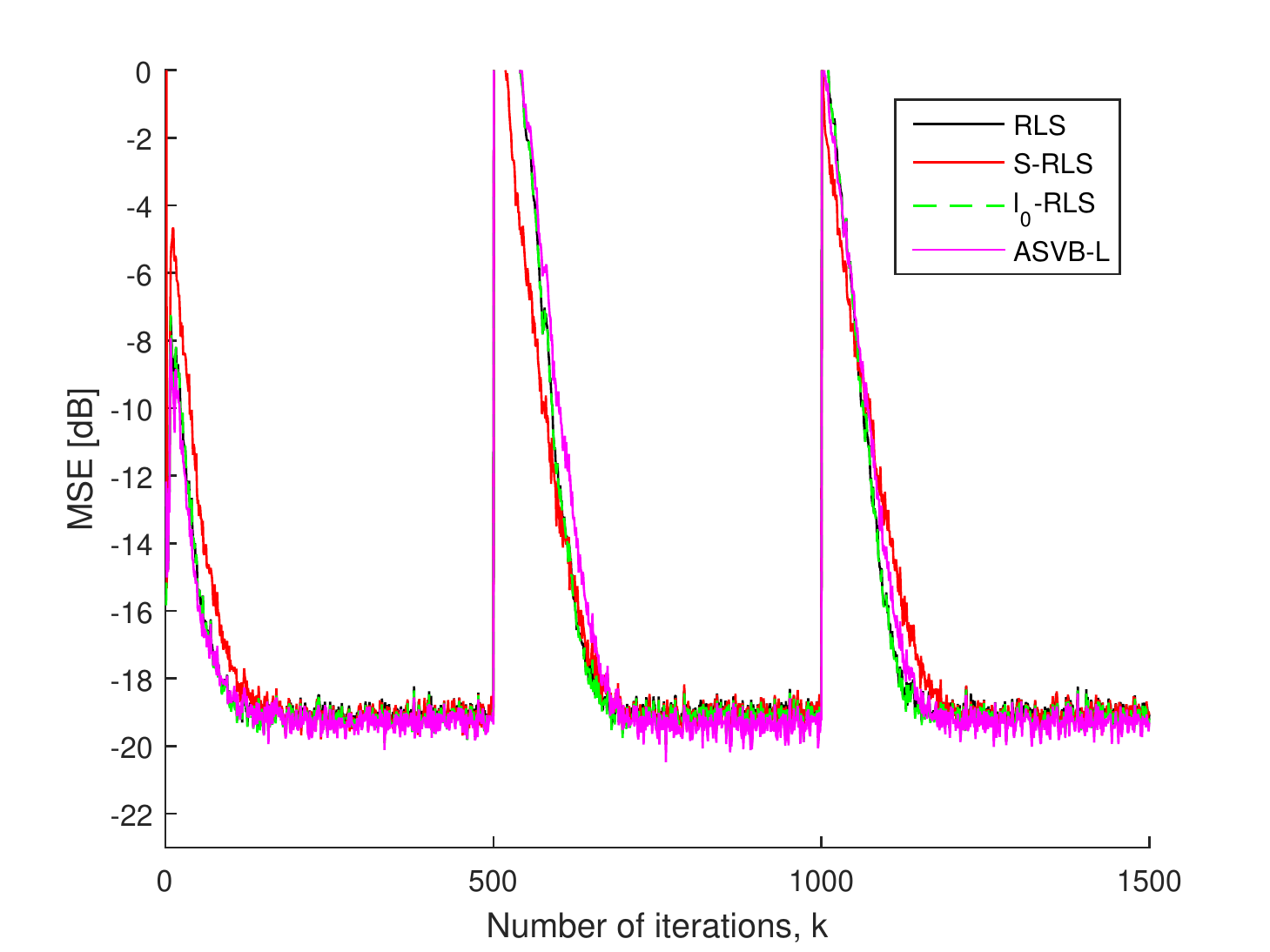}
\label{fig:RLS_sys3_sparse}}
\caption{The learning curves of the RLS, the S-RLS, the $l_0$-RLS, and the ASVB-L algorithms applied to identify: (a) $\wbf_o$; (b) $\wbf'_o$; (c) $\wbf'''_o$.  \label{fig:RLS_sparse}}
\end{figure}

\subsection{Simulation results of the RLS-based algorithms} \label{sub:simulation_rls_based_sparse}

Here, the RLS\abbrev{RLS}{Recursive Least-Squares}, the S-RLS\abbrev{S-RLS}{RLS Algorithm for Sparse System}, the AS-RLS\abbrev{AS-RLS}{Alternative S-RLS}, the $l_0$-RLS\abbrev{$l_0$-RLS}{$l_0$ Norm RLS}, the A-$l_0$-RLS\abbrev{A-$l_0$-RLS}{Alternative $l_0$-RLS}, the  ASVB-L~\cite{Themelis_BayesianAP_tsp2014,Giampouras_Bayesian_LR_Subspace_eusipco2015,Themelis_Bayesian_GIGMC_eusipco2015}\abbrev{ASVB-L}{Adaptive Sparse Variational Bayes Iterative Scheme Based on Laplace Prior}, the DS-S-RLS\abbrev{DS-S-RLS}{Data-Selective S-RLS}, the DS-$l_0$-RLS\abbrev{DS-$l_0$-RLS}{Data-Selective $l_0$-RLS}, and the data-selective ASVB-L (DS-ASVB-L)\abbrev{DS-ASVB-L}{Data-Selective ASVB-L} algorithms are tested to identify three unknown sparse systems of order 14. The first model is an arbitrary sparse system $\wbf_o$, the second model is a block sparse 
system $\wbf'_o$, and the third model, $\wbf'''_o$, is a sparse system which its coefficients changes at $500th$ and $1000th$ iterations.  
The coefficients of $\wbf_o$ and $\wbf'_o$ are listed in Table~\ref{tab2-eusipco}. 
The input is an autoregressive signal generated by $x(k)=0.95x(k-1)+n(k-1)$.  
The signal-to-noise ratio (SNR)\abbrev{SNR}{Signal-to-Noise Ratio} is set to be 20 dB, meaning that the noise variance is $\sigma_n^2=0.01$. 
The bound on the estimation error is set to be $\gammabar=\sqrt{5\sigma_n^2}$. The initial vector $\wbf(0)$ and $\lambda$ are $[1,\cdots,1]^T$ and 
$0.97$, respectively. The parameter $\delta$ is $0.2$ and the constant $\epsilon$ is chosen as $0.015$. For the DS-$l_0$-RLS\abbrev{DS-$l_0$-RLS}{Data-Selective $l_0$-RLS} and the $l_0$-RLS\abbrev{$l_0$-RLS}{$l_0$ Norm RLS} algorithms, the parameters $\alpha$ and $\beta$ are chosen as 0.005 and 5, respectively. We have chosen the GMF\abbrev{GMF}{Geman-McClure Function} as the approximation of the $l_0$ norm. The depicted learning curves represent the results of averaging of the outcomes of 500 trials.

\begin{figure}[t!]
\centering
\subfigure[b][]{\includegraphics[width=.48\linewidth,height=7cm]{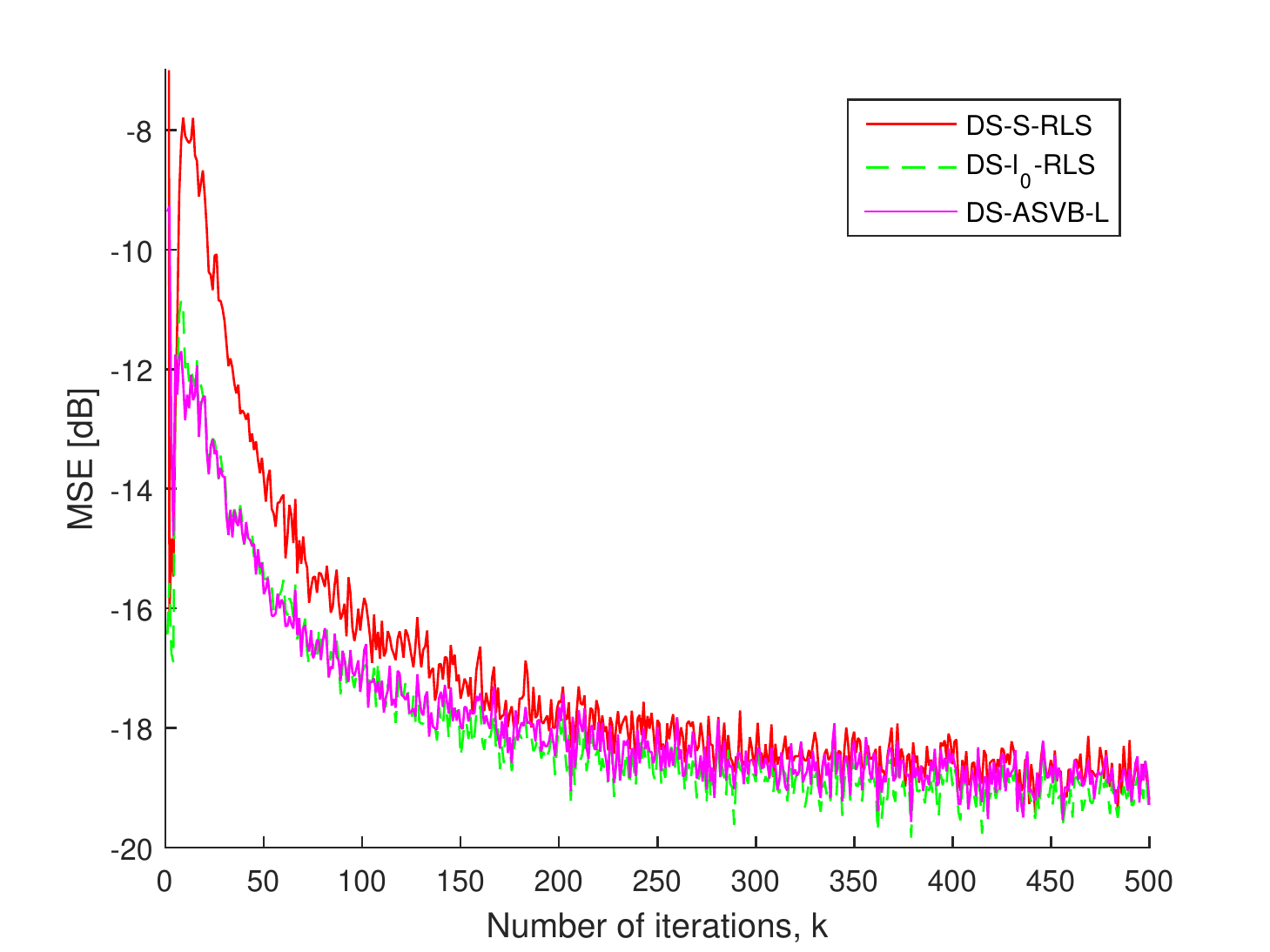}
\label{fig:DS_RLS_sys1_sparse}}
\subfigure[b][]{\includegraphics[width=.48\linewidth,height=7cm]{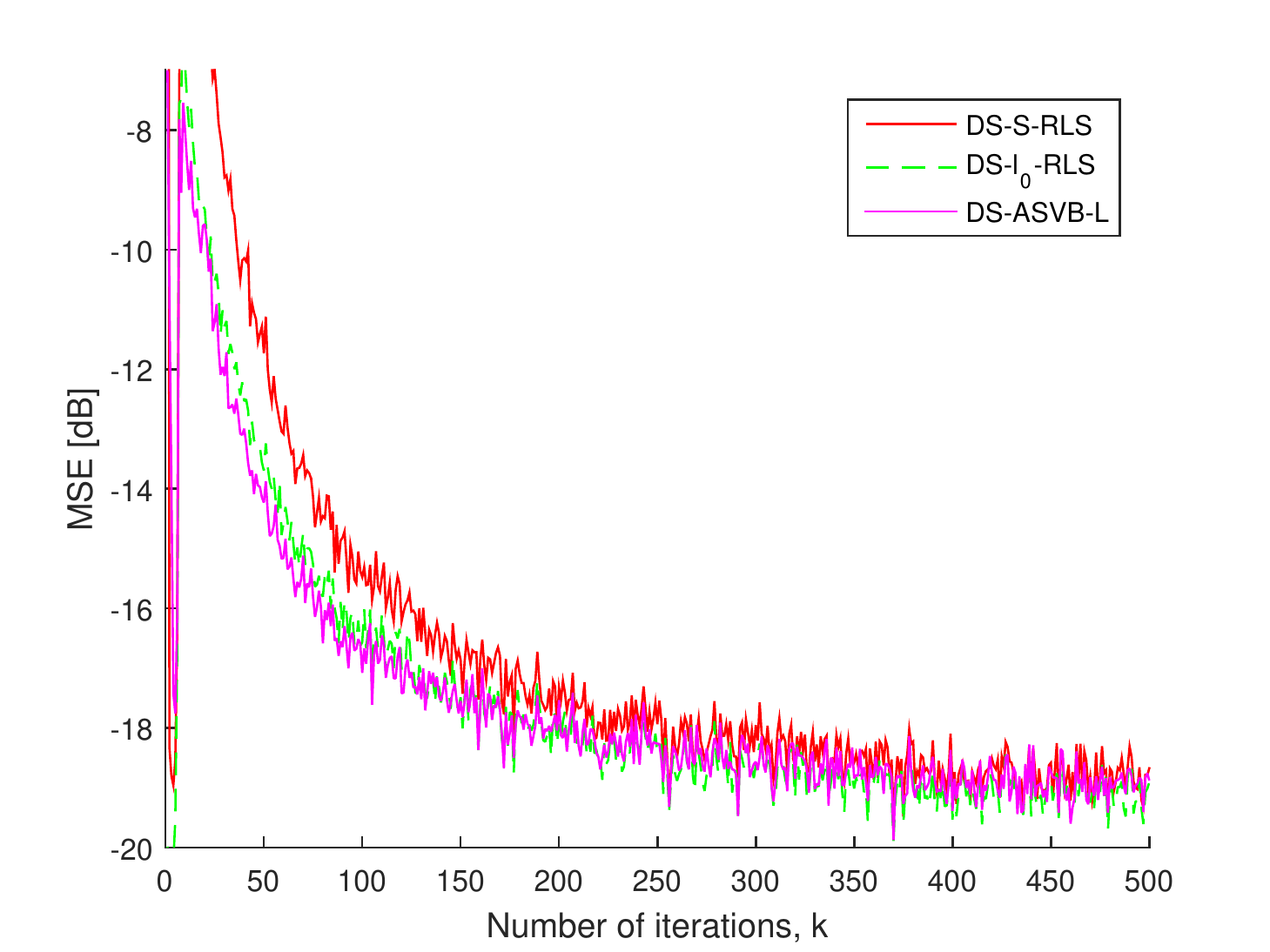}
\label{fig:DS_RLS_sys2_sparse}}
\subfigure[b][]{\includegraphics[width=.48\linewidth,height=7cm]{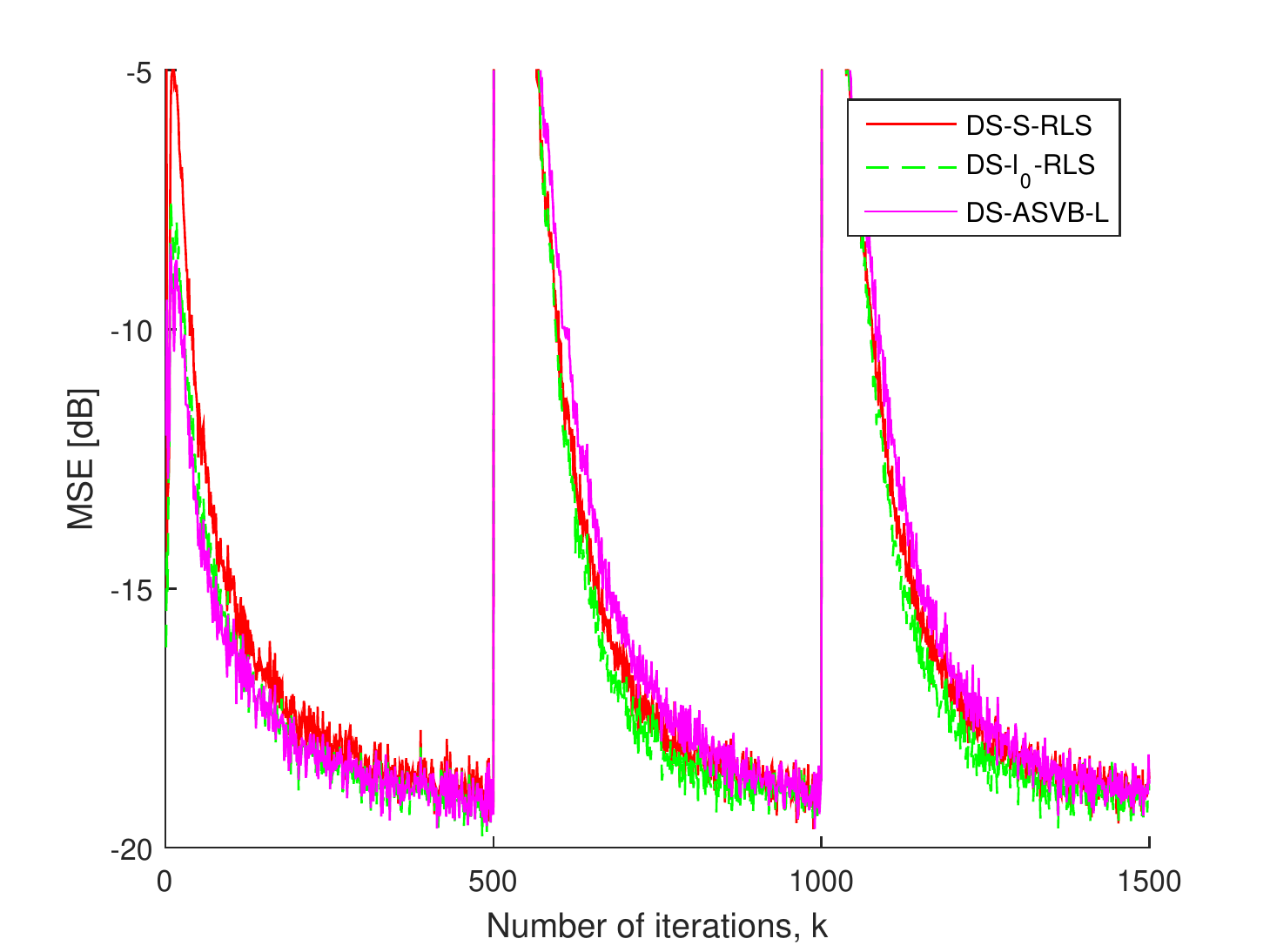}
\label{fig:DS_RLS_sys4_sparse}}
\caption{The learning curves of the DS-S-RLS, the DS-$l_0$-RLS, and the DS-ASVB-L algorithms applied to identify: (a) $\wbf_o$; (b) $\wbf'_o$; (c) $\wbf'''_o$.  \label{fig:DS_RLS_sparse}}
\end{figure}

Figures~\ref{fig:RLS_sys1_sparse}, \ref{fig:RLS_sys2_sparse}, and \ref{fig:RLS_sys3_sparse} show the learning curves for the RLS\abbrev{RLS}{Recursive Least-Squares}, the S-RLS\abbrev{S-RLS}{RLS Algorithm for Sparse System}, the $l_0$-RLS\abbrev{$l_0$-RLS}{$l_0$ Norm RLS}, and the ASVB-L\abbrev{ASVB-L}{Adaptive Sparse Variational Bayes Iterative Scheme Based on Laplace Prior} algorithms to identify the unknown systems $\wbf_o$, $\wbf'_o$, and $\wbf'''_o$, respectively. Figures~\ref{fig:DS_RLS_sys1_sparse}, \ref{fig:DS_RLS_sys2_sparse}, and \ref{fig:DS_RLS_sys4_sparse} illustrate the learning curves for the DS-S-RLS\abbrev{DS-S-RLS}{Data-Selective S-RLS}, the DS-$l_0$-RLS\abbrev{DS-$l_0$-RLS}{Data-Selective $l_0$-RLS}, and the DS-ASVB-L\abbrev{DS-ASVB-L}{Data-Selective ASVB-L} algorithms to identify the unknown systems $\wbf_o$, $\wbf'_o$, and $\wbf'''_o$, respectively. The average number of updates implemented by the DS-S-RLS\abbrev{DS-S-RLS}{Data-Selective S-RLS}, the DS-$l_0$-RLS\abbrev{DS-$l_0$-RLS}{Data-Selective $l_0$-RLS}, and the DS-ASVB-L\abbrev{DS-ASVB-L}{Data-Selective ASVB-L} algorithms are presented in columns 2 to 4 of Table~\ref{tab:update-rate-DS-RLS-sparse}.

Observe that, in every scenario we tested, the S-RLS\abbrev{S-RLS}{RLS Algorithm for Sparse System} and the $l_0$-RLS\abbrev{$l_0$-RLS}{$l_0$ Norm RLS} algorithms performed as well as the RLS\abbrev{RLS}{Recursive Least-Squares} algorithm. The S-RLS\abbrev{S-RLS}{RLS Algorithm for Sparse System} algorithm has lower computational complexity compared to the $l_0$-RLS\abbrev{$l_0$-RLS}{$l_0$ Norm RLS} algorithm. As can be seen, the performances of the S-RLS\abbrev{S-RLS}{RLS Algorithm for Sparse System} and the DS-S-RLS\abbrev{DS-S-RLS}{Data-Selective S-RLS} algorithms are close to the ASVB-L\abbrev{ASVB-L}{Adaptive Sparse Variational Bayes Iterative Scheme Based on Laplace Prior} and the DS-ASVB-L\abbrev{DS-ASVB-L}{Data-Selective ASVB-L} algorithms, respectively, while the former ones require lower computational resources.

Finally, Figures~\ref{fig:A_RLS_sys1_sparse} and~\ref{fig:A_RLS_sys2_sparse} depict the learning curves of the S-RLS, the AS-RLS, the $l_0$-RLS, and the A-$l_0$-RLS algorithms, when they are applied to identify the unknown systems $\wbf_o$ and $\wbf_o'$, respectively. As can be seen, the performances of the AS-RLS and the A-$l_0$-RLS algorithms are similar to the S-RLS and the $l_0$-RLS algorithms, respectively.

\begin{table*}
\caption{The average number of updates implemented by 
the DS-S-RLS, the DS-$l_0$-RLS, and the DS-ASVB-L algorithms \label{tab:update-rate-DS-RLS-sparse}}
\begin{center}
\begin{tabular}{|l|c|c|c|}\hline
Algorithm&$\wbf_o$ &$\wbf'_o$ &$\wbf'''_o$ \\\hline
DS-S-RLS&11.95$\%$&14.13$\%$&19.40$\%$\\
DS-$l_0$-RLS&8.72$\%$&10.90$\%$&17.74$\%$\\
DS-ASVB-L&9.18$\%$&10.53$\%$&19.69$\%$\\\hline
\end{tabular}
\end{center}
\end{table*}

\begin{figure}[t!]
\centering
\subfigure[b][]{\includegraphics[width=.48\linewidth,height=7cm]{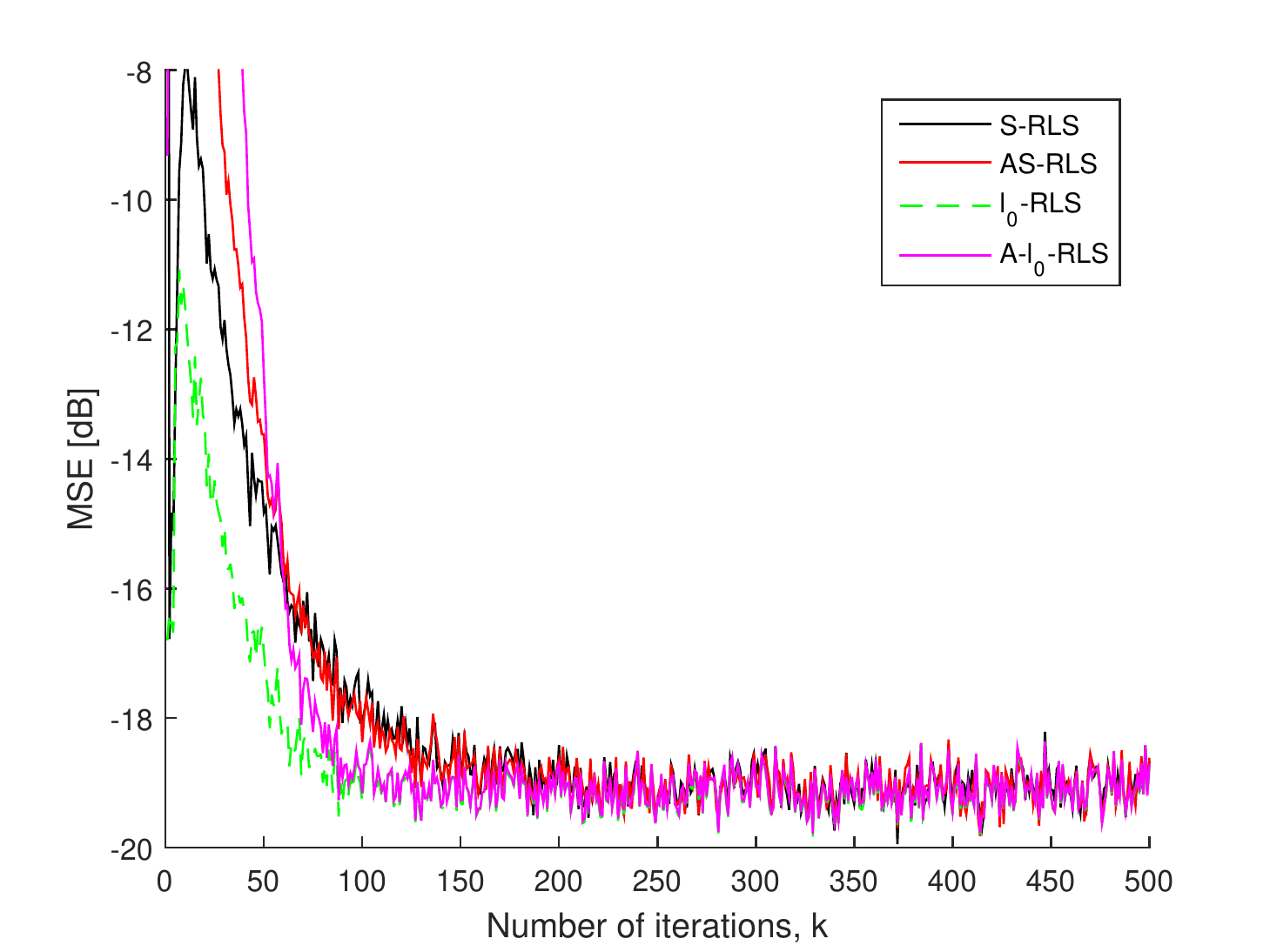}
\label{fig:A_RLS_sys1_sparse}}
\subfigure[b][]{\includegraphics[width=.48\linewidth,height=7cm]{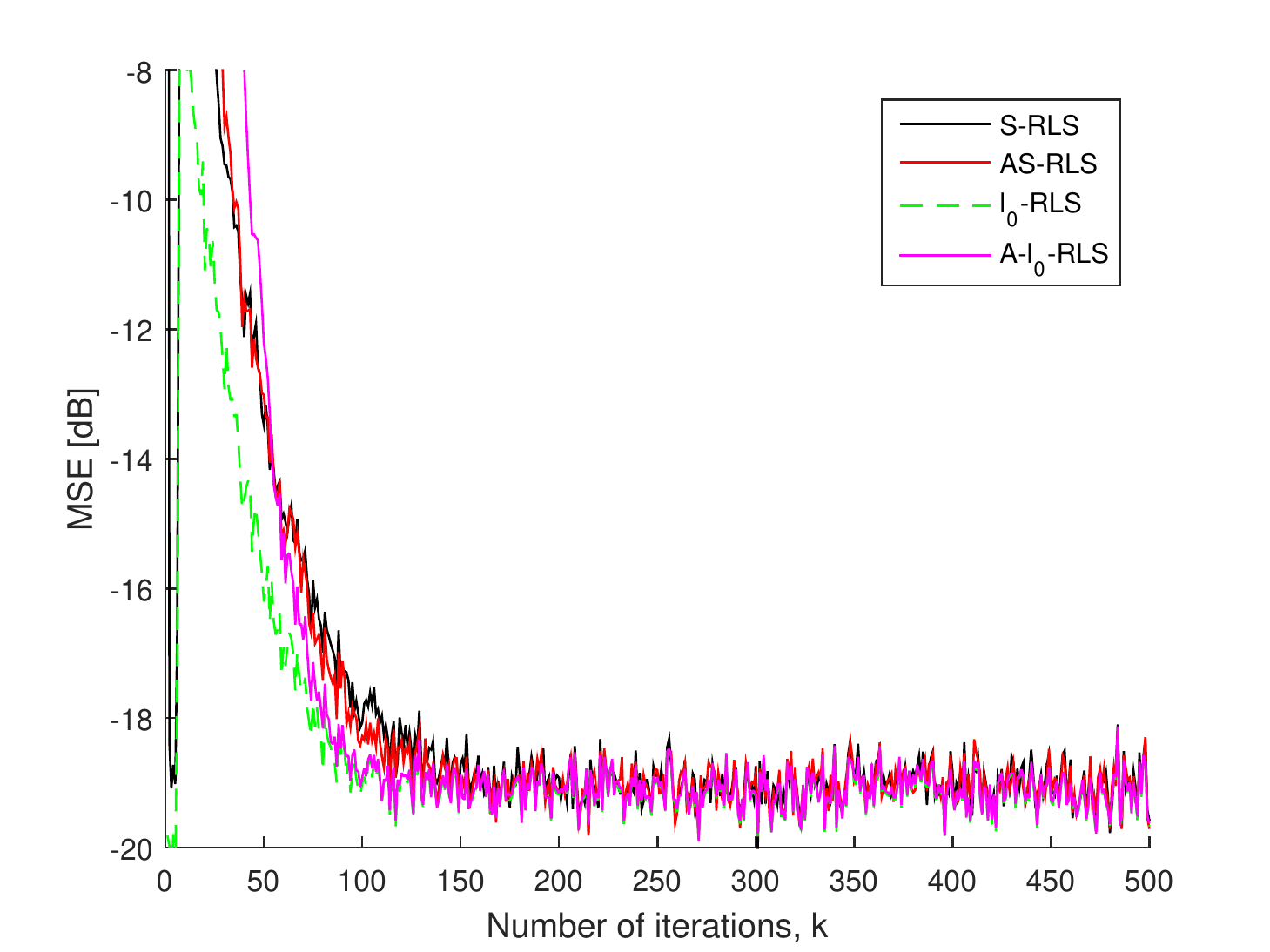}
\label{fig:A_RLS_sys2_sparse}}
\caption{The learning curves of the S-RLS, the AS-RLS, the $l_0$-RLS, and the A-$l_0$-RLS algorithms applied to identify: (a) $\wbf_o$; (b) $\wbf'_o$.  \label{fig:A_RLS_sparse}}
\end{figure}


\section{Conclusions} \label{sec:conclusions-eusipco}

In this chapter, we have proposed the S-SM-AP\abbrev{S-SM-AP}{Simple SM-AP} and the IS-SM-AP\abbrev{IS-SM-AP}{Improved S-SM-AP} algorithms 
to take advantage of sparsity in the signal models while attaining low computational 
complexity. 
To reach this target, we have derived a simple update equation which only updates 
the filter coefficients whose magnitudes are greater than a predetermined value. 
Also, this method is jointly applied with the well-known set-membership approach 
aiming at obtaining even lower computational complexity and better convergence rate. 
The simulation results have shown the excellent performance of the algorithm and 
lower computational complexity as compared to some other sparsity-aware data-selective 
adaptive filters. 
{Indeed, the IS-SM-AP\abbrev{IS-SM-AP}{Improved S-SM-AP} algorithm performed as well as the SM-PAPA\abbrev{SM-PAPA}{Set-Membership Proportionate AP Algorithm} algorithm while requiring fewer arithmetic operations 
(for the scenarios in Section \ref{sec:simulations-eusipco}, it entailed about 38$\%$ of the operations spent by the SM-PAPA).}\abbrev{SM-PAPA}{Set-Membership Proportionate AP Algorithm} Also, the numerical results in Section~\ref{sec:simulations-eusipco} confirm the importance of SMF\abbrev{SMF}{Set-Membership Filtering} technique for the proposed algorithm.

Moreover, we have used the discard function and the $l_0$ norm in order to propose the S-RLS\abbrev{S-RLS}{RLS Algorithm for Sparse System} and the $l_0$-RLS\abbrev{$l_0$-RLS}{$l_0$ Norm RLS} algorithms, respectively, to exploit the sparsity in the involved signal models. Also, we have employed the data-selective strategy to implement an update when the output estimation error is greater than a pre-described positive value leading to 
reduced update rate and lower computational complexity. The simulation results have shown the excellent performance of the proposed algorithms as compared to the standard RLS\abbrev{RLS}{Recursive Least-Squares} algorithm being
competitive with the new proposed state-of-the-art ASVB-L\abbrev{ASVB-L}{Adaptive Sparse Variational Bayes Iterative Scheme Based on Laplace Prior} algorithm which requires much more computations. It is worthy to mention that there are many RLS-based algorithms to exploit sparsity in signal and system models~\cite{Angelosante_rls-sparse_cd_tsp2010,Angelosante_rls_lasso_sparse_icassp2009,Valdman_rls_lar_eusipco2014}; however, their update equation is entirely different from the algorithms proposed in this chapter. Therefore, we avoid comparing the RLS-based algorithms proposed here with other RLS-based algorithms in the literature.

  \chapter{Feature LMS algorithms}

Among the adaptive filtering algorithms, the popular least-mean-square (LMS)\abbrev{LMS}{Least-Mean-Square} algorithm, first introduced in 1960~\cite{Widrow_lms_1960,Maloberti_history_book2016}, has been widely considered as the most used in the field. Elaborate studies of the LMS\abbrev{LMS}{Least-Mean-Square} algorithm were presented in~\cite{Widrow_adaptiveFiltering_book1985,Diniz_adaptiveFiltering_book2013}.
Also, the LMS\abbrev{LMS}{Least-Mean-Square} and its variants solve real problems including active noise control~\cite{Rupp_active_noise_control_eusipco2014}, digital equalization~\cite{Rebhi_digital_equalizer_ICTON2016}, continuous-time filter tuning~\cite{Westwick_continuous_time_filter_tuning_IEECDS2005}, system identification~\cite{Ciochina_LMS_system_identification_eusipco2016}, among others.

In the previous chapter, some adaptive filtering algorithms exploiting the sparsity in the system parameters were proposed. Also, a number of adaptive filtering algorithms exploiting the sparsity in the model coefficients has been introduced by imposing some constraints 
in the cost function~\cite{Markus_sparseSMAP_tsp2014,Candes_reweightedl1_fourier2008,Gasso_nonconvex_penalties_tsp2009,Vitor_SparsityAwareAPA_sspd2011}. This strategy relies on the attraction of some coefficient values to zero enabling the detection of nonrelevant parameters of the model.

In this chapter, we introduce the feature LMS (F-LMS)\abbrev{F-LMS}{Feature LMS} family of algorithms inducing simple sparsity properties hidden in the parameters. The type of feature to seek determines the structure of the feature matrix $\Fbf (k)$\symbl{$\Fbf(k)$}{Feature matrix} to be applied in the constraints of the F-LMS\abbrev{F-LMS}{Feature LMS} algorithm. In fact, a plethora of featured algorithms is possible to be defined by applying smart combinations of feature matrices to the coefficient vector. In this work, some simple cases are discussed whereas many more advanced solutions will be exploited in future publications. Moreover, by introducing {\it feature function}, we propose the low-complexity F-LMS (LCF-LMS) algorithm to reduce the computational complexity of the F-LMS algorithms. The LCF-LMS algorithm implements less multiplication in calculating the output signal.

The content of this chapter was partially published in~\cite{Hamed_Flms_ICASSP2018}. This chapter is organized as follows. 
Section~\ref{sec:F-LMS-chap7} proposes the F-LMS\abbrev{F-LMS}{Feature LMS} family of algorithms. 
Some examples of F-LMS\abbrev{F-LMS}{Feature LMS} algorithms for systems with lowpass and highpass spectrum are introduced in Section~\ref{sec:example_algorithms-chap7}. The LCF-LMS and the alternative LCF-LMS (ALCF-LMS) algorithms are derived in Sections~\ref{sec:low_comp_f_lms-chap7} and~\ref{sec:a-lcf-lms}, respectively. The matrix representation of the feature function is explained in Section~\ref{sec:matrix-trend-function}. 
Simulation results are presented in Section~\ref{sec:simulations-chap7} and the conclusions are drawn in Section~\ref{sec:conclusions-chap7}.


\section{The Feature LMS algorithms} \label{sec:F-LMS-chap7}

Feature LMS (F-LMS)\abbrev{F-LMS}{Feature LMS} refers to a family of LMS-type\abbrev{LMS}{Least-Mean-Square} algorithms capable of exploiting the features inherent to the unknown systems to be identified. These algorithms minimize the general objective function \symbl{${\cal P}(\cdot)$}{Sparsity-promoting penalty function}
\begin{align}
 \xi_{\text{F-LMS}} (k) = \underbrace{ \frac{1}{2}|e(k)|^2 }_{\text{standard LMS term}} + \underbrace{\alpha  {\cal P} \left( \Fbf(k) \wbf(k) \right) }_{\text{feature-inducing term}} , \label{eq:objective_function_general-chap7}
\end{align}
where $\alpha\in\mathbb{R}_+$ stands for the weight given to the {\it sparsity-promoting penalty function} ${\cal P}$, which maps a vector to the nonnegative reals $\mathbb{R}_+$, and $\Fbf(k)$ is the so-called {\it feature matrix} responsible for revealing the hidden sparsity, i.e., the result of applying $\Fbf(k)$ to $\wbf(k)$ should be a sparse vector (in the sense that most entries of the vector $\Fbf(k)\wbf(k)$ should be close or equal to zero).

The penalty function ${\cal P}$ can be any sparsity-promoting penalty function that is almost everywhere differentiable in order to allow for 
gradient-based methods.  
Examples of suitable functions are: 
(i) vector norms, especially the widely used $l_1$ norm~\cite{Candes_reweightedl1_fourier2008,Vitor_SparsityAwareAPA_sspd2011}; 
(ii) vector norms combined with shrinking strategies~\cite{Hamed_eusipco2016}; 
(iii) a function that approximates the $l_0$ norm \cite{Markus_sparseSMAP_tsp2014,Markus_apssi_icassp2013}.

The feature matrix $\Fbf(k)$ can vary at each iteration and it represents any linear combination that when applied to $\wbf(k)$ results in a sparse vector. 
In practice, $\Fbf(k)$ should be chosen based on some previous knowledge about the unknown system $\wbf_o$.  
For instance, $\wbf_o$ can represent a lowpass or a highpass filter, it can have linear phase, it can be an upsampled or downsampled signal, etc. 
All these features can be exploited by the F-LMS\abbrev{F-LMS}{Feature LMS} algorithm in order to accelerate convergence and/or achieve lower mean-squared error (MSE).\abbrev{MSE}{Mean-Squared Error}

The resulting gradient-based algorithms using the objective function given in~\eqref{eq:objective_function_general-chap7} are known as F-LMS\abbrev{F-LMS}{Feature LMS} algorithms, and their 
recursions have the general form
\begin{align}
\wbf(k+1)=\wbf(k)+\mu e(k)\xbf(k) -\mu\alpha\pbf(k), \label{eq:update_equation-chap7}
\end{align}
where $\mu \in \mathbb{R}_+$ is the step size, which should be small enough to ensure convergence~\cite{Diniz_adaptiveFiltering_book2013}, 
and $\pbf(k) \in \mathbb{R}^{N+1}$ is the gradient of function ${\cal P} \left( \Fbf(k) \wbf(k) \right)$. \symbl{$\pbf(k)$}{Gradient of ${\cal P} \left( \Fbf(k) \wbf(k) \right)$}


\section{Examples of F-LMS algorithms} \label{sec:example_algorithms-chap7}

From Section~\ref{sec:F-LMS-chap7}, it is clear  that the F-LMS\abbrev{F-LMS}{Feature LMS} family contains infinitely many algorithms. So, in this section we introduce some of these algorithms in order to illustrate how some specific features of the unknown system can be exploited. For the sake of clarity, we focus on simple algorithms and, therefore, we choose function ${\cal P}$ to be the $l_1$ norm and the feature matrix to be time-invariant $\Fbf$ so that the cost function in~\eqref{eq:objective_function_general-chap7} simplifies to  
\begin{align}
\xi_{\text{F-LMS}} (k) = \frac{1}{2}|e(k)|^2  + \alpha  \|\Fbf\wbf(k)\|_1, \label{eq:objective_function-chap7}
\end{align}
where $\|\cdot\|_1$ denotes the $l_1$-norm and for a vector $\wbf\in\mathbb{R}^{N+1}$ it is given by $\|\wbf\|_1=\sum_{i=0}^N|w_i|$. As a consequence, the reader will notice that the computational complexity of the algorithms proposed in this section is only slightly superior to the complexity of the LMS\abbrev{LMS}{Least-Mean-Square} algorithm, as the computation of $\pbf(k)$ required in~\eqref{eq:update_equation-chap7} is very simple (does not involve multiplication or division). 


\subsection{The F-LMS algorithm for lowpass systems} \label{sub:F-LMS-lowpass-chap7}

Most systems found in practice have their energy concentrated mainly in the low frequencies. If the unknown system has lowpass narrowband spectrum, then its impulse response $\wbf_o$ is smooth, meaning that the difference between adjacent coefficients is small (probably close to zero).

The adaptive filtering algorithm can take advantage of this feature present in the unknown system by selecting the feature matrix properly.  
Indeed, by selecting $\Fbf$ as $\Fbf_l$, where $\Fbf_l$ is a $N \times N+1$ matrix defined as \symbl{$\Fbf_l$}{Feature matrix for systems with lowpass narrowband spectrum}
\begin{align}
\Fbf_l=\left[\begin{array}{ccccc}1&-1&0&\cdots&0\\0&1&-1&\cdots&0\\\vdots&&\ddots&\ddots&\\0&0&\cdots&1&-1\end{array}\right], \label{eq:F_lowpass}
\end{align}
and $\|\Fbf_l\wbf(k)\|_1=\sum_{i=0}^{N-1}|w_i(k)-w_{i+1}(k)|$, 
the optimization problem in~\eqref{eq:objective_function-chap7} can be interpreted as: we seek for $\wbf(k)$ that minimizes both the squared error (LMS\abbrev{LMS}{Least-Mean-Square} term) and the distances between adjacent coefficients of $\wbf(k)$. In other words, the F-LMS\abbrev{F-LMS}{Feature LMS} algorithm for lowpass systems acts like the LMS\abbrev{LMS}{Least-Mean-Square} algorithm, but enforcing $\wbf(k)$ to be a lowpass system. It is worth mentioning that if $\wbf_o$ is indeed a lowpass system, then matrix $\Fbf_l$ yields a sparse vector $\Fbf_l\wbf(k)$.\footnote{A matrix similar to the $\Fbf_l$ in~\eqref{eq:F_lowpass} is already known by the statisticians working on a field called {\it trend filtering}~\cite{Wang_Trend_Graphs_jmlr2016}.}

Thus, the F-LMS\abbrev{F-LMS}{Feature LMS} algorithm for lowpass systems is defined by the recursion given in~\eqref{eq:update_equation-chap7}, but replacing 
vector $\pbf(k)$ with $\pbf_l(k)$ defined as
\begin{align}
\left\{\begin{array}{ll}p_{l,i}(k)={\rm sgn}(w_0(k)-w_1(k))&{\rm if~} i=0,\\
p_{l,i}(k)=-{\rm sgn}(w_{i-1}(k)-w_i(k))+{\rm sgn}(w_i(k)-w_{i+1}(k))&{\rm if~} i=1,\cdots,N-1,\\
p_{l,i}(k)=-{\rm sgn}(w_{N-1}(k)-w_{N}(k))&{\rm if~}i=N,\end{array}\right. \label{eq:p_lowpass-chap7}
\end{align}
where ${\rm sgn}(\cdot)$ denotes the sign function.

As previously explained, the F-LMS\abbrev{F-LMS}{Feature LMS} algorithm above tries to reduce the distances between consecutive coefficients of $\wbf(k)$, i.e., 
matrix $\Fbf_l$ can be understood as the process of windowing $\wbf(k)$ with a window of length $2$ (i.e., two coefficients are considered at a time). We can increase the window length, in order to make a smoothing considering more coefficients simultaneously, by nesting linear combinations as follows
\begin{align}
 \Fbf_l^{M{\rm -nested}} = \prod_{m=1}^{M} \Fbf_l^{(m)}   \Fbf_l ,
\end{align}
where $\Fbf_l^{(m)}$ has the same structure given in~\eqref{eq:F_lowpass}, but losing $m$ rows and $m$ columns in relation to the dimensions of $\Fbf_l$.

In addition to the previous examples, suppose that the unknown system is the result of upsampling a lowpass system by a factor of $L$. In this case, we should use matrix $\Fbf_l^*$, whose rows have $L-1$ zeros between the $\pm 1$ entries, in~\eqref{eq:objective_function-chap7}. For $L=2$, we have the following matrix
\begin{align}
\Fbf_l^*=\left[\begin{array}{cccccc}1&0&-1&0&\cdots&0\\0&1&0&-1&\cdots&0\\\vdots&&\ddots&\ddots&\ddots&\\0&0&\cdots&1&0&-1\end{array}\right], \label{eq:F*_lowpass-chap7}
\end{align}
and $\|\Fbf_l^*\wbf(k)\|_1=\sum_{i=0}^{N-2}|w_i(k)-w_{i+2}(k)|$.

Next the F-LMS\abbrev{F-LMS}{Feature LMS} algorithm using such $\Fbf_l^*$ has the update rule given in~\eqref{eq:update_equation-chap7}, but replacing $\pbf(k)$ with $\pbf_l^*(k)$ defined as
\begin{align}
\left\{\begin{array}{ll}p_{l,i}^*(k)={\rm sgn}(w_i(k)-w_{i+2}(k))&{\rm if~} i=0,1,\\
p_{l,i}^*(k)=-{\rm sgn}(w_{i-2}(k)-w_i(k))+{\rm sgn}(w_i(k)-w_{i+2}(k))&{\rm if~} i=2,\cdots,N-2,\\
p_{l,i}^*(k)=-{\rm sgn}(w_{i-2}(k)-w_{i}(k))&{\rm if~}i=N-1,N.\end{array}\right.  \label{eq:p*_lowpass-chap7} 
\end{align}


\subsection{The F-LMS algorithm for highpass systems} \label{sub:F-LMS-highpass-chap7}

If the unknown system $\wbf_o$ has a highpass narrowband spectrum, then adjacent coefficients tend to have similar absolute values, but with opposite signs. 
Therefore, the sum of two consecutive coefficients is close to zero and we can exploit this feature in the learning process by minimizing the sum of 
adjacent coefficients of $\wbf(k)$. 
This can be accomplished by selecting $\Fbf$ as $\Fbf_h$, where $\Fbf_h$ is an $N \times N+1$ feature matrix defined as \symbl{$\Fbf_h$}{Feature matrix for systems with highpass narrowband spectrum}
\begin{align}
\Fbf_h=\left[\begin{array}{ccccc}1&1&0&\cdots&0\\0&1&1&\cdots&0\\\vdots&&\ddots&\ddots&\\0&0&\cdots&1&1\end{array}\right], \label{eq:D_highpass-chap7}
\end{align}
such that $\|\Fbf_h\wbf(k)\|_1=\sum_{i=0}^{N-1}|w_i(k)+w_{i+1}(k)|$.

The F-LMS\abbrev{F-LMS}{Feature LMS} algorithm for highpass systems is characterized by the recursion given in~\eqref{eq:update_equation-chap7}, but replacing $\pbf(k)$
with $\pbf_h(k)$, which is defined as
\begin{align}
\left\{\begin{array}{ll}p_{h,i}(k)={\rm sgn}(w_0(k)+w_1(k))&{\rm if~} i=0,\\
p_{h,i}(k)={\rm sgn}(w_{i-1}(k)+w_i(k))+{\rm sgn}(w_i(k)+w_{i+1}(k))&{\rm if~} i=1,\cdots,N-1,\\
p_{h,i}(k)={\rm sgn}(w_{N-1}(k)+w_{N}(k))&{\rm if~}i=N.\end{array}\right. \label{eq:p_highpass-chap7}
\end{align}

Similar to the lowpass case, let us consider that the unknown system is the result of interpolating a highpass system by a factor $L=2$.  
The set of interpolated highpass systems leads to a notch filter with zeros at $z=\pm \jmath$. 
In this case, we can utilize $\Fbf_h^*$ in the objective function~\eqref{eq:objective_function-chap7}, where $\Fbf_h^*$ is described by
\begin{align}
\Fbf_h^*=\left[\begin{array}{cccccc}1&0&1&0&\cdots&0\\0&1&0&1&\cdots&0\\\vdots&&\ddots&\ddots&\ddots&\\0&0&\cdots&1&0&1\end{array}\right], \label{eq:F*_highpass-chap7}
\end{align}
and $\|\Fbf_h^*\wbf(k)\|_1=\sum_{i=0}^{N-2}|w_i(k)+w_{i+2}(k)|$.

Using $\Fbf_h^*$, the F-LMS\abbrev{F-LMS}{Feature LMS} recursion in~\eqref{eq:update_equation-chap7} should substitute $\pbf(k)$ by $\pbf_h^*(k)$ defined as
\begin{align}
\left\{\begin{array}{ll}p_{h,i}^*(k)={\rm sgn}(w_i(k)+w_{i+2}(k))&{\rm if~} i=0,1,\\
p_{h,i}^*(k)={\rm sgn}(w_{i-2}(k)+w_i(k))+{\rm sgn}(w_i(k)+w_{i+2}(k))&{\rm if~} i=2,\cdots,N-2,\\
p_{h,i}^*(k)={\rm sgn}(w_{i-2}(k)+w_{i}(k))&{\rm if~} i=N-1,N.\end{array}\right. \label{eq:p*_highpass-chap7}
\end{align} 


\section{Low-complexity F-LMS Algorithms} \label{sec:low_comp_f_lms-chap7}

In this section, we derive the low-complexity feature LMS (LCF-LMS)\abbrev{LCF-LMS}{Low-Complexity Feature LMS} algorithm to exploit sparsity in the linear combination of the parameters, as the F-LMS algorithms do, while also reducing the computational cost of calculating the output signal. 

Here, the idea is to reduce the number of multiplications required for computing the output signal when there is a strong relation between neighboring coefficients. In systems with lowpass frequency content, for example, neighboring coefficients vary smoothly. Therefore, when the input signal is highly correlated, we can fix the value of the neighboring coefficients where the distances (the absolute value of their differences) between any two consecutive coefficients are less than a small constant $\epsilon>0$. As a result, we reduce the number of multiplications in the calculation of $y(k)\triangleq\wbf^T(k)\xbf(k)$. For instance, if for nonnegative integers $m$ and $j$, where $m,j<N$, the discrepancies between the coefficients with indexes $m$ to $m+j$ are less than $\epsilon$, then we can use the $m$th coefficient as a reference. Mathematically, if the value of $|w_{m+i+1}(k)-w_{m+i}(k)|\leq\epsilon$ for $i=0,1,2,\cdots,j-1$, then in the calculation of the output signal instead of computing
\begin{align}
y(k)=w_m(k)x_m(k)+\cdots+w_{m+j}(k)x_{m+j}(k),
\end{align}
we can approximate $y(k)$ as
\begin{align}
\hat{y}(k)\triangleq\underbrace{w_m(k)x_m(k)+\cdots+w_m(k)x_m(k)}_{(j+1)-{\rm times}}. \label{eq:output_approx}
\end{align}
As a result, we decrease the number of multiplications from $j+1$ to one. Hence, for a block of coefficients in which the distance between any two consecutive coefficients is less than $\epsilon$, we can use the first parameter of the block as the reference parameter. As soon as the distance between two consecutive coefficients becomes greater than $\epsilon$, we will use the new one as a reference for the new block of coefficients.

To this end, for each block of coefficients in which the distance of any two consecutive coefficients is less than $\epsilon$, we have to preserve the first coefficient of the block, and the rest of them will be replaced by zero. Furthermore, when the absolute value of a coefficient is less than $\epsilon$, we can replace it with zero to avoid additional multiplication~\cite{Hamed_eusipco2016,Hamed_S_RLS_ICASSP2017}. Therefore, two subsets of parameters will be replaced by zero: (I) the coefficients whose absolute values are less than $\epsilon$, and (II) the coefficients whose distances from their antecessor are less than $\epsilon$.

The above reasoning can be implemented by means of the {\it feature function}, $\mathbb{F}_\epsilon:\mathbb{R}^{N+1}\rightarrow\mathbb{R}^{N+1}$, \symbl{$\mathbb{F}_\epsilon$}{Feature function} applied to the weigh vector of the adaptive filter. The $i$th element of the feature function, for $i=0,1,\cdots,N$, is defined as
\begin{align}
\mathbb{F}_{\epsilon,i}(\wbf(k))\triangleq\left\{\begin{array}{ll}f_\epsilon(w_0(k))&{\rm if~}i=0,\\
f_\epsilon(w_i(k))&{\rm if~}|w_i(k)-w_{i-1}(k)|>\epsilon~\&~i\neq0,\\
0&{\rm if~}|w_i(k)-w_{i-1}(k)|\leq\epsilon~\&~i\neq0, \end{array}\right. \label{eq:trend_function}
\end{align}
where $f_\epsilon$ is the discard function defined in~\eqref{eq:f_epsilon-eusipco}. As can be observed, the feature function replaces the subsets (I) and (II) of the coefficients of $\wbf(k)$ with zero. Let us define $\wbf_s(k)\triangleq\mathbb{F}_\epsilon(\wbf(k))$. Figure~\ref{fig:stem_explain} shows an example for the impulse response of $\wbf(k)$ and $\wbf_s(k)$ when $\epsilon=0.02$. As can be observed, $\wbf(k)$ has fifteen nonzero coefficients, and after using the feature function twelve of them are replaced by zero.

\begin{figure}[t!]
\centering
\subfigure[b][]{\includegraphics[width=.48\linewidth,height=7cm]{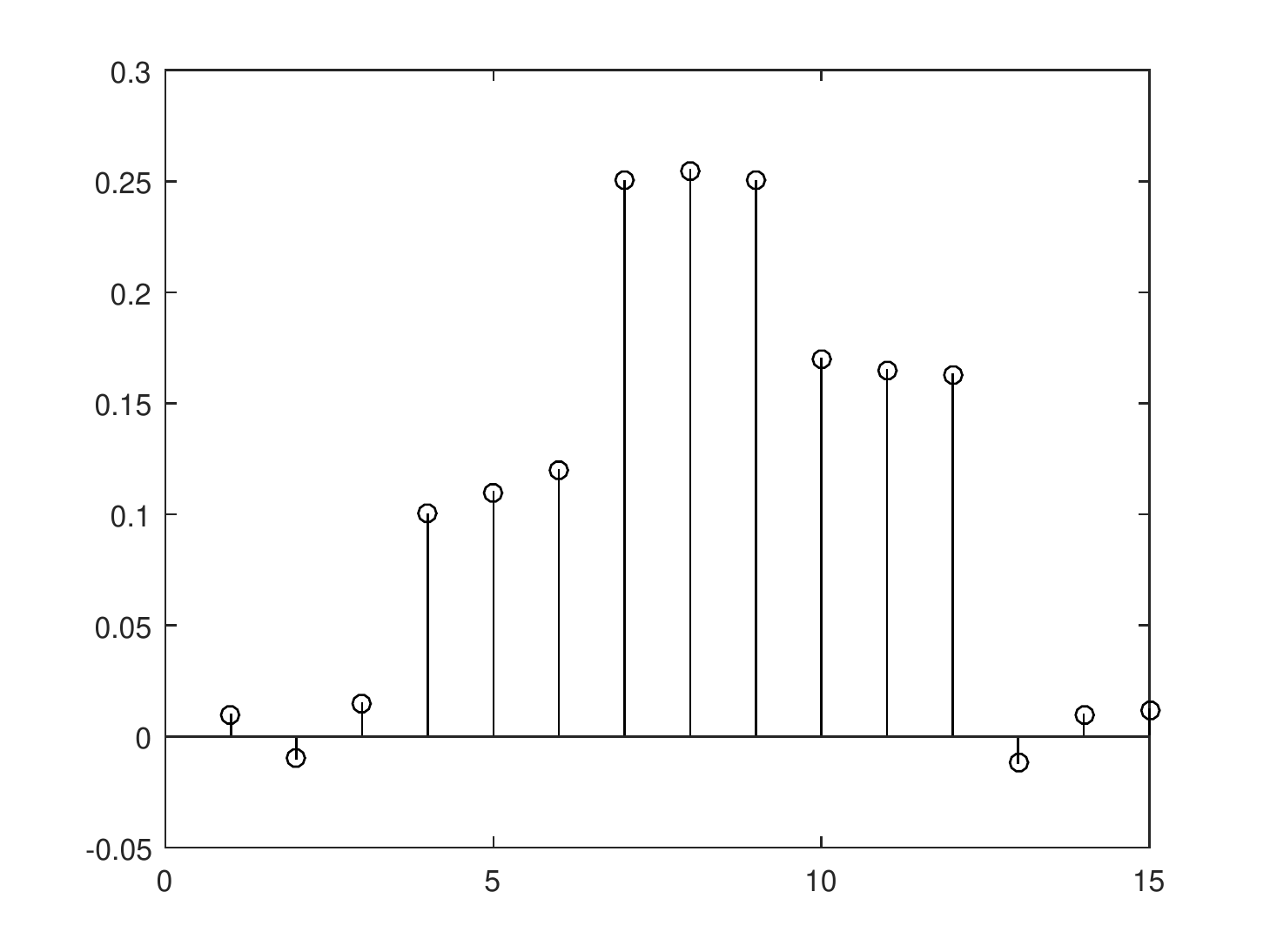}
\label{fig:stem_explain_prio}}
\subfigure[b][]{\includegraphics[width=.48\linewidth,height=7cm]{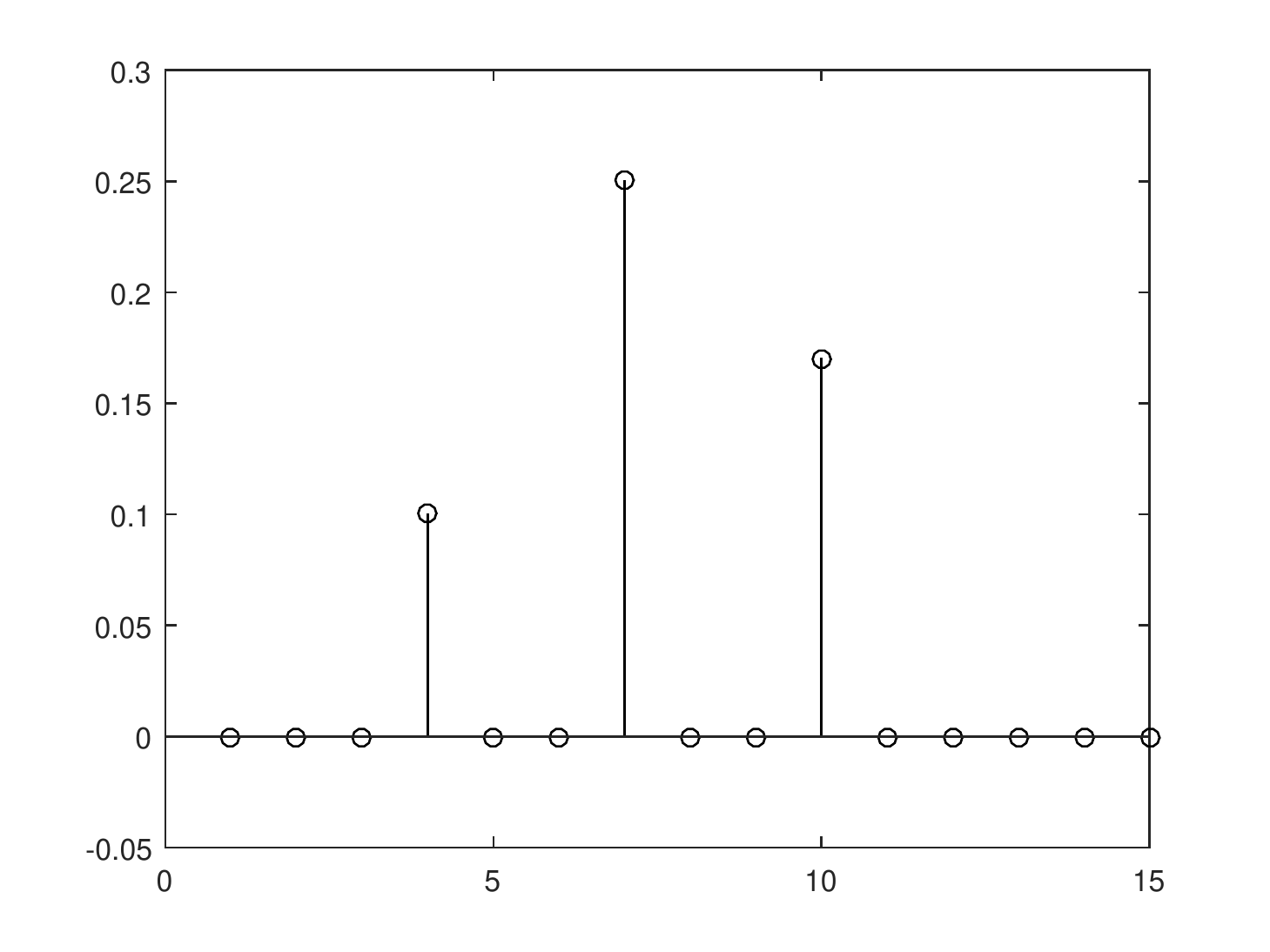}
\label{fig:stem_explain_pos}}
\caption{The impulse response of (a) $\wbf(k)$; (b) $\wbf_s(k)=\mathbb{F}_\epsilon(\wbf(k))$ for $\epsilon=0.02$. \label{fig:stem_explain}}
\end{figure}

Our goal is to utilize $\wbf_s(k)=\mathbb{F}_\epsilon(\wbf(k))$ in the calculation of the output signal. However, we must determine from which subset of coefficients of $\wbf(k)$ the zero elements of $\wbf_s(k)$ came, i.e., subsets (I) or (II). In fact, for some $i$, $w_{s_i}(k)$ is zero if and only if $w_i(k)$ belongs to the subsets (I) or (II). If $w_i(k)$ belongs to the subset (I), then we can directly apply $w_{s_i}(k)$ to calculate the output signal, i.e., we use $w_{s_i}(k)x_i(k)=0$. However, if $w_i(k)$ belongs to the subset (II), then we must apply the last nonzero coefficient of $\wbf_s(k)$ before the $i$th index to compute the output signal. Assume that this nonzero coefficient has index $m$, then we use $w_{s_m}(k)$ instead of $w_i(k)$ since their values are close to each other. Hence, in the calculation of the output signal, we use $w_{s_m}(k)x_m(k)$ instead of $w_{s_i}(k)x_i(k)$. 

In order to determine the background of the zero coefficients in $\wbf_s(k)$, we define a binary vector $\bbf(k)\in\{0,1\}^{N+1}$ as $\bbf(k)=\fbf_\epsilon(\wbf(k))$, where $\fbf_\epsilon$ is the discard vector function. Then, for some $i$, if $w_{s_i}(k)$ and $b_i(k)$ are zero, we infer that $w_i(k)$ belongs to the subset (I). However, if $w_{s_i}(k)=0$ and $b_i(k)=1$, then we conclude that $w_i(k)$ belongs to the subset (II).

Finally, we can present the LCF-LMS\abbrev{LCF-LMS}{Low-Complexity Feature LMS} algorithm in Table~\ref{tb:LCF-LMS}. This algorithm implements less multiplication as compared to the LMS algorithm.

\begin{table}[t!]
\caption{Low-complexity feature LMS algorithm}
\begin{center}
\begin{footnotesize}
\begin {tabular}{|l|} \hline\\ \hspace{0.7cm}{\bf LCF-LMS Algorithm}\\ \\
\hline\\
Initialization
\\
$\wbf_s(0)=\bbf(0)=\wbf(0)=[0~\cdots~0]^T$\\
choose $\mu$ in the range $0<\mu\ll 1$\\
choose small constant $\epsilon>0$\\
Do for $k\geq0$\\
\hspace*{0.15cm} ${\rm temp}=0$, $y(k)=0$\\
\hspace*{0.15cm} for $i=0$ to $N$\\
\hspace*{0.3cm} if $w_{s_i}(k)\neq0$\\
\hspace*{0.45cm} ${\rm temp}=w_{s_i}(k)x_i(k)$\\
\hspace*{0.45cm} $y(k)=y(k)+{\rm temp}$\\
\hspace*{0.3cm} else\\
\hspace*{0.45cm} $y(k)=y(k)+({\rm temp}\times b_i(k))$\\
\hspace*{0.3cm} end\\
\hspace*{0.15cm} end\\
\hspace*{0.15cm} $e(k)=d(k)-y(k)$\\
\hspace*{0.15cm} $\wbf(k+1)=\wbf(k)+\mu e(k)\xbf(k)$\\
\hspace*{0.15cm} $\wbf_s(k+1)=\mathbb{F}_\epsilon(\wbf(k+1))$\\
\hspace*{0.15cm} $\bbf(k+1)=\fbf_\epsilon(\wbf(k+1))$\\
end  \\
\\
\hline
\end {tabular}
\end{footnotesize}
\end{center}
\label{tb:LCF-LMS}
\end{table}

As mentioned earlier, for proposing the LCF-LMS algorithm, we assumed that the input signal is highly correlated. This assumption restricts the use of the LCF-LMS algorithm. To avoid this assumption, instead of approximating $y(k)$ by~\eqref{eq:output_approx}, we can approximate $y(k)$ as
\begin{align}
\hat{y}(k)\triangleq w_m(k)(x_m(k)+x_{m+1}(k)+\cdots+x_{m+j}(k)). \label{eq:output_approx_modified}
\end{align}
In other words, when $w_m(k)$ represents a block of coefficients of length $j+1$, the LCF-LMS algorithm sums $j+1$ copies of $w_m(k)x_m(k)$; however, in Equation~\eqref{eq:output_approx_modified}, we multiply $w_m(k)$  by the sum of the input signal components corresponding to the coefficients represented by $w_m(k)$. Note that the number of required arithmetic operations in~\eqref{eq:output_approx_modified} and~\eqref{eq:output_approx} are identical; i.e., both equations implement one multiplication and $j$ additions. The algorithm using Equation~\eqref{eq:output_approx_modified} in calculating output signal is called the improved LCF-LMS (I-LCF-LMS) \abbrev{I-LCF-LMS}{Improved LCF-LMS} algorithm, and its application is not limited to cases with correlated input signals. The I-LCF-LMS algorithm is presented in Table~\ref{tb:ILCF-LMS}.

\begin{table}[t!]
\caption{Improved low-complexity feature LMS algorithm}
\begin{center}
\begin{footnotesize}
\begin {tabular}{|l|} \hline\\ \hspace{0.8cm}{\bf I-LCF-LMS Algorithm}\\ \\
\hline\\
Initialization
\\
$\wbf_s(0)=\bbf(0)=\wbf(0)=[0~\cdots~0]^T$\\
choose $\mu$ in the range $0<\mu\ll 1$\\
choose small constant $\epsilon>0$\\
Do for $k\geq0$\\
\hspace*{0.15cm} ${\rm temp}_x=0$, ${\rm temp}_w=0$, $y(k)=0$\\
\hspace*{0.15cm} for $i=0$ to $N$\\
\hspace*{0.3cm} if $w_{s_i}(k)\neq0$\\
\hspace*{0.45cm} $y(k)=y(k)+({\rm temp}_w\times{\rm temp}_x)$\\
\hspace*{0.45cm} ${\rm temp}_w=w_{s_i}(k)$\\
\hspace*{0.45cm} ${\rm temp}_x=x_i(k)$\\
\hspace*{0.3cm} else\\
\hspace*{0.45cm} ${\rm temp}_x={\rm temp}_x+(x_i(k)\times b_i(k))$\\
\hspace*{0.3cm} end\\
\hspace*{0.15cm} end\\
\hspace*{0.15cm} $y(k)=y(k)+({\rm temp}_w\times{\rm temp}_x)$\\
\hspace*{0.15cm} $e(k)=d(k)-y(k)$\\
\hspace*{0.15cm} $\wbf(k+1)=\wbf(k)+\mu e(k)\xbf(k)$\\
\hspace*{0.15cm} $\wbf_s(k+1)=\mathbb{F}_\epsilon(\wbf(k+1))$\\
\hspace*{0.15cm} $\bbf(k+1)=\fbf_\epsilon(\wbf(k+1))$\\
end  \\
\\
\hline
\end {tabular}
\end{footnotesize}
\end{center}
\label{tb:ILCF-LMS}
\end{table}


\section{Alternative LCF-LMS Algorithm} \label{sec:a-lcf-lms}

In the LCF-LMS\abbrev{LCF-LMS}{Low-Complexity Feature LMS} algorithm, when $\wbf(k)$ contains a long sequence of coefficients with almost similar absolute values, then $\wbf_s(k)$ contains a long block of zeros. Therefore, when calculating the output signal, all parameters of this block are represented by the first element of the block. As a result, since we are using a fixed coefficient to represent many ones, we could have an accumulated error. In this section, we introduce the alternative LCF-LMS (ALCF-LMS)\abbrev{ALCF-LMS}{Alternative Low-Complexity Feature LMS} algorithm to address this problem.

To avoid accumulated error because of many adjacent zeros in $\wbf_s(k)$, for some natural number $p<N$, we can force the feature function to keep every $p$ coefficients of $\wbf(k)$ in $\wbf_s(k)$ if the absolute value of the coefficient is greater than $\epsilon$. In other words, no parameter can represent a block of coefficients with more than $p$ elements. The only exception is the case when the parameters of the block have absolute values smaller than $\epsilon$ (i.e., they are really close to zero; therefore, they must be replaced by zero). Let us denote by $\mathbb{F}^a_\epsilon:\mathbb{R}^{N+1}\rightarrow\mathbb{R}^{N+1}$ \symbl{$\mathbb{F}^a_\epsilon$}{Alternative feature function} the new feature function, and it is called the {\it alternative feature function}. The $i$th element of $\mathbb{F}^a_\epsilon$, for $i=0,1,\cdots,N$, is defined by
\begin{align}
\mathbb{F}^a_{\epsilon,i}(\wbf(k))\triangleq\left\{\begin{array}{ll}f_\epsilon(w_i(k))&{\rm if~mod}(i,p)=0,\\
f_\epsilon(w_i(k))&{\rm if~}|w_i(k)-w_{i-1}(k)|>\epsilon~\&~{\rm mod}(i,p)\neq0,\\
0&{\rm if~}|w_i(k)-w_{i-1}(k)|\leq\epsilon~\&~{\rm mod}(i,p)\neq0, \end{array}\right. \label{eq:a_trend_function}
\end{align}
where ${\rm mod}(i,p)$ stands for the remainder of $\frac{i}{p}$. Therefore, the ALCF-LMS\abbrev{ALCF-LMS}{Alternative Low-Complexity Feature LMS} algorithm is similar to the LCF-LMS\abbrev{LCF-LMS}{Low-Complexity Feature LMS} one in Table~\ref{tb:LCF-LMS}, but the feature function is replaced by the alternative feature function (i.e., $\wbf_s(k+1)=\mathbb{F}^a_\epsilon(\wbf(k+1))$).

By using the same argument, we can propose the alternative I-LCF-LMS (AI-LCF-LMS) \abbrev{AI-LCF-LMS}{Alternative I-LCF-LMS} algorithm. Indeed, if we replace the feature function in Table~\ref{tb:ILCF-LMS} with the alternative feature function, then we obtain the AI-LCF-LMS algorithm. 


\section{Matrix Representation of the Feature Function} \label{sec:matrix-trend-function}

In this section, we show how to generate $\wbf_s(k)$ through matrix operations. Indeed, presenting $\wbf_s(k)$ through matrix operations is helpful for future mathematical analysis.

To generate $\wbf_s(k)$, we use quantization matrices $\Qbf_t(k)$ for $t=1,2,3$, and two feature matrices $\Fbf_1$ and $\Fbf_2(k)$, all matrices belong to $\mathbb{R}^{(N+1)\times(N+1)}$. The matrices $\Fbf_1$ and $\Fbf_2(k)$ are responsible for exploiting the sparsity in the linear combination of the parameters and reconstructing the weight vector after exploiting the sparsity, respectively. Therefore, to exploit the hidden sparsity in the parameters of $\wbf(k)$ and their linear combinations, we introduce $\wbf_s(k)$ as follows
\begin{align}
\wbf_s(k)\triangleq\Qbf_3(k)\Fbf_2(k)\Qbf_2(k)\Fbf_1\Qbf_1(k)\wbf(k). \label{eq:exploit_sparsity-chap7}
\end{align}

In the following, we describe the matrices and justify their actions. We define the quantization matrix $\Qbf_1(k)$ as the Jacobian matrix of $\fbf_\epsilon(\wbf(k))$. Therefore, $\Qbf_1(k)$ is a diagonal matrix whose entries are zero or one. For the coefficients of $\wbf(k)$ where their absolute values are less than $\epsilon$, the corresponding entries on the diagonal of $\Qbf_1(k)$ are zero, otherwise they are one. Similarly, the matrices $\Qbf_2(k)$ and $\Qbf_3(k)$ are defined as the Jacobian matrices of $\fbf_\epsilon(\Fbf_1\Qbf_1(k)\wbf(k))$ and $\fbf_\epsilon(\Fbf_2(k)\Qbf_2(k)\Fbf_1\Qbf_1(k)\wbf(k))$, respectively. Thus $\Qbf_2(k)$ is a diagonal matrix with zero and one. Its diagonal entries are zero (one) for the corresponding elements of $\Fbf_1\Qbf_1(k)\wbf(k)$ with the absolute value lower (greater) than $\epsilon$. Also, $\Qbf_3(k)$ is a diagonal matrix similar to $\Qbf_2(k)$; however, it is derived from the vector $\Fbf_2(k)\Qbf_2(k)\Fbf_1\Qbf_1(k)\wbf(k)$. The diagonal entries of $\Qbf_3(k)$ are one for the corresponding elements of $\Fbf_2(k)\Qbf_2(k)\Fbf_1\Qbf_1(k)\wbf(k)$ with absolute value greater than $\epsilon$, and zero for the others.

The feature matrix $\Fbf_1$ has to find the difference between the coefficients of the vector $\Qbf_1(k)\wbf(k)$. In fact, it keeps the first parameter unchanged, and for other coefficients replaces them with the differences between them and the previous one. Thus, it can be represented as
\begin{align}
\Fbf_1\triangleq\left[\begin{array}{cccccc}1&0&0&0&\cdots&0\\-1&1&0&0&\cdots&0\\0&-1&1&0&\cdots&0\\\vdots&0&\ddots&\ddots&0&\vdots\\0&\cdots&0&-1&1&0\\0&0&\cdots&0&-1&1\end{array}\right]. \label{eq:F_matrix}
\end{align}

The function of the feature matrix $\Fbf_2(k)$ is to reconstruct the weight vector from the vector $\rbf(k)\triangleq\Qbf_2(k)\Fbf_1\Qbf_1(k)\wbf(k)$. The structure of $\Fbf_2(k)$ is a little complicated. In the following steps, we explain how to construct $\Fbf_2(k)$:
\begin{enumerate}
\item Assume that the first nonzero element of $\rbf(k)$ is $r_{i_1}(k)$, thus all rows of $\Fbf_2(k)$ before the $i_1$th row are zero vectors.
\item For $i_1$th row, the element corresponding to the $r_{i_1}(k)$ is one, and other entries of this row are zero.
\item If the next element of $\rbf(k)$ is nonzero, then the next row of $\Fbf_2(k)$ contains one more nonzero entry equal to one corresponding to these nonzero coefficients of $\rbf(k)$. We repeat this step as far as a zero element appears in $\rbf(k)$.
\item As soon as a zero element appears in $\rbf(k)$, we look for the next nonzero element, and assume that it is $r_{i_2}(k)$. Then the next row of $\Fbf_2(k)$ is similar to the previous row, but the element corresponding to $r_{i_2}(k)$ must be equal to one.
\item Suppose that the first nonzero element of $\rbf(k)$ after $r_{i_2}(k)$ is $r_{i_3}(k)$. Then next rows of $\Fbf_2(k)$ until the $(i_3-1)$th row are identical to the last constructed row. Note that if it does not exist some nonzero element as $r_{i_3}(k)$, the remaining rows of $\Fbf_2(k)$ are identical to the last constructed row.
\item The $i_3$th row of $\Fbf_2(k)$ contains only one nonzero element equal to one, and it must be placed on column $i_3$. This row is similar to the $i_1$th row (step 2); however, the position of one is different. Now, we go back to the step 3 and repeat the same process to construct the next rows of $\Fbf_2(k)$.
\end{enumerate}
      
In Equation~\eqref{eq:exploit_sparsity-chap7}, the matrix $\Qbf_1(k)$ replaces the coefficients of $\wbf(k)$ which has absolute value lower than $\epsilon$ with zero. Then matrix $\Fbf_1$ keeps the first coefficient unchanged. For the other components, this matrix subtracts the previous component from each of them. Hence, for the resulting vector, the matrix $\Qbf_2(k)$ changes the elements with an absolute value lower than $\epsilon$ to zero. Afterwards, the matrix $\Fbf_2(k)$ reconstructs the weight vector and, in some sense, it inverts the effect of $\Fbf_1$. Finally, for the resulting vector, the matrix $\Qbf_3(k)$ replaces the coefficients inside $[-\epsilon,\epsilon]$ with zero. The final result is identical to $\mathbb{F}_\epsilon(\wbf(k))$.

To clarify the process above, we describe the details for $\wbf(k)=[0~0.5~0.51~0.01~0.6~0.7~0.8~0.81~0~-0.01]^T$, as an example, when $\epsilon=0.02$. $\Qbf_1(k)$ is a diagonal matrix, where its diagonal is
$[0~1~1~0~1~1~1~1~0~0]^T$. Therefore, $\Qbf_1(k)\wbf(k)=[0~0.5~0.51~0~0.6~0.7~0.8~0.81~0~0]^T$. Then $\Fbf_1\Qbf_1(k)\wbf(k)=[0~0.5~0.01~-0.51~0.6~0.1~0.1~0.01~-0.81~0]^T$. The diagonal of $\Qbf_2(k)$ is $[0~1~0~1~1~1~1~0~1~0]^T$, and $\Qbf_2(k)\Fbf_1\Qbf_1(k)\wbf(k)=[0~0.5~0~-0.51~0.6~0.1~0.1~0~-0.81~0]^T$. Following the procedure explained to construct $\Fbf_2(k)$, we obtain the matrix $\Fbf_2(k)$ as follows
\begin{align}
\Fbf_2(k)=\left[\begin{array}{cccccccccc}0&0&0&0&0&0&0&0&0&0\\
0&1&0&0&0&0&0&0&0&0\\
0&1&0&1&0&0&0&0&0&0\\
0&1&0&1&0&0&0&0&0&0\\
0&0&0&0&1&0&0&0&0&0\\
0&0&0&0&1&1&0&0&0&0\\
0&0&0&0&1&1&1&0&0&0\\
0&0&0&0&1&1&1&0&1&0\\
0&0&0&0&1&1&1&0&1&0\\
0&0&0&0&1&1&1&0&1&0\end{array}\right].
\end{align}
Then $\Fbf_2(k)\Qbf_2(k)\Fbf_1\Qbf_1(k)\wbf(k)=[0~0.5~-0.01~-0.01~0.6~0.7~0.8~-0.01~-0.01~-0.01]^T$. The diagonal of $\Qbf_3(k)$ is $[0~1~0~0~1~1~1~0~0~0]^T$. Hence, $\wbf_s(k)=\Qbf_3(k)\Fbf_2(k)\Qbf_2(k)\Fbf_1\Qbf_1(k)\wbf(k)=[0~0.5~0~0~0.6~0.7~0.8~0~0~0]^T$. Also, if we use the feature function with $\epsilon=0.02$, then we obtain $\wbf_s(k)=\mathbb{F}_\epsilon(\wbf(k))=[0~0.5~0~0~0.6~0.7~0.8~0~0~0]^T$.


\section{Simulations} \label{sec:simulations-chap7}

In this section, we apply the LMS, the F-LMS, the LCF-LMS, and the ALCF-LMS algorithms to system identification problems. In scenario 1, we utilize the LMS and the F-LMS algorithms. Then, in scenario 2, we use the LMS, the LCF-LMS, and the ALCF-LMS algorithms.

In both scenarios, the order of all the unknown systems is 39, i.e., they have 40 coefficients. The signal-to-noise ratio (SNR)\abbrev{SNR}{Signal-to-Noise Ratio} is chosen as 20 dB. For all algorithms, the initial vector is $\wbf(0) = [0~\cdots~0]^T$, and the MSE\abbrev{MSE}{Mean-Squared Error} learning curves are computed by averaging the outcomes of 200 independent trials.


\subsection{Scenario 1} 

In this scenario, we apply the LMS\abbrev{LMS}{Least-Mean-Square} and the F-LMS\abbrev{F-LMS}{Feature LMS} algorithms to identify some unknown lowpass and highpass systems. The first example considers predominantly lowpass and highpass systems defined as $\wbf_{o,l}= [0.4,\cdots,0.4]^T$ and $\wbf_{o,h}= [0.4,-0.4,0.4,\cdots,-0.4]^T$, respectively. The second example uses the interpolated models $\wbf_{o,l}'=[0.4,0,0.4,\cdots,0,0.4,0]^T$ and 
$\wbf_{o,h}'=[0.4,0,-0.4,0,0.4,\cdots,0]^T$. The third example uses block-sparse lowpass and block-sparse highpass models, $\wbf_{o,l}''$ and $\wbf_{o,h}''$, whose entries are defined in~\eqref{eq:second_wo_lowpass-chap7} 
and~\eqref{eq:second_wo_highpass-chap7}, respectively.
\begin{align}
w_{o,l_i}''&=\left\{\begin{array}{ll}0 & {\rm if~}0 \leq i \leq 9,\\
0.05(i-9) & {\rm if~} 10\leq i \leq 14,\\
0.3 & {\rm if~} 15 \leq i \leq 24,\\
0.3-0.05(i-24) & {\rm if~} 25 \leq i \leq 29,\\
0 & {\rm if~} 30 \leq i \leq 39,\end{array}\right. \label{eq:second_wo_lowpass-chap7}\\
w_{o,h_i}''&=(-1)^{i+1}w_{o,l_i}''. \label{eq:second_wo_highpass-chap7}
\end{align}

The input signal is a zero-mean white Gaussian noise with unit variance. The value of $\alpha$ for the F-LMS algorithm is chosen as 0.05. The values of the step size $\mu$ are informed later for each simulated scenario. The MSE\abbrev{MSE}{Mean-Squared Error} learning curves of the LMS\abbrev{LMS}{Least-Mean-Square} and the F-LMS\abbrev{F-LMS}{Feature LMS} algorithms are depicted in Figures~\ref{fig:LP-chap7} to~\ref{fig:Block-chap7}.

\begin{figure}[t!]
\centering
\subfigure[b][]{\includegraphics[width=.48\linewidth,height=7cm]{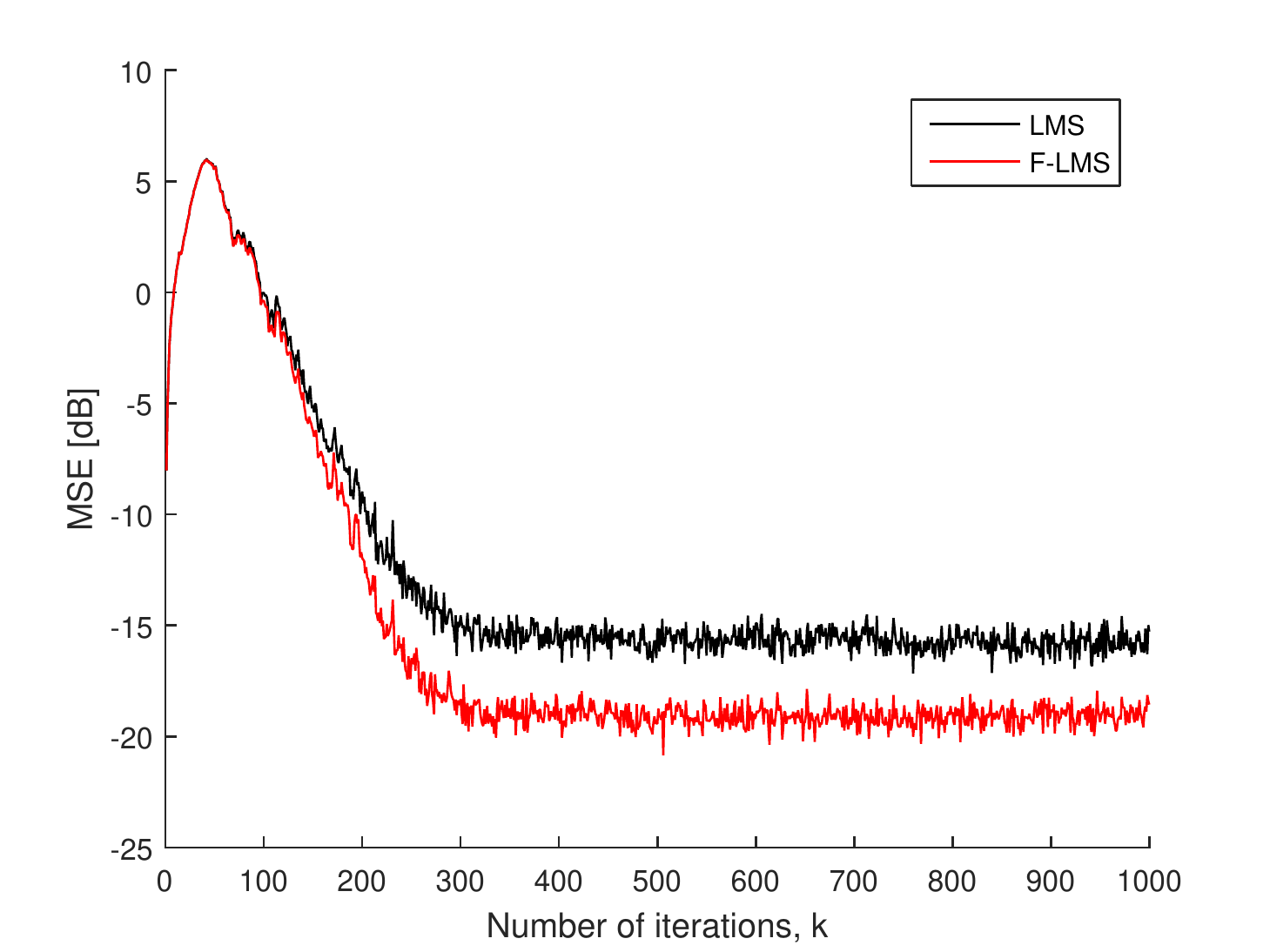}
\label{fig:LP_same_mu-chap7}}
\subfigure[b][]{\includegraphics[width=.48\linewidth,height=7cm]{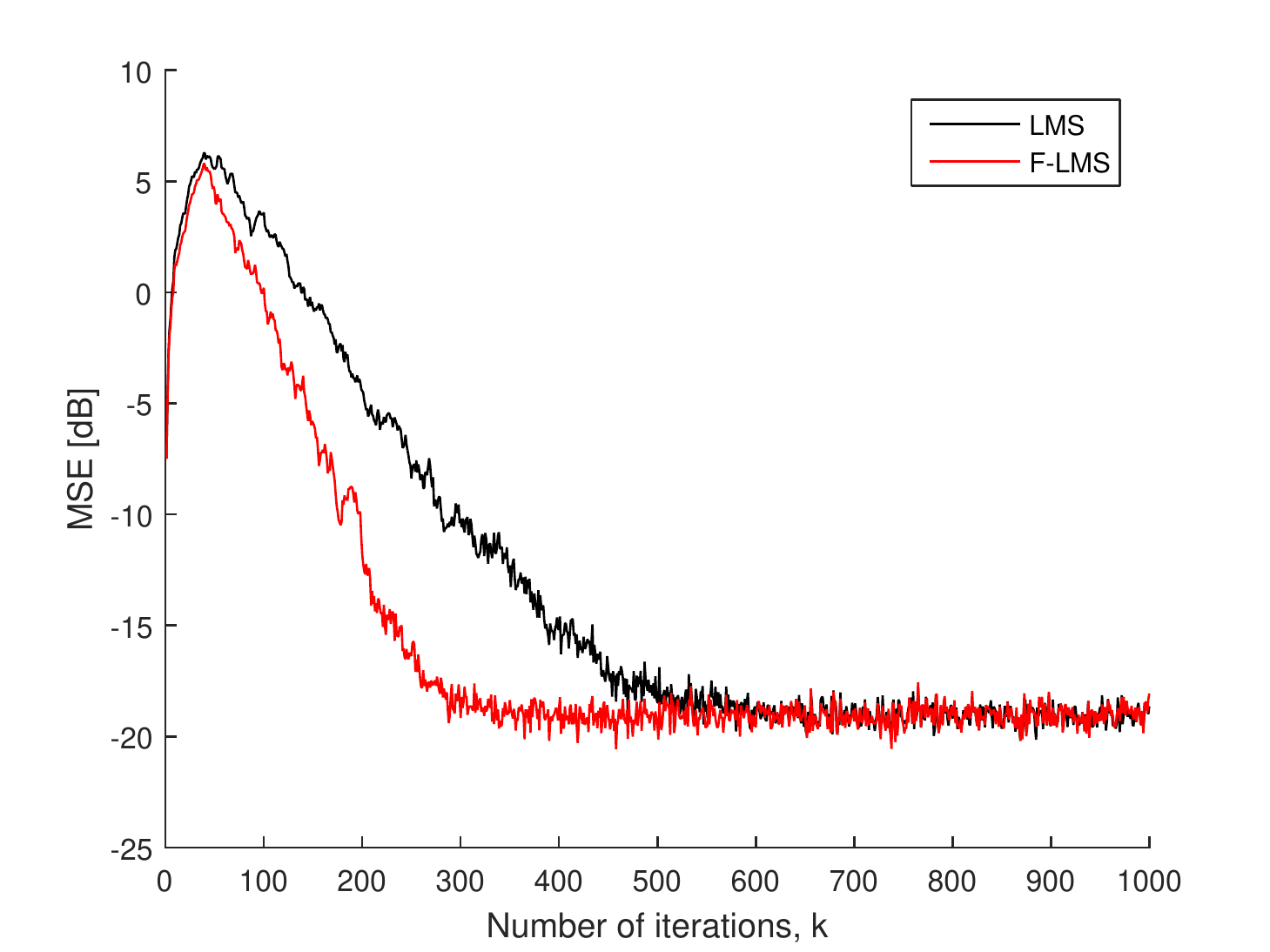}
\label{fig:LP_diff_mu-chap7}}
\caption{MSE learning curves of the LMS and F-LMS algorithms considering $\wbf_{o,l}$: 
(a) both algorithms with the same step size: $\mu = 0.03$; (b) LMS and F-LMS with step sizes equal to 0.01 and 0.03, respectively. \label{fig:LP-chap7}}
\end{figure}

\begin{figure}[t!]
\centering
\subfigure[b][]{\includegraphics[width=.48\linewidth,height=7cm]{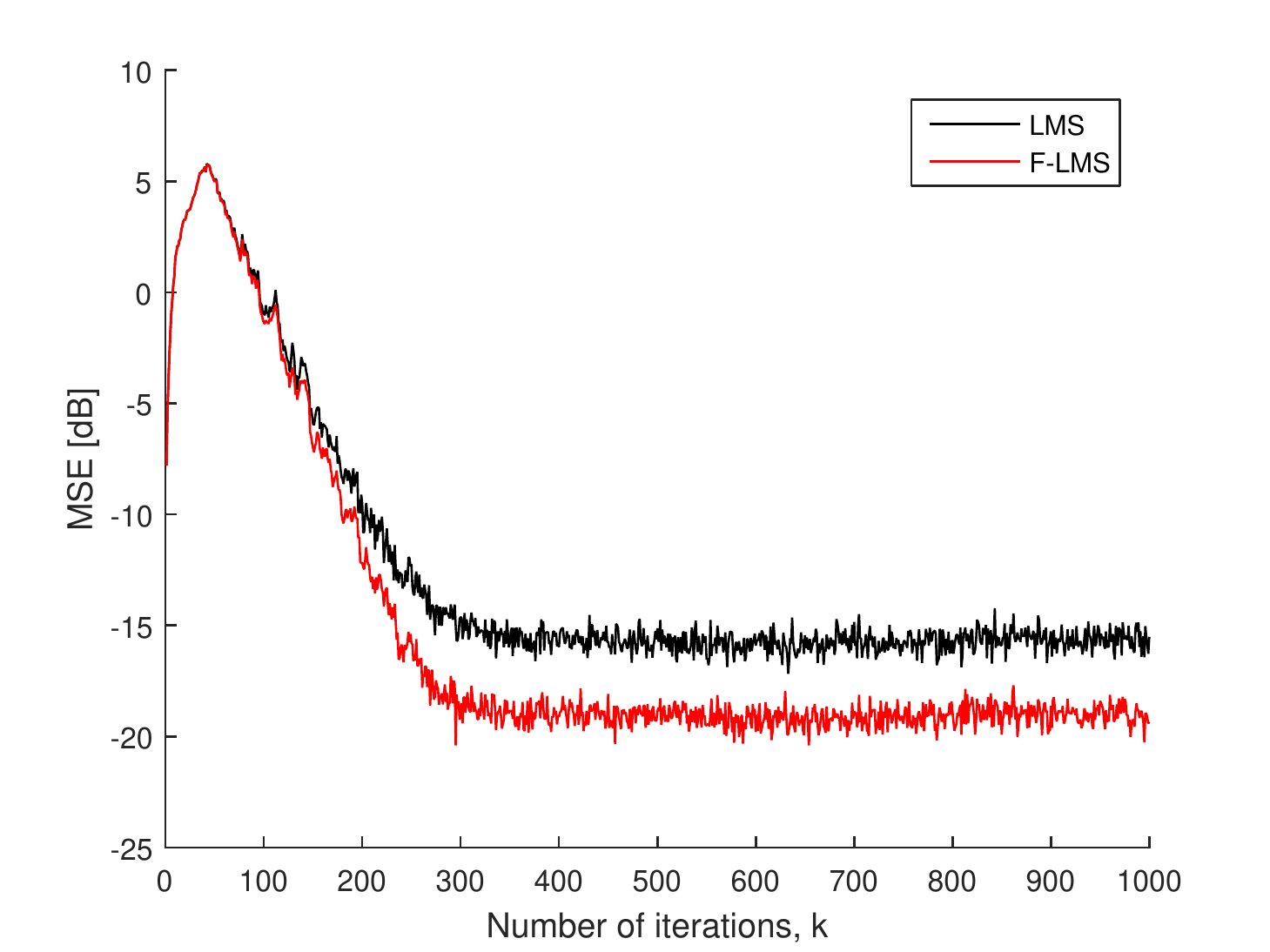}
\label{fig:HP_same_mu-chap7}}
\subfigure[b][]{\includegraphics[width=.48\linewidth,height=7cm]{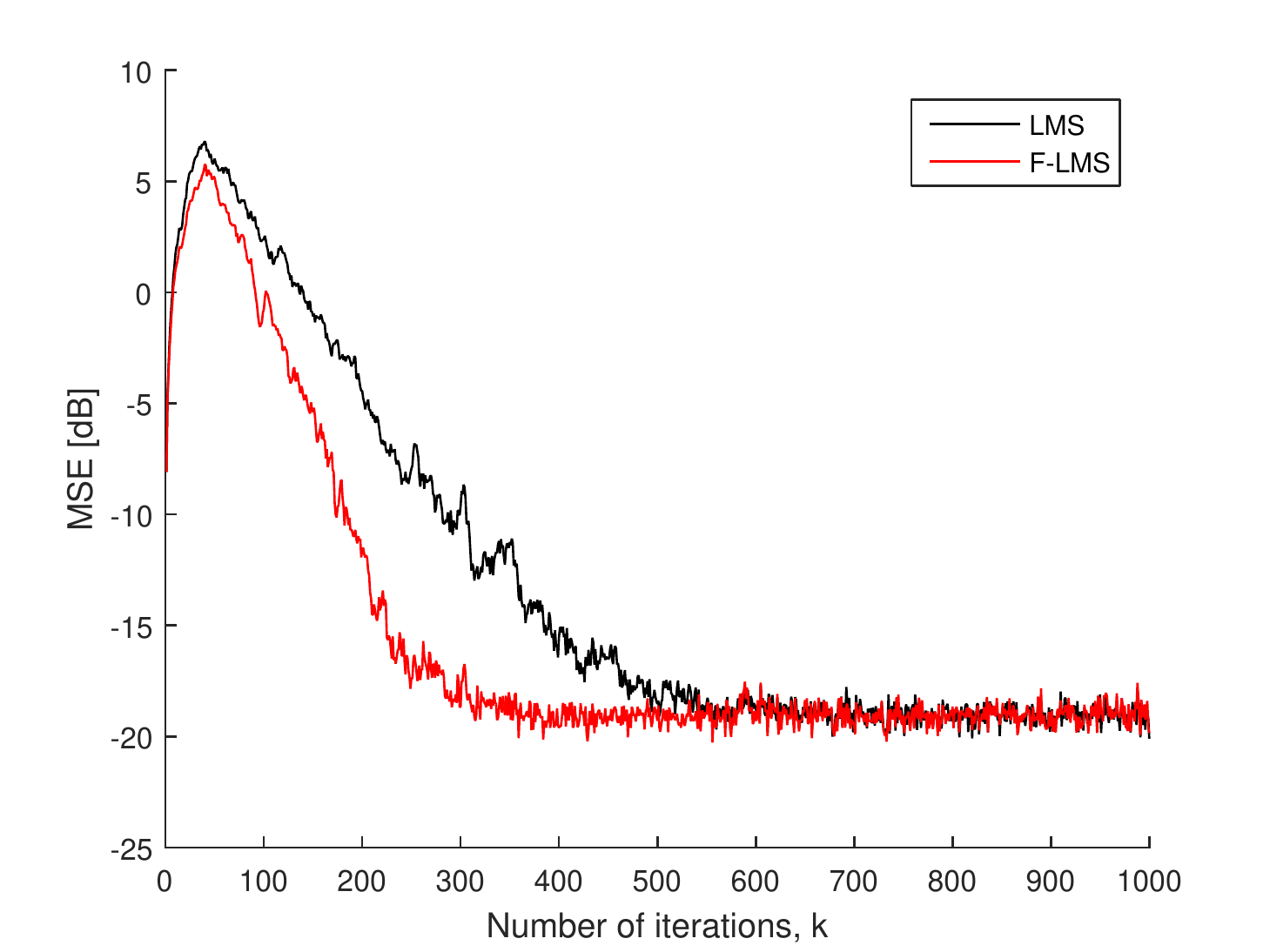}
\label{fig:HP_diff_mu-chap7}}
\caption{MSE learning curves of the LMS and F-LMS algorithms considering $\wbf_{o,h}$: 
(a) both algorithms with the same step size: $\mu = 0.03$; (b) LMS and F-LMS with step sizes equal to 0.01 and 0.03, respectively. \label{fig:HP-chap7}}
\end{figure}

Figure~\ref{fig:LP-chap7} depicts the MSE\abbrev{MSE}{Mean-Squared Error} learning curves of the LMS\abbrev{LMS}{Least-Mean-Square} and the F-LMS\abbrev{F-LMS}{Feature LMS} algorithms considering the lowpass system $\wbf_{o,l}$. In Figure~\ref{fig:LP_same_mu-chap7}, both algorithms use the same step size $\mu = 0.03$ so that they exhibit similar convergence speeds. In this figure, we can observe that the F-LMS\abbrev{F-LMS}{Feature LMS} algorithm achieved a steady-state MSE\abbrev{MSE}{Mean-Squared Error} which is more than $3$~dB lower than the MSE\abbrev{MSE}{Mean-Squared Error} results of the LMS\abbrev{LMS}{Least-Mean-Square} algorithm. In Figure~\ref{fig:LP_diff_mu-chap7}, the steady-state MSE\abbrev{MSE}{Mean-Squared Error} of the algorithms are fixed in order to compare their convergence speeds. 
Thus, we set the step sizes of the LMS\abbrev{LMS}{Least-Mean-Square} and the F-LMS\abbrev{F-LMS}{Feature LMS} algorithms as 0.01 and 0.03, respectively. We can observe, in this figure, that the F-LMS\abbrev{F-LMS}{Feature LMS} algorithm converged much faster than the LMS\abbrev{LMS}{Least-Mean-Square} algorithm.

In Figure~\ref{fig:HP-chap7}, we present results equivalent to the ones presented in Figure~\ref{fig:LP-chap7}, but considering the highpass system $\wbf_{o,h}$. Once again, when the step sizes of both algorithms are the same ($\mu = 0.03$), refer to Figure~\ref{fig:HP_same_mu-chap7}, the F-LMS\abbrev{F-LMS}{Feature LMS} algorithm achieved 
lower steady-state MSE;\abbrev{MSE}{Mean-Squared Error} whereas the F-LMS\abbrev{F-LMS}{Feature LMS} algorithm (with $\mu = 0.03$) converged much faster than the LMS\abbrev{LMS}{Least-Mean-Square} algorithm (with $\mu = 0.01$) when their steady-state MSEs\abbrev{MSE}{Mean-Squared Error} are fixed, as illustrated in Figure~\ref{fig:HP_diff_mu-chap7}. 

\begin{figure}[t!]
\centering
\subfigure[b][]{\includegraphics[width=.48\linewidth,height=7cm]{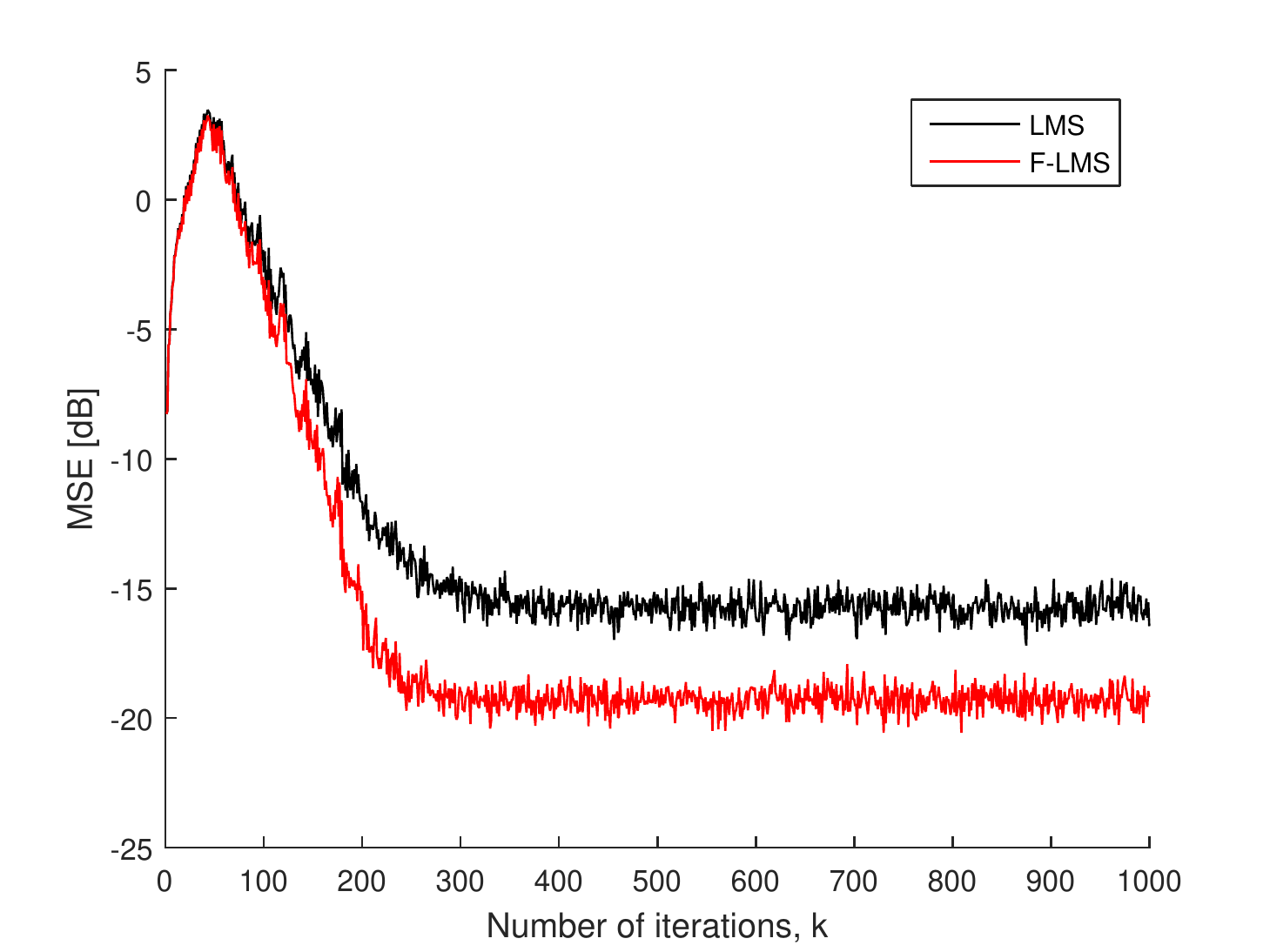}
\label{fig:LP_int-chap7}}
\subfigure[b][]{\includegraphics[width=.48\linewidth,height=7cm]{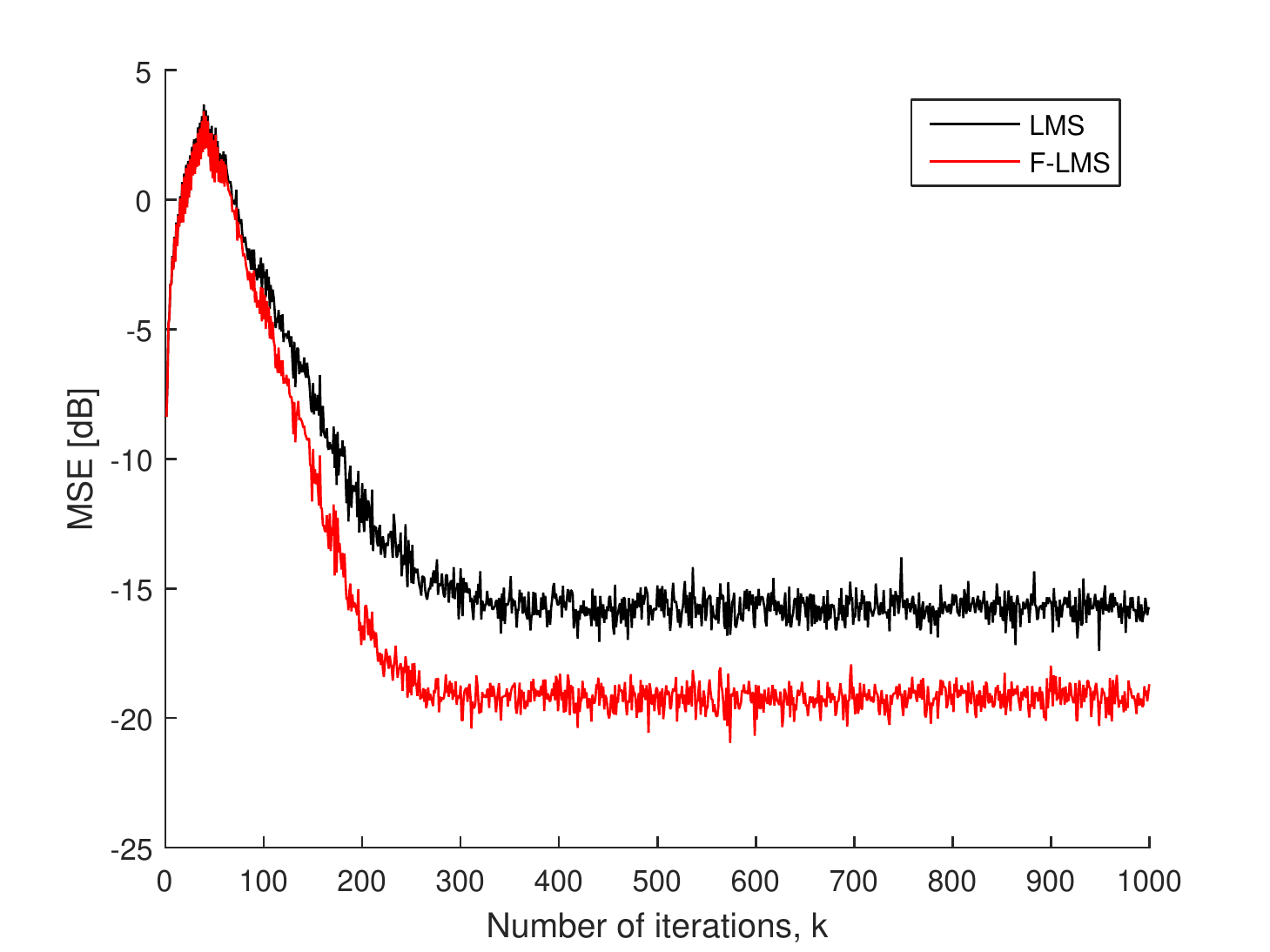}
\label{fig:HP_int-chap7}}
\caption{MSE learning curves of the LMS and F-LMS algorithms, both with step size $\mu = 0.03$, considering 
the unknown systems: (a) $\wbf_{o,l}'$ and (b) $\wbf_{o,h}'$. \label{fig:int-chap7}}
\end{figure}

\begin{figure}[t!]
\centering
\subfigure[b][]{\includegraphics[width=.48\linewidth,height=7cm]{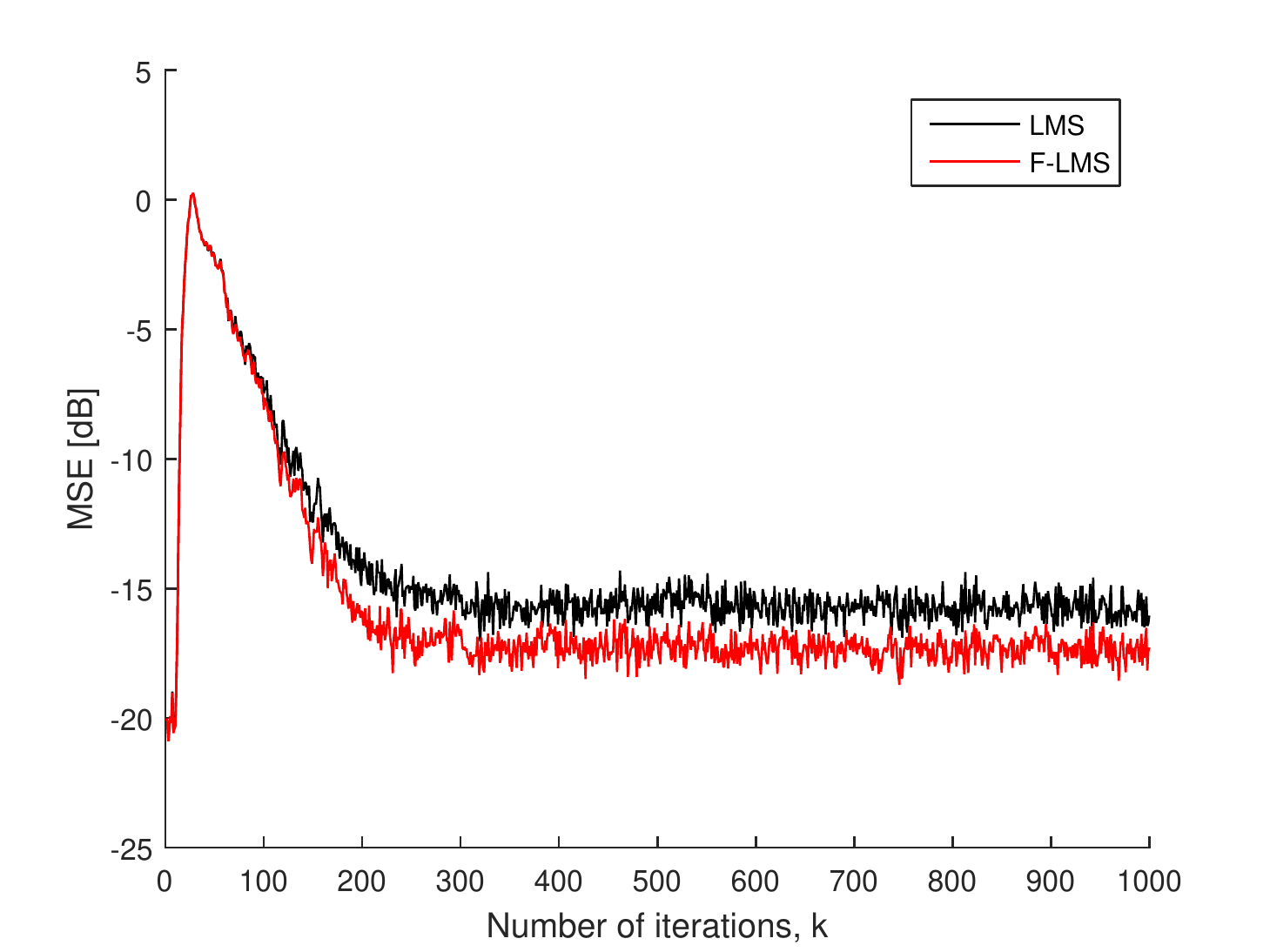}
\label{fig:LP_Block-chap7}}
\subfigure[b][]{\includegraphics[width=.48\linewidth,height=7cm]{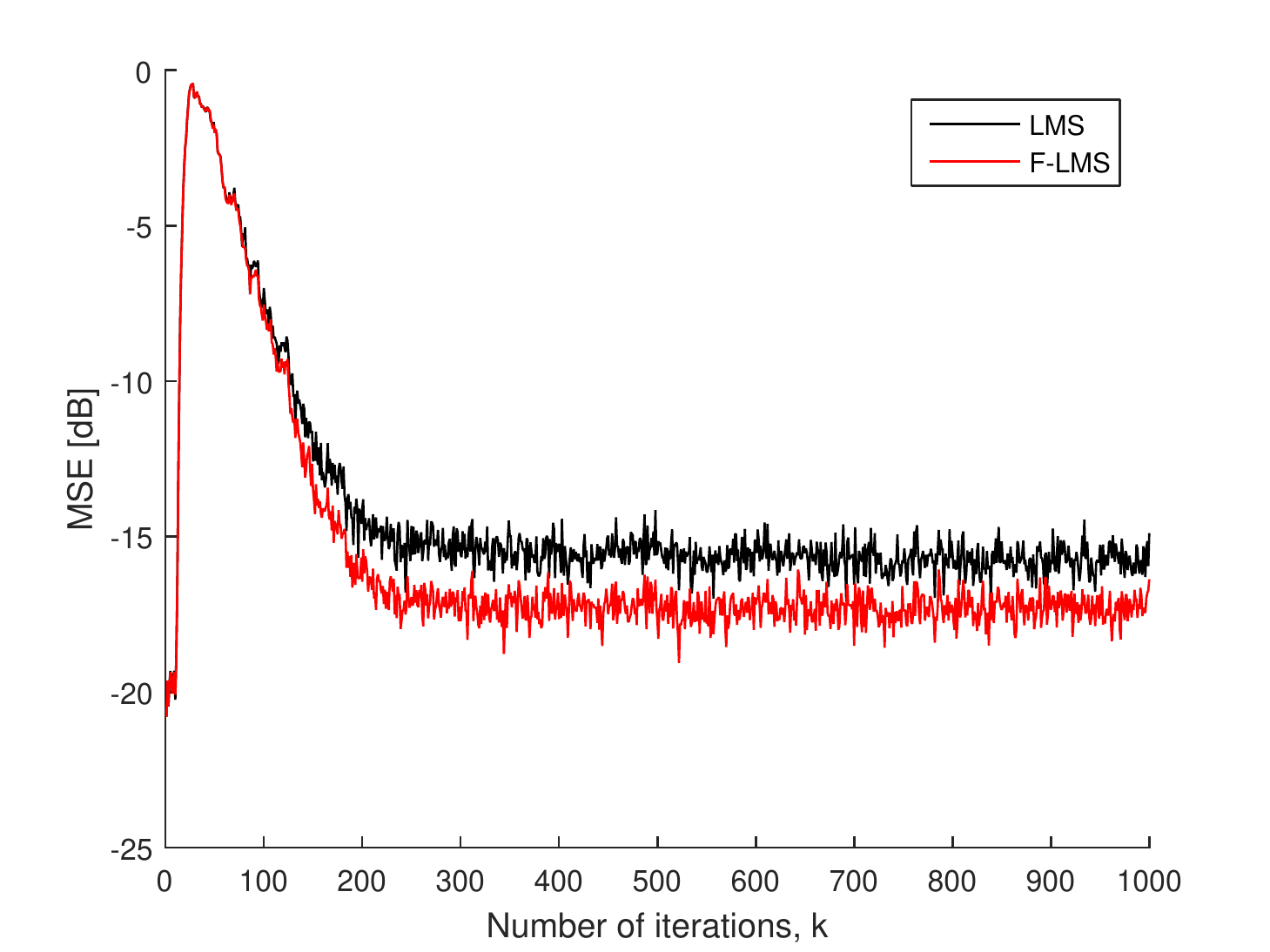}
\label{fig:HP_Block-chap7}}
\caption{MSE learning curves of the LMS and F-LMS algorithms, both with step size $\mu = 0.03$, considering 
the unknown systems: (a) $\wbf_{o,l}''$ and (b) $\wbf_{o,h}''$. \label{fig:Block-chap7}}
\end{figure}

Figures~\ref{fig:LP_int-chap7} and~\ref{fig:HP_int-chap7} depict the MSE\abbrev{MSE}{Mean-Squared Error} learning curves of the LMS\abbrev{LMS}{Least-Mean-Square} and the F-LMS\abbrev{F-LMS}{Feature LMS} algorithms, both using $\mu = 0.03$, considering the interpolated systems $\wbf_{o,l}'$ and $\wbf_{o,h}'$, respectively. Notice, in both figures, that the F-LMS\abbrev{F-LMS}{Feature LMS} algorithm achieved lower steady-state MSE,\abbrev{MSE}{Mean-Squared Error} thus outperforming the LMS\abbrev{LMS}{Least-Mean-Square} algorithm.

Figures~\ref{fig:LP_Block-chap7} and~\ref{fig:HP_Block-chap7} depict the MSE\abbrev{MSE}{Mean-Squared Error} learning curves of the LMS\abbrev{LMS}{Least-Mean-Square} and the F-LMS\abbrev{F-LMS}{Feature LMS} algorithms, both using $\mu = 0.03$, considering the block-sparse systems $\wbf_{o,l}''$ and $\wbf_{o,h}''$, respectively. In both cases, the F-LMS\abbrev{F-LMS}{Feature LMS} algorithm achieved lower steady-state MSE,\abbrev{MSE}{Mean-Squared Error} thus outperforming the LMS\abbrev{LMS}{Least-Mean-Square} algorithm.


\subsection{Scenario 2}

In this scenario, we apply the LMS, the LCF-LMS, the ALCF-LMS, the I-LCF-LMS, and the AI-LCF-LMS algorithms to identify two unknown systems. The first unknown system is the predominantly lowpass system $\wbf_{o,l}$. The second unknown model is a block-sparse model, $\wbf_{o,l}'''$, defined as follows
\begin{align}
w_{o,l_i}'''&=\left\{\begin{array}{ll}0 & {\rm if~}0 \leq i \leq 9,\\
0.04+0.01(i-9) & {\rm if~} 10\leq i \leq 17,\\
0.5 & {\rm if~} 18 \leq i \leq 21,\\
0.13-0.01(i-21) & {\rm if~} 22 \leq i \leq 29,\\
0 & {\rm if~} 30 \leq i \leq 39.\end{array}\right. \label{eq:second_lowpass_for_LCF-LMS-chap7}
\end{align}
In the case of the LCF-LMS and the ALCF-LMS algorithms, the input signal is an autoregressive signal generated by $x(k)=0.99x(k-1)+n(k-1)$. However, we do not have any restrictions on the input signal when utilizing the I-LCF-LMS and the AI-LCF-LMS algorithms. Thus, we use a zero-mean white Gaussian noise with unit variance as the input signal when implementing the I-LCF-LMS and the AI-LCF-LMS algorithms. The step size $\mu$ for the all algorithms is 0.003. Also, we adopt $\epsilon$ equal to 0.02. 

Figures~\ref{fig:ALCF_LMS-chap7} and~\ref{fig:ALCF_LMS_Block-chap7} show the MSE learning curves of the LMS, the LCF-LMS, and the ALCF-LMS algorithms. Furthermore, the MSE learning curves of the LMS, the I-LCF-LMS, and the AI-LCF-LMS algorithms are illustrated in Figures~\ref{fig:Improved_Lowpass-chap7} and~\ref{fig:Improved_Block_Lowpass-chap7}.

\begin{figure}[t!]
\centering
\subfigure[b][]{\includegraphics[width=.48\linewidth,height=7cm]{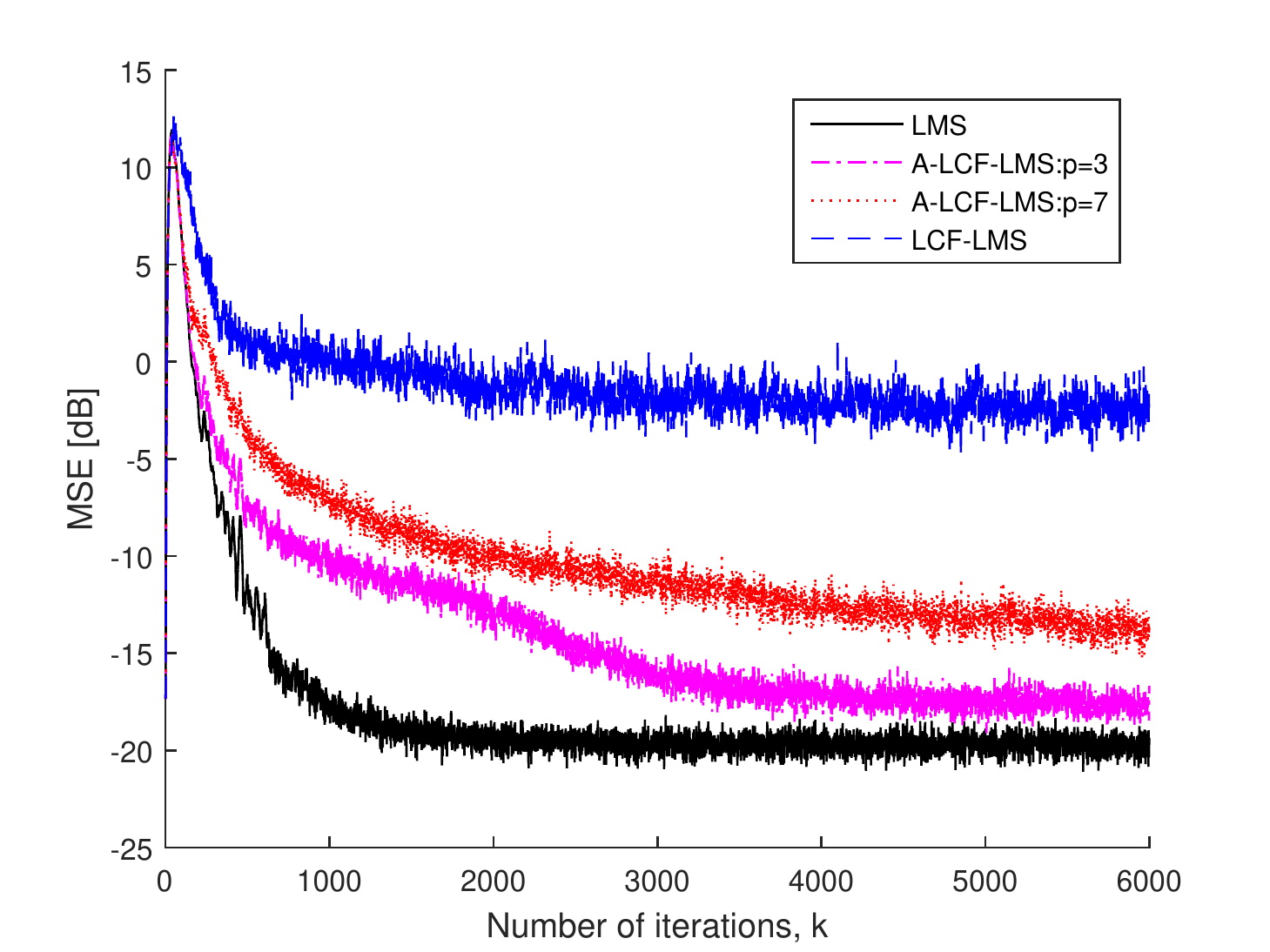}
\label{fig:ALCF_LMS-chap7}}
\subfigure[b][]{\includegraphics[width=.48\linewidth,height=7cm]{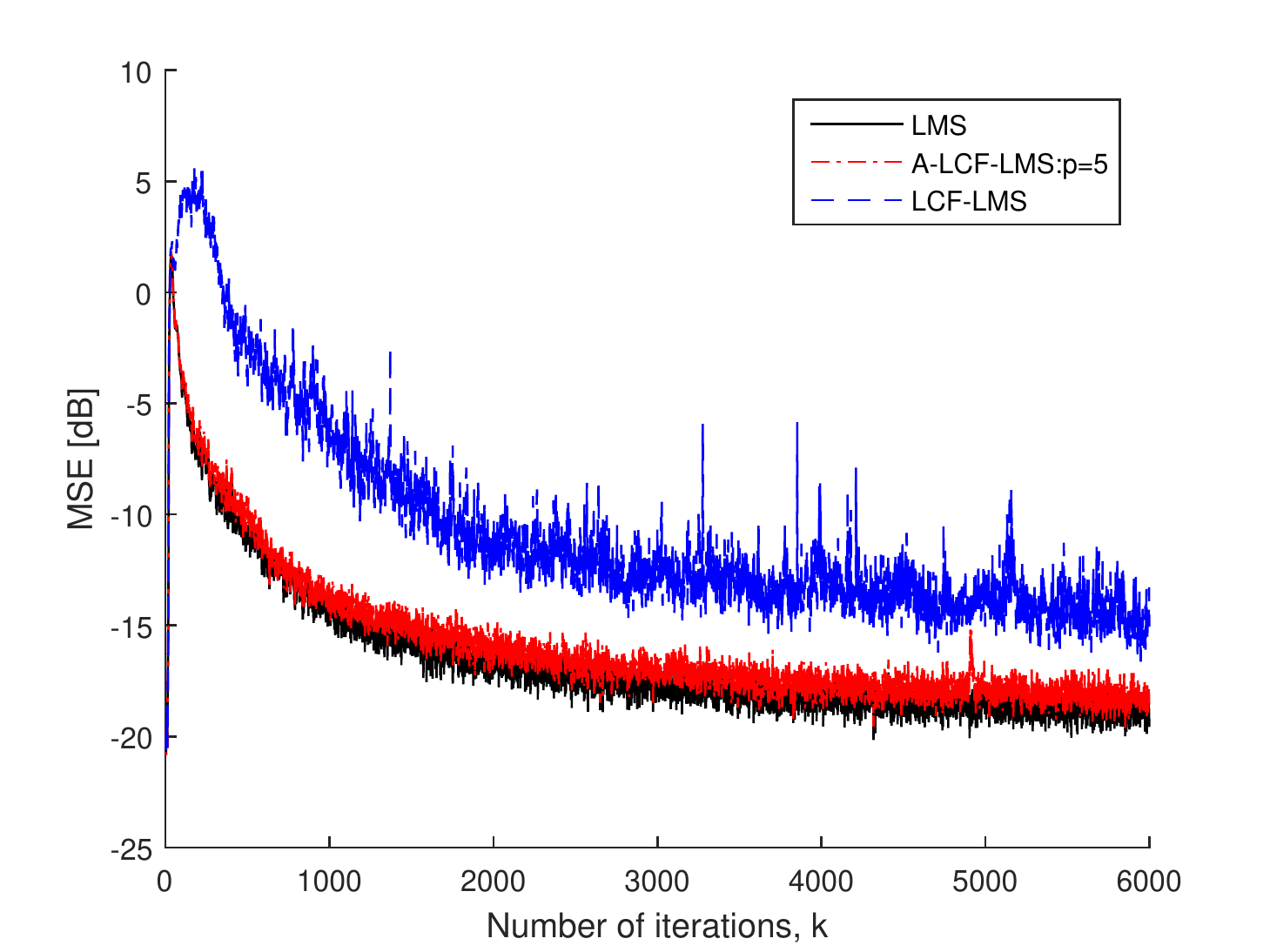}
\label{fig:ALCF_LMS_Block-chap7}}
\caption{MSE learning curves of the LMS, the LCF-LMS, and the ALCF-LMS algorithms considering 
the unknown systems: (a) $\wbf_{o,l}$ and (b) $\wbf_{o,l}'''$. \label{fig:LCF_LMS-chap7}}
\end{figure}

Figure~\ref{fig:ALCF_LMS-chap7} shows the learning curves of the mentioned algorithms when they are applied to identify the predominantly lowpass unknown system $\wbf_{o,l}$. We can observe that the LCF-LMS algorithm, the blue curve, has high MSE but it has the lowest computational complexity. In the steady-state environment, it implements only one multiplication to calculate the error signal. However, the LMS algorithms, the black curve, requires forty multiplication to compute the error signal, and it has the highest computational burden. The ALCF-LMS algorithms have acceptable performances and, using $p=3$ and 7, they need thirteen and six multiplication to calculate the error signal, respectively.

Figure~\ref{fig:ALCF_LMS_Block-chap7} depicts the learning curves of the algorithms, when they are applied to identify the block-sparse lowpass unknown model $\wbf_{o,l}'''$. As can be seen, the LCF-LMS algorithm, the blue curve, has the highest MSE but it executes only three multiplication to compute the error signal. The red curve illustrates the remarkable performance of the ALCF-LMS algorithm. Indeed, its learning curve is extremely close to the learning curve of the LMS algorithm. However, in the steady-state environment, it implements only six multiplication to calculate the error signal.

\begin{figure}[t!]
\centering
\subfigure[b][]{\includegraphics[width=.48\linewidth,height=7cm]{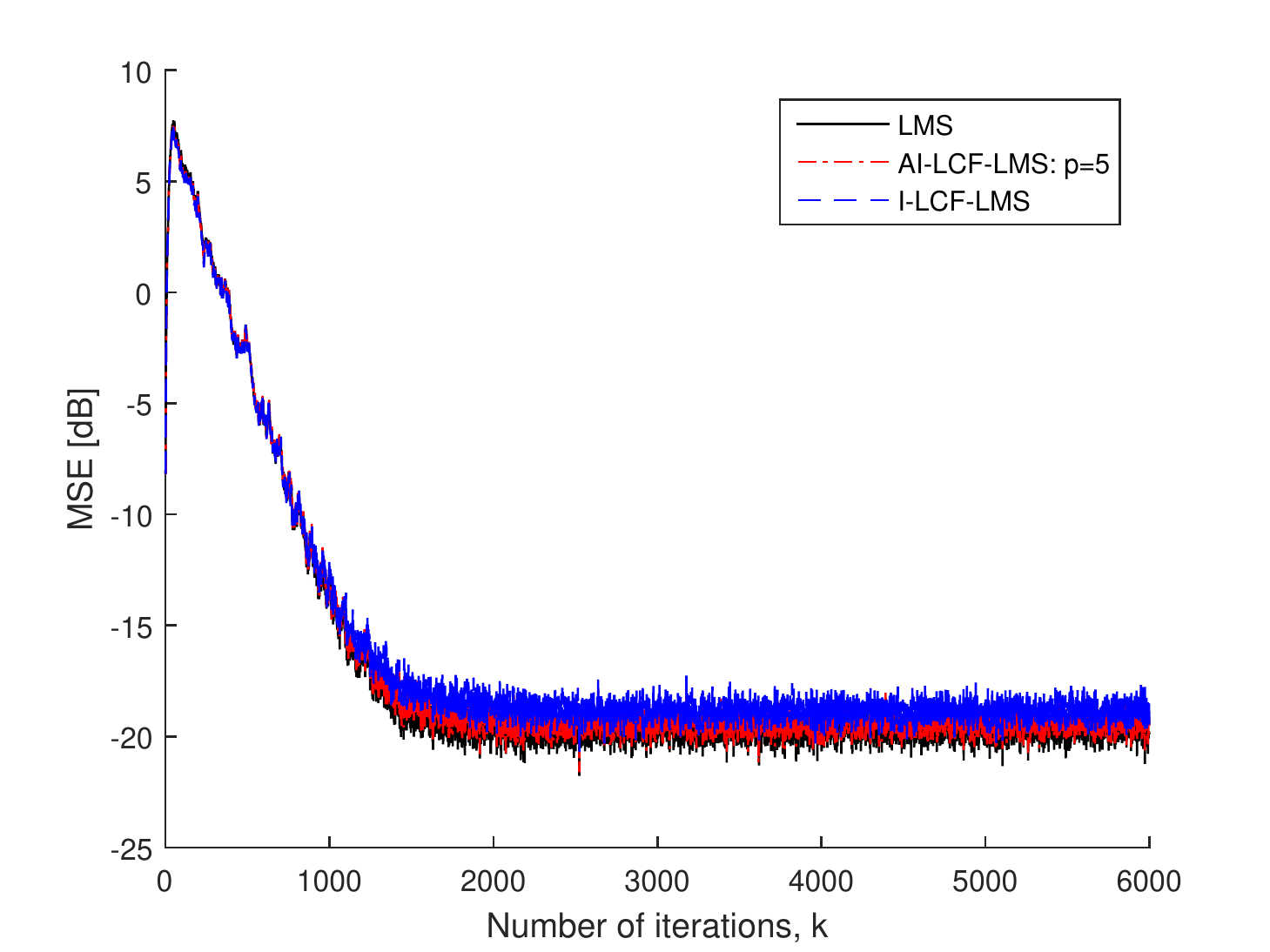}
\label{fig:Improved_Lowpass-chap7}}
\subfigure[b][]{\includegraphics[width=.48\linewidth,height=7cm]{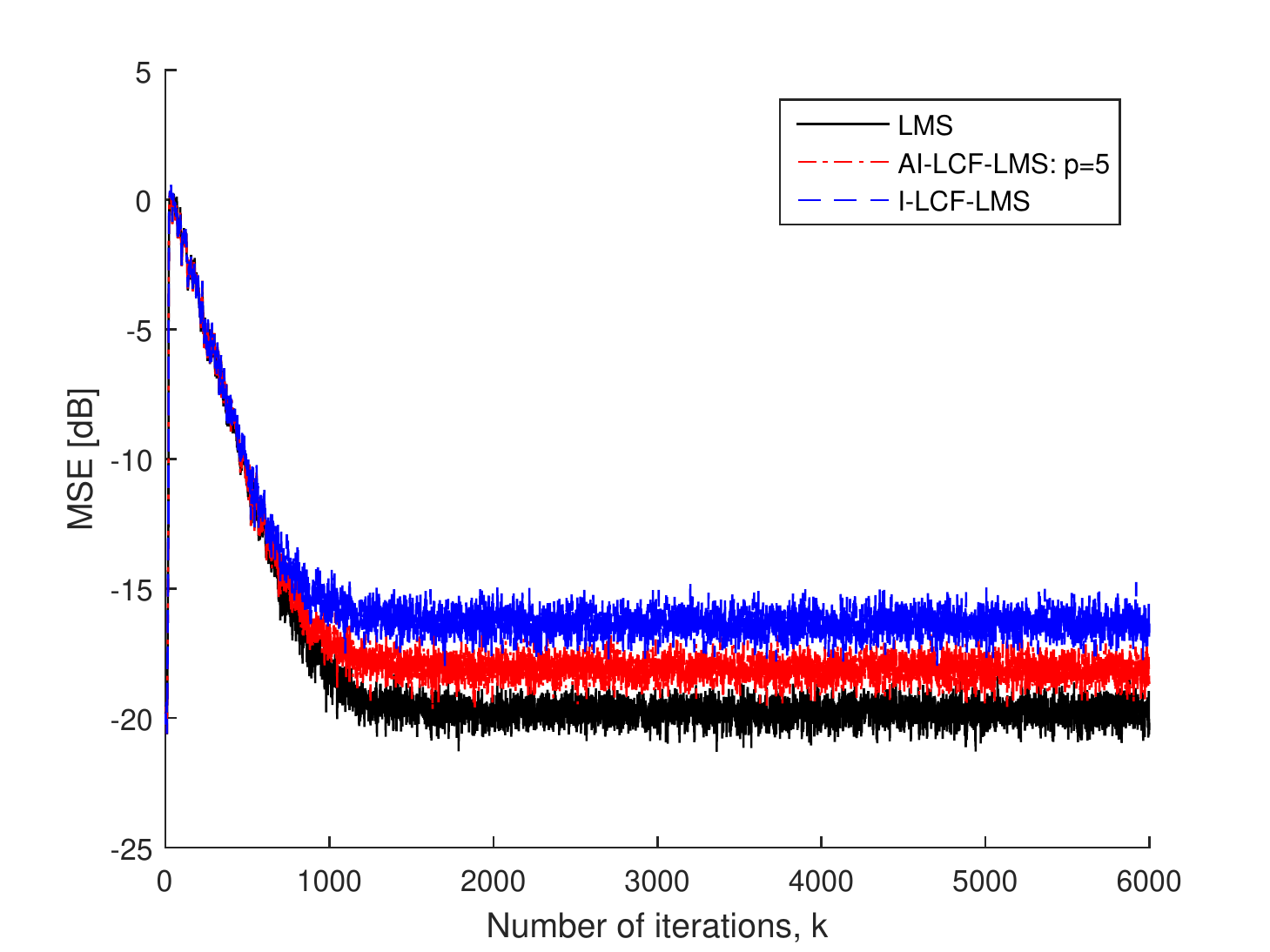}
\label{fig:Improved_Block_Lowpass-chap7}}
\caption{MSE learning curves of the LMS, the I-LCF-LMS, and the AI-LCF-LMS algorithms considering 
the unknown systems: (a) $\wbf_{o,l}$ and (b) $\wbf_{o,l}'''$. \label{fig:Improved_LCF_LMS-chap7}}
\end{figure}

Figure~\ref{fig:Improved_Lowpass-chap7} illustrates the learning curves of the LMS, the I-LCF-LMS, and the AI-LCF-LMS algorithms when they are utilized in the identification of the predominantly lowpass unknown system $\wbf_{o,l}$. The three algorithms have the same convergence rate; however, the LMS algorithm has the best MSE, followed by the AI-LCF-LMS and the I-LCF-LMS algorithms. As  can be seen, the superiority of the MSE of the LMS algorithm to the MSE of the other two algorithms is not remarkable but the LMS algorithm has higher computational load. In the steady-state environment, for the calculation of the error signal, the LMS algorithm implements 40 multiplication, whereas the I-LCF-LMS and the AI-LCF-LMS algorithms execute one and eight multiplication, respectively.

The MSE learning curves of the LMS, the I-LCF-LMS, and the AI-LCF-LMS algorithms, when they are applied to identify the block-sparse unknown system $\wbf_{o,l}'''$, are presented in Figure~\ref{fig:Improved_Block_Lowpass-chap7}. The curves shown in this figure indicate that the LMS algorithm has the best misadjustment, followed by the AI-LCF-LMS and the I-LCF-LMS algorithms. Moreover, we can observe that the three algorithms have similar convergence speed. We must note that the computational complexity of the LMS algorithm is higher than that of the I-LCF-LMS and of the AI-LCF-LMS algorithms. In other words, to compute the error signal in the steady-state environment, the LMS algorithm requires 40 multiplication; however, the I-LCF-LMS and the AI-LCF-LMS algorithms need three and six multiplication, respectively.

As can be seen, in Scenario 1, the learning curves of the F-LMS algorithm are lower than that of the LMS algorithm. However, in Scenario 2, the learning curves of the LCF-LMS, the ALCF-LMS, the I-LCF-LMS, and the AI-LCF-LMS algorithms are higher than that of the LMS algorithm. It is worthwhile to mention that the computational complexity of the F-LMS algorithm is higher than that of the LMS algorithm, whereas the LCF-LMS, the ALCF-LMS, the I-LCF-LMS, and the AI-LCF-LMS algorithms require lower computational resources as compared to the LMS algorithm. Therefore, higher MSE in the performance of the low-complexity F-LMS algorithms is compensated by their lower computational complexity.


\section{Conclusions} \label{sec:conclusions-chap7}

In this chapter, we have proposed a family of algorithms called Feature LMS (F-LMS)\abbrev{F-LMS}{Feature LMS}. The F-LMS\abbrev{F-LMS}{Feature LMS} algorithms are capable of exploiting specific features of the unknown system to be identified in order to accelerate convergence speed and/or reduce steady-state MSE,\abbrev{MSE}{Mean-Squared Error} obtaining a more accurate estimate. The main idea is to apply a sparsity-promoting function to a linear combination of the parameters, in which this linear combination should reveal the sparsity hidden in the parameters, i.e., the linear combination exploits the specific structure/feature in order to generate a sparse vector. Some examples of the F-LMS\abbrev{F-LMS}{Feature LMS} algorithms having low computational complexity and exploiting the lowpass and highpass characteristics of unknown systems were introduced. Simulation results confirmed the superior performance of the F-LMS\abbrev{F-LMS}{Feature LMS} algorithm in comparison with the LMS\abbrev{LMS}{Least-Mean-Square} algorithm. 

Furthermore, we have introduced the low-complexity F-LMS (LCF-LMS) and the alternative LCF-LMS (ALCF-LMS) algorithms in order to exploit hidden sparsity in the parameter with low computational cost. For this purpose, we have defined the feature function. The proposed algorithms have lower computational burden compared to the LMS algorithm; however, they have competitive performance. Also, we have introduced the improved versions of the LCF-LMS and the ALCF-LMS algorithms. Numerical results showed the competitive performance of the AI-LCF-LMS algorithm while requiring less multiplication to compute the error signal.

In future works, we intend to investigate other choices for the sparsity-promoting penalty function and the feature matrix. Also, we want to analyze the stability and MSE\abbrev{MSE}{Mean-Squared Error} of the F-LMS\abbrev{F-LMS}{Feature LMS} and the LCF-LMS algorithms.  
  \chapter{Conclusions, and Future Works}

In this thesis, we have investigated a number of data-selective adaptive filtering algorithms. It is generally accepted that data selection is an effective strategy to reduce the computational resources of the adaptive algorithms. To benefit from data selection in adaptive filtering algorithms, we have utilized the set-membership filtering (SMF) approach. 

In set-membership (SM) adaptive filtering algorithms, the inclusion of {\it a priori} information, such as the noise bound, into the objective function leads to some noticeable advantages. The SM adaptive algorithms evaluate, choose, and process data at each iteration of their learning process. These algorithms have the potential to outperform the conventional adaptive filtering algorithms. Indeed, they retain the advantages of their traditional counterparts; however, they are more accurate, more robust against noise, and have lower computational load.

Moreover, we incorporate some sparsity-aware techniques into the SM adaptive algorithms. Thus, we introduced some sparsity-aware set-membership adaptive filtering algorithms. In order to exploit the sparsity in system models, we utilized the $l_0$ norm approximation, the discard function, and the feature matrices. The $l_0$ norm approximation and the discard function exploit the sparsity in coefficients close to zero; however, the feature matrices exploit the sparsity in linear combination of the parameters.

\section{Contributions} \label{sec:contributions}

The thesis started by reviewing the classical adaptive filtering algorithms. Also, we have introduced the SM normalized least-mean-square (SM-NLMS) and the SM affine projection (SM-AP) algorithms briefly. Then we have analyzed the robustness (in the sense of $l_2$ stability) of the SM-NLMS and the SM-AP algorithms. One of the major drawbacks of adopting the conventional algorithms is that one cannot guarantee the convergence of the algorithm independent of the choice of the parameters. However, when the additional noise is bounded, we have proved that the SM algorithms never diverge.

Moreover, the SMF approach has been generalized to trinion and quaternion numbers. Whenever the problem at hand suits both the quaternion and trinion solutions, the trinion algorithms clearly have an advantage over the quaternion ones in terms of computational burden. Furthermore, we have derived a new set-membership partial-update affine projection algorithm. This algorithm can improve the convergence rate significantly, particularly in a nonstationary environment.

In addition, some data-selective adaptive filtering algorithms have been proposed in order to exploit sparsity in systems with low computational cost. The key idea is to apply the discard function and the $l_0$ norm approximation. In particular, the use of discard function can effectively decrease the computational complexity. Finally, we have derived some feature least-mean-square (F-LMS) algorithms to exploit hidden sparsity in models when adjacent coefficients have a strong relation. To this end, the feature matrices and the feature function play fundamental roles.

\section{Future Works} \label{sec:future}

In this section, we list our future works. Indeed, research into studying and analyzing the F-LMS and the low-complexity (LCF-LMS) algorithms is already in progress. We are investigating some mathematical properties, such as the stability and MSE, of the F-LMS and the LCF-LMS algorithms. Also, we are currently in the process of investigating other choices for the sparsity-promoting penalty function and the feature matrix. 

A possible topic for research is to employ distinct feature matrices in an online basis aiming at verifying the best one for a given iteration. It is also possible to derive a multitude for feature matrices inspired by previous knowledge of the spectral content of the unknown system model.

Another future work will concentrate on proposing some set-membership quaternion-valued adaptive filtering algorithms to exploit sparsity in system models. Also, further works need to be performed in order to analyze the performance of the proposed trinion- and quaternion-valued and partial-update adaptive algorithms.

  \backmatter
  \bibliographystyle{coppe-unsrt}
  \bibliography{thesis}

\end{document}